\documentclass[12pt, a4paper,twoside,openright]{book} 

\usepackage{bm} 
\usepackage{xcolor}

\usepackage{amssymb, amsthm, paralist, amsmath, mathrsfs} 
\usepackage{braket}
\usepackage{tikz} 
\usetikzlibrary{calc}
\usepackage{enumerate} 
\usepackage{enumitem} 

\usepackage{graphicx} 
\usepackage{hyperref} 

\usepackage{url}

\usepackage{setspace} 
\usepackage{msor_thesis} 
\usepackage{fancyhdr}
\usepackage{lipsum}
\usetikzlibrary{decorations.pathreplacing,backgrounds}
\usepackage{circuitikz}
\usepackage{graphicx}
\graphicspath{ {./images/} }

\usepackage{libertine}
\usepackage[libertine]{newtxmath}

\usepackage{empheq}
\usepackage[utf8]{inputenc}
\usepackage{mleftright}
\usepackage{booktabs}
\usepackage{caption}

\newcommand*\widefbox[1]{\fbox{\hspace{2em}#1\hspace{2em}}}

\makeatletter
\renewcommand{\@chapapp}{}
\newenvironment{chapquote}[2][2em]
{\setlength{\@tempdima}{#1}%
	\def\chapquote@author{#2}%
	\parshape 1 \@tempdima \dimexpr\textwidth-2\@tempdima\relax%
	\itshape}
{\par\normalfont\hfill--\ \chapquote@author\hspace*{\@tempdima}\par\bigskip}
\makeatother

\definecolor{purple}{rgb}{1,0,1}
\definecolor{lime}{HTML}{A6CE39} 

\newtheorem{theorem}{Theorem}[chapter]

\newtheorem{proposition}{Proposition}[chapter]

\begin{document}
\frontmatter 

\author{Del Rajan}
\title{Quantum Entanglement in Time}


\abstract{This thesis is in the field of quantum information science, which is an area that reconceptualizes quantum physics in terms of information.  Central to this area is the quantum effect of entanglement in space.  It is an interdependence among two or more spatially separated quantum systems that would be impossible to replicate by classical systems.  Alternatively, an entanglement in space can also be viewed as a resource in quantum information in that it allows the ability to perform information tasks that would be impossible or very difficult to do with only classical information.  Two such astonishing applications are quantum communications which can be harnessed for teleportation, and quantum computers which can drastically outperform the best classical supercomputers. 
	
In this thesis our focus is on the theoretical aspect of the field, and we provide one of the first expositions on an analogous quantum effect known as entanglement in time.  It can be viewed as an interdependence of quantum systems across time, which is stronger than could ever exist between classical systems.  We explore this temporal effect within the study of quantum information and its foundations as well as through relativistic quantum information.

An original contribution of this thesis is the design of one of the first quantum information applications of entanglement in time, namely a quantum blockchain.  We describe how the entanglement in time provides the quantum advantage over a classical blockchain.  Furthermore, the information encoding procedure of this quantum blockchain can be interpreted as non-classically influencing the past, and hence the system can be viewed as a `quantum time machine.'}


\ack{Somtimes a few words carry more weight: There have been an eclectic set of mentors and books in my life that were pivotal to my positive development.  However there is truly one mentor who I owe so much to, and who has impacted my life the most in a positive way.  And that is my supervisor, Professor Matt Visser.  Matt gave me an opportunity when I needed it the most in my life; I am not really sure where I would have been without it.  He taught me so much about theoretical physics, supervised me with care, and his work especially on wormholes inspired me.  I am extraordinarily grateful to be trained under this great scientist.  
	
Thank you Matt.}


\phd

\maketitle

\thispagestyle{empty}
\null\vspace{\stretch {1}}
\begin{flushright}
	Dedicated to Albert Einstein\footnote{It was a temporal narrative (the twin paradox) from his theories of relativity that shocked a less than ordinary boy to first open the door to the magical world of theoretical physics.}
\end{flushright}
\vspace{\stretch{2}}\null

\tableofcontents

\mainmatter 

\begin{figure}
	\includegraphics[width=\textwidth]{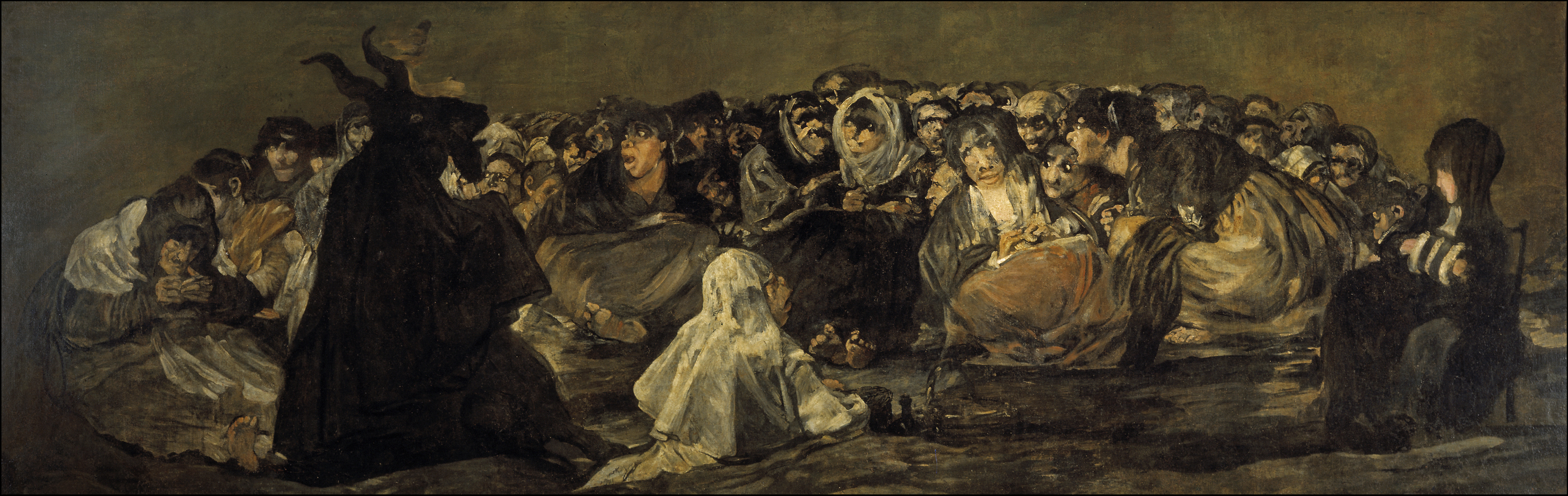}
\caption*{One of the fourteen untitled pieces from  Francisco Goya's \textit{Black Paintings} series.}
\caption*{\footnotesize{\textit{``Quantum mechanics: Real Black Magic Calculus"} - Albert Einstein}}
\end{figure}


\chapter{Introduction}

\begin{chapquote}{Niels Bohr, co-inventor of quantum theory}
	``Anyone who is not shocked by quantum theory has not understood it.''
\end{chapquote}

\textsf{THIS THESIS} explores the most shocking temporal effects in quantum physics.  These recently discovered phenomena violently overthrow the classical picture of the world \cite{jennings2016no}.  Each of these effects has been described as an `entanglement in time' and yet arise from different contexts within quantum physics.  This thesis collects these results to provide one of the first systematic expositions on entanglement in time, which also entails a comparison with the extensively researched entanglement in space. Furthermore, this project is carried out using the theoretical framework of quantum information science \cite{nielsen2002quantum}.

Quantum information science reconceptualizes quantum physics in terms of information.  It is the theoretical and experimental study of quantum information and its applications.  It can be regarded as a fundamental subject in that it distils questions on the nature of quantum physics to distinctions between quantum information and classical information.  One very powerful advantage of quantum information is that it can contain a resource known as entanglement in space.  As a result, quantum information has the ability to perform information tasks that would be impossible or very difficult to do with only classical information.  One prominent example of such tasks is quantum communication which can be used to teleport quantum information.  Another remarkable example are quantum computers which can be shown to efficiently solve problems that would be infeasible to perform on any classical computer that could ever be built.  Both of these applications of entanglement in space are major research programs, and hence emphasizes the crucial role of this resource in quantum information.

Apart from its importance, entanglement in space is a perplexing phenomenon when expressed in terms of the physical systems that instantiate the quantum information.  It can be described as an interdependence among two or more spatially separated quantum information systems in which any one system can instantaneously affect the other systems that can in principle be arbitrarily far.  

Entanglement in time, of which we bring to light in this thesis, corresponds to an analogous effect.  It can be viewed as an interdependence of quantum information systems across time, which is stronger than could ever exist between classical information systems; it can even arise for the case of a single system (across multiple times).  Moreover, the interpretations associated with each manifestation of the effect are far more bizarre; for example, a newly created photon can \textit{affect the physical description of a photon in the past} that has long since been destroyed; in another scenario, a quantum detector that is switched on and off at say quarter to 12:00 can form a non-classical interdependence with another detector at the same spatial location in the future, but \textit{only if the future detector waits} to be switched on and off at precisely quarter past 12:00. 

To gain a deeper understanding of these effects,  a natural question that arises is if there is a common trait that marks these various entanglement phenomena as shocking? This is a challenging question given that these effects are theoretically expressed using different mathematical areas; we can conceive of a large number of factors that contribute to its radical departure from classical properties.  This project provides an insight towards an answer, and captures this in the form of the following overarching theme of this thesis:\textit{ The interdependence in any entanglement in space is shocking due to the absense of a time interval involved.  The interdependence in any entanglement in time is shocking due to the existence of a time interval involved.}  To elaborate on this observation, in an entanglement in space the ability of a system to instantaneously affect a distant system signifies a lack of a time interval.  Introducing a time interval in this scenario will only make the effect clash less harshly with our classical intuition (as it allows for an explanation involving hidden causal signals of some form or the other); such an insight regarding this spatial case was first expressed in the concluding remarks in \cite{schrodinger1935discussion}.  However in an entanglement in time, it is precisely the introduction of a time interval that makes the effect completely unpalatable to the mind of a classical physicist.  This is non-trivial as it was already noted that a time interval allows for a classical dependence among systems across time in the form of a causal relationship. We aim to provide a compelling case for this theme.

In quantum information science, both quantum communications and quantum computers are established applications of entanglement in space.  A central aim of the field is the creation of new quantum information applications.  In this thesis, we make an original contribution by designing one of the first novel applications of entanglement in time, namely a quantum blockchain.  In our mathematical design, we show that the entanglement in time (as opposed to an entanglement in space) provides the quantum advantage over a classical blockchain.  Furthermore, the information encoding procedure of this quantum blockchain can be interpreted as non-classically influencing the past, and hence the system can be viewed as a `quantum time machine.'  This advancement forms one piece of the various original works and insights presented in this thesis.

Rather than provide a chronological presentation, our thesis will place the entanglements in quantum information science as the conceptual core and coherently organize the diverse topics as backgrounds or extensions of this core.  We believe this approach achieves the most clarity.  Hence the structure of this thesis is as follows:  In Chapter \ref{chap: classical}, we provide mathematical descriptions of classical information and highlight three of its applications.  These fall under the respective sections of classical communications, classical computers and classical blockchains.  In Chapter \ref{chap: QInfo}, we introduce quantum information and contrast this with the properties of classical information. As a result, it contains a description of the mathematical tools of quantum information science, and this will closely follow the material in \cite{nielsen2002quantum} along with recent developments; from the perspective of a theoretical physicist, this subject can be viewed as an information-theoretic reformulation of non-relativistic quantum mechanics.  Chapter \ref{chap: QEnt} is the core chapter of this thesis which starts by describing entanglement in space using the tools obtained in the previous chapter.  Furthermore, it introduces quantum communications and quantum computers which are the applications of entanglement in space.  We proceed to describe the central topic of entanglement in time, and its various non-classical properties.  We conclude this chapter by presenting a mathematical design of the quantum blockchain which is an application of entanglement in time.  All three quantum applications will be contrasted with the classical case.  We proceed to Chapter \ref{chap: QFound}, where the notions of both entanglement in space and entanglement in time are extended to the subject of quantum foundations.  These are respectively described in sections titled nonlocality in space and nonlocality in time.  The emphasis will be placed on the aspects of quantum foundations which shares an interface with quantum information science.  In Chapter \ref{chap: RQI}, we see entanglement in space and entanglement in time manifesting itself in the relativistic regime.  These effects are respectively termed spacelike entanglement and timelike entanglement in the subject of relativistic quantum information.  One can think of this new subject as placing quantum information science within the broader framework of quantum field theory (in flat and curved spacetimes).  Finally in Chapter \ref{chap: Conclusion}, we provide a conclusion which includes a discussion on future research projects concerning entanglement in time.

Though entanglement in time may resemble some of the concepts of closed timelike curves, we exclude a detailed study of the latter for two reasons.  The first reason is that closed timelike curves have already been extensively reviewed within relativity \cite{visser1996lorentzian} as well as in quantum theory \cite{ralph2012relativistic}; whereas this thesis is one of the first to systematically compile the diverse literature on entanglement in time.  The second reason is that unlike closed timelike curves, the effects of entanglement in time have been experimentally verified for a number of cases,  thereby warranting itself as a distinct phenomenon.  

Although we refer to various literature in experimental physics, information theory, and computer science, this thesis falls under the theoretical physics aspect of quantum information science.  Hence, the focus is solely on the mathematics with an emphasis on the physical theories.


\chapter{Classical Information}\label{chap: classical}

\begin{chapquote}{Microsoft's founding vision statement}
	``A [classical] computer on every desk and in every home.''
\end{chapquote}

\textsf{THE MODERN CONCEPT} of classical information will be articulated against the backdrop of three technological applications.  The associated mathematical models in each of the three cases provide an abstraction for how this information is represented and transformed.  For the sole study of information, this abstraction provides the necessary framework to ignore the details of the various physical systems used which store that information.   

Of particular importance to this thesis is the notion of systems exhibiting \textit{interdependence} with each other.  In the realm of classical information, we will find this to be naturally captured through the constructs of probability theory.

\section{Review of Probability Theory}

In this thesis, probability theory will prove to be essential in two primary ways:
\begin{enumerate}[noitemsep, topsep=0pt, label=\roman*)]
	\item To assist in mathematically defining classical information.
	\item To understand the probabilities derived from quantum information.
\end{enumerate}

\subsection{Single random variable}

The fundamental object of probability theory is the random variable, which we denote as $X$.  The random variable can take one of a number of values, $x$, with respective probabilities $p(X=x)$; we limit ourselves to the case where the values form a finite set; we also use the convention that $p(x)$ can represent $p(X=x)$.

The expectation value of $X$ is defined as
\begin{equation}\label{expectation}
\mathbb{E}(X) \equiv  \sum_{x}p(x) \, x.
 \end{equation}
It is a particular type of mean for all the values the random variable can take.  Furthermore if $a$ and $b$ are constants, then it can be shown that $\mathbb{E}(aX + b) = a \mathbb{E}(X) + b$.

The variance and standard deviation are respectively defined as
\begin{equation}
\text{Var}(X) \equiv \mathbb{E}[(X-\mathbb{E}(X))^{2}] = \mathbb{E}(X^{2}) - \mathbb{E}(X)^{2},
\end{equation}
\begin{equation}\label{standarddeviation}
\Delta(X) \equiv \sqrt{\text{Var}(X)}.
\end{equation}
Both are statistical measures of the `spread' of values about the average. One advantage of using the standard deviation, as opposed to the variance, is that it has the same units as the expectation value.

\subsection{Multiple random variables}

When considering the case of more than one random variable, several new constructions can be introduced.  Suppose $X$ and $Y$ are random variables.  Then the probability that $X =x$ \textit{and} $Y=y$ is known as the joint probability,
\begin{equation}
	p(X=x, Y=y).    
\end{equation}
An equivalent notation is simply $p(x,y)$.  

The conditional probability that $X=x$ given that $Y=y$ is defined as
\begin{equation}\label{conditional}
	p(X=x \,|\, Y = y) \equiv \frac{p(X=x, Y=y)}{p(Y=y)}.
\end{equation}
Bayes' rule allows one to `invert' conditional probabilities
\begin{equation}
	p(X=x \,|\, Y = y) = p(Y=y \,|\, X = x) \,\frac{p(X=x)}{p(Y=y)}.  
\end{equation}

Given two random variables, the law of total probability provides an alternative way to calculate probabilities of one of the variables,
\begin{equation}
	p(Y=y) = \sum_{x} p(Y=y \,|\, X = x) \, p(X=x).
\end{equation}
The sum is over all values that the other random variable can take.  

The expectation value for two random variables is the rather simple result
\begin{equation}
	\mathbb{E}(X + Y) = \mathbb{E}(X) + \mathbb{E}(Y).
\end{equation}

\subsection{Independent random variables}
The pertinent question, from the view of this thesis, is how to describe random variables where the realization of one has no effect on the other?  This can be enunciated by the following mathematical definition:  Random variables $X$ and $Y$ are said to be \textit{independent} if
\begin{empheq}[box=\widefbox]{align}\label{classicalindependence}
	p(X=x, Y=y) = p(X=x) \, p(Y=y), \quad \forall x, y.
\end{empheq}
When viewed through the concepts from the previous subsection, it is not difficult to see that if $X$ and $Y$ are independent random variables, then:
\begin{align}
	p(Y=y \,|\, X=x)=& \, p(Y=y) \quad \forall x, y \\
\mathbb{E}(X Y)=&  \, \mathbb{E}(X) \, \mathbb{E}(Y), \\
\text{Var}(X + Y)=& \, \text{Var}(X) + \text{Var}(Y).
\end{align}
A further consequence from considering independent random variables is the following theorem.
\begin{theorem}\label{lawlarge}
	\textbf{(Law of large numbers)}  Suppose $X_{1}$, $X_{2}$, $\dots$ are independent random variables that all have identical probability distributions as $X$, where $\lvert \mathbb{E}(X)\rvert < \infty$ and $\lvert \mathbb{E}(X^{2})\rvert < \infty$.  Then for any $\epsilon > 0$, and where $S_{n} \equiv \sum_{i=1}^{n} X_{i}/n$, we have that $p(\lvert S_{n} - \mathbb{E}(X) \rvert > \epsilon) \rightarrow 0$ as $n \rightarrow \infty$.  (See e.g. \cite{nielsen2002quantum}.)\hfill$\Box$
\end{theorem}
The utility of this result arises in many applications such as in games involving chance.  Hence, it is of no surprise that origins of probability theory can be traced to systemic study of dice games \cite{applebaum1996probability}. 

\subsection{Application:  Monty Hall game}

We devote this subsection to applying some of the reviewed concepts to the well known Monty Hall game \cite{rodriguez2018probability, rosenthal2008monty, gill2010monty, lucas2009monty}.  This game is perhaps the most bizarre application of classical probabilities.

\textbf{a) Classic Monty Hall game:} A character named Monty hosts a game show.  There are three doors respectively labelled $\{1,2,3\}$.  There is a car prize behind one door, and goats behind the remaining two.  We let a random variable $A$ represent the prize door which can take value $i$ from the set of door labels.  We assume in the game that when a random choice needs to be made, all options are chosen with equal probability.  Note that this implies that the choice for prize door has probabilities $p(A = i) = 1/3$ for each value $i$.

The contestant on the show, who doesn't know which door the prize is behind, is given a choice to pick a door;  we represent the chosen door using a random variable denoted $B$ which can take value $j$ from the set of door labels.  Provided this is a random choice, we have $p(B=j \,|\, A=i) = 1/3$ for all values $i,j$.   

Next, Monty who knows where the prize is, has to open a goat door, which we represent using random variable $C$; in a similar manner, the variable takes value $k$ from the set of door labels.  But unlike the previous choices, Monty's decision is constrained through the game rule that he is not allowed to open the door chosen by the contestant. From this, we derive the following conditional probabilities:
\begin{equation}\label{Monty}
p(C=k \, | \, B=j, A=i) = 
\begin{cases}
\frac{1}{2},& \text{if } i =j \neq k\\
1,              & \text{if } i \neq j \neq k \\
0, & \text{otherwise}
\end{cases}
\end{equation}
Once a goat door is opened, Monty offers the contestant the option to stick with the original choice, or alternatively switch to the other unopened door.  By sticking, the contestant's probability of opening the prize door is $1/3$.  Counter-intuitively, by switching doors, the probability of winning increases to $2/3$.

This can be seen by proceeding to compute the non-zero joint probabilities
\begin{equation}
p(A=i, B=j, C=k)= p(C=k \, | \, B=j, A=i) \, p(B=j \,|\, A=i) \, p(A=i).
\end{equation}   
Then we sum the joint probabilities corresponding to the combination of door labels where the contestant would win by switching.  This leads to the desired result
\begin{equation}
p(\text{win if switch}) = \sum_{i \neq j \neq k} p(A=i, B=j, C=k) = \frac{2}{3}.
\end{equation}

\textbf{b) Ignorant Monty Hall game:} Let us consider the case where Monty does not know what lies behind any of the doors.  Nonetheless, we still have $p(A=i) = 1/3$, and also $p(B=j \,|\, A=i)= 1/3$ for all values $i,j$.  The only constraint as in the Classic game is that Monty cannot open the door chosen by the contestant.  This means that (\ref{Monty}) is modified to   
\begin{equation}\label{IgnorantMontyHallgame}
p(C=k \, | \, B=j, A=i) = 
\begin{cases}
0,& \text{if } j=k \\
\frac{1}{2}, & \text{otherwise}
\end{cases}
\end{equation}
Unlike the previous case, there is a probability in this scenario that Monty opens the prize door by accident; this be seen as the set of cases where $i=k$:
\begin{equation}\label{prize}
p(\text{opens prize door}) = \sum_{i = k \neq j} p(A=i, B=j, C=k) = \frac{1}{3}.
\end{equation}
By respecting that probabilities sum to unity, we derive from (\ref{prize}) that the probability Monty opens a goat door is $2/3$.  The joint probability that Monty opens a goat door and the contestant wins by switching doors can be computed to be $1/3$.  Substituting the last two values into the conditional probability formula (\ref{conditional}), we obtain
\begin{equation}
p(\text{win if switch} \, | \, \text{opens goat door}) = \frac{1/3}{2/3} = \frac{1}{2}.
\end{equation}
Thus, in this modified game, the contestant essentially acquires the same probability of winning whether a choice to switch is made or not.

\section{Classical Communication}

Classical information theory \cite{shannon1948mathematical} is a powerful application of probability theory. It arose from considering engineering problems associated with classical communication systems.  The central mathematical object of the subject is the Shannon entropy.  It turns out that there are two rather separate ways to interpret this quantity; the first is derived on an intuitive notion of what properties information should have; the second is based on an operational definition in terms of data compression.  For an extensive treatment on the subject, refer to \cite{nielsen2002quantum, cover2012elements, wilde2017quantum}.

\subsection{Information content}
Consider a random variable $X$ which can take one of the values $x$ with respective probabilities $p(x)$.  The information content of $x$ is defined as
\begin{equation}\label{infocontent}
I(x) \equiv -\log_{2} (p(x)).
\end{equation}
This mathematical definition captures the intuition that the occurrence of a value associated with a lower probability provide a greater `information' gain than the occurrence of a value associated with a larger probability.

\textbf{a) Single random variable:}  Generalizing (\ref{infocontent}) to the case of the random variable gives $I(X) = -\log_{2}(p(X))$.  The Shannon entropy of $X$ is defined as the expectation value of $I(X)$:
\begin{equation}\label{Shannon}
H(X) \equiv \mathbb{E}(I(X)) = -\sum_{x} p(x) \log_{2} p(x).
\end{equation}
This quantity is a function of only the probability distribution.  We take the convention that $0 \log_{2} 0 \equiv 0$, which is supported through $\lim_{x \to 0}$ $x \log_{2} x = 0$. Within this intuitive definition, there are three ways to view the Shannon entropy:   
\begin{enumerate}[noitemsep, topsep=0pt, label=\roman*)]
	\item It represents the information content of random variable $X$.
	\item It quantifies the information gained after we learn the value of $X$.
	\item It measures the uncertainty before we know the value of $X$.
\end{enumerate}
It can be shown that the entropy has the bounds $0 \leq H(X) \leq \log_{2} d$, where $d$ is the number of values $X$ can take.  

\textbf{b) Multiple random variables:} By extracting the notions developed in probability theory, one can develop various information-theoretic constructions for multiple random variables.  An example of this is the joint entropy of random variables $X$ and  $Y$, which is defined as
\begin{equation}\label{(joint)}
H(X,Y) \equiv - \sum_{x,y} p(x,y) \log_{2} p(x,y).
\end{equation} 
The joint entropy corresponds to the total uncertainty of both the variables considered.  Using (\ref{(joint)}), we define the conditional entropy, of $X$ conditioned on $Y$, as
\begin{equation}\label{conditionalentropy}
H(X|Y) \equiv H(X,Y) - H(Y).
\end{equation}
It can be interpreted as the remaining uncertainty of $X$ once the value of $Y$ is known.  A quantity of great importance is the mutual information as it provides a way to measure how much information $X$ and $Y$ have in common:
\begin{equation}\label{mutualinformation}
H(X:Y) \equiv H(X) + H(Y) - H(X,Y).
\end{equation}
Next, consider the case where $p(x)$ and $q(x)$ are probability distributions over the same index set, $x$.  The relative entropy provides measure of `distance' between these distributions; it is defined (from $p(x)$ to $q(x)$) as
\begin{equation}\label{relativeentropy}
H(p(x) || q(x)) \equiv \sum_{x} p(x) \log_{2} \frac{p(x)}{q(x)}.
\end{equation}
Given our emphasis on temporal phenomena in this thesis, we want to consider the relationships between random variables across time.  This can be exemplified by a Markov chain, which is a sequence of random variables $X_{1} \rightarrow X_{2} \rightarrow \cdots$ such that
\begin{equation}
p(X_{n+1} = x_{n+1} \,|\, X_{n} = x_{n}, \dots, X_{1} = x_{1}) = p(X_{n+1} = x_{n+1} \,|\, X_{n} = x_{n}).
\end{equation}

\textbf{c) Properties:}  We list out some elementary properties including how the considered quantities relate to one another:
\begin{align}
H(X:Y)&=H(X) - H(X|Y), \\
H(X,Y)&=H(Y,X), \\
H(X:Y)&=H(Y:X), \\
H(Y|X)&\geq0, \\
H(X)&\leq H(X,Y),\\
H(X,Y)&\leq H(X) + H(Y), \\
H(Y|X)&\leq H(Y), \\
H(X:Y)&= H(p(x,y)||p(x)p(y)), \\
H(X,Y,Z) + H(Y) &\leq H(X,Y) + H(X,Z) \\
H(X|Y,Z) &\leq H(X|Y).
\end{align} 
Along with that, the chaining rule for conditional entropies is a result that relates random variable $Y$ to a set of random variables $X_{1}, \dots , X_{n}$ in the following way
\begin{equation}
H(X_{1}, \dots , X_{n}|Y) = \sum_{i=1}^{n}H(X_{i}|Y, X_{1}, \dots , X_{i-1}).
\end{equation}
With respect to temporal relationships, we expect that once information is lost over time, it is gone forever.  This idea is mathematically captured by the data processing inequality:  If $X \rightarrow Y \rightarrow Z$ is a Markov chain, then 
\begin{equation}
H(X) \geq H(X:Y) \geq H(X:Z).
\end{equation}
\textbf{d) Independent random variables:}  For the special case of random variables that are independent, each of the following are a biconditional property:
\begin{align}
	H(X,Y)&=H(X) + H(Y), \\
	H(Y|X)&=H(Y), \\
	H(X:Y)&=0.
\end{align}

\subsection{Data compression}

The fundamental results of classical information theory are the noiseless channel coding theorem and the noisy channel coding theorem; the former is concerned with the problem of compressing a message in a communication channel; the latter quantifies the reliability of transmitting that message over a noisy channel.  

However, our focus is solely on the noiseless coding theorem as it provides an operational definition of the Shannon entropy.  Instead of viewing $H(X)$ as the information content of $X$, it will be seen as the minimal physical resource necessary and sufficient to reliably store the output of a classical information source.

\textbf{a) Defining an information source:}  In order to derive the noiseless coding theorem, we define a classical information source as a sequence of random variables ($X_{1}, X_{2}, \dots$).  The output of the source are the values the variables take.  Furthermore, we assume the variables are independent and have identical distributions, which we abbreviate as \textit{i.i.d}.  Hence we have $H(X) \equiv H(X_{1}) = H(X_{2})= \dots$  Developing on this model, we define a compression scheme of rate $R$ as mapping output $x=(x_{1},\dots, x_{n})$ to a string of length $nR$, which we represent by $C^{n}(x) = C^{n}(x_{1},\dots, x_{n})$.  Conversely, the corresponding decompression scheme, $D^{n}(C^{n}(x))$, maps the string of length $nR$ to a string of length $n$.  The compression-decompression scheme is defined to be reliable if the probability that $D^{n}(C^{n}(x)) = x$ goes to one as $n$ goes to $\infty$.         

\textbf{b) Defining typical sequences:}  The possible outputs of the information source can divided into two sets, namely typical sequences and its complement, atypical sequences.  More precisely, given $\epsilon > 0$, a sequence $x_{1}, \ldots , x_{n}$ is $\epsilon$-typical if it satisfies 
\begin{equation}\label{typical}
2^{-n(H(X) + \epsilon)} \leq p(x_{1}, \ldots , x_{n}) \leq 2^{-n(H(X) - \epsilon)}.
\end{equation}
We can reformulate (\ref{typical}) as
\begin{equation}\label{typical2}
\biggl\lvert \frac{1}{n} \log \frac{1}{p(x_{1}, \ldots , x_{n})} - H(X) \, \biggl\rvert  \leq \, \epsilon.
\end{equation}
We also denote $T(n, \epsilon)$ as the set of of all $\epsilon$-typical sequences of length $n$. 

\textbf{c) Application of the law of large numbers:}  In the case of large $n$, it can be observed that most sequences are typical.  This hypothesis is rigorously proved in the following theorem using the law of large numbers.
\begin{theorem}\label{typicaltheorem}
	\textbf{(Theorem of typical sequences)}   
	\begin{enumerate}[noitemsep, topsep=0pt, label=\roman*)]
		\item Fix $\epsilon > 0$.  Then for any $\delta > 0$, for sufficiently large $n$, the probability that a sequence is $\epsilon$-typical is at least $1-\delta$.
		\item For any fixed $\epsilon > 0$ and $\delta > 0$, for sufficiently large $n$, the number of $\epsilon$-typical sequences, $\lvert T(n,\epsilon) \rvert$, satisfies  
		\begin{equation}
		(1-\delta)\, 2^{n(H(X) - \epsilon)} \leq \lvert T(n,\epsilon) \rvert \leq 2^{n(H(X) + \epsilon)}.
		\end{equation}
		\item Suppose $R<H(X)$.  Let $S(n)$ be a collection of size at most $2^{nR}$, of length $n$ sequences from the source.  Then for any $\delta > 0$ and for sufficiently large $n$, 
		\begin{equation}
		\sum_{x \in S(n)} p(x) \leq \delta.  
		\end{equation}
	\end{enumerate} \hfill$\Box$ 
\end{theorem}

\begin{proof} (See e.g. \cite{nielsen2002quantum}.)
\begin{enumerate}[noitemsep, topsep=0pt, label=\roman*)]
	\item Given that $X_{i}$ are a set of \textit{i.i.d} random variables, this implies $-\log p(X_{i})$ are also a set of \textit{i.i.d} random variables.  Using the law of large numbers (Theorem (\ref{lawlarge})), we have for any $\epsilon > 0$ and $\delta > 0$ for sufficiently large $n$ that
	\begin{equation}\label{typicalformula}
	p\Bigg( \biggl\lvert \sum_{i=1}^{n} \frac{-\log p(X_{i})}{n} - \mathbb{E}(-\log_{2} p(X)) \, \biggl\rvert \leq \epsilon \Bigg) \geq \, 1-\delta.
	\end{equation}    
	Using (\ref{Shannon}), we can substitute $H(X)$ for $\mathbb{E}(-\log_{2} p(X))$. Furthermore using the product property of logarithms, we have that $\sum_{i=1}^{n} {\log p(X_{i})} = \log(p(X_{1}, \dots, X_{n}))$.  This modifies (\ref{typicalformula}) to give the desired result that the probability a sequence is $\epsilon$-typical is at least $1-\delta$: 
	\begin{equation}\label{typicalprob}
	p\Bigg( \biggl\lvert \frac{1}{n} \log \frac{1}{p(X_{1}, \ldots , X_{n})} - H(X) \, \biggl\rvert  \leq \, \epsilon \Bigg) \geq 1-\delta.
	\end{equation}
	\item The sum of the probabilities of the typical sequences cannot be greater than one.  Along with (\ref{typical}), we see that
	\begin{align}
	1 &\geq   \sum_{x \in T(n,\epsilon)} p(x) \\
	&\geq   \sum_{x \in T(n,\epsilon)} 2^{-n(H(X) + \epsilon)} \\
	&= \lvert T(n, \epsilon) \rvert  2^{-n(H(X) + \epsilon)}. 
	\end{align}
	Therefore, we obtain that $\lvert T(n, \epsilon) \rvert \leq 2^{n(H(X) + \epsilon)}$.  Conversely, from (\ref{typicalprob}), we can also deduce that the sum of the probabilities of typical sequences must be at least $1-\delta$.  Under this requirement, along with (\ref{typical}), we can write 
	\begin{align}
	1-\delta  &\leq  \sum_{x \in T(n,\epsilon)} p(x) \\
	&\leq   \sum_{x \in T(n,\epsilon)} 2^{-n(H(X) - \epsilon)} \\
	&= \lvert T(n, \epsilon) \rvert  2^{-n(H(X) - \epsilon)}. 
	\end{align} 
	Hence, we can compute that $\lvert T(n, \epsilon) \rvert \geq (1-\delta)2^{n(H(X) - \epsilon)}$. 
	\item Fix an $\epsilon$ such that $R < H(X) - \epsilon$, and $0 < \epsilon < \delta/2$.  The total probability for $\epsilon$-atypical sequences in $S(n)$ can be made small, ie less than $\delta/2$, for large enough $n$.  The total number of $\epsilon$-typical sequences is at most $2^{nR}$ since that is the upper bound for the total number of sequences in $S(n)$.  Furthermore, each $\epsilon$-typical sequence has probability at most $2^{-n(H(X)-\epsilon)}$. Therefore, the total probability of $\epsilon$-typical sequences in $S(n)$ is $2^{-n(H(X)-\epsilon-R)}$. Given  $R < H(X)- \epsilon$, we can see that $2^{-n(H(X)-\epsilon-R)} \rightarrow 0$ as $n \rightarrow \infty$. Hence the total probability of sequences in set $S(n)$ is less than $\delta$ for sufficiently large $n$.        
	\end{enumerate}  
\end{proof}

\textbf{d) Application of theorem of typical sequences:}  The usefulness of Theorem (\ref{typicaltheorem}) becomes apparent when proving the main result:
\begin{theorem}\label{Shannoncoding}
	\textbf{(Shannon's noiseless channel coding theorem)} 
	Consider an \textit{i.i.d.} information source represented by $\{X_{i}\}$, with entropy rate $H(X)$:
	\begin{enumerate}[noitemsep, topsep=0pt, label=\roman*)]
		\item  If $R>H(X)$, then there exists a reliable compression scheme of rate $R$ for the information source.
		\item  Conversely, if $R<H(X)$, then any compression scheme will not be reliable. 
	\end{enumerate} \hfill$\Box$ 
\end{theorem}

\begin{proof} (See e.g. \cite{nielsen2002quantum}.)
	\begin{enumerate}[noitemsep, topsep=0pt, label=\roman*)]
		\item Consider the case $R>H(X)$.  We choose an $\epsilon$ such that $H(X) + \epsilon < R$.  From Theorem (\ref{typicaltheorem}), we have that for any $\delta > 0$ and for sufficiently large $n$, there are at most $2^{n(H(x)+\epsilon)} < 2^{nR}$ $\epsilon$-typical sequences produced by the information source.  Given that there are most $2^{nR}$ of such sequences, it only requires $nR$ bits to uniquely identify a particular $\epsilon$-typical output.  Hence we can compress the $\epsilon$-typical output, using some scheme, to a string of $nR$ bits which can be decompressed later.  Furthermore, using Theorem (\ref{typicaltheorem}), we have that the probability of producing such an $\epsilon$-typical sequences is at least $1-\delta$.  If on the other hand, we have an  $\epsilon$-atypical sequence, we declare an error and give up on compression.
		\item Consider the case $R<H(X)$.  There are at most $2^{nR}$ outputs for the combined compression-decompression scheme.  Using Theorem (\ref{typicaltheorem}), the probability, for sufficiently large $n$, of the information output belonging to a subset of the $2^{nR}$ sequences tends to zero.  Hence any compression scheme for this case will not be reliable.        
	\end{enumerate}  
\end{proof}

\textbf{e) Comments:} 

\begin{enumerate}[noitemsep, topsep=0pt, label=\roman*)]
	\item The entropy can be operationally defined as the minimum physical resource required to reliably store the output of a classical information source.
	\item The idea is that we only need to compress typical sequences, as they are the outputs that are overwhelmingly likely to occur in the asymptotic limit.
\end{enumerate}  

\section{Classical Computing}

The wide proliferation of digital computers across the globe has led to a period in human history known as the `Information Age.'  However, the conception of these physical devices stemmed from abstract work in the foundations of mathematics \cite{alan1936turing}.  This investigation brought about a mathematical model of computation known as Turing machine, which has since had a profound influence across different spheres of thought\cite{downey2014turing}.

Surprisingly, there are a number of different models of computation which are equivalent to the Turing machine.  One such example is the circuit model which we briefly cover in this section.  We also look at how one can probe at the resources required for a model to solve a computational problem; this can quantitatively captured by a framework known as the asymptotic notation.  For a broader survey on the theory of computation, we refer the reader to \cite{moore2011nature}.
 
\subsection{Circuit model}

An enormous range of computations can be performed by using a combination of circuits.  Circuits are abstractions which can be physically instantiated, most commonly through classical electrical systems.  They are composed of three primary elements.  The first is that they encode the information in a bit, whose state is either a $0$ or a $1$.  The second element is that circuits are made up of `wires' which carry that information through space or time.  The final piece is that circuits contain logic gates which are a particular application of Boolean logic; more precisely a logic gate is a function $f: \{0,1\}^{k} \rightarrow \{0,1\}^{l}$ where $k$ and $l$ respectively denote the number of input and output bits.  

We briefly describe various elementary logic gates as follows:

\textbf{a) NOT: }The NOT gate inverts the input value
\begin{equation}
f(a)=1 \oplus a
\end{equation}
where $\oplus$ represents modulo $2$ addition.

\textbf{b) AND: } The AND gate outputs bit $1$ if both input values are $1$.

\textbf{c) OR: } The OR gate produces output $1$ if at least one of the input values are $1$.

\textbf{d) XOR: } The XOR gate outputs bit $1$ if only one of the input values are $1$.
	
\textbf{e) NAND :} The NAND gate produces the negation of an AND gate.
	
\textbf{f) NOR :} The NOR gate produces the negation of an OR gate.
	
Using a combination of these gates, one can construct integrated circuits to solve computational problems with sophisticated mathematical structures.  The particular step by step procedure to do so are collectively known as an algorithm for that problem.  

However, a related issue to consider are what are the minimal number of gates required to solve a particular problem of interest?  More broadly speaking, how does one quantify the resources required by a specific algorithm? Furthermore, is there a limit to the computational capabilities provided by classical resources?

\subsection{Asymptotic notation}

Computational resources can be measured in a multitude of forms depending on the nature of the  problem in question.  Common examples include the number of evaluations of a function, space requirements (say in the form of memory), time requirements (in terms of run time of an algorithm) or even energy.  

For an appropriate framework to analyze specific algorithms, an important consideration is that one cares only about how the resource consumed scales with the `size' of the corresponding problem. Roughly speaking, each problem has a quantity of interest that can be used to describe the problem, and the magnitude of that quantity represents the size of the problem.  As an example, $n$ could be the number of input bits for an algorithm which takes $30n + \log_{2} n$ gates to execute.  The only term that dominates for large sizes is $30n$ hence we say that the number of operations required scales like $n$.  The asymptotic notation captures this idea.

Suppose $f(n)$ and $g(n)$ are two functions where $n$ is a non-negative integer.  With this in mind, one can define the three tools provided by the asymptotic notation.

\textbf{a) The `big $O$': }The first tool in the asymptotic notation is the $O$ notation.  It quantifies the upper bound on the behaviour of a function.  A function $f(n)$ is $O(g(n))$ if there are constants $c$ and $n_{0}$ such that for all values of $n$ greater than $n_{0}$, $f(n) \leq c g(n)$. 

\textbf{b) The `big Omega': }Conversely, the $\Omega$ notation provides a lower bound.  A function $f(n)$ is in $\Omega(g(n))$ if there are constants $c$ and $n_{0}$ such that for all values of $n$ greater than $n_{0}$, $c g(n) \leq f(n)$. 

\textbf{c) The `big Theta': }The final tool is the $\Theta$ notation which corresponds to the notion that $f(n)$ and $g(n)$ are similar in the asymptotic regime.  More precisely, $f(n)$ is in $\Theta(g(n))$ if it is both $O(g(n))$ and $\Omega(g(n))$. 

The asymptotic notation provides a way to quantify the resources used by an algorithm for a specific problem.  By harnessing this framework, it allows superior algorithms to be quantitatively expressed in that they use fewer resources than previous ways of solving the relevant problem.  The design of such powerful algorithms is one of the central aims in the field of classical computation.

\section{Classical Blockchain}

Information security systems harness concepts from both communication and computing.  One prominent example of this class of technologies is the classical blockchain system which stores data securely over time.  Furthermore, this task is accomplished among computer nodes in a communication network that do not necessarily trust each other.  

The pioneering invention of the blockchain system was first described pseudonymously in \cite{nakamoto2008bitcoin}.  However, many of the individual subsystems draw their inspiration from a large body of disconnected theoretical research \cite{narayanan2017bitcoin}.  Over recent years, countless variants have been proposed \cite{cachin2017blockchain}, but we devote this section to describing the original design, with an emphasis on the mathematical concepts.

The aim of a blockchain system is to have a single database of records about the past that every node in the network can agree on.  Furthermore, it should not require a centralized management node.  We start with describing the two primary elements of such a system. The first is the blockchain data structure which encodes the classical information using an algorithm.  The second component involves a communication network to provide the decentralization feature.  We conclude this section by conveying the essential ideas of public key cryptography; this is used in various tasks within the blockchain system.

\subsection{Blockchain data structure}

Records about the past, which occurred at around the same time, are received and collected into a data block.  These blocks are time-stamped to ensure that the data existed at the specified time. Furthermore, the blocks are linked in chronological order through mathematical functions known as cryptographic hash functions \cite{wenbo2004modern}.  We provide a more careful treatment of the linked blocks as follows.

A cryptographic hash function, $h$, is a deterministic function that maps a string of arbitrary length to a string of fixed length (eg 256 bits).  The output is known as the hash digest, $h(x)$.  This computing task can be accomplished by various cryptographic hash algorithms (eg SHA-$256$).

The function, $h$, satisfies the following properties:

\textbf{a) Preimage resistant: } It is infeasible through classical computation that given output $d = h(x)$, one can derive the input string $x$.  This gives the implication that the function is one-way.  This is based on the assumption that the search space of outputs is large.

\textbf{b) Second preimage resistant: } It is infeasible through classical computation given that given input $x$, one can find $y \neq x$, such that $h(x)=h(y)$.  

\textbf{c) Collision resistant: } It is infeasible through classical computation to find any two inputs, $x$ and $y$, that produce the same digest $h(x)=h(y)$.  

\textbf{d) Efficient: } It requires polynomial (ideally linear) computational resources to compute the digest, $d$, given the size of input $x$.  

\textbf{e) Pseudo-random: }  If one modifies any of the bits in the input, $x$, it has a significant unpredictable change in the output of $h(x)$. 

Using these mathematical properties, each block, with its string of bits, is mapped using the hash function to a specific digest.  More crucially, each block's data contains the hash digest of the previous block.  This latter property provides the required notion of a `chain' of blocks, resulting in the term blockchain.

This \textit{interdependence of the time-stamped blocks}, through the cryptographic hash functions, provides the necessary sensitivity for the role of securing the records in a blockchain.  Any party that attempts to falsify the past records in a block would need to find a way to alter the data such that it does not change the digest of that block. This task, as we have mentioned, is computationally infeasible.  Hence, the resulting change in the digest of the tampered block would cause all subsequent blocks to have different digests.  This is due to the design that each block's data contains the digest of the previous block.  Hence, the consequence of this sensitivity is that altering the data in a block would tamper all subsequent blocks and hence invalidate them.  Furthermore, given that only future blocks following the tampered block are invalidated, this implies that the the older the time stamp on the block, the more secure it is in the blockchain.  In summary, the blockchain data structure provides a tamper proof system for storing records, precisely because tampering with it can easily be detected.

\subsection{Network consensus protocol}
Along with a blockchain data structure, the second part to the system is a classical communication network.  Each node on the network carries a local copy of the blockchain data structure.  This provides the mechanism if one local copy is destroyed, other nodes with local copies would serve to provide replication.

However, the primary objective of the network component is to add valid blocks to each local copy without a centralized management node.  The challenge of the task is that it must be accomplished without the assumption that all the nodes are `honest.' Typically, this involves invoking a node on the network to confirm the validity of records in a new block, and then communicating that block to other nodes on the network. The different nodes accept the block if the block is valid and they can successfully link it to their own local copy of the blockchain data structure through the cryptographic hash functions. For this procedure to maintain ongoing accuracy, the validating node gets chosen at random for each block; this prevents preplanned node-specific attacks. Furthermore, the validating node is also incentivised through the network for carrying out these tasks. Despite some dishonest nodes, this is all successfully accomplished through a non-trivial consensus protocol.

In the original design, the consensus protocol is coined `proof-of-work' or is also known as the Nakamoto consensus.  In this scenario, the node that successfully validates the block has to expend a specific amount of computational resource.  This resource is used to solve a tractable problem involving the hash digest associated to the new block in question.  After the node verifies the validity of the block, it is rewarded by an economic incentive.

However, the consensus protocols in the blockchain systems do not fit into the traditional framework of fault-tolerant distributed computing \cite{narayanan2017bitcoin, bano2017consensus}.  More specifically, it is not rigorously clear that `proof-of-work' satisfies a security standard known as BFT (Byzantine Fault Tolerance) \cite{lamport1982byzantine}.  In this setting, byzantine nodes refer to computer nodes that may take arbitrary actions such as sending faulty messages, as opposed to crash failure nodes which fail by stopping.  An well known example of a BFT protocol in the fault-tolerant literature is PBFT (practical Byzantine fault tolerance) \cite{castro1999practical}.

\subsection{Public key cryptography}

Public key cryptography forms the security backbone of the classical information infrastructure of the modern world.  In the specific case of blockchain technologies, it is most notably implemented for digitally signing the records in a block \cite{yaga2019blockchain}.  The subject of public key cryptography is infeasible to cover in a short section, and hence we refer the reader to \cite{galbraith2012mathematics} for a deeper mathematical coverage.  We limit our discussion to the RSA (Rivest--Shamir--Adleman) public key cryptosystem which relies on ideas extracted from number theory.  Furthermore, our aim is to articulate the essential concepts by focusing within the simplified context of two parties wishing to communicate in private.

\textbf{a) Number-theoretic preliminaries: } We briefly digress to results regarding prime numbers and modular arithmetic.  Two integers $a$ and $b$ are defined as co-prime if their greatest common divisor is one.  The Euler $\varphi (n)$ function is defined to be the number of positive integers less than $n$ which are co-prime to $n$.  

Suppose that $n$ has prime factorization $n = p_{1}^{\alpha_{1}} \cdots p_{k}^{\alpha_{k}}$ where $p_{1}, \cdots, p_{k}$ represent the distinct prime numbers, and $\alpha_{1},\cdots, \alpha_{k}$ are positive integers.  Then one can derive the formula
\begin{equation}
\varphi(n) = \prod_{j=1}^{k}p_{j}^{\alpha_{j}-1}(p_{j}-1).
\end{equation}
Furthermore, it can be proven that if $a$ is co-prime to $n$, then 
\begin{equation}\label{coprime}
a^{\varphi(n)} = 1 \, (\text{mod } n).
\end{equation}
\textbf{b) Communication problem: } Suppose a party, say `Alice', wants to transmit a message to another party, say `Bob', over a classical communication channel.  More crucially, they want to ensure that no other party can access the contents of the message.  This can be accomplished, with a significant degree of confidence, by invoking the mathematical notions of a public and private key. 

\textbf{c) Encrypting the message: }  The message Alice wants to transmit is denoted $m$.  She is said to have encrypted her message to $c$ if she performs the computation
\begin{equation}
c = m^{e} \, (\text{mod } n),
\end{equation}
where the values $n$ and $e$ are collectively known as the public key.  These values are generated by Alice.  She first selects two large prime numbers $p$ and $q$.  Then she computes $n=pq$.  From this, Alice picks an $e \in \mathbb{N}$ such that $e$ is co-prime to $n$ and also satisfies $1 < e < \varphi(n)$ where 
\begin{equation}
\varphi(n) = (p-1) (q-1).
\end{equation}

\textbf{d) Decrypting the message: } Bob receives the encrypted message $c$ over the communication channel.  He is said to have decrypted the message back to $m$ if he performs the computation
\begin{equation}
m = c^{d} (\text{mod }n),
\end{equation}
where the values $n$ and $d$ are collectively known as the private key.  The value $d \in \mathbb{N}$ is generated by
\begin{equation}\label{privatekey}
d = \frac{1 \, \text{mod }\varphi (n)}{e}.
\end{equation}

The decryption procedure can seen more clearly by considering the specific case that $m$ is co-prime to $n$ (although this can be generalized to the case when $m$ is not co-prime to $n$).  From (\ref{privatekey}), we have $ed = 1 + k\varphi(n)$ for some $k \in \mathbb{N}$.  Using result (\ref{coprime}), we find that $m^{k \varphi(n)} = 1 \, (\text{mod }  n)$.  Substituting this into the decryption procedure results in
\begin{align}
(c)^{d} &= (m^{e})^{d} \, (\text{mod } n) \\
&= m^{ed} \, (\text{mod } n) \\
&= m^{1 + k\varphi(n)} \, (\text{mod } n) \\
&= m \cdot m^{k\varphi(n)} \, (\text{mod } n) \\
&= m \, (\text{mod } n).
\end{align}

Using the symmetry property of modular arithmetic, this implies the desired result that $m = \, c^{d} (\text{mod }n)$. 

\textbf{e) Breaking encryption: }  The private key is kept in secret by the intended party.  This is in contrast with the public key which is available to anyone.  Despite this wide access, there is no increase in the security vulnerability as we shall describe below.  The outside party that aims to eavesdrop to the transmission between Alice and Bob is commonly referred to as `Eve'.  If Eve has access to the private key, she can extract the message $m$ from $c$.  One way to obtain the private key would be if she could derive $p$ and $q$ by factoring $n=pq$.  She would then be able to compute $\varphi(n) = (p-1) (q-1)$, and consequently obtain the private key $(d, n)$.

However, the problem of prime factorization with classical computation is currently believed to require exponential resources (but this hypothesis is not formally proven).  More accurately, the best known classical algorithm for this task is the NFS (Number Field Sieve) algorithm which has a performance of $\exp (\Theta \, (n^{1/3} \log^{2/3} n))$ operations for an $n$-bit integer.  It is precisely the on-going computational difficulty of this problem that ensures durability of this information security system.


\chapter{Quantum Information}\label{chap: QInfo}

\begin{chapquote}{Philip Ball, \textit{Quantum teleportation is even weirder than you think}}
	``Is it, for example, information about some underlying reality, or about the effects of our intervention in it? Information universal to all observers, or personal to each? And can it be meaningful to speak of quantum information as something that flows, like liquid in a pipe, from place to place? No one knows (despite what they might tell you).''
\end{chapquote}

\textsf{QUANTUM INFORMATION SCIENCE} is the theoretical and experimental study of quantum information and its applications.  The field is largely concerned with designing quantum systems to perform information tasks.  This novel exploration has the following consequences that make the subject fundamental:
\begin{enumerate}[noitemsep, topsep=0pt, label=\roman*)]
	\item It reconceptualizes the probability amplitude of quantum theory as a quantity that can be harnessed for representing and transforming information; it is precisely this quantity that is termed `quantum information.'
	\item Analogous to the study of classical information, a generalized framework is developed that abstracts away from the physical (quantum mechanical) systems that could be used to store the quantum information.
	\item It distils questions on the nature of quantum physics to distinctions between quantum information and classical information.
\end{enumerate}
In this chapter, we look at three theoretical tools of quantum information science.

\section{Review of Linear Algebra}

Prior to examining the three main topics in this chapter, we provide a brief overview of linear algebra with an emphasis on the use of the Dirac notation.  

\subsection{Vector spaces}

The vector space that is commonly used in quantum information science is $\mathbb{C}^{n}$.  An element of the space, namely a vector, can be denoted $\ket{\psi}$ (referred to as a ket), where $\psi$ is simply a label for the vector.  The vector can have a column matrix representation of its $n$-tuples of complex numbers.  Vector addition in $\mathbb{C}^{n}$ proceeds as
\begin{equation}
\begin{pmatrix}
a_{1}  \\
\vdots  \\
a_{n} 
\end{pmatrix} + 
\begin{pmatrix}
b_{1}  \\
\vdots  \\
b_{n} 
\end{pmatrix} \equiv 
\begin{pmatrix}
a_{1} + b_{1}  \\
\vdots \\
a_{n} + b_{n}
\end{pmatrix}.
\end{equation}
Scalar multiplication is computed as
\begin{equation}\label{scalarmultiplication}
\alpha 
\begin{pmatrix}
a_{1}  \\
\vdots  \\
a_{n} 
\end{pmatrix} \equiv
\begin{pmatrix}
\alpha a_{1}  \\
\vdots  \\
\alpha a_{n} 
\end{pmatrix}.
\end{equation}
Note that it does not make a difference if a scalar stands on the left or the right of a ket, $\alpha\ket{\psi} = \ket{\psi}\alpha$.  We exclude the use of the ket notation for the zero vector and rather denote it as $0$.  A vector subspace of a vector space is a subset of the vector space such that the subset is also a vector space.

\subsection{Basic definitions}

A spanning set for a vector space is a set of vectors $\ket{v_{1}}, \dots ,\ket{v_{n}}$ such that any vector $\ket{v}$ in the vector space can be written as $\ket{v}=\sum_{i} \alpha_{i}\ket{v_{i}}$.  Another core concept is that a set of non-zero vectors $\ket{v_{1}}, \dots ,\ket{v_{n}}$ is said to be linearly dependent if the equation
\begin{equation}
\alpha_{1} \ket{v_{1}} + \alpha_{2} \ket{v_{2}} + \cdots + \alpha_{n} \ket{v_{n}} = 0, 
\end{equation}
has a solution where $\alpha_{i} \neq 0$ for at least one value of $i$.  A set of vectors is linearly independent if it is not linearly dependent.  A set of vectors that spans the vector space and is linearly independent is called a basis for the vector space.  The dimension of the vector space is the number of elements in a basis set.  With the exception of \autoref{chap: RQI}, this thesis is only concerned with finite dimensional vector spaces.  

An example of a basis for $\mathbb{C}^{2}$ is the computational basis set 
\begin{equation}\label{compbasis}
\ket{0} \equiv
\begin{pmatrix}
1  \\
0 
\end{pmatrix}, \quad
\ket{1} \equiv
\begin{pmatrix}
0  \\
1 
\end{pmatrix}.
\end{equation} 
Another basis for the space is
\begin{equation}\label{plusbasis}
\ket{+} \equiv \frac{1}{\sqrt{2}}
\begin{pmatrix}
1  \\
1 
\end{pmatrix}, \quad
\ket{-} \equiv \frac{1}{\sqrt{2}}
\begin{pmatrix}
1  \\
-1 
\end{pmatrix}.
\end{equation} 

\subsection{Operators and Matrices}

Suppose $V$ and $W$ are vector spaces.  A linear operator between $V$ and $W$ is defined to be any function $A: V \rightarrow W$ which is linear in inputs
\begin{equation}
A \, \Biggl(\sum_{i}\alpha_{i} \ket{v_{i}}\Biggl) = \sum_{i} \alpha_{i} \, A(\ket{v_{i}}).
\end{equation}
We can write $A\ket{v_{i}}$ to denote $A(\ket{v_{i}})$.  We say a linear operator $A$ is defined on a vector space $V$ if $A:V\rightarrow V$.   The identity operator $I$ maps all vectors to their respective self, $I\ket{v}=\ket{v}$.  The zero operator $0$ maps any vector to the zero vector, $0\ket{v}=0$.  The composition of two operators, say $A$ and $B$, on a vector is defined as $(AB)(\ket{v}) \equiv A(B\ket{v})$.  

Operator addition is commutative, $A+B=B+A$, and associative $A+(B+C) = (A+B+C)$.  However operator multiplication is not commutative $AB \neq BA$ but is associative $A(BC) = (AB)C = ABC$.  

Operators have an equivalent matrix representation.  A $m$ by $n$ complex matrix $A$ with entries $A_{ij}$ can be thought as a linear operator that maps vectors from $\mathbb{C}^{n}$ to $\mathbb{C}^{m}$ under matrix multiplication.  Conversely to view operators as matrices, suppose $V$ and $W$ are vector spaces with operator $A: V \rightarrow W$.  More crucially, let $\ket{v_{1}}, \dots ,\ket{v_{m}}$ be a basis for $V$, and let $\ket{w_{1}}, \dots ,\ket{w_{n}}$ be a basis for $W$. Then for every $j$ between $1$ and $m$, there exists complex coefficients $A_{1j}$ through $A_{nj}$ such that 
\begin{equation}
A\ket{v_{j}} = \sum_{i}A_{ij} \ket{w_{i}}.
\end{equation}  
The complex numbers $A_{ij}$ form the matrix representation of the operator $A$.

Of critical importance are the topics of eigenvectors and eigenvalues.  An eigenvector of operator $A$ is a non-zero vector $\ket{v}$ that satisfies the equation $A\ket{v} = \lambda \ket{v}$, where $\lambda$ is a complex number known as the eigenvalue corresponding to $\ket{v}$.  The solution to the characteristic equation $c(\lambda) = 0$, where $c(\lambda) \equiv \text{det}\lvert A-\lambda I\rvert$, are the eigenvalues of operator $A$.  The eigenspace corresponding to eigenvalue $\lambda$, is a vector subspace on which $A$ acts, that contains all the eigenvectors which have $\lambda$ as its eigenvalue.  When the dimension of the eigenspace is greater than one, we say it is degenerate.

\subsection{Types of products}

One can go beyond the basic abstraction of a vector space with its scalar multiplication; we will discuss four types of products that occur between vectors. 

\textbf{a) Inner product: }  An inner product maps two vectors, say $\ket{v}$ and $\ket{w}$, to a complex number.  We denote this complex number as $\braket{v|w}$.  The notation $\bra{v}$ is referred to as the dual vector (or a bra).  A vector space with an inner product is called an inner product space.  An inner product satisfies properties:
\begin{align}
\bra{v} \Bigg(\sum_{i} \alpha_{i} \ket{w_{i}}\Bigg) &= \sum_{i} \alpha_{i} \braket{v|w_{i}},  \\
\braket{v|w} &= \braket{w|v}^{*}, \\
\braket{v|v} &\geq 0. 
\end{align}
One can define the following inner product for $\mathbb{C}^{n}$: For two vectors with respective column matrix entries $(a_{1}, \dots, a_{n})$ and  $(b_{1}, \dots, b_{n})$, an inner product is given by $\sum_{i} a_{i}^{*} b_{i}$.  In the case of finite dimensional complex vector spaces, an inner product space is also referred to as a Hilbert space.

Using the inner product, one can develop several useful notions.  Vectors $\ket{v}$ and $\ket{w}$ are said to be orthogonal if $\braket{v|w}= 0$. The norm of a vector $\ket{v}$ is defined as $\lvert \lvert \ket{v} \rvert \rvert = \sqrt{\braket{v|v}}$.  A unit vector has a norm of value one; any vector with this property is said to be normalized.  Furthermore, for any non-zero vector $\ket{v}$, its normalized form is given by $\ket{v}/\lvert \lvert \ket{v} \rvert \rvert$.  A set of vectors $\ket{v_{i}}$ with index $i$ is said to be orthonormal if $\braket{v_{i}|v_{j}}$ = $\delta_{ij}$.  The Gram-Schmidt procedure transforms an arbitrary basis of a vector space with an inner product, to an orthonormal basis; suppose $\ket{w_{1}}, \dots ,\ket{w_{d}}$ is an arbitrary basis; then an orthonormal basis $\ket{v_{1}}, \dots ,\ket{v_{d}}$ is computed first by     $\ket{v_{1}} \equiv \ket{w_{1}}/\lvert \lvert \ket{w_{1}} \rvert \rvert$, and then the rest inductively obtained through formula,
\begin{equation}
\ket{v_{k+1}} \equiv \frac{\ket{w_{k+1}} - \sum_{i=1}^{k} \braket{v_{i}|w_{k+1}} \ket{v_{i}}}{\lvert \lvert \ket{w_{k+1}} - \sum_{i=1}^{k} \braket{v_{i}|w_{k+1}} \ket{v_{i}} \rvert \rvert}.
\end{equation}    
An orthonormal basis has the advantage of simplifying various computations.  Let $\ket{i}$ be an orthonormal basis, with the following vectors, $\ket{w} = \sum_{i} w_{i} \ket{i}$ and $\ket{v} = \sum_{i} v_{i} \ket{i}$.  Then the inner product is given by
\begin{equation}
\braket{v|w} = \sum_{ij}v_{i}^{*}w_{j}\delta_{ij} = \sum_{i}v_{i}^{*}w_{i} = \begin{pmatrix}
v_{1}^{*} \dots v_{n}^{*} \\
\end{pmatrix} 
\begin{pmatrix}
w_{1}  \\
\vdots  \\
w_{n} 
\end{pmatrix}.
\end{equation} 		
The dual vector can be interpreted as a row vector whose elements are complex conjugates of the components of the column vector form of $\ket{v}$.		

\textbf{b) Outer product: }  Suppose $\ket{v}$ and $\ket{w}$ are vectors from respective inner product spaces, $V$ and $W$.  Then the outer product $\ket{w}\bra{v}$ is a linear operator from $V \rightarrow W$ which is defined by $(\ket{w}\bra{v})\ket{u} \equiv \ket{w} \braket{v|u} = \braket{v|u} \ket{w}$.  This is a valid operation as long as we are dealing with `legal' products.  This property is also referred to as the associative axiom \cite{sakurai1995modern} as it is an extension of the associativity of operator multiplication.  More generally,
\begin{equation}
\Bigg(\sum_{i}\alpha_{i} \ket{w_{i}}\bra{v_{i}}\Bigg)\ket{u} = \sum_{i}\alpha_{i} \ket{w_{i}}\braket {v_{i}|u}.
\end{equation} 
An application of the outer product is the completeness relation: If $\ket{i}$ is an orthonormal basis, then the identity operator can be written as $I =\sum_{i}\ket{i}\bra{i}$.  Using this property, one can obtain an outer product representation of operator $A:V \rightarrow W$:
\begin{align}
A &= I_{W} \, A \, I_{V} \\
& = \sum_{ij}\ket{w_{j}}\bra{w_{j}} A \ket{v_{i}}\bra{v_{i}} \\
&= \sum_{ij} \braket{w_{j}|A|v_{i}}\ket{w_{j}}\bra{v_{i}}.
\end{align}
The quantity $ \braket{w_{j}|A|v_{i}}$ is the matrix element in the $j$th row and $i$th column; the matrix representation is with respect to basis $\ket{v_{i}}$ and $\ket{w_{j}}$.  The completeness relation is also used to prove the Cauchy-Schwarz inequality which states that for any two vectors in a Hilbert space, $\ket{v}$ and $\ket{w}$, we have $\lvert \braket{v|w} \rvert^{2} \leq \braket{v|v} \braket{w|w}$. 

Suppose $\ket{i}$ is an orthonormal set of eigenvectors for operator $A$ with corresponding eigenvalues $\lambda_{i}$.  Then a diagonal representation (or an orthonormal decomposition) for $A$ is given by $A = \sum_{i} \lambda_{i} \ket{i}\bra{i}$.  An operator that has a diagonal representation is said to be diagonalizable.  

\textbf{c) Tensor product: } One can construct a larger vector space from two or more different vector spaces.  The mathematical machinery for such a construction is named the tensor product.  To be more precise, suppose $V$ and $W$ are Hilbert spaces with respective dimensions $m$ and $n$.  Then $V \otimes W$ is a vector space with dimension $mn$.  The elements of $V \otimes W$ are linear combinations of $\ket{v} \otimes \ket{w}$, which is a tensor product of elements $\ket{v}$ of $V$, and $\ket{w}$ of $W$.  For the case that $\ket{i}$ and $\ket{j}$ are respective orthonormal bases for $V$ and $W$, $\ket{i} \otimes \ket{j}$ forms a basis for $V \otimes W$.  The tensor product has the following properties:
\begin{align}
&z(\ket{v} \otimes \ket{w}) = (z\ket{v}) \otimes \ket{w} = \ket{v} \otimes (z\ket{w}), \\
&(\ket{v_{1}} + \ket{v_{2}}) \otimes \ket{w} = \ket{v_{1}} \otimes \ket{w} + \ket{v_{2}} \otimes \ket{w}, \\
&\ket{v} \otimes (\ket{w_{1}} + \ket{w_{2}}) = \ket{v} \otimes \ket{w_{1}}  + \ket{v} \otimes \ket{w_{2}}, 
\end{align}   
where $z$ is an arbitrary scalar, and the rest are vectors from their respective vector spaces.  One can extend the tensor product to operators; suppose $\ket{v}$ and $\ket{w}$ are vectors in $V$ and $W$, and $A$ and $B$ are linear operators respectively on $V$ and $W$; then one can define a linear operator $A \otimes B$ which acts on $V \otimes W$ as
\begin{equation}
(A \otimes B)(\ket{v} \otimes \ket{w}) \equiv A\ket{v} \otimes B\ket{w}.
\end{equation}
More generally, one has
\begin{equation}
(A \otimes B) \, \Bigg(\sum_{i} \alpha_{i} \ket{v_{i}} \otimes \ket{w_{i}}\Bigg) \equiv \sum_{i} \alpha_{i} A\ket{v_{i}} \otimes B\ket{w_{i}}.
\end{equation}
The inner product on $V \otimes W$ is defined as follows; suppose we have two vectors $\sum_{i} \alpha_{i} \ket{v_{i}} \otimes \ket{w_{i}}$ and $\sum_{j} \beta_{j} \ket{v'_{j}} \otimes \ket{w'_{j}}$, then the inner product is defined as
\begin{equation}
\sum_{ij} \alpha_{i}^{*}\beta_{j}\braket{v_{i}|v'_{j}} \braket{w_{i}|w'_{j}}.
\end{equation}
The tensor product can also be computed in terms of matrices.  If $A$ is an $m$ by $n$ matrix, and $B$ is an $p$ by $q$ matrix, then we have 
\begin{equation}
A \otimes B \equiv
\begin{pmatrix}
A_{11}B & A_{12}B & A_{13}B & \dots  & A_{1n}B \\
A_{21}B & A_{22}B & A_{23}B & \dots  & A_{2n}B \\
\vdots & \vdots & \vdots & \ddots & \vdots \\
A_{m1}B & A_{m2}B & A_{m3} & \dots  & A_{mn}B
\end{pmatrix}.
\end{equation}
For tensor product $\ket{v} \otimes \ket{w}$, one can use equivalent notations $\ket{v}\ket{w}$, or $\ket{v,w}$, or simply $\ket{v w}$.  Additionally, one often writes $\ket{\psi}^{\otimes n}$ to signify that $\ket{\psi}$ is tensored with itself $n$ times. 

\textbf{d) Illegal products: } Certain products are nonsensical in the Dirac notation \cite{sakurai1995modern} and should be avoided.  Unlike the tensor product, if vectors $\ket{v}$ and $\ket{w}$ belong to the same vector space, then the product $\ket{v}\ket{w}$ is illegal; a similar condition holds for the dual vectors.  Furthermore, operators always stand on the left of a ket and to the right of a bra; hence, the products $\ket{v}B$ and $A \bra{v}$ are illegal.

\subsection{Common operations}   

\textbf{a) Hermitian conjugate: }  If $A$ is a linear operator on $V$, then the Hermitian conjugate (or adjoint) of $A$ is denoted $A^{\dagger}$ and it satisfies
\begin{equation}
\braket{v|A^{\dagger}|w} = \braket{w|A|v}^{*}, 
\end{equation}  
for all vectors $\ket{v}$, $\ket{w}$ in $V$.
In terms of a matrix representation of operator $A$, the Hermitian conjugation can be defined as $A^{\dagger} \equiv (A^{*})^{T}$ where $*$ represents complex conjugation and $T$ represents the transpose operation.  For the case of a scalar, the Hermitian conjugate reduces to the complex conjugate.  For the case of a vector, we have $\ket{v}^{\dagger} \equiv \bra{v}$.  We list a number of further properties:
\begin{align}
&(AB)^{\dagger} = B^{\dagger} A^{\dagger}, \\
&(A\ket{v})^{\dagger} = \bra{v}A^{\dagger}, \\
& (\ket{w}\bra{v})^{\dagger} = \ket{v}\bra{w}, \\
&(A^{\dagger})^{\dagger}=A, \\
&\Bigg(\sum_{i}\alpha_{i} A_{i} \Bigg)^{\dagger} = \sum_{i}\alpha_{i}^{*} A_{i}^{\dagger}.
\end{align}   

\textbf{b) Function of an operator: }  Suppose we have a function $f: \mathbb{C}\rightarrow \mathbb{C}$.  If linear operator $A$ has a diagonal representation $A = \sum_{i} \lambda_{i} \ket{i}\bra{i}$, then the corresponding operator function is defined as 
\begin{equation}
f(A) \equiv \sum_{i} f(\lambda_{i}) \ket{i}\bra{i}.
\end{equation}

\textbf{c) Trace: }  The trace of a matrix is the sum of its diagonal elements.  Furthermore, the trace of an operator is defined as the trace of any matrix representation of the operator. Hence, for the case of an operator $A$, we have
\begin{equation}
\text{tr}(A) = \sum_{i} A_{ii}.
\end{equation}   
This operation has the following properties:
\begin{align}
&\text{tr}(AB) = \text{tr}(BA), \\
&\text{tr}(A+B) =\text{tr}(A) + \text{tr}(B), \\
&\text{tr}(\alpha A) = \alpha \text{tr}(A), \\
&\text{tr}(A\ket{v}\bra{v}) = \braket{v|A|v}.
\end{align}

\textbf{d) Commutator: }  The commutator of two operators, $A$ and $B$ is defined as
\begin{equation}
[A,B] \equiv AB - BA.
\end{equation} 
The anti-commutator for the two operators is computed as $\{A,B\} \equiv AB + BA$.  The important case of $[A,B]=0$, is expressed by saying $A$ commutes with $B$.

\subsection{Types of Operators}

Using the Hermitian conjugate, operators can be classified into certain classes.

\textbf{a) Hermitian: }  A Hermitian (or self-adjoint) operator is an operator that is equal to its Hermitian conjugate, $A^{\dagger} = A$.  One of the most useful theorems regarding Hermitian operators is,
\begin{theorem}\label{diagonalization}
	\textbf{(Simultaneous diagonalization theorem)}  Suppose $A$ and $B$ are two Hermitian operators.  Then $[A, B] = 0$ if and only if there exists an orthonormal basis such that both $A$ and $B$ are diagonal with respect to that basis. (See e.g. \cite{nielsen2002quantum}.)
\end{theorem}
Hence, for simultaneous diagonalizable Hermitian operators, $A$ and $B$, we express them as $A=\sum_{i}\alpha_{i} \ket{i}\bra{i}$ and $B=\sum_{i}\beta_{i} \ket{i}\bra{i}$ for some common orthonormal set of eigenvectors $\ket{i}$.
 
\textbf{b) Projectors: }  A particular subset of Hermitian operators are known as projectors (or projection operators).  Suppose we have vector space $V$, along with a vector subspace $W$ that has orthonormal basis $\ket{1}, \dots , \ket{k}$.  Then a projector onto $W$ is defined as 
\begin{equation}
P \equiv \sum_{i=1}^{k}\ket{i}\bra{i}.
\end{equation}   
A projector satisfies the property $P^{2}=P$.  Furthermore, all the eigenvalues of a projector are all either $0$ or $1$.

\textbf{c) Positive: }  Another class of Hermitian operators are known as positive operators.  An operator $A$ is said to be a positive if for every vector $\ket{v}$ we have $\braket{v|A|v} \geq 0$.  An even stricter case is that an operator is said to be positive definite if $\braket{v|A|v} > 0$ for every non-zero vector $\ket{v}$.  An interesting property is that if $A$ is any operator, then $A^{\dagger}A$ is positive.

\textbf{d) Unitary: } A operator $U$ is unitary if $U U^{\dagger}=U^{\dagger}U=I$.  Alternatively an operator is unitary if and only if each of its matrix representations are unitary matrices.  Furthermore, all the eigenvalues of a unitary matrix take the form $e^{i\theta}$ for some real $\theta$.  Of importance is the result that any unitary operator $U$ can be formulated as 
\begin{equation}\label{unihermi}
U=\exp(iA)
\end{equation}
for some Hermitian operator $A$.  Aside from the algebraic properties, unitary operators are geometrically significant in that they preserve the inner product between vectors; as an example the inner product between  $U\ket{v}$ and $U\ket{w}$ is computed as $\braket{v|U^{\dagger}U|w} = \braket{v|I|w} = \braket{v|w}$.

\textbf{e) Normal: } An operator $A$ is said to be normal if $AA^{\dagger}=A^{\dagger}A$.  Both Hermitian and unitary operators are normal.  One of the most important results in linear algebra is the spectral decomposition theorem:
\begin{theorem}\label{spectral}
	\textbf{(Spectral decomposition)}  Any normal operator $M$ on a vector space $V$ is diagonal with respect to some orthonormal basis for V.  Conversely, any diagonalizable operator is normal (See e.g. \cite{nielsen2002quantum}.)
\end{theorem}\hfill$\Box$

More explicitly, this can be expressed as
\begin{equation}\label{spectralformula}
M = \sum_{i} \lambda_{i} \ket{i}\bra{i},
\end{equation}
where $\lambda_{i}$ are the eigenvalues of $M$ with each $\ket{i}$ signifying the corresponding eigenvector.  Furthermore, the set of eigenvectors form an orthonormal basis for the vector space.  One can also derive the projectors $P_{i}=\ket{i}\bra{i}$ which results in $M = \sum_{i} \lambda_{i} P_{i}$.  The set of projectors in this `spectral expansion' of $M$ satisfy both $\sum_{i} P_{i} = I$ and $P_{i}P_{j} = \delta_{ij}P_{i}$.

\section{Qubits}

The postulates of quantum theory \cite{sakurai1995modern, shankar2012principles} are most commonly framed through state vectors.  An information-theoretic view of these mathematical objects results in the quantum circuit model for qubits. At a coarse level, this framework can be viewed as a quantum analogue of the classical circuit model described in Chapter \ref{chap: classical}.  There are four concepts to the quantum circuit model; we provide a description of each concept, their difference to the classical counterpart, and their inception from the postulates of quantum theory.  We conclude this section with noting implications that portray further distinctions between quantum information and classical information.  

\subsection{Single qubit}

\textbf{a) Description: } A bit can be physically manifested by a classical two state system.  A qubit is a quantum analogue of a bit.  It corresponds to an abstraction, that relates to a classical bit, which can be physically instantiated by a two-level quantum system.  More precisely, a qubit is a unit vector in a two-dimensional Hilbert space which takes the general form,
\begin{equation}\label{qubit}
\ket{\psi} = \alpha\ket{0} + \beta\ket{1},
\end{equation}  
where we have used the computational basis set (\ref{compbasis}), and where $\alpha, \beta \in \mathbb{C}$.  It is these complex numbers that are referred to as \textit{quantum information}.  Given $\braket{\psi|\psi}=1$, known as the normalization condition, it can easily be shown that
\begin{equation}
\lvert \alpha \rvert^{2} + \lvert \beta \rvert^{2}=1.
\end{equation}  
Depending on the values of $\alpha$ and $\beta$, a qubit is in one of the orthogonal computational basis vectors ($\ket{0}$ or $\ket{1}$), or in some linear combination of those vectors (\ref{qubit}) which is referred to as a superposition.  For the former case, a qubit would then map to the notion of a classical bit.  This alludes to the idea that orthogonal vectors can be thought of as the different states of classical information.

\textbf{b) Difference to classical information: } The classical information of a bit, namely $0$ or $1$, can directly correspond to some physical feature of the classical system such as the voltage value of an electrical circuit.  This is in vast contrast to quantum information, such as $\alpha$ and $\beta$ in (\ref{qubit}), which does not have a direct correspondence with the physical properties of the quantum system.  The fundamental mystery \cite{bokulich2010philosophy,infeld1971evolution} is: What do these complex numbers physically represent?  We do not exactly know what quantum information is!  Nevertheless, these values do carry direct experimental consequences.  From a historical view, this problem is known as the issue of the interpretation of quantum mechanics.

\textbf{c) Quantum-theoretic origin: } The relationship between a two-level quantum system and a qubit (\ref{qubit}) stems from a postulate of quantum theory which states that: Associated to any isolated quantum system is a Hilbert space known as the state space of the system; the system is completely described by its state vector (also known as the quantum state), which is a unit vector in the system's state space.  In regards to terminology, if a state vector is represented as $\sum_{i}\alpha_{i}\ket{v_{i}}$ where it is a linear combination of basis states $\ket{v_{i}}$, then the complex coefficients $\alpha_{i}$ are referred to as its probability amplitudes.  The central tenet of quantum theory is that to describe the state of a system, one needs to assign one amplitude for each possible configuration that you would the find the system in upon measuring it.  For the case of a two-level quantum system its state vector (or its quantum state) is adapted as a qubit, and its amplitudes are referred to as its quantum information.  The power of the quantum circuit framework can be seen in that a qubit can be physically instantiated by a diverse range of two-level quantum systems \cite{nielsen2002quantum, flamini2018photonic}.  A few examples include the spin of a spin-$1/2$ particle, the polarization of a photon or the energy levels of a two-state atom.  

\subsection{Multiple qubits}

\textbf{a) Description: } A `string' of qubits is connected by a tensor product structure.  As an example, a two qubit system can be in one of the four computational basis vectors $\ket{0} \otimes \ket{0}$, $\ket{0} \otimes \ket{1}$, $\ket{1} \otimes \ket{0}$, $\ket{1} \otimes \ket{1}$, or in some linear combination of these vectors
\begin{equation}
\ket{\psi} = \alpha_{00}\ket{00} + \alpha_{01}\ket{01} + \alpha_{10}\ket{10} + \alpha_{11}\ket{11}.
\end{equation}
The vector $\ket{\psi}$ satisfies the normalization condition $\sum_{x \in \{0,1\}^{2}}\lvert \alpha_{x} \rvert^{2} = 1$, where $\{0,1\}^{2}$ refers to `the set of strings of length two with each letter being either $0$ or $1$.' More generally for a system of $n$ qubits, the associated state vector $\ket{\psi}$ is referred to as its quantum state, with the computational basis states of the form $\ket{x_{1}x_{2}\dots x_{n}}$ with $x_{i}\in \{0,1\}$.  The number of complex coefficients involved is $2^{n}$ and it is these coefficients that are the quantum information. 

\textbf{b) Difference to classical information: } The superposition property of a qubit provides the key distinction from a classical bit.  Moreover, it has a remarkable consequence for multiple qubits; for a relatively small number of qubits such as $n=500$, the superposition property gives $2^{n}$ values of quantum information.  These are more complex numbers than can be stored on any classical computer that could ever feasibly be built. Fortunately, this exponential relationship between the number of qubits and the amount of quantum information makes quantum systems a compelling platform to design information technologies on.

\textbf{c) Quantum-theoretic origin: } The idea that the tensor product is the appropriate mathematical machinery for multiple qubits comes from a postulate of quantum theory concerning composite systems.  It assumes that the state space of a composite quantum system is the tensor product of the states spaces of the component quantum systems.  Furthermore, if we have systems numbered $1$ through to $n$, and system number $i$ is in state $\ket{\psi_{i}}$, then the joint state vector of the total system is given by $\ket{\psi_{1}} \otimes \ket{\psi_{2}} \otimes \dots \ket{\psi_{n}}$.

\subsection{Transforming qubits}

\textbf{a) Description: }Information as an abstraction is useful when it can be transformed.  Classical gates transform the classical information through the algebra of boolean logic.  The quantum circuit model introduces the concept of a quantum gate as a means to transform quantum information.  The only constraint on the notion of a quantum gate is that it be a unitary operator, $U^{\dagger}U=U^{\dagger}U=I$.  These gates are applied to qubits as operators acting on vectors.  The most important single qubit gates are the Pauli operators.  With respect to basis set (\ref{compbasis}), they represented as
\begin{equation}\label{paulioperators}
\sigma_{x} \equiv \sigma_{1} \equiv X \equiv
\begin{pmatrix}
0 & 1 \\
1 & 0
\end{pmatrix}; \quad
\sigma_{y} \equiv \sigma_{2} \equiv Y \equiv
\begin{pmatrix}
0 & -i \\
i & 0
\end{pmatrix}; \quad
\sigma_{z} \equiv \sigma_{3} \equiv Z \equiv
\begin{pmatrix}
1 & 0 \\
0 & -1
\end{pmatrix}.
\end{equation}
It is also standard to include the identity operator as part of this set which we label as $\sigma_{0}$. In terms of outer products, the Pauli operators are expressed as
\begin{align}
&\sigma_{0} = \ket{0}\bra{0} + \ket{1}\bra{1}, \\
&\sigma_{x} = \ket{0}\bra{1} + \ket{1}\bra{0}, \\
&\sigma_{y} = -i\ket{0}\bra{1} + i\ket{1}\bra{0}, \\
&\sigma_{z} = \ket{0}\bra{0} - \ket{1}\bra{1}. 
\end{align}
The commutators between the different Pauli operators equate to
\begin{equation}
[X,Y] = 2iZ; \quad [Y,Z]=2iX; \quad [Z,X]=2iY.
\end{equation}
The Hadamard gate, phase gate, and $\pi/8$ gate (denoted $T$) are respectively
\begin{equation}\label{hadamard}
H \equiv \frac{1}{\sqrt{2}}
\begin{pmatrix}
1 & 1 \\
1 & -1
\end{pmatrix}; \quad
S \equiv
\begin{pmatrix}
1 & 0 \\
0 & i
\end{pmatrix}; \quad
T \equiv
\begin{pmatrix}
1 & 0 \\
0 & e^{(i\pi/4)}
\end{pmatrix}.
\end{equation}
The Hadamard gate turns the computational basis states into particular superposition states as follows
\begin{equation}
H\ket{0} = \frac{\ket{0} + \ket{1}}{\sqrt{2}} = \ket{+}, \quad H\ket{1} = \frac{\ket{0} - \ket{1}}{\sqrt{2}} = \ket{-}.
\end{equation}
The states $\ket{+}$ and $\ket{-}$ have column vector representation (\ref{plusbasis}).  The quantum gates mentioned so far satisfy the following well known identities
\begin{align}
&XYX=-Y, \\
&HXH=Z, \\
&HYH=-Y, \\
&HZH=X, \\
&H=\frac{(X+Z)}{\sqrt{2}}, \\
&S= T^{2}.
\end{align}
Another important set of quantum gates, which are derived from the Pauli operators, are known as the rotation operators:
\begin{align}
R_{x}(\theta) &\equiv e^{-i\theta X/2} = \cos \frac{\theta}{2} \, I - i \sin \frac{\theta}{2} \, X = \begin{pmatrix}
\cos \frac{\theta}{2}  & -i \sin \frac{\theta}{2} \\
-i \sin \frac{\theta}{2} & \cos \frac{\theta}{2}
\end{pmatrix}, \\
R_{y}(\theta) &\equiv e^{-i\theta Y/2} = \cos \frac{\theta}{2} \, I - i \sin \frac{\theta}{2} \, Y = \begin{pmatrix}
\cos \frac{\theta}{2}  & -\sin \frac{\theta}{2} \\
\sin \frac{\theta}{2} & \cos \frac{\theta}{2}
\end{pmatrix}, \\
R_{z}(\theta) &\equiv e^{-i\theta Z/2} = \cos \frac{\theta}{2} \, I - i \sin \frac{\theta}{2} \, Z = \begin{pmatrix}
e^{-i\theta/2}  & 0 \\
0 & e^{i\theta/2}
\end{pmatrix}.
\end{align}
The significance of these rotation operators is that we can express an arbitrary single qubit quantum gate $U$ as
\begin{equation}
U = e^{i\alpha}R_{z}(\beta)\, R_{y}(\gamma)\, R_{z}({\delta}),
\end{equation} 
for some real numbers $\alpha$, $\beta$, $\gamma$, and $\delta$.  

For the case of two qubits, an important quantum gate is the controlled-NOT operator.  This unitary operator has the matrix representation
\begin{equation}\label{CNOT}
U_{CN} \equiv
\begin{pmatrix}
1 & 0 & 0 & 0 \\
0 & 1 & 0 & 0 \\
0 & 0 & 0 & 1 \\
0 & 0 & 1 & 0
\end{pmatrix}.
\end{equation}
The action of the operator on a quantum state $\ket{a, b}$ is to transform it into $\ket{a, b\oplus a}$ where $\oplus$ denotes addition modulo two.  A more explicit description is this gate acts on two registers where the first qubit is known as the control qubit and the second as the target qubit; if the control qubit is in state $\ket{0}$, then the target qubit is left unchanged; however if the control qubit is in state $\ket{1}$, then an $X$ (NOT) operator is applied to the target qubit.  One can generalize the essence of the controlled-NOT operator to any other gate in that the $X$ operator is replaced by the appropriate gate.

The importance of the controlled-NOT operator can be stated by the result that any multiple qubit quantum gate may be composed from controlled-NOT gates and single qubit gates.

\textbf{b) Difference to classical information: } The mathematical difference between boolean functions and unitary operators is clearly self-evident.  However the non-trivial differences between classical and quantum gates are subtle. Some classical gates such as the NAND gate or the XOR gate are non-invertible; it is not possible to derive the input given the output.  In contrast, all quantum gates are invertible as the inverse of a unitary matrix is also a unitary matrix, hence a valid quantum gate.  In the classical case, the only non-trivial single bit gate is the NOT gate; in the quantum model, we have several important single qubit gates.  It is interesting to note that there are some subtle similarities.  The Pauli $X$ operator can be thought of as a quantum analogue of classical NOT gate since it inverts the computational basis states
\begin{equation}
X\ket{0} = \ket{1}, \quad X\ket{1} = \ket{0}.
\end{equation}

\textbf{c) Quantum-theoretic origin: } In the quantum circuit model, we have seen the use of unitary operators as a means to transform qubits.  It turns out that this is directly connected to a postulate of quantum theory regarding dynamics.  Namely that the continuous time evolution of a state vector of a closed quantum system is governed by the Schr\"{o}dinger equation
\begin{equation}\label{schrodinger}
i\hbar \frac{\partial \ket{\psi(t)}}{\partial t} = H \ket{\psi(t)}.
\end{equation}
The Hamiltonian $H$ is a Hermitian operator which specifies the physics of the system.  The solution to (\ref{schrodinger}) is 
\begin{equation}
\ket{\psi(t_{2})} = \exp \Bigg[ \frac{-iH(t_{2}-t_{1})}{\hbar} \Bigg]\ket{\psi(t_{1})}.
\end{equation}
Using relationship (\ref{unihermi}), one can naturally define a unitary operator
\begin{equation}\label{unitarytransformation}
U(t_{1}, t_{2}) \equiv \exp \Bigg[ \frac{-iH(t_{2}-t_{1})}{\hbar} \Bigg].
\end{equation}
Hence a discrete time transformation of states is provided by unitary operators.

\subsection{Measuring qubits}

\textbf{a) Description: } The final element of the quantum circuit model is measuring the qubits to extract their values.  One way to mathematically represent the measurement of qubits is using any orthonormal bases.  We have seen a qubit (\ref{qubit}) represented using the computational basis states (\ref{compbasis}).  More generally, suppose a qubit is represented using an arbitrary orthonormal basis, $\ket{x}$ and $\ket{y}$
\begin{equation}\label{generalqubit}
\ket{\psi} = \alpha_{x}\ket{x} + \beta_{y}\ket{y},
\end{equation}
where $\alpha_{x}, \beta_{y} \in \mathbb{C}$.  Then by measuring the qubit, with respect to the $\ket{x}$, $\ket{y}$ basis, we find the qubit is in state $\ket{x}$ or $\ket{y}$; we never find it in the superposition state (\ref{generalqubit}); hence measurement is said to instantaneously `collapse' the state into one of the basis states.  Furthermore, the probability of finding the qubit in state $\ket{x}$ is given by the modulus square of its coefficient, $\lvert \alpha_{x} \rvert^{2}$; similarly the probability of finding it in state $\ket{y}$ is given by $\lvert \beta_{y} \rvert^{2}$.  Due to the normalization condition, these `quantum' probabilities (that are derived from quantum information) sum to one.  As an example, if the qubit is in state $\ket{+} = (1/\sqrt{2})\ket{0} + (1/\sqrt{2})\ket{1}$, then the probability of finding it in state $\ket{0}$ upon measurement is $1/2$, and the probability of finding it in state $\ket{1}$ is also $1/2$.  One can generalize this technique to multiple qubits by using an arbitrary orthonormal basis of the respective Hilbert space.

\textbf{b) Difference to classical information: }Unlike classical information, quantum information such as $\alpha_{x}, \beta_{y}$ in (\ref{generalqubit}) can be described as `hidden.'  No single measurement allows us to directly extract those values.  On a related matter, measurement can generally be viewed as a process that converts quantum information into classical information in the form of an orthogonal basis state.  This presents us with another fundamental mystery: Why does this `collapse' occur? Or perhaps can we derive `measurement' from a unitary process.  Like with the first mystery, a deep answer still unknown. This issue from a quantum-theoretic perspective is referred to as the measurement problem \cite{penrose2007road}.

\textbf{c) Quantum-theoretic origin: } In quantum theory, an alternative way of evolving a state forward in time is through the measurement of quantum states (as opposed to the Schr\"{o}dinger equation).  We describe the relevant postulate: Quantum measurements are denoted by a set of measurement operators $\{M_{m}\}$ which satisfy 
\begin{equation}\label{completemeasure}
\sum_{m}M_{m}^{\dagger}M_{m} = I.
\end{equation}
This is known as the completeness relation. These operators act on the relevant state space.  The index $m$ refers to the outcome obtained from the measurement; if the quantum system is in state $\ket{\psi}$, then probability that result $m$ occurs upon measurement is given by 
\begin{equation}
p(m) = \braket{\psi{|M_{m}^{\dagger}M_{m}|\psi}}.
\end{equation} 
These quantum probabilities sum to one
\begin{equation}
\sum_{m} p(m) = \sum_{m}\braket{\psi{|M_{m}^{\dagger}M_{m}|\psi}} = 1
\end{equation}
Furthermore, given result $m$, the post-measurement state can be written as
\begin{equation}
\frac{M_{m}\ket{\psi}}{\sqrt{ \braket{\psi{|M_{m}^{\dagger}M_{m}|\psi}}}}.
\end{equation}
The quantity $e^{i\theta}$ known as a global phase factor, where $\theta \in \mathbb{R}$, is irrelevant with respect to measurement.  The state $\ket{\psi}$ and $e^{i\theta}\ket{\psi}$ are equivalent from the perspective of observation since
\begin{equation}
\braket{\psi|{e^{-i\theta}M_{m}^{\dagger}M_{m}e^{i\theta}|\psi}} = \braket{\psi{|M_{m}^{\dagger}M_{m}|\psi}}.
\end{equation}
A special case of these measurement operators are projective measurements.  Projective measurements satisfy (\ref{completemeasure}) as well as carry the property that $M_{m}$ are Hermitian and that $M_{m}M_{m'}=\delta_{m, m'}M_{m}$.  A projective measurement corresponds to an observable (a physical quantity that can be measured).  An observable is represented as a Hermitian operator $M$ on the state space.  Using spectral decomposition (\ref{spectralformula}), one can state this more precisely as
\begin{equation}\label{projectivemeasure}
M=\sum_{m}m \, P_{m},
\end{equation}
where $P_{m}$ represents the orthogonal projector onto the eigenspace of $M$ with eigenvalue $m$.  The measurement outcomes are the eigenvalues of $M$.  Furthermore, the probability of obtaining result $m$ if system was in state $\ket{\psi}$ before measurement, is given by
\begin{equation}\label{projectiveprob}
p(m) =\braket{\psi|P_{m}|\psi}.
\end{equation}  
This is often referred to as the Born rule.  The post-measurement state after obtaining outcome $m$ becomes
\begin{equation}\label{projectivepost}
\frac{P_{m}\ket{\psi}}{\sqrt{\braket{\psi|P_{m}|\psi}}}.
\end{equation}
It is convention to sometimes not emphasize the observable, but rather focus on the orthogonal projectors in (\ref{projectivemeasure}) or its associated kets.  To be more precise, the phrase `measure in basis $\ket{m}$' means to measure any observable that has $\ket{m}$ as its eigenbasis.  The corresponding projectors are $P_{m}=\ket{m}\bra{m}$ which satisfy $\sum_{m}P_{m}=I$ and $P_{m}P_{m'} = \delta_{m,m'}P_{m}$.  The Born rule can be re-written as
\begin{equation}\label{Bornruleeigenbasis}
p(m) =\lvert \braket{m|\psi}\rvert^{2}.
\end{equation}
Hence it can be seen that every orthonormal basis of the Hilbert space corresponds to a quantum measurement, and has outcome probabilities given by the Born rule.  In the quantum circuit model, we use projective measurements, and therefore equations (\ref{projectiveprob}) and (\ref{projectivepost}) are implicit in our explanation regarding the measurement of qubits.  Furthermore, it is often the case that we measure in the computational basis states (\ref{compbasis}).

\subsection{Further distinctions from classical information}

From our treatment on the qubits, a number of implications arise.  These results further emphasize the non-trivial distinctions between quantum information and classical information. 

\textbf{a) No-cloning: }  An essential task of classical information systems is to copy bits.  For an unknown quantum state, this operation is impossible to carry out.  The no-cloning theorem \cite{dieks1982communication, wootters1982single, mcmahon2007quantum} states that there exists no unitary operator that can clone an unknown quantum state.  To see this as true, suppose such a unitary operator $U$ did exist, and we have two unknown quantum states, $\ket{\psi}$ and $\ket{\phi}$.  Furthermore, let $\ket{s}$ denote a blank state to copy in.  We have mappings of the form
\begin{align}
U(\ket{\psi} \otimes \ket{s}) &= \ket{\psi} \otimes \ket{\psi}, \\
U(\ket{\phi} \otimes \ket{s}) &= \ket{\phi} \otimes \ket{\phi}.
\end{align} 
We then compute the inner product of the left hand side of both equations.  We carry the same task on the right hand side.  By equating the quantities, we obtain 
\begin{equation}
\braket{\psi|\phi} = (\braket{\psi|\phi})^{2}.
\end{equation}
This implies that either the two unknown states are orthogonal ($\braket{\psi|\phi}=0$), or they are equal to each other ($\ket{\psi}=\ket{\phi}$).  Hence a general cloning machine is impossible.  However, it is not surprising that a set of orthogonal states can be copied since these can be viewed as different states of classical information.  The no-cloning theorem is a basic result in quantum information science, and is related to other fundamental constraints in physics such as in regards to closed timelike curves \cite{ahn2013quantum, brun2013quantum}.  Having introduced the no-cloning theorem, it seems appropriate to say that other related no-go theorems exist including a no-deletion theorem \cite{pati2000impossibility}.  Both no-cloning and no-deletion collectively allude to the conservation of quantum information.

\textbf{b) Indistinguishability: } In principle, one can always distinguish classical bits from each other.  This is an impossible task for non-orthogonal quantum states.  As an example, there is no single measurement process that can reliably distinguish states $\ket{0}$ or $\ket{+}$.  We provide a rough argument.  Consider the simpler case of a fixed set of orthogonal quantum states denoted $\ket{\psi_{i}}$.  One can define general measurement operators consisting of $M_{i} \equiv \ket{\psi_{i}}\bra{\psi_{i}}$ as well as the positive square root of $I-\sum_{i \neq 0} \ket{\psi_{i}}\bra{\psi_{i}}$.  In this case, if state $\ket{\psi_{i}}$ is prepared, then $p(i) = \braket{\psi_{i}|M_{i}|\psi_{i}} = 1$.  This implies that this set of states can be reliably distinguished.  If on the other hand $\ket{\psi_{i}}$ denotes a set of non-orthogonal states, then a crucial property is that say state $\ket{\psi_{2}}$ can be broken into a component parallel to say state $\ket{\psi_{1}}$ as well as a component orthogonal to $\ket{\psi_{1}}$.  Due to the component of $\ket{\psi_{2}}$ that is parallel to $\ket{\psi_{1}}$, there is a non-zero probability of mistaking $\ket{\psi_{1}}$ as the state when in fact it was $\ket{\psi_{2}}$ that was prepared.  Thus these non-orthogonal states cannot be reliably distinguished.  This means a measurement has a limit on its ability to exact information which conveys additional support to the notion that quantum information is hidden.  Surprisingly, quantum indistinguishability is also related to the constraints regarding closed timelike curves \cite{brun2009localized}.  

\textbf{c) Uncertainty: }  Unlike classical bits, qubits have a probabilitic property that is intrinsic to them (\ref{projectiveprob}).  One can view these quantum probabilities in terms of an expectation value (\ref{expectation}). More precisely, the expectation value of an observable, $M$, (with respect to quantum state $\ket{\psi}$) in (\ref{projectivemeasure}), is defined as
\begin{align}\label{quantumexpectation}
\langle M \rangle &\equiv \sum_{m}m \, p(m) \\
&= \sum_{m} m \braket{\psi|P_{m}|\psi} \\
&= \bra{\psi}\Bigg(\sum_{m} m \, P_{m}\Bigg) \ket{\psi} \\
&= \braket{\psi|M|\psi}.
\end{align} 
In the derivation we have used quantum probabilities (\ref{projectiveprob}).  We can invoke further classical probabilistic concepts, and introduce the standard deviation (\ref{standarddeviation}) of an observable,
\begin{equation}
\Delta(M) = \sqrt{\langle M^{2} \rangle - \langle M \rangle^{2}}.
\end{equation}
This quantity is also known as the uncertainty of the observable, and it represents a statistical measure of the spread of measurements about the expectation value.  The Heisenberg uncertainty principle states that for observables $A$ and $B$ we have
\begin{equation}\label{heisenberguncertainty}
\Delta(A)\Delta(B) \geq \frac{\lvert \bra{\psi} [A,B] \ket{\psi}\rvert}{2}.
\end{equation}
This result has a shocking implication; suppose we had two observables that do not commute and we performed a large number of these measurements on systems which are in identical states $\ket{\psi}$; then if we make the uncertainty on the results of $B$ decrease, then the uncertainty of the results of $A$ must increase, regardless of the sophistication of the measurement.  The uncertainty principle is also related to constraints regarding closed timelike curves \cite{pienaar2013open}.

\section{Density Operators}

Density operators are a widely used mathematical tool in the study of open quantum systems and quantum statistical mechanics \cite{gardiner}.  Within quantum information science, it provides a natural framework for quantifying the information concerning subsystems.  In this section, we briefly describe quantum information in the language of density operators, while relaying its relationship to the quantum circuit model.

\subsection{Single density operator}

Associated to any isolated quantum system is a Hilbert space with an operator known as the density operator, $\rho$, which acts on the space.  The density operator is a positive operator with $\text{tr}(\rho)=1$.  The relationship to the quantum circuit model is as follow:  If a system is known to be in state $\ket{\psi_{i}}$ with associated classical probability $p_{i}$ then the density operator of the system is given by 
\begin{equation}
\rho = \sum_{i}p_{i}\rho_{i},
\end{equation} 
where $\rho_{i}\equiv \ket{\psi_{i}}\bra{\psi_{i}}$.  The set $\{\ket{\psi_{i}}, p_{i}\}$ is referred to as an ensemble.  For the limited case where $\ket{\psi}$ is the only member of an ensemble (like in the circuit model), we have $\rho=\ket{\psi}\bra{\psi}$ which is then called a pure state.  Otherwise it is known as a mixed state, which means we do not know with certainty what quantum state it is in.  In terms of computations, a pure state satisfies 
\begin{equation}
\text{tr}(\rho^{2})=1,
\end{equation}
whereas a mixed state results in 
\begin{equation}
\text{tr}(\rho^{2})<1.
\end{equation}
The largest statistical ignorance is expressed by the maximally mixed state 
\begin{equation}\label{completelymixedstate}
\rho = (1/n)\, I_{n},
\end{equation}
where $n$ is the dimension of the Hilbert space.  The density operator is a broader framework than the quantum circuit model.  It captures both the quantum information with its quantum probabilities, as well as the classical probabilities related to our ignorance of that quantum information.

\subsection{Multiple density operators}

\textbf{a) Composite system: } In the quantum circuit model, we have employed the tensor product as a means to describe multiple qubits.  Similarly, in the density operator language, a composite system is represented as a tensor product of the Hilbert spaces of the component systems.  If system number $i$ is in state $\rho_{i}$, then the composite system is described by density operator 
\begin{equation}
\rho_{1} \otimes \rho_{2} \otimes \dots \rho_{n}.  
\end{equation}

\textbf{b) Subsystem: } A central role of the density operator in quantum information science is as an information tool to describe subsystems.  This particular task is carried out by the reduced density operator.  Suppose $\rho^{AB}$ describes a composite system made up of system $A$ and system $B$.  Then the reduced density operator for $A$ is defined as
\begin{equation}\label{reduceddensityoperator}
\rho^{A} \equiv \text{tr}_{B}(\rho^{AB}),
\end{equation}  
where $\text{tr}_{B}$ is known as the partial trace over system $B$.  The partial trace is defined as
\begin{equation}
\text{tr}_{B}(\ket{a_{1}}\bra{a_{2}} \otimes \ket{b_{1}}\bra{b_{2}}) \equiv \ket{a_{1}}\bra{a_{2}}\text{tr}(\ket{b_{1}}\bra{b_{2}})
\end{equation}
where we have used the usual trace operation on the right hand side; the vectors $\ket{a_{1}}$ and $\ket{a_{2}}$ are any vectors from the Hilbert space associated to system $A$; similarly $\ket{b_{1}}$ and $\ket{b_{2}}$ are any vectors from the Hilbert space associated to system $B$.  As an example, one can apply this operation to the trivial composite system $\rho^{AB} =\rho_{1} \otimes \rho_{2}$, and obtain as expected
\begin{equation}
\rho^{A} = \text{tr}_{B}(\rho_{1} \otimes \rho_{2}) = \rho_{1} \text{tr}(\rho_{2}) = \rho_{1}.
\end{equation}  

\subsection{Transforming density operators}

\textbf{a) Closed system: }  Like in the quantum circuit model, density operators associated with closed systems also transform according to unitary operators.  This is computed as
\begin{equation}\label{closeddensity}
\rho' = U \rho U^{\dagger},
\end{equation}
where $U$ is a unitary operator.

\textbf{b) Open system: } For open quantum systems, a generalized framework for dynamics known as quantum operations is employed.  The density operators transform as 
\begin{equation}\label{quantumoperation}
\rho' = \Phi(\rho) = \sum_{k} A_{k} \rho A_{k}^{\dagger},
\end{equation}
where $\Phi$ is known as a quantum operation.  The operators $A_{k}$ are known as operation elements or as the Krauss operators.  These are not necessarily unitary, but rather satisfy the condition  
\begin{equation}
\sum_{k}A_{k}^{\dagger}A_{k}\leq I.
\end{equation}
The mapping (\ref{closeddensity}) can be regarded as a quantum operation where $\Phi(\rho) = U \rho U^{\dagger}$.  But the utility of the framework is best captured when considering open systems such in the case of environmental noise on a qubit.  Suppose a qubit flips from $\ket{0}$ to $\ket{1}$ (or vice versa) with classical probability $1-p$.  The associated operation elements are $A_{0} = \sqrt{p}I$, and $A_{1} = \sqrt{1-p}X$.  The quantum operation, known as the bit flip channel, is written as
\begin{equation}
\Phi(\rho) = p \rho + (1-p)X\rho X.
\end{equation}
More generally, quantum operations have been used to quantify a broad range of noise-related phenomenon on qubits.

\subsection{Measuring density operators}

\textbf{a) General measurement: } The measurement postulates of quantum theory can be reformulated for density operators.  Quantum measurements are described by a collection of measurement operators $\{M_{m}\}$ which satisfy the completeness relation (\ref{completemeasure}).  If the quantum system is in state $\rho$ before measurement, then the probability of obtaining result $m$ is given by 
\begin{equation}\label{generalmeasuredensityprob}
p(m) = \text{tr}(M_{m}^{\dagger}M_{m} \rho).
\end{equation} 
The post-measurement operator is expressed as
\begin{equation}\label{generalmeasuredensitypost}
\frac{M_{m}\rho M_{m}^{\dagger}}{\text{tr}(M_{m}^{\dagger}M_{m} \rho)}.
\end{equation}

\textbf{b) POVM: } POVM stands for positive operator valued measure and is a formalism that is usually expressed with the language of density operators.  Given general measurement operators $\{M_{m}\}$, one can define 
\begin{equation}\label{povmgeneral}
E_{m} \equiv M_{m}^{\dagger}M_{m}.
\end{equation}
These positive operators, $E_{m}$, are known as the POVM elements.  It can be shown that
\begin{equation}
\sum_{m} E_{m} = I.
\end{equation}
If the density operator prior to measurement is denoted $\rho$, then probability of obtaining outcome $m$ is given by
\begin{equation}\label{povmprob}
p(m) = \text{tr}(E_{m} \rho).
\end{equation}
One example of a POVM are projection measurements which are described by projectors, $P_{m}$, and satisfy $E_{m} =  P_{m}^{\dagger}P_{m} = P_{m}$.  For this specific case, (\ref{povmprob}), equates to
\begin{equation}\label{povmproprob}
p(m) = \text{tr}(P_{m} \rho).
\end{equation}
which is just a reformulation of the Born rule (\ref{projectiveprob}).

\textbf{c) Tomography: }  We have seen that a single measurement on a qubit does not allow us to obtain the quantum information.  This means in general it is impossible to characterize an unknown state $\rho$ if we are given a single copy.  Quantum state tomography is a procedure to estimate the unknown quantum state with many measurements.  Suppose we have many copies of the density operator $\rho$ of an unknown qubit.  Using the Pauli operators (\ref{paulioperators}), one can express $\rho$ as    
\begin{equation}\label{quantumtomography}
\rho = \frac{\text{tr}(\rho)I + \text{tr}(X\rho)X + \text{tr}(Y\rho)Y + \text{tr}(Z\rho)Z}{2}.
\end{equation}
For large sample sizes, one can obtain a reasonable estimation of the values of $\text{tr}(X\rho)$, $\text{tr}(Y\rho)$ and $\text{tr}(Z\rho)$ and identify the quantum information of the qubit.  Generalizing this procedure to $n$ qubits results in the expression,
\begin{equation}\label{quantumtomographyN}
\rho = \sum_{\vec{v}}\frac{\text{tr}(\sigma_{v_{1}} \otimes \sigma_{v_{2}} \otimes  \cdots \otimes \sigma_{v_{n}}\rho)  \sigma_{v_{1}} \otimes \sigma_{v_{2}} \otimes  \cdots \otimes \sigma_{v_{n}}}{2^{n}},
\end{equation}
where $\vec{v} = (v_{1}, \dots, v_{n})$ with entries $v_{i}$ chosen from the set $0,1,2,3$

\subsection{Further properties}

\textbf{a) No-broadcasting: }  The no-broadcast theorem \cite{Broadcast96}  generalizes the no-cloning theorem to the case of mixed states.  It states that given state $\rho_{1}$, it is not possible to create a composite system $\rho^{AB}$ such that $\text{tr}_{A}(\rho^{AB}) = \rho_{1}$ and $\text{tr}_{B}(\rho^{AB}) = \rho_{1}$.

\textbf{b) Antidistinguishability: }  In the previous section, we looked at the general case of distinguishing non-orthogonal quantum states.  In terms of density operators, distinguishability can be stated as the existence of a POVM $E_{j}$ for set of states $\rho_{k}$ such that 
\begin{equation}
\text{tr}(E_{j}\rho_{k}) = \delta_{ij},
\end{equation}
for all $j$ and $k$.  A related property is the notion of antidistinguishability \cite{caves2002conditions, heinosaari2018antidistinguishability}.  A set of states $\rho_{k}$ is antidistinguishable if there exists a POVM $E_{j}$ such that for each $j$,  
\begin{equation}\label{antidisting}
\text{tr}(E_{j}\rho_{k}) = 0.
\end{equation} 
Distinguishability lets us know that a particular state was definitely prepared.  This in contrast to antidistinguishability which lets us know that a particular state was definitely not prepared.

\textbf{c) Distance measures: } To quantitatively capture the idea of how `close' two quantum states are, there are two useful tools that we proceed to describe.  The first is the trace distance between two density operators $\rho$ and $\sigma$, which is defined as
\begin{equation}
D(\rho, \sigma) \equiv \frac{1}{2} \text{tr}\lvert \rho - \sigma \rvert,
\end{equation}
where
\begin{equation}
\lvert A \rvert \equiv \sqrt{A^\dagger A}.
\end{equation}
The second method is known as the fidelity which is given by
\begin{equation}\label{fidelity}
F(\rho, \sigma) \equiv \text{tr}\sqrt{\rho^{1/2}\sigma \rho^{1/2}}
\end{equation}
for density operators $\rho$ and $\sigma$.  The fidelity is invariant under unitary transformations
\begin{equation}
F(U\rho U^{\dagger}, U\sigma U^{\dagger} ) = F(\rho, \sigma).
\end{equation}
If both density operators represent pure states, $\rho = \ket{\psi_{\rho}}\bra{\psi_{\rho}}$ and $\sigma = \ket{\psi_{\sigma}}\bra{\psi_{\sigma}}$, the fidelity reduces to
\begin{equation}\label{fidelitypure}
F(\rho, \sigma) = \lvert \braket{\psi_{\rho}|\psi_{\sigma}} \rvert^{2} 
\end{equation}
The quantity (\ref{fidelitypure}) measures the probability of confusing the two states if one is only able to carry out only one measurement on one system which is prepared in one of the two states.  If the two states are orthogonal, then the fidelity is computed to be zero and the states can be fully distinguished.

\section{Entropy}

An alternative approach to view quantum information science is based on entropy.  In chapter \ref{chap: classical}, we introduced the Shannon entropy of a random variable as a means of describing classical information.  In this section, we define the von Neumann entropy of a quantum density operator.  A limited perspective is that the Shannon entropy applies only in the classical realm, whereas the von Neumann entropy strictly conveys quantum properties.  

Rather in the modern setting of quantum information, we'll see that the Shannon entropy can employed with respect to classical probabilities (in a mixed state) as well quantum probabilities (derived from quantum information).  Moreover, the von Neumann entropy can represent classical ignorance (in the case of a mixed state) as well as signify a reliable storage of quantum information (through the quantum analogue of data compression).  All of the material in this section reformulates or builds on concepts seen in the previous sections.

\subsection{Indistinguishability using Shannon entropy}

Although we have treated quantum indistinguishability in the qubit and the density operator frameworks, a description through entropy is most insightful \cite{han2018residual}.  

\textbf{a) Scenario: }Suppose a state is prepared from an ensemble $\xi$ of density operators $\{\rho_{j}: j= 1, 2, \dots, N\}$ with a prior classical probability distribution $\{\eta_{j}: j= 1, 2, \dots, N\}$.  Hence, the resulting operator can be written as
\begin{equation}\label{residual0}
\rho = \sum_{j=1}^{N}\eta_{j}\rho_{j},
\end{equation}

with $\text{tr}(\rho)=1$.  The task of distinguishability is to identify which state was prepared through a single measurement.  We perform this measurement using POVM elements which we denote by $\Pi \equiv \{\Pi_{k}: k= 1, 2, \dots, M \}$, where $M \geq N$.

\textbf{b) Probabilistic quantities: }To develop an entropic model of this task, we proceed to derive several quantities.  The joint probability that the state $\rho_{j}$ is prepared and that the outcome obtained is $\Pi_{k}$, is given by
\begin{equation}\label{residual1}
P(\rho_{j}, \Pi_{k}) = \eta_{j} \, \text{tr}(\Pi_{k}\rho_{j}),
\end{equation}
where we have used (\ref{povmprob}).  The total probability of obtaining outcome $\Pi_{k}$ is computed as
\begin{equation}
P_{\Pi_{k}} = \sum_{j=1}^{N}\eta_{j} \, \text{tr}(\Pi_{k}\rho_{j}) = \text{tr}(\Pi_{k}\rho).
\end{equation}
Summing over $k$ in (\ref{residual1}) results in 
\begin{equation}
\sum_{k=1}^{M} P(\rho_{j}, \Pi_{k}) = \eta_{j}.
\end{equation}
\textbf{b) Entropic quantities: }The quantities $\eta_{j}$ signify a probability distribution in (\ref{residual0}).  Therefore, we can evaluate the Shannon entropy (\ref{Shannon}) of this distribution,
\begin{equation}\label{residual4}
H(\xi) = -\sum_{j=1}^{N}\eta_{j} \log_{2} \eta_{j}.
\end{equation} 
Recall the mutual information (\ref{mutualinformation}) and its property $H(X:Y)=H(X) - H(X|Y)$.  Using (\ref{residual4}), we have the following mutual information associated to the measurement process
\begin{equation}\label{residualmutual}
H(\xi : \Pi) = H(\xi) - \sum_{k=1}^{M} P_{\Pi_{k}} H(\xi | \Pi_{k}).
\end{equation} 
It quantifies how much information is gained about inferring the state that was prepared through the measurement.  Moreover, the quantity $H(\xi | \Pi_{k})$ signifies the conditional entropy (\ref{conditionalentropy}) of the remaining ignorance after outcome $\Pi_{k}$ is obtained.  Therefore, a reasonable goal for this task is to choose a measurement that maximizes $H(\xi : \Pi)$.

\textbf{c) Accessible information: }Of crucial importance is the accessible information which is defined as the maximum mutual information attainable over all possible POVM measurements,
\begin{equation}
I_{\text{acc}} = H(\xi) - \underset{\text{min}}{{\text{all}  \, \Pi}}\sum_{k=1}^{M} P_{\Pi_{k}} H(\xi | \Pi_{k}). 
\end{equation}
The accessible information is a marker of how well a measurement can do at identifying the state prepared.  Moreover, it has an upper bound known as the Holevo bound \cite{nielsen2002quantum}.  From this point of view, the accessible information quantitatively captures the notion that quantum information has a hidden nature.

\textbf{d) Subsequent measurements: } Suppose measurement outcome $\Pi_{k}$ is obtained. After our first measurement, there may be subsequent measurements performed to extract further accessible information.  To compute the relevant entropic quantity, recall (\ref{povmgeneral}); each of the POVM elements corresponds to a general measurement operator $M_{k}$, where $\Pi_{k} = M_{k}^{\dagger}M_{k}$.  Then with respect to state $\rho_{j}$, the normalized postmeasurement states (\ref{generalmeasuredensitypost}) are written as
\begin{equation}\label{subsequent}
\rho_{j}^{(k)} = \frac{M_{k}\rho_{j}M_{k}^{\dagger}}{\text{tr}(\Pi_{k}\rho_{j})}.
\end{equation}
The respective new probabilities (\ref{generalmeasuredensityprob}) are found as
\begin{equation}\label{subsequent2}
\eta_{j}^{(k)} = \frac{\eta_{j}\text{tr}(\Pi_{k} \rho_{j})}{\text{tr}(\Pi_{k}\rho)}.
\end{equation}
Moreover, we let $\xi^{(k)}$ denote the postmeasurement ensemble consisting of states (\ref{subsequent}) with respective probabilities (\ref{subsequent2}).  The Shannon entropy of $\xi^{(k)}$ using (\ref{subsequent2}) is equal to value of $H(\xi | \Pi_{k})$ in (\ref{residualmutual}).  

If one performs a optimal subsequent POVM on $\xi^{(k)}$, this reduces the remaining ignorance for distinguishability in $\xi$ to $H(\xi^{(k)}) - I_{\text{acc}}(\xi^{k})$.  Hence the maximum mutual information between the original ensemble $\xi$ and the outcomes of optimal subsequent measurements is given by 
\begin{equation}
I'_{\text{max}}(\xi, \Pi) = H(\xi) - \sum_{k=1}^{M} P_{\Pi_{k}} [H(\xi^{(k)}) - I_{\text{acc}}(\xi^{k})].
\end{equation} 

\textbf{e) Efficiency of a measurement: } Using the computed quantities, one can characterize a quantum measurement using the following framework.  The amount of extracted information from a measurement is defined as 
\begin{equation}
\bar{E} \equiv \frac{H(\xi : \Pi)}{I_{\text{acc}}(\xi)}.
\end{equation}
The residual information is defined as the information that can be potentially extracted from subsequent measurements
\begin{equation}
\bar{R} \equiv \frac{I'_{\text{max}}(\xi, \Pi) - H(\xi : \Pi)}{I_{\text{acc}}(\xi)}.
\end{equation}
This leaves us with a definition of the destroyed information, which quantifies the reduction of the accessible information due to measurement $\Pi$:
\begin{equation}
\bar{D} \equiv \frac{I_{\text{acc}}(\xi) - I'_{\text{max}}(\xi, \Pi)}{I_{\text{acc}}(\xi)}
\end{equation}
The conservation of the total accessible information can thus be expressed as
\begin{equation}
\bar{E} + \bar{R} + \bar{D} = 1.
\end{equation}
For the task of distinguishability, these entropic quantities express the idea of the `efficiency' of a single quantum measurement.

\subsection{Uncertainty using Shannon entropy}

In the context of a large number of measurements, an unavoidable consequence of quantum information is the Heisenberg uncertainty principle (\ref{heisenberguncertainty}).  However from an information-theoretic perspective, the entropy is a preferred quantity over the standard deviation to measure uncertainty.  Indeed, it can be seen therefore, that the uncertainty principle can be reformulated in terms of the Shannon entropy \cite{deutsch1983uncertainty, kraus1987complementary, maassen1988generalized, coles2017entropic}.  

\textbf{a) Entropic uncertainty relation: } In the uncertainty principle, the standard deviation of observables, $X$ and $Z$, must satisfy (\ref{heisenberguncertainty}).  Using the spectral expansion (\ref{spectralformula}), one obtains the corresponding eigenvectors and their eigenvalues
\begin{align}
X &= \sum_{x} x \ket{x}\bra{x}, \\ 
Z &= \sum_{z} z \ket{z}\bra{z}. 
\end{align}
Suppose we measure either one of these observables on a system represented by density operator $\rho$.  Through (\ref{povmproprob}), one obtains a distribution for the quantum probabilities, denoted $p(x)$, associated with the measurement of $X$; likewise, one obtains a probability distribution, denoted $q(z)$, associated with the measurement of $Z$.  The Shannon entropy (\ref{Shannon}) is a function of only a probability distribution.  Hence it is not too difficult to see that we can derive an entropy from $p(x)$, as well as entropy from $q(z)$; these are respectively labelled $H(X)$ and $H(Z)$.  The entropic uncertainty relation states that
\begin{equation}\label{entropicuncertainty}
H(X) + H(Z) \geq \log \frac{1}{c},
\end{equation}
where $c$ is the maximum value of the possible quantities, $c_{xz} = \lvert \braket{x|z} \rvert^{2}$.  Moreover, for a system with an associated Hilbert space of dimension $d$ we have the following bounds,
\begin{equation}\label{entropicuncertaintybounds}
0 \leq \text{log }\frac{1}{c} \leq \text{log }d.
\end{equation}

\textbf{b) Guessing game: } One can view the entropic uncertainty relation through the lens of a game.  Suppose we have two players whom we name Alice and Bob.  The initial role of Bob is to prepare a system in state $\rho$, and send it to Alice.  Alice proceeds to measure either observable $A$ or $B$ with equal probability; the measurement choice is stored in bit $\Theta$ whereas the outcome is stored in bit $K$. The final step of the game is that Alice reveals the choice $\Theta$ to Bob.  The aim of the game is for Bob to guess $K$, given the value of $\Theta$.  

It can be shown \cite{coles2017entropic}, that regardless of the state $\rho$ prepared, the entropic uncertainty relation (\ref{entropicuncertainty}) implies that Bob will not be able to perfectly guess $K$ if $\log (1/c) > 0$.

\textbf{c) Temporal version: } Recently \cite{coles2019entropic}, it was shown that an entropic uncertainty relation can be formulated for energy and time.  The Hamiltonian, $H$, in (\ref{unitarytransformation}) corresponds to the energy of a system.  However, capturing the temporal aspect is non-trivial as there does not exist a Hermitian time operator.  Hence, an entropic uncertainty relation was formulated through the construction of a `quantum clock.' The uncertainty about time corresponds to how well one can `read off' the time from measuring this clock. 

It would be illuminating to view this in terms of a guessing game.  Bob prepares a quantum clock in state $\rho$.  He then sends this to Alice.  In this modified scenario, Alice either measures the clock's energy or randomly sets the clock's time; the choicce is made with equal probability; the latter task is accomplished by applying $\exp(-iHt)$ using a random chosen $t$ from a set of values.  Depending on what Alice chose to do, Bob's task is either to guess the clock's energy or estimate the value of $t$ by reading the clock.  The entropic energy-time uncertainty relation limits Bob's ability to win this guessing game.

\subsection{The von Neumann entropy}

We have witnessed the application of the Shannon entropy in settings involving quantum information.  An alternative entropic tool is the von Neumann entropy.  The usual treatment of this quantity is found in the subject of quantum statistical mechanics.  The approach taken by quantum information science is to describe this quantity in relation to the concepts in classical information theory.

\textbf{a) Single system: } The von Neumann entropy of a quantum density operator $\rho$ is defined as
\begin{equation}\label{vonneumann}
S(\rho) \equiv -\text{tr}(\rho \log_{2} \rho).
\end{equation}
One finds that this entropy can re-written as
\begin{equation}\label{vonneumanneigen}
S(\rho) = -\sum_{x}\lambda_{x} \log_{2} \lambda_{x},
\end{equation}
where $\lambda_{x}$ are the eigenvalues of $\rho$.  With the latter form, $S(\rho)$ can be seen as a Shannon entropy (\ref{Shannon}) where the eigenvalues are substituted for the probabilities.  We also take the convention that $0 \log_{2} 0 \equiv 0$.  

The bounds of the von Neumann entropy are $0 \leq S(\rho) \leq \log_{2}d$, where $d$ is the dimension of the Hilbert space.  Moreoever, the case of $S(\rho) = 0$ corresponds to a pure state, whereas for a completely mixed state (\ref{completelymixedstate}) we have $S(\rho) = \log_{2}d$.  Hence a non-zero von Neumann entropy signifies an ignorance (through classical probabilities) as to what the state of the system is.

\textbf{b) Multiple systems: } Suppose we have composite system with two components denoted $A$ and $B$.  This system is collectively described by density operator $\rho^{AB}$.  Analogous to (\ref{(joint)}), we define the von Neumann joint entropy of this system as
\begin{equation}
S(\rho^{AB}) \equiv -\text{tr}(\rho^{AB}\log_{2}(\rho^{AB})).
\end{equation}
Following (\ref{conditionalentropy}), we can define the von Neumann conditional entropy as 
\begin{equation}\label{quantumconditionalentropy}
S(A|B) \equiv S(A,B) - S(B).
\end{equation}
In classical communications, the quantity $H(X|Y)$ can interpreted as the number of additional bits that need to be transmitted to have full knowledge of $X$, after knowing $Y$.  In an analogous manner, it was recently \cite{witten2018mini} shown that $S(A|B)$ can be interpreted as a number of qubits that needs to be transmitted to make the task of quantum teleportation (which we'll discuss in the next chapter) possible.    
The von Neumann mutual information is defined as 
\begin{equation}
S(A:B) \equiv S(A) + S(B) - S(A,B),
\end{equation}
and resembles the form of (\ref{mutualinformation}).  Furthermore, it can shown that
\begin{align}
S(A:B) &= S(A) - S(A|B) \\
&= S(B) - S(B|A).
\end{align}
By considering (\ref{relativeentropy}), we are then led to define the von Neumann relative entropy (of $\rho$ to $\sigma$) as
\begin{equation}\label{quantumrelativeentropy}
S(\rho || \sigma) \equiv \text{tr}(\rho \log_{2} \rho) - \text{tr}(\rho \log_{2} \sigma),
\end{equation}
where it can be derived that $S(\rho || \sigma) = 0$ if and only if $\rho = \sigma$.

\textbf{c) Transformation: } For a density operator, recall that a unitary transformation is given by
\begin{equation}
\rho' = U\rho U^{\dagger}.
\end{equation}
The von Neumann entropy is invariant under this unitary transformation, hence
\begin{equation}
S(\rho') = S(U\rho U^{\dagger}).
\end{equation}

\textbf{d) Measurement: }  Suppose we have a system in state $\rho$ that we would like to perform a projective measurement on.  Let $P_{i}$ be the corresponding orthogonal projectors for that measurement.  If we never learn the result of the measurement, the post-measurement state can be represented as
\begin{equation}
\rho' = \sum_{i} P_{i}\rho P_{i}.
\end{equation}
It can be shown that this procedure in general increases the entropy,
\begin{equation}
S(\rho') \geq S(\rho),
\end{equation}
with equality if and only if $\rho = \rho'$.

\textbf{e) Properties: } The first important property regarding von Neumann entropies is the subadditivity inequality 
\begin{equation}
S(A,B) \leq S(A) + S(B).
\end{equation}
A related property is the triangle inequality which is written as
\begin{equation}
S(A,B) \geq \lvert S(A) - S(B) \rvert. 
\end{equation}
Of considerable importance is the strong subadditivity inequality 
\begin{equation}\label{strongsubadditivity}
S(A, B, C) + S(B) \leq S(A,B) + S(B,C)
\end{equation}
which applies for a system composed of three components denoted $A$, $B$, and $C$.  For the conditional entropy associated to a trio of systems, we have the result 
\begin{equation}
S(A|B,C) \leq S(A,B).
\end{equation}
In regards to mutual information, one finds that
\begin{equation}
S(A:B) \leq S(A:B,C).
\end{equation}
The monotocity of the relative entropy is a result regarding subsystems
\begin{equation}
S(\rho^{A} || \sigma^{A}) \leq S(\rho^{AB}|| \sigma^{AB}),
\end{equation}
where $\rho^{AB}$ and $\sigma^{AB}$ are any two density operators of a system $AB$.    Another significant result is the concavity of the von Neumann entropy which is expressed as
\begin{equation}
S\Bigg(\sum_{i} p_{i} \rho_{i} \Bigg) \geq \sum_{i}p_{i}S(\rho_{i})
\end{equation}
for probabilities $p_{i}$ (which sum to unity) and their corresponding density operators $\rho_{i}$.

\subsection{Quantum data compression}

Data compression plays a fundamental role in classical information theory; the noiseless channel coding theorem (Theorem \ref{Shannoncoding}) forms the basis for an operational definition of the Shannon entropy.  In this subsection, we provide a brief overview of the quantum noiseless channel coding theorem \cite{schumacher1995quantum}, which provides an operational definition of the von Neumann entropy.  A large part of the development towards the theorem relies on the mathematical machinery associated with classical data compression.  However, the pioneering nature of the work stems from the conceptual shift of treating the states of quantum theory as information in the most genuine manner.  Hence, the significance of this quantum coding theorem cannot be understated for the development of quantum information theory, which is also referred to as the quantum Shannon theory \cite{wilde2017quantum}. 

\textbf{a) Defining an information source: }  As in the classical case, the first step is to construct a valid notion of an information source.  We define a \textit{i.i.d} quantum information source,$\{H, \rho\}$, as one that can described by a Hilbert space $H$, and a density operator $\rho$ on that Hilbert space.  Furthermore, we utilize the framework of quantum operations (\ref{quantumoperation}) to help us define a compression scheme of rate $R$.  The compression operation, $C^{n}$ maps states in $H^{\otimes n}$ to states in a $2^{nR}$-dimensional state space.  Conversely, $D^{n}$ represents a decompression operation which takes states in the compressed space back to states in the original Hilbert space.

\textbf{b) Defining typical states: }  It will be necessary to recall the definition of a typical sequence (\ref{typical2}) that was described in Chapter \ref{chap: classical}.  To harness this result, we note that the density operator associated with our information source has a spectral expansion
\begin{equation}
\rho = \sum_{x} p(x) \ket{x}\bra{x},
\end{equation}
where $\ket{x}$ are the eigenvectors with associated eigenvalues $p(x)$.  Of crucial importance is that the eigenvalues, in this case, behave like a probability distribution in that they are non-negative and sum to unity.  Thus, $S(\rho)$ can be viewed as the Shannon entropy of the set of eigenvalues.  Therefore, by using (\ref{typical2}) we obtain the $\epsilon$-typical sequence $x_{1}, x_{2}, \dots, x_{n}$ where
\begin{equation}\label{quantumtypical2}
\biggl\lvert \frac{1}{n} \log \frac{1}{p(x_{1})p(x_{2}) \ldots , p(x_{n})} - S(\rho) \, \biggl\rvert  \leq \, \epsilon.
\end{equation}
We define an $\epsilon$-typical state $\ket{x_{1}}\ket{x_{2}}\dots\ket{x_{n}}$ as one for which $x_{1}, x_{2}, \dots, x_{n}$ is an $\epsilon$-typical sequence. Related to this concept is the definition of an $\epsilon$-typical subspace, denoted $T(n, \epsilon)$; this is a subspace spanned by all $\epsilon$-typical states, $\ket{x_{1}}, \dots, \ket{x_{n}}$.  Moreover, to project onto the subspace $T(n, \epsilon)$, we can use the operator,
\begin{equation}
P(n, \epsilon) = \sum_{x \,\, \epsilon-\text{typical}} \ket{x_{1}}\bra{x_{1}} \otimes \ket{x_{2}}\bra{x_{2}} \otimes \dots \ket{x_{n}}\bra{x_{n}}.
\end{equation}  

\textbf{c) Application of Theorem \ref{typicaltheorem}: } One can use the classical theorem regarding typical sequences to prove the following quantum theorem: 

\begin{theorem}\label{typicalsubspacetheorem}
	\textbf{(Typical subspace theorem)}   
	\begin{enumerate}[noitemsep, topsep=0pt, label=\roman*)]
		\item Fix $\epsilon > 0$.  Then for any $\delta > 0$, for sufficiently large $n$, 
		\begin{equation}
		\text{tr}(P(n, \epsilon)\rho^{\otimes n}) \geq 1-\delta.
		\end{equation}
		\item For any fixed $\epsilon > 0$ and $\delta > 0$, for sufficiently large $n$, the dimension of the subspace, $\lvert T(n,\epsilon) \rvert = \text{tr}(P(n, \epsilon))$, satisfies  
		\begin{equation}
		(1-\delta)\, 2^{n(S(\rho) - \epsilon)} \leq \lvert T(n,\epsilon) \rvert \leq 2^{n(S(\rho) + \epsilon)}.
		\end{equation}
		\item Let $S(n)$ be a projector onto any subspace of $H^{\otimes n}$ of dimension at most $2^{nR}$, where $R<S(\rho)$ is fixed.  Then for any $\delta > 0$ and for sufficiently large $n$, 
		\begin{equation}
		\text{tr}(S(n)\rho^{\otimes n}) \leq \delta.  
		\end{equation}
	\end{enumerate}  
	(See e.g. \cite{nielsen2002quantum}.)\hfill$\Box$
\end{theorem}

\textbf{d) Application of typical subspace theorem: }  The utility of Theorem \ref{typicalsubspacetheorem} manifests by its use in proving the quantum analogue of noiseless channel coding theorem (Theorem \ref{Shannoncoding}).  For the sake of briefness, we simply state the end result:

\begin{theorem}\label{Schumachercoding}
	\textbf{(Schumacher's noiseless channel coding theorem)} 
	Let $\{H, \rho\}$ be an \textit{i.i.d.} quantum information source:
	\begin{enumerate}[noitemsep, topsep=0pt, label=\roman*)]
		\item  If $R>S(\rho)$ then there exists a reliable compression scheme of rate $R$ for the information source
		\item  Conversely, if $R<S(\rho)$, then any compression scheme will not be reliable. 
	\end{enumerate} 
(See e.g. \cite{nielsen2002quantum}.) \hfill$\Box$
\end{theorem}

\textbf{e) Comments:} 

\begin{enumerate}[noitemsep, topsep=0pt, label=\roman*)]
	\item From Theorem \ref{Schumachercoding}, the von Neumann entropy can be operationally defined as the minimum physical resource required to reliably store the output of a quantum information source.  Recall that the Shannon entropy is the minimum physical resource required to reliably store the output of a classical information source.  Hence in this precise manner, the von Neumann entropy can be considered a quantum generalization of the Shannon entropy.  More importantly, we see that entropies in both information theories play the role of signifying optimal data compression.
	\item We have seen that most of the quantum results rely on the mathematics of classical data compression.  This is part of a broader framework in which quantum information theory can be seen as a generalization of classical information theory.
\end{enumerate}


\chapter{Quantum Entanglement}\label{chap: QEnt}

\begin{chapquote}{Albert Einstein, co-inventor of quantum theory}
	``I cannot seriously believe in [the quantum theory] because it cannot be reconciled with the idea that physics should represent a reality in time and space, free from spooky actions at a distance.''
\end{chapquote}

\textsf{THE INTERDEPENDENCE }among classical information systems is developed on the violation of probabilistic independence (\ref{classicalindependence}) described in Chapter \ref{chap: classical}.  We portrayed this property of independence only after introducing the case of a single variable followed by the consideration of multiple variables. Our presentation of quantum information science will proceed in an analogous manner.  In Chapter \ref{chap: QInfo}, we examined single and multiple quantum information systems through a variety of theoretical tools.  Hence, in this chapter we are led to introduce the mathematical description of `independent' quantum information systems; the notion of interdependence arises naturally in a form known as entanglement; it turns out that entanglement exists across spatial distances (entanglement in space) as well as across temporal intervals (entanglement in time).  In Chapter \ref{chap: classical}, we also described three applications namely classical communication, classical computing, and classical blockchain.  In this chapter, we introduce their quantum information analogues using entanglement.  Both quantum communications and quantum computing rely on an entanglement in space.  The quantum blockchain is one of the first novel applications of an entanglement in time.

\section{Entanglement in Space}

Entanglement, or more precisely entanglement in space, was first theoretically discovered in the Einstein-Podolsky-Rosen (EPR) paradox \cite{einstein1935can}.  They attempted to dismiss the framework of quantum theory by assuming that such an effect could not reasonably exist in the physical world, due to the bizarre implications associated with it.  However, the effect has been experimentally well established, most recently to spatial distances exceeding a $1000$ kilometers \cite{yin2017satellite}. Entanglement in space has also been historically described as \textit{the} single property that radically distinguishes quantum physics from classical physics  \cite{schrodinger1935discussion}.  From a modern perspective, such a statement has manifested itself in that the property plays a central and pervasive role in quantum information science.  It can be seen as an interdependence among two or more spatially separated quantum information systems that would be impossible to replicate by classical information systems.  

In this thesis, we observe that \textit{the interdependence in any entanglement in space is shocking due to the absense of a time interval involved}.   Introducing a time interval in the relevant scenario will only make the effect clash less harshly with our classical intuition.  This observation was first described in \cite{schrodinger1935discussion}, where it was crucially noted that ``\textit{The [EPR] paradox would be shaken, though, if an observation did not relate to a definite moment}."  

Our description of entanglement in space will be introduced through the theoretical tools of qubits, density operators and entropy.  Each provides a different perspective into the perplexing nature of the spatial interdependence.  For detailed reviews on the subject of entanglement in space, we refer the reader to \cite{horodecki2009quantum, guhne2009entanglement, vedral2006introduction, brunner2014bell} whose material we follow closely.  For the rest of this section, we use the term entanglement to solely mean an entanglement in space.

\subsection{Through qubits}

\textbf{a) Bipartite definition: } The entanglement among pure states can easily be described using the quantum circuit model.  We constrain our focus even further by considering the bipartite case which is a system composed of two quantum information subsystems.  These can be respectively labelled $A$ and $B$.  The Hilbert space associated to each subsystem is written as $H_{A}$ with dimension $d_{A}$, and $H_{B}$ with dimension $d_{B}$.  Then any state vector, representing the composite system, in the Hilbert space $H= H_{A} \otimes H_{B}$ is given by
\begin{equation}
\ket{\psi} = \sum_{i,j=1}^{d_{A}, d_{B}} c_{ij} \ket{a_{i}} \otimes \ket{b_{j}},
\end{equation}       
with a $d_{A} \times d_{B}$ matrix $C$ consisting of complex numbers $c_{ij}$. 

A pure state $\ket{\psi} \in H$ is known as \textit{separable}, or as a \textit{product state}, if we can obtain states $\ket{\phi^{A}} \in H_{A}$ and $\ket{\phi^{B}} \in H_{B}$ such that
\begin{empheq}[box=\widefbox]{align}\label{bipartiteseparable}
\ket{\psi} = \ket{\phi^{A}} \otimes \ket{\phi^{b}}. 
\end{empheq}
Otherwise the state $\ket{\psi}$ is referred to as \textit{entangled} or as \textit{nonseparable}.

Quantum separability can be seen to be comparable in some respects to the definition of classical independence (\ref{classicalindependence}).  By looking ahead, we can generalize separability to multipartite systems which consist of multiple subsystems.  

\textbf{b) Multipartite definition: } Consider a pure $N$-partite state $\ket{\psi}$.  We refer to the state $\ket{\psi}$ as \textit{fully separable} if it can be written as 
\begin{empheq}[box=\widefbox]{align}\label{fullyseparable}
\ket{\psi} = \bigotimes_{i=1}^{N} \ket{\phi_{i}}. 
\end{empheq}
If a state does not satisfy the condition of fully separable, then it contains some entanglement.  A pure state is called \textit{m-separable} where $1 < m < N$, if there exists a division of the $N$ parties into $m$ parts $P_{1}, \dots , P_{m}$ such that
\begin{empheq}[box=\widefbox]{align}\label{mseparable}
\ket{\psi} = \bigotimes_{i=1}^{m} \ket{\phi_{i}}_{P_{i}}. 
\end{empheq}
The $m$-separable state may still contain some entanglement.  A state is referred to as truly \textit{N-partite entangled} when it is neither fully separable, nor $m$-separable, for any $m>1$.
  
As an example, consider the case of $N=3$ where the respective quantum information subsystems are labelled $A$, $B$, and $C$.  The pure three-qubit state are fully separable if they can be written as
\begin{equation}
\ket{\phi^{fs}}_{A|B|C} = \ket{\alpha}_{A} \otimes \ket{\beta}_{B} \otimes \ket{\gamma}_{C}.
\end{equation}
Let $m=2$ to consider the associated biseparable states:
\begin{align}
\ket{\phi^{bs}}_{A|BC} &= \ket{\alpha}_{A} \otimes \ket{\delta}_{BC}, \\
\ket{\phi^{bs}}_{B|AC} &= \ket{\beta}_{B} \otimes \ket{\delta}_{AC}, \\
\ket{\phi^{bs}}_{C|AB} &= \ket{\gamma}_{C} \otimes \ket{\delta}_{AB}. 
\end{align}
Note that the state $\ket{\delta}$ may contain entanglement.

\textbf{c) Implications: } We proceed to describe some properties of well known entangled pure states starting with the bipartite case.  

The simplest entangled states are the four \textit{Bell states} (also known as EPR states or EPR pairs)
\begin{empheq}[box=\widefbox]{align}
\ket{\Phi_{AB}^{\pm}} &= \frac{1}{\sqrt{2}}(\ket{0_{A}0_{B}} \pm \ket{1_{A}1_{B}} ), \label{Bellstatebasis} \\
\ket{\Psi_{AB}^{\pm}} &= \frac{1}{\sqrt{2}}(\ket{0_{A}1_{B}} \pm \ket{1_{A}0_{B}} ). \label{Bellstatebasis2}
\end{empheq}
We can describe the generation of these states using the quantum circuit model.  Consider starting with the computational basis state $\ket{0_{A}0_{B}}$.  After applying the Hadamard gate to the first qubit, we obtain state $(\ket{0_{A}} + \ket{1_{A}})\ket{0_{B}}/\sqrt{2}$.  The next step of applying the CNOT gate results in the desired output $\ket{\Phi_{AB}^{+}}$.  Similar procedures can produce the remaining Bell states. 

An entanglement (in space) has an associated interdependence among quantum information systems across spatial distances.  This can be portrayed in the following scenario.  Suppose we have a bipartite system in Bell state 
\begin{equation}\label{Bellstateshock}
\ket{\Phi_{AB}^{+}} = \frac{1}{\sqrt{2}}(\ket{0_{A}0_{B}} + \ket{1_{A}1_{B}} ),
\end{equation}
where subsystem $A$ can be arbitrarily far from subsystem $B$.  The state $\ket{\Phi_{AB}^{+}}$ has the property that if we make a projective measurement \textit{only} on subsystem $A$ (in the computational basis), then the post-measurement result for the system is either $\ket{0_{A}0_{B}}$ or $\ket{1_{A}1_{B}}$ (each occuring with probability $1/2$).  The point we want to stress is that the state of subsystem $B$ will equate to whatever binary state that subsystem $A$ `collapses' to.  The measurement outcomes are correlated.  It also is important to emphasize that prior to the measurement on $A$, both subsystems are in a superposition in (\ref{Bellstateshock}) and neither can be described to be in a definite state.  (Note that a similar analysis occurs for the inverted case where the measurement is on subsystem $B$).  This is remarkable in that subsystem $B$, who is arbitrarily far away from system $A$, \textit{instantaneously} takes whatever value that subsystem $A$ is measured to be found in.  How is it that subsystem $B$ instantaneously `knows' the measurement outcome of subsystem $A$ and follows accordingly?  This property is what Einstein referred \cite{born2005born} to as ``spooky action at a distance."  This interdependence of quantum information systems across space is ``spooky" precisely due to the instantaneous aspect of it.  In other words, it is the \textit{lack of a time interval involved that makes this spatial interdependence shocking}.  However, it is important to note that the measurement outcomes $\ket{0_{A}0_{B}}$ or $\ket{1_{A}1_{B}}$ occur randomly.  Hence such an effect cannot be used to send classical information instantaneously across vast distances.  

It turns out the measurements results are always interdependent.  We have witnessed the case of correlated results.  Consider the Bell state 
\begin{equation}
\ket{\Psi^{-}} = \frac{1}{\sqrt{2}}(\ket{01} - \ket{10}),
\end{equation}
where the measurements are anti-correlated with respect to the computational basis states.  If $\vec{v}$ is any real three-dimensional unit vector, then we can define the observable,
\begin{equation}
\vec{v}\cdot \vec{\sigma} \equiv v_{1}\sigma_{x} + v_{2}\sigma_{y} + v_{3}\sigma_{z},
\end{equation}
which is referred to as a measurement of spin along the $\vec{v}$ axis.  Let the eigenvectors of the observable be denoted $\ket{a}$ and $\ket{b}$.  Then it can be shown that 
\begin{equation}
\frac{1}{\sqrt{2}}(\ket{01} - \ket{10}) = \frac{1}{\sqrt{2}}(\ket{ab} - \ket{ba}),
\end{equation}
up to a global phase factor which we can ignore.  This quantitatively shows that the measurement outcomes, for this Bell state, are always anti-correlated.

The Bell states, (\ref{Bellstatebasis}) and (\ref{Bellstatebasis2}), also form an orthonormal basis for a two qubit four dimensional Hilbert space.  Hence, one can perform a joint quantum measurement of two qubits that determine which of the four Bell states the two qubits are in. This is known as a Bell state measurement.  On a related matter, an important class of operations are LOCC which is an acronym for local operations and classical communications.  This means that operations can only be performed locally on the individual subsystems and the subsystems can communicate classically with each other.  An example of this is the local application of the Pauli operators (\ref{paulioperators}) to change between any of the Bell states
\begin{align}
(\sigma_{x} \otimes I)\ket{\Phi_{AB}^{\pm}} = \ket{\Psi_{AB}^{\pm}}, \\
(\sigma_{x} \otimes I)\ket{\Psi_{AB}^{\pm}} = \ket{\Phi_{AB}^{\pm}}, \\
(\sigma_{z} \otimes I)\ket{\Phi_{AB}^{\pm}} = \ket{\Phi_{AB}^{\mp}}, \\
(\sigma_{z} \otimes I)\ket{\Psi_{AB}^{\pm}} = \ket{\Psi_{AB}^{\mp}}.
\end{align} 
In contrast to Bell state measurements, the ability to distinguish the four Bell states using LOCC is an impossible task and its violation is related to notions of closed timelike curves \cite{moulick2016timelike}.   

Moving from the bipartite case, we proceed to briefly list some well known examples of multipartite entangled pure states.  The first of these are the \textit{GHZ (Greenberger-Horne-Zeilinger) states} which are perhaps the most well studied.  The GHZ state for $N$ qubits is defined as
\begin{empheq}[box=\widefbox]{align}\label{GHZ}
\ket{GHZ_{N}} = \frac{1}{\sqrt{2}}(\ket{0}^{\otimes N} + \ket{1}^{\otimes N}).
\end{empheq}
The second example we wish highlight are the graph states which are defined as follows.  Let $G$ be a graph with a set of $N$ vertices and certain number of edges connecting them.  For each vertex $i$, let neigh($i$) be defined as the neighborhood of $i$, which is the set of vertices that are connected to $i$ by an edge.  Then for each vertex $i$, one can construct what is known as as a stabilizer operator,
\begin{equation}
g_{i} = X_{i} \bigotimes_{j \in \text{neigh}(i)} Z_{j},
\end{equation}
where $X_{i}$, $Y_{i}$, and $Z_{i}$ represent Pauli matrices (\ref{paulioperators}) applied to the $i$-th qubit.  Using this notation, the \textit{graph state} $\ket{G}$ associated with graph $G$ is the unique common eigenvector to all stabilizing operators $g_{i}$,
\begin{empheq}[box=\widefbox]{align}
g_{i}\ket{G} = \ket{G}, \quad \text{for }i \in \{1, \dots, N\}.
\end{empheq}
Notice the important property that
\begin{equation}
\braket{G|g_{i}|G} = 1 \quad \text{for }i \in \{1, \dots, N\}.
\end{equation}
An important subset of graph states are \textit{cluster states} which are based on square lattice graphs.  An example of this is the four qubit cluster state
\begin{equation}
\ket{CL_{4}} = \frac{1}{2}(\ket{0000} + \ket{0011} + \ket{1100} - \ket{1111}).
\end{equation}

Our third and final example of multipartite entangled pure states are the \textit{Dicke states} which are physically associated with the light emission of a cloud of atoms in excited states.  In relation to quantum information science, the most important are the \textit{symmetric Dicke states}, which for $N$ qubits and and $k$ excitations is given by
\begin{empheq}[box=\widefbox]{align}
\ket{D_{k, N}} = \begin{pmatrix}
N \\
k 
\end{pmatrix}^{-\frac{1}{2}} \sum_{j} P_{j} \Bigg\{ \ket{1}^{\otimes k} \otimes \ket{0}^{\otimes N-k}  \Bigg\},
\end{empheq}
where $\sum_{j} P_{j} \{\dots \}$ represents the sum over all possible permutations of the qubits.  An example of such a Dicke state is the \textit{$W$ state} which is the symmetric state of $N$ particles with a single excitation,
\begin{equation}
\ket{W_{n}} = \frac{1}{\sqrt{n}}(\ket{0 \dots 01} + \dots \ket{10 \dots 0}).
\end{equation}

\textbf{d) Detection: } An important question is how do we show that a state is entangled?  For bipartite systems, we consider two types of entanglement detection.

The first is known as the Schmidt decomposition.  Suppose we have the pure state
\begin{equation}
\ket{\psi} = \sum_{i,j=1}^{d_{A}, d_{B}} c_{ij} \ket{a_{i}b_{j}},
\end{equation}  
which is a state vector in the space $H_{A} \otimes H_{B}$.  Moreover, we have an associated  $d_{A} \times d_{B}$ matrix $C$ consisting of the complex numbers $c_{ij}$. Then the Schmidt decomposition states that there exists an orthonormal basis $\ket{\alpha_{i}}$ of $H_{A}$ and an orthonormal basis $\ket{\beta_{j}}$ of $H_{B}$ such that
\begin{equation}
\ket{\psi} = \sum_{k=1}^{R} \lambda_{k} \ket{\alpha_{k}\beta_{k}},
\end{equation}
where $\lambda_{k}$ are positive real coefficients.  The values of $\lambda_{k}$ are the unique square roots of the eigenvalues of the matrix $CC^{\dagger}$.  The number $R \leq \text{min}\{d_{A}, d_{B}\}$ is known as the Schmidt rank of $\ket{\psi}$.  Pure product states correspond to states of Schmidt rank one.  If it is greater than one, then the state is entangled.

The second method is known as the Bell inequality or more precisely the CHSH inequality \cite{clauser1969proposed}.  Suppose we have a bipartite system, composed of $A$ and $B$, in which each subsystem can be measured in two quantities;  for system $A$, this is denoted by $A_{1}$ and $A_{2}$ and similarly for system $B$, we have $B_{1}$ and $B_{2}$; each can take either value $+1$ or $-1$.  The CHSH inequality states that 
\begin{equation}\label{CHSH}
\langle A_{1}B_{1} \rangle  + \langle A_{2}B_{1} \rangle + \langle A_{2}B_{2} \rangle - \langle A_{1}B_{2} \rangle \leq 2.
\end{equation}  
We will see in Chapter \ref{chap: QFound} that the violation of this result has profound implications for fundamental physics.  However from an operational perspective, the violation of this inequality (and its generalization) detects all pure entangled states.  More precisely, for any entangled pure state it is possible to find local measurements such that it violates the CHSH inequality.  Furthermore, the only states that do not violate it are product states.  To see an explicit example, consider the entangled state 
\begin{equation}\label{CHSHviolationBell}
\ket{\Psi^{-}} = \frac{1}{\sqrt{2}}(\ket{01} - \ket{10}), 
\end{equation}
and let 
\begin{equation}
A_{1} = \sigma_{z}, \quad A_{2} = \sigma_{x}, \quad B_{1} = \frac{-\sigma_{z} - \sigma_{x}}{\sqrt{2}}, \quad B_{2} = \frac{\sigma_{z} - \sigma_{x}}{\sqrt{2}}. 
\end{equation}
From this we can compute the expectation values of each observable through (\ref{quantumexpectation}).  We find that the violation of (\ref{CHSH}) occurs by the left hand side of the inequality equating to $2\sqrt{2}$.  

For the case of multiple subsystems, it can be shown that all pure entangled $N$-partite states violate a generalization of this Bell inequality \cite{brunner2014bell, popescu1992generic}.

\subsection{Through density operators}

\textbf{a) Bipartite definition: } Expressing the definition of entanglement through density operators allows the property to be extended to mixed states.  We begin by constraining our attention to the bipartite case, with the subsystems labelled $A$ and $B$.  Suppose we have the density operator
\begin{equation}
\sigma = \sum_{i}p_{i}\ket{\phi_{i}}\bra{\phi_{i}},
\end{equation}
where the state of the the system is known to be in one of $\ket{\phi_{i}} \in H = H_{A} \otimes H_{B}$ with respective classical probabilities $p_{i}$.  In the literature regarding entanglement, it is often the case that the probabilities which satisfy 
\begin{equation}
p_{i} \geq 0, \quad \quad  \sum_{i}p_{i} = 1,
\end{equation}
are referred to as convex weights; this terminology stems from a geometric interpretation.  Moreover, a convex combination of density operators $\sigma_{i}$ refers to the quantity
\begin{equation}
\sum_{i} p_{i} \sigma_{i}.
\end{equation}
We say that $\sigma$ is a \textit{product state} if there exists state $\sigma^{A}$ for subsystem $A$, and state $\sigma^{B}$ for subsystem $B$, such that
\begin{empheq}[box=\widefbox]{align}\label{bipartitedensityproductstate}
\sigma = \sigma^{A}\otimes \sigma^{B}.
\end{empheq}
The density operator $\sigma$ is called \textit{separable} if there exists convex weights $p_{i}$ and product states $\sigma_{i}^{A} \otimes \sigma_{i}^{B}$ such that
\begin{empheq}[box=\widefbox]{align}\label{bipartitedensityseparable}
\sigma = \sum_{i} p_{i} \sigma_{i}^{A}\otimes \sigma_{i}^{B}.
\end{empheq}
Otherwise the density operator $\sigma$ is referred to as \textit{entangled}.

\textbf{b) Multipartite definition: } For an $N$-partite system, a density operator $\sigma$ is \textit{fully separable} if it can be written as a convex combination of pure fully separable states
\begin{empheq}[box=\widefbox]{align}\label{multidensityseparable}
\sigma = \sum_{i} p_{i} \ket{\phi_{i}^{fs}}\bra{\phi_{i}^{fs}},
\end{empheq}
which can also be written as
\begin{empheq}[box=\widefbox]{align}\label{multidensityseparable2}
\sigma = \sum_{k} p_{k} \sigma_{k}^{(1)} \otimes \sigma_{k}^{(2)} \otimes \dots \sigma_{k}^{(N)}. 
\end{empheq}
A density operator is called \textit{m-separable}, where $1 < m < N$, if it can be written as a convex combination of pure $m$-separable states.  The density operator is said to be \textit{N-partite entangled} when it is neither fully separable, nor $m$-separable for any $m > 1$.

\textbf{c) Implications: } Through the qubit framework, we witnessed some non-trivial properties regarding entanglement best exemplified through the Bell state (\ref{Bellstateshock}).  Other than extending the definition of entanglement to mixed states, density operators provide a widely different perspective on the puzzling nature of entanglement.  To elaborate on this point, consider once again the Bell state (\ref{Bellstateshock}).  This can be expressed through the density operator
\begin{align}
\rho &= \Bigg( \frac{\ket{00} + \ket{11}}{\sqrt{2}} \Bigg)\Bigg( \frac{\bra{00} + \bra{11}}{\sqrt{2}} \Bigg) \\
&= \frac{\ket{00}\bra{00} + \ket{11}\bra{00} + \ket{00}\bra{11} + \ket{11}\bra{11}}{2}.
\end{align}
One can compute the reduced density operator (\ref{reduceddensityoperator}) of the first qubit as
\begin{align}
\rho^{1} &= \text{tr}_{2}(\rho) \\
&= \frac{\text{tr}_{2}(\ket{00}\bra{00}) + \text{tr}_{2}(\ket{11}\bra{00}) + \text{tr}_{2}(\ket{00}\bra{11}) + \text{tr}_{2}(\ket{11}\bra{11})}{2} \\
&= \frac{\ket{0}\bra{0}\braket{0|0} + \ket{1}\bra{0}\braket{0|1} + \ket{0}\bra{1}\braket{1|0} + \ket{1}\bra{1}\braket{1|1}}{2} \\
&= \frac{\ket{0}\bra{0} + \ket{1}\bra{1}}{2} \\
&= \frac{1}{2}I.
\end{align}
The result is that we obtain a maximally mixed state (\ref{completelymixedstate}) for its subsystem.  We can verify this by computing $\text{tr}((I/2)^{2}) = 1/2 < 1$.  Of more interest is the interpretation of this computation.  This result is truly perplexing in that the joint state of the system is known \textit{exactly} ($\rho$ is a pure state), and yet at the \textit{at the same time}, we do not have maximal knowledge about its subsystem ($\rho^{A}$ is a mixed state)!  If there was a time interval involved, then perhaps such a property could be explained by a loss or transfer of information among the systems during some period of time.  Hence, it is precisely the \textit{lack of a time interval involved that makes this interdependence among the system and its subsystems shocking}.  

More broadly speaking, a pure bipartite state is said to be \textit{maximally entangled} if the reduced density matrix on either system is maximally mixed.

\textbf{d) Detection: } Detecting entanglement in mixed states is non-trivial.  One way to articulate this is that the the test of the Bell inequality or CHSH inequality (\ref{CHSH}) fails for some entangled mixed states; they do not violate the inequality.  An example of such mixed states are a subset of the Werner states
\begin{equation}\label{Wernerstates}
\sigma_{W} = F\ket{\Psi^{-}}\bra{\Psi^{-}} + \frac{1-F}{3}(\ket{\Psi^{+}}\bra{\Psi^{+}} + \ket{\Phi^{+}}\bra{\Phi^{+}} + \ket{\Phi^{-}}\bra{\Phi^{-}}),
\end{equation}
where we have used Bell states (\ref{Bellstatebasis}) and (\ref{Bellstatebasis2}), and where $0 \leq F \leq 1$.  When $F> 0.5$, the density operator $\sigma_{w}$ is entangled, and yet these mixed states only violate the Bell inequality when $F>0.78$.

From such an example, it becomes readily apparent that one needs a new set of theoretical tools.  However the question of whether a given density operator is separable or entangled has no known general solution.  This problem is called the separability problem.  The challenge in mixed states is in detecting the quantum interdependence while ignoring the classical interdependence.  Nevertheless we introduce two methods that succeed for certain scenarios.

For the case of bipartite entanglement, there is a tool known as the PPT criterion which is also known as the Peres-Horodecki criterion.  Suppose we have a density operator for a composite system and this is expanded in terms of a product basis such that
\begin{equation}
\sigma = \sum_{i,j}^{N}\sum_{k,l}^{M}\sigma_{ij, kl}\ket{i}\bra{j} \otimes \ket{k}\bra{l}.
\end{equation}
We define the partial transposition of $\sigma$ as the transposition with respect to one of its subsystems.  An example is that the partial transposition with respect to subsystem $A$ is written as
\begin{equation}\label{PPT}
\sigma^{T_{A}} =  \sum_{i,j}^{N}\sum_{k,l}^{M}\sigma_{ji, kl}\ket{i}\bra{j} \otimes \ket{k}\bra{l},
\end{equation}
where we have exchanged the indices $i$ and $j$.  In a similar manner, one can define $\sigma^{T_{B}}$ by exchanging $k$ and $l$.  Moreover, a density operator $\sigma$ is said to have a PPT (positive partial transpose) if its partial transposition has no negative eigenvalues.  It is important to note that the spectrum of the density matrix does not depend on what product basis the density operator was expanded in.  

The PPT criterion states that if $\sigma$ is a bipartite separable state, then $\sigma$ is PPT.  Hence, this provides us with a method to detect entanglement.  If for a given density matrix, we compute the partial transpose with its spectrum and obtain negative eigenvalues, then the state is entangled.  However, this method does not provide a general sufficient criteria for separability.  Nevertheless, we can see its utility on detecting the entanglement such as for the case of Werner states (\ref{Wernerstates}).  We have 
\begin{equation}
\sigma_{W} = F\ket{\Psi^{-}}\bra{\Psi^{-}} + \frac{1-F}{3}(\ket{\Psi^{+}}\bra{\Psi^{+}} + \ket{\Phi^{+}}\bra{\Phi^{+}} + \ket{\Phi^{-}}\bra{\Phi^{-}}),
\end{equation}
and we can compute the partial tranposition with respect to subsystem $B$ in the following manner.  For the Bell states we obtain
\begin{align}
\ket{\Phi^{+}}\bra{\Phi^{+}}^{T_{B}} &= \frac{1}{2}(\ket{00}\bra{00} + \ket{01}\bra{10} + \ket{11}\bra{00} + \ket{11}\bra{11}), \\
\ket{\Phi^{-}}\bra{\Phi^{-}}^{T_{B}} &= \frac{1}{2}(\ket{00}\bra{00} - \ket{01}\bra{10} - \ket{10}\bra{01} + \ket{11}\bra{11}), \\
\ket{\Psi^{+}}\bra{\Psi^{+}}^{T_{B}} &= \frac{1}{2}(\ket{01}\bra{01} + \ket{00}\bra{11} + \ket{11}\bra{00} + \ket{10}\bra{10}), \\
\ket{\Psi^{-}}\bra{\Psi^{-}}^{T_{B}} &= \frac{1}{2}(\ket{01}\bra{01} - \ket{00}\bra{11} - \ket{11}\bra{00} + \ket{10}\bra{10}).
\end{align}
From this we can compute
\begin{equation}
\frac{1-F}{3}(\ket{\Psi^{-}}\bra{\Psi^{-}} + \ket{\Psi^{+}}\bra{\Psi^{+}})^{T_{B}} = \frac{1-F}{3}(2\ket{01}\bra{01} + 2\ket{10}\bra{10}),
\end{equation}
and 
\begin{align}
\Bigg(\frac{3F}{3}\ket{\Phi^{-}}\bra{\Phi^{-}} + \frac{1-F}{3}\ket{\Psi^{+}}\bra{\Psi^{+}}\Bigg)^{T_{B}} &= \frac{1}{3}((2F+1)(\ket{00}\bra{00} + \ket{11}\bra{11}) \\
&+ (4F-1)(\ket{01}\bra{01} + \ket{10}\bra{01})). \nonumber
\end{align}
Combining these quantities, the partial transpose of $\sigma$ in a matrix can be obtained as
\begin{equation}
\sigma^{T_{B}} = \frac{1}{3}
\begin{pmatrix}
	2F+1 & 0 & 0 & 0 \\
	0 & 2-2F & 4F-1 & 0 \\
	0 & 4F-1 & 2-2F & 0 \\
	0 & 0 & 0 & 2F+1 
\end{pmatrix}
\end{equation}
with eigenvalues equating to $(2F+1)$ and to $(3-6F)$. Therefore, we can correctly identify that entanglement occurs when $F> 0.5$, as this results in $(3-6F)$ becoming negative.  This is in contrast to the Bell inequality which is only violated when $F>0.78$.

Another partial solution to the separability problem are through what are known as entanglement witnesses.  These are widely used in experimental settings.  Theoretically, these are Hermitian operators (observables) that assist in determining whether a density operator is entangled or not.  More formally, an observable $W$ is defined as an entanglement witness if 
\begin{align}
\text{tr}(W \sigma_{s}) &\geq 0 \quad \text{for all separable }\sigma_{s}, \\
\text{tr}(W \sigma_{e}) &< 0 \quad \text{for at least one entangled }\sigma_{e}.
\end{align}
The underlying mathematical reasoning is based on the Hahn-Banach theorem regarding Hilbert spaces.  Physically what is important is that for any entangled state $\sigma_{e}$ there always exists an entanglement witness that detects it.  However, constructing an entanglement witness is a difficult problem.  One construction of an entanglement witness is given by
\begin{equation}
W = \alpha I - \ket{\psi}\bra{\psi},
\end{equation}     
where $\ket{\psi}$ represents an entangled pure state, and where the value of $\alpha$ is specific to the case in question.  As an example, in the tripartite case an entanglement witness for GHZ is given by
\begin{equation}
W_{GHZ_{N}} = \frac{3}{4}I - \ket{GHZ_{3}}\ket{GHZ_{3}},
\end{equation}
where for mixed states $\sigma$ we have
\begin{align}
\text{tr}(W_{GHZ_{N}}\sigma) &< 0 \quad \rightarrow \quad \sigma \text{ is in the GHZ class}, \\
\text{tr}(W_{GHZ_{N}}\sigma) &\geq 0 \quad \rightarrow \quad \sigma \text{ is not detected.} 
\end{align}

\subsection{Through entropy}

\textbf{a) Definition: } Another interpretation of the von Neumann entropy (\ref{vonneumann}) is in relation to entanglement.  More precisely, suppose we have a bipartite system with the subsystems labelled $A$ and $B$.  Moreover, let $\ket{AB}$ denote a pure state of this composite system.  Then $\ket{AB}$ is \textit{entangled} if and only if 
\begin{empheq}[box=\widefbox]{align}\label{vonneumannentangled}
S(A|B) < 0,
\end{empheq}
where we have used the conditional von Neumann entropy (\ref{quantumconditionalentropy}), which we rewrite here as 
\begin{equation}\label{quantumconditionalentropy2}
S(A|B) \equiv S(A,B) - S(B). 
\end{equation}

\textbf{b) Implications: } We aim to examine two properties regarding entangled states from the perspective of entropy. 

The first is that the inequality (\ref{vonneumannentangled}) implies that for entangled states
\begin{equation}
S(B) > S(A, B),
\end{equation} 
which means the uncertainty about the subsystem $B$ is greater than the uncertainty of the composite system $AB$.  This characteristic was expressed earlier through our analysis via density operators.  However, the implications of this entropic inequality are far more interesting when we consider the strong subadditivity inequality (\ref{strongsubadditivity}).  For a tripartite system this can be written as 
\begin{equation}
S(A, B, C) + S(B) \leq S(A, B) + S(B, C),
\end{equation}
which can be shown to be equivalent to
\begin{equation}
S(A) + S(B) \leq S(A,C) + S(B,C).
\end{equation}
For entangled systems, it is possible to obtain counter-intuitive results such as $S(A) > S(A, C)$ or $S(B) > S(B,C)$.  However we see that the strong subadditivity constrains this freedom in that both of these cases cannot be true \textit{at the same time}.  Hence the \textit{lack of a time interval in this tripartite scenario makes the interdependence among these three quantum systems extremely non-trivial}.

The second property we wish to consider is how entanglement may influence the entropic uncertainy relation (\ref{entropicuncertainty}).  We refer the reader to \cite{coles2017entropic} for a detailed analysis.  To briefly see this, we rewrite the entropic uncertainty relation as
\begin{equation}\label{entropicuncertaintyrelation2}
H(X) + H(Z) \geq \log \frac{1}{c}.
\end{equation}
More specific to the scenario is how would the uncertainty relation be modified if one is able to have access to entangled states.  These would serve as memory or side information that assists in predicting the results of the measurement of $X$ and $Z$.  To answer this we need to introduce what is known as a classical-quantum state which is a classical register $X$ correlated with a quantum memory $B$, modelled by density operator
\begin{equation}
\rho_{XB} = \sum_{x} p(x) \ket{x}\bra{x} \otimes \rho_{B}^{x}.
\end{equation}
Note that $p(x)$ refers to the probability distribution associated with $X$, and $\rho_{B}^{x}$ is the quantum state of the memory conditioned on the classical register taking value $X=x$.  From this quantity, we can compute the classical-quantum entropy which is the von Neumann entropy of $X$ conditioned on $B$,
\begin{equation}\label{classicalquantumentropy}
S(X|B) \equiv S(\rho_{XB}) - S(\rho_{B}),
\end{equation}
where 
\begin{equation}
\rho_{B} = \text{tr}_{X}(\rho_{XB}) = \sum_{x}p(x)\rho_{B}^{x}.  
\end{equation}
The classical-quantum entropy (\ref{classicalquantumentropy}) is a specific form of the conditional von Neumann entropy (\ref{quantumconditionalentropy2}).  From these constructions, one can prove the following entropic uncertainty relation 
\begin{equation}\label{entropicuncertaintyrelationsentangled}
S(X|B) + S(Z|B) \geq \text{log }\frac{1}{c} + S(A|B),
\end{equation}
for bipartite quantum state $\rho^{AB}$, for observables $X$ and $Z$, and 
where $c$, as in (\ref{entropicuncertaintyrelation2}), is the maximum value of the possible quantities, $c_{xz} = \lvert \braket{x|z} \rvert^{2}$, where 
\begin{equation}
0 \leq \text{log }\frac{1}{c} \leq \text{log }d.
\end{equation}
Both classical-quantum conditional entropies $S(X|B)$ and $S(Z|B)$ quantify the uncertainty of $X$ and $Z$ given that one has access to quantum memory $B$.  For a maximally entangled state it can be shown that $S(A|B) = -\log d$ where $d$ is the dimensionality of the respective Hilbert space.  Hence we have 
\begin{equation}
\log \frac{1}{c} + S(A|B) = \log \frac{1}{c} - \log d \leq 0.
\end{equation} 
To interpret this result, recall the guessing game between Alice and Bob associated with (\ref{entropicuncertainty}). If we allow Bob access to a maximally entangled quantum memory, then it can be shown that Bob can win the game with probability one.  This highlights how entanglement allows one to perform tasks that would be impossible to carry out with only classical resources.

Finally suppose we have a tripartite system $ABC$ represented by density operator $\rho_{ABC}$.  Moreover we have associated observables $X$ and $Z$.  Then it can be shown that 
\begin{equation}
S(X|B) + S(Z|C) \geq \log d,
\end{equation}  
where $d$ is the dimension of the Hilbert space associated with subsystem $A$.  More generally, one can obtain    
\begin{equation}\label{tripartiteentropicrelation}
S(X|B) + S(Z|C) \geq \text{log }\frac{1}{c},
\end{equation}  
where $c$ is defined as in (\ref{entropicuncertaintyrelation2}).

\textbf{c) Measures: } Through the qubit and density operator framework, we were introduced to methods that detected whether a state was entangled or not.  Using entropic concepts, we can develop tools to quantify the amount of entanglement in an entangled state.  Such tools are known as entanglement measures.  We expect for a density operator, $\sigma$, an entanglement measure, denoted $E(\sigma)$, satisfies the following properties:
\begin{enumerate}[noitemsep, topsep=0pt, label=\roman*)]
	\item For a separable state $\sigma$, we have $E(\sigma) = 0$.
	\item It is invariant under unitary transformation, that is 
	\begin{equation}
	E(\sigma) = E(U_{A} \otimes U_{B} \sigma U_{A}^{\dagger} \otimes U_{B}^{\dagger})
	\end{equation}
	for a unitary transformation of the form $U_{A} \otimes U_{B}$.
	\item $E(\sigma)$ should not increase under an LOCC operations.
\end{enumerate}
We briefly list four common entanglement measures discussed in the literature.  Our focus is on the bipartite case (labelled $AB$), and how they are related to entropic concepts.  

The first of these is the entanglement of formation which for density operator $\sigma$ is written as
\begin{equation}
E_{F}(\sigma) \equiv \text{min }\sum_{i}p_{i} S(\sigma_{i}^{A}),
\end{equation}
where we use von Neumann entropy
\begin{equation}
S(\sigma^{A}) = -\text{tr }\sigma^{A} \log \sigma^{A}.
\end{equation}
The minimum is over all possibilities of state
\begin{equation}
\sigma^{AB} = \sum_{j}p_{j}\ket{\psi_{j}}\bra{\psi_{j}},
\end{equation}
where 
\begin{equation}
\sigma_{i}^{A} = \text{tr}_{B}(\ket{\psi_{i}}\bra{\psi_{i}}).
\end{equation}
It can be interpreted as the minimum number of maximally entangled states that is required to to obtain a certain number of copies of the given state by LOCC.

The second quantity is known as the entanglement of distillation which for a pure state is given by the von Neumann entropy of the reduced state $\sigma_{A}$,
\begin{equation}
E_{D}(\ket{\psi}) = S(\sigma_{A}) = -\text{tr} (\sigma_{A} \log \sigma_{A}).
\end{equation}
It can be interpreted as the number of maximally entangled states that can be derived from an initial number of non-maximally entangled states using LOCC.

Another useful measure is known as the relative entropy of entanglement which is defined as
\begin{equation}
E_{R}(\sigma) \equiv \underset{\varrho \in D}{\text{min }} S(\sigma || \varrho),
\end{equation}
where $S(\sigma || \varrho)$ is the von Neumann relative entropy (\ref{quantumrelativeentropy}), and $D$ is the set of all disentangled states.  It quantifies the amount of entanglement through a distance measure.

Finally the concurrence for a pure state is given by
\begin{equation}\label{concurrence}
C(\ket{\psi}) = \sqrt{2(1-\text{tr}(\sigma_{A}^{2}))},
\end{equation} 
where $\sigma_{A}$ is the reduced subsystem of $\ket{\psi}$.  For the two qubit case, the concurrence is related to the entanglement of formation
\begin{equation}
E_{F}(\sigma) = h\Bigg(\frac{1 + \sqrt{1-C^{2}(\sigma)}}{2}\Bigg),
\end{equation}
where we use binary version of the Shannon entropy, $h(p) = -p \log p - (1-p) \log (1-p)$.

\section{Application: Quantum Communication}

Entanglement in space can be seen as a resource in quantum information in that it allows the ability to perform information tasks that would be impossible or very difficult to do with only classical information.  The three different communication protocols described in this section serve to illustrate this point.  Each protocol is described in the context of two parties, named Alice and Bob, who are some arbitrary distance apart.  More crucially, each share a qubit from a spatial Bell state.  It is also common in these protocols to design a code that relates the classical and quantum information.  These applications are instrumental for the construction of a useful quantum communications network \cite{simon2017towards,wehner2018quantum}.  

\subsection{Superdense coding}

\textbf{a) Protocol: }  This information task requires Alice to send two bits of classical information to Bob using a single qubit \cite{bennett1992communication}.  The protocol starts by assuming Alice and Bob share the spatial Bell state
\begin{equation}
\ket{\Phi^{+}} = {1\over\sqrt2}(\ket{0}\ket{0}+\ket{1}\ket{1}).
\end{equation}
Moreover, they have agreed to encode the classical information in the following way:  The bit string $xy$, where $xy=00, 01, 10, 11$ corresponds to Bell state
\begin{equation}\label{superdensecode}
\ket{\beta_{xy}} = \frac{1}{\sqrt{2}}(\ket{0}\ket{y} + (-1)^{x}\ket{1}\ket{\bar{y}}), 
\end{equation}
where $\bar{y}$ is the negation of $y$.

The protocol is as follows: If Alice wants to send bit string $00$, she simply sends her qubit to Bob.  However if Alice wants to send string $01$, she applies the $X$ operator on her qubit before sending it to Bob
\begin{equation}
(X \otimes I)\ket{\psi} = {1\over\sqrt2}(\ket{1}\ket{0}+\ket{0}\ket{1}) = \ket{\beta_{01}}.
\end{equation}
For the case of sending bits $10$, she applies a $Z$ operator,
\begin{equation}
(Z \otimes I)\ket{\psi} = {1\over\sqrt2}(\ket{0}\ket{0}-\ket{1}\ket{1}) = \ket{\beta_{10}}.
\end{equation}
And for the last case of $11$, she applies the $iY$ gate before sending her qubit to Bob
\begin{equation}
(iY \otimes I)\ket{\psi} = {1\over\sqrt2}(\ket{0}\ket{1}-\ket{1}\ket{0}) =\ket{\beta_{11}}.
\end{equation}
Once Bob receives the qubit from Alice, he performs a projective measurement, in the Bell basis on both qubits.  From that, he is able to recover bit string $xy$ from identifying state $\ket{\beta_{xy}}$.

\textbf{b) Comments: }
\begin{enumerate}[noitemsep, topsep=0pt, label=\roman*)]
	\item This information task would be impossible to perform, in the classical case, had Alice only transmitted a single classical bit.  
	\item Superdense coding has recently been experimentally demonstrated within an optical fiber infrastructure \cite{williams2017superdense}.
\end{enumerate}

 \subsection{Quantum teleportation}
 
\textbf{a) Protocol: } The following task \cite{bennett1993teleporting} requires Alice to send a particular set of quantum information to Bob without that information traversing the space between them.  More precisely, Alice wants to send Bob a qubit $\ket{\psi} = \alpha\ket{0} + \beta\ket{1}$, where the values of $\alpha$ and $\beta$ are unknown to both parties.  They both share the Bell state 
\begin{equation}\label{teleportationbellstate}
\ket{\beta_{00}} = {1\over\sqrt2}(\ket{00}+\ket{11}),
\end{equation} 
as well as have access to a classical communications channel which transmits bits.  The initial state of this scenario can written as
\begin{equation}
\ket{\psi_{0}}= \ket{\psi}\ket{\beta_{00}} = {1\over\sqrt2} [\alpha \ket{0}(\ket{00}+\ket{11}) + \beta \ket{1} (\ket{00}+\ket{11})],
\end{equation}
where the first two qubits are in Alice's possession, while the third qubit belongs to Bob.  The first step is that Alice applies a CNOT gate (\ref{CNOT}) to both of her qubits, in which case the state transforms to
\begin{equation}
\ket{\psi_{1}}= {1\over\sqrt2} [\alpha \ket{0}(\ket{00}+\ket{11}) + \beta \ket{1} (\ket{10}+\ket{01})].
\end{equation}
From there, she proceeds to apply a Hadamard gate to her first qubit.  This produces the overall state
\begin{equation}
\ket{\psi_{2}}= \frac{1}{2} [\alpha (\ket{0} + \ket{1})(\ket{00}+\ket{11}) + \beta (\ket{0} - \ket{1}) (\ket{10}+\ket{01})],
\end{equation} 
which can re-written as
\begin{eqnarray}\label{Teleportationpremeasure}
\ket{\psi_{2}}= \frac{1}{2} \Big( \ket{00}(\alpha\ket{0} + \beta\ket{1}) + \ket{01} (\alpha \ket{1} + \beta \ket{0})
\nonumber
\\
+ \ket{10}(\alpha\ket{0} - \beta\ket{1}) + \ket{11} (\alpha \ket{1} - \beta \ket{0}) \Big).
\end{eqnarray}
When Alices measures her qubits, in her computational basis states, she gets one of the results on the left in  (\ref{Alice}).  Bob would then apply the corresponding Pauli operator (\ref{paulioperators}) on his qubit to obtain $\ket{\psi}$:    
\begin{eqnarray}\label{Alice}
00&  \rightarrow &\text{Does nothing},
\nonumber
\\
01& \rightarrow &\text{Applies } \sigma_{x} = \ket{0}\bra{1} + \ket{1}\bra{0}, 
\nonumber
\\
10& \rightarrow &\text{Applies }  \sigma_{z} = \ket{0}\bra{0} - \ket{1}\bra{1},
\nonumber
\\
11& \rightarrow &\text{Applies }   \sigma_{z}\sigma_{x}.
\end{eqnarray}
Bob receives the two bits from Alice in (\ref{Alice}) through the classical channel.  In this case, one view (\ref{Alice}) as a code to relate the classical and quantum information.

\textbf{b) Comments: }
\begin{enumerate}[noitemsep, topsep=0pt, label=\roman*)]
	\item The notion of teleportation can be seen in that the quantum information disappears from Alice's location and re-appears in Bob's location. Of crucial necessity to perform this task is the initial Bell state (\ref{teleportationbellstate}).  This emphasizes the point that entanglement in space can be regarded as a resource in quantum information; it allows us to carry out information tasks that would be impossible to do with only classical resources. 
	\item In Chapter \ref{chap: QInfo} we saw measurement as a process that converts quantum information into classical information.  This protocol alludes to a more general property in that one can also convert the classical information back to the quantum information as long as the quantum measurement does not reveal any information about the state being measured.
	\item Throughout the duration of the protocol, there is always at most one copy of $\ket{\psi}$.  Hence at no time is the no-cloning theorem violated.
	\item When Alice performs a measurement on her qubits in (\ref{Teleportationpremeasure}), the quantum information residing in Bob's qubit is instantaneously affected.  The lack of a time interval involved in this process suggests a violation of relativity.  However, such a concern can be largely alleviated in that Bob still requires the two bits from the classical channel (whose transmission is limited by the speed of light) to obtain $\ket{\psi}$.  The density operator framework clearly illustrates this point:  Each of the outcomes in (\ref{Alice}) from measuring (\ref{Teleportationpremeasure}) occur with probability $1/4$.  Hence the density operator of the system after Alice's measurement is given by 	
\begin{eqnarray}\label{Teleportationpremeasuredensity}
\rho =& \frac{1}{4} \Big( \ket{00}\bra{00}(\alpha\ket{0} + \beta\ket{1})(\alpha^{*}\bra{0} + \beta^{*}\bra{1})   \\ \nonumber
+& \ket{01}\bra{01}(\alpha\ket{1} + \beta\ket{0})(\alpha^{*}\bra{1} + \beta^{*}\bra{0}) \\ \nonumber
+& \ket{10}\bra{10}(\alpha\ket{0} - \beta\ket{1})(\alpha^{*}\bra{0} - \beta^{*}\bra{1}) \\ \nonumber
 +& \ket{11}\bra{11}(\alpha\ket{1} - \beta\ket{0})(\alpha^{*}\bra{1} - \beta^{*}\bra{0}) \Big).
\end{eqnarray}
By using (\ref{reduceddensityoperator}), we obtain the reduced density operator of Bob's system which can be computed as
\begin{eqnarray}\label{TeleportBobsystem}
\rho^{B} =& \frac{1}{4} \Big( (\alpha\ket{0} + \beta\ket{1})(\alpha^{*}\bra{0} + \beta^{*}\bra{1})   \\ \nonumber
+& (\alpha\ket{1} + \beta\ket{0})(\alpha^{*}\bra{1} + \beta^{*}\bra{0}) \\ \nonumber
+& (\alpha\ket{0} - \beta\ket{1})(\alpha^{*}\bra{0} - \beta^{*}\bra{1}) \\ \nonumber
+& (\alpha\ket{1} - \beta\ket{0})(\alpha^{*}\bra{1} - \beta^{*}\bra{0}) \Big).
\end{eqnarray}
One can simplify this expression to
\begin{align}
\rho^{B} &= \frac{2(\lvert \alpha \rvert^{2} + \lvert \beta \rvert^{2})\ket{0}\bra{0} + 2(\lvert \alpha \rvert^{2} + \lvert \beta \rvert^{2})\ket{1}\bra{1}}{4} \\
&= \frac{\ket{0}\bra{0} + \ket{1}\bra{1}}{2} \\
&= \frac{I}{2}.
\end{align}
This means that prior to receiving the classical measurement results from Alice, the state appears totally random to Bob.  Nevertheless, there is an instantaneous effect across space on the quantum information held by Bob when Alice makes a measurement.  This remaining issue would be resolved where a time interval introduced into that process. 

\item  Quantum teleportation has been demonstrated experimentally, most recently from a ground station to a space-based satellite \cite{ren2017ground}
\end{enumerate}

\textbf{c) Monty Hall teleportation: } The teleportation protocol has been extended to probabilistic scenarios~\cite{li2000probabilistic, lu2000teleportation, agrawal2002probabilistic}.  In this section, we present a probabilistic version \cite{rajan2019quantum} of quantum teleportation that is \textit{part of the original component of this thesis} (which was done in collaboration with my supervisor).  We  modify the standard teleportation protocol into the  Monty Hall game which was described in detail in Chapter \ref{chap: classical}.  Alice can be viewed as Monty, and Bob as the contestant.  The four doors are respectively labelled $(00, 01, 10, 11)$.
This coincides with Alice's possible measurement results in (\ref{Alice}); the prize door is Alice's actual result, whose bits we denote $ab$, and what Bob would need get the desired state $\ket{\psi}$. The contestant's initial choice of door would be equivalent to what Bell state was used at the start of the protocol. In this modification, the contestant is allowed to choose any of the four doors $(00, 01, 10, 11)$, which we denote $xy$.  This event coincides with using Bell state
\begin{equation}
\ket{\beta_{xy}} = \frac{1}{\sqrt{2}}(\ket{0}\ket{y} + (-1)^{x}\ket{1}\ket{\bar{y}}), 
\end{equation}
where $\bar{y}$ is the negation of $y$.
As an example, if the contestant chooses door $10$, then a way to implement this is that Bob applies the operator $(\sigma_{0} \otimes \sigma_{z})\ket{\beta_{00}} = \ket{\beta_{10}}$, and communicates that to Alice; the last step would be analogous to Monty being aware of what door the contestant chooses.
In this modified protocol, the initial state is represented as
\begin{align}
\ket{\psi}\ket{\beta_{xy}} =& \frac{1}{\sqrt{2}}(\alpha\ket{0} + \beta\ket{1})  (\ket{0}\ket{y} + (-1)^{x}\ket{1}\ket{\bar{y}})  \\ \nonumber
=& \frac{\alpha(\ket{00y} + (-1)^{x}\ket{01\bar{y}}) + \beta(\ket{10y} + (-1)^{x}\ket{11\bar{y}})}{\sqrt{2}}.
\end{align}
After Alice applies a CNOT gate to her qubits, the state can be found in
\begin{equation}
\frac{\alpha(\ket{00y} + (-1)^{x}\ket{01\bar{y}}) + \beta(\ket{11y} + (-1)^{x}\ket{10\bar{y}})}{\sqrt{2}}. 
\end{equation} 
This is equivalent to 
\begin{equation}
\frac{\alpha\ket{0}(\ket{0y}+(-1)^{x}\ket{1\bar{y}})}{\sqrt{2}} + \frac{\beta\ket{1}(\ket{1y}+(-1)^{x}\ket{0\bar{y}})}{\sqrt{2}}.
\end{equation}
Alice proceeds to apply the relevant Hadamard gate which provides the result
\begin{eqnarray}
\frac{1}{2} \Big( \ket{00}(\alpha\ket{y} + \beta(-1)^{x}\ket{\bar{y}}) + \ket{01} (\alpha (-1)^{x} \ket{\bar{y}} + \beta \ket{y})
\nonumber
\\
+ \ket{10}(\alpha\ket{y} - \beta(-1)^{x}\ket{\bar{y}}) + \ket{11} (\alpha (-1)^{x} \ket{\bar{y}} - \beta \ket{y}) \Big).
\nonumber
\end{eqnarray}
At this step, Alice measures her qubits to get her result.  If Alice's result is $ab=xy$, meaning it coincides with the Bell state used $\ket{\beta_{xy}}$, then Bob has to do nothing and he has the desired state $\ket{\psi}$ (the exception is if the initial Bell state used was $\ket{\beta_{11}}$ in which case Bob has to apply operator ($-\sigma_{0}$) to get $\ket{\psi}$ if result is $11$).  This is why the contestant's initial choice relates to the Bell state used. As an example, if the initial Bell state was $\ket{\beta_{01}}$ and Alice's measurement outcome was bits $01$, then Bob's state is automatically in $\ket{\psi}=\alpha\ket{0} + \beta\ket{1}$.

In this Monty Hall protocol, Alice sends Bob two bits as in (\ref{Alice}) with the following modification: she chooses two bits denoted $cd$  (ie goat door) that are not $xy$ (ie contestant's initial choice) and are not $ab$ (ie prize door). Should Bob do nothing, or apply one of the possible operators (which depend on what Bell state was used) to get $\ket{\psi}$ ie should the contestant stick or switch?

To answer this, let $B_{xy}$ be the door chosen by contestant.  For this example, assume we use $\ket{\beta_{00}}$, hence $P(B_{00})=1$.  Let $A_{ab}$ be the prize door and due to Born probabilities we have $P(A_{ab})= 1/4$.  Let $C_{cd}$ be the goat door opened by Monty whose probabilities, from the protocol description, work out as:   
\begin{equation}
P(C_{cd} \, | \, B_{00} , A_{ab}) = 
\begin{cases}
\frac{1}{3},& \text{if } 00 = ab \neq cd,\\
\frac{1}{2},& \text{if } 00 \neq ab \neq cd, \\
0, & \text{otherwise}.
\end{cases}
\end{equation}
If Bob always does nothing (ie, stick strategy), then 
\begin{equation}
P(\text{win if stick}) = \sum_{ab = 00 \neq cd} P(A_{ab} , B_{00} , C_{cd}) =  \frac{2}{8}.
\end{equation}
Suppose Bob decides to always apply one of the two operators (ie, switch strategy). 
Then there are one of two possibilities which we denote $ef$ and given its a random choice, each occur with probability $1/2$.  Let $D_{ef}$ represent that door, and $P(\text{win if switch})$ is
\begin{equation} \sum_{ab = ef \neq cd \neq 00} P(A_{ab} , B_{00} , C_{cd} , D_{ef}) =  \frac{3}{8}.
\end{equation}
This means Bob should apply one of the two operators (switch) rather than do nothing (stick) to get state $\ket{\psi}$.

\textbf{d) Unreliable teleportation: }  The effect of noise has been widely analyzed for teleportation~\cite{fortes2015fighting, fortes2016probabilistic, knoll2014noisy, carlo2003teleportation, kumar2003effect}.  In this part, we present a second modification \cite{rajan2019quantum} of quantum teleportation, involving noise, that is \textit{part of the original component of this thesis} (which was done in collaboration with my supervisor).  Consider the standard teleportation protocol with the following unreliability: one of the two bits (either the first or second) Alice sends to Bob in (\ref{Alice}) is received but the other is lost; each event occurs with probability $1/2$.  If the initial Bell state is $\ket{\beta_{00}}$ and Alice's result is $00$, then Bob can do nothing. But in this scenario, if Bob receives the single bit as $1$, then the possible options are $01, 10, \text{or }11$; in this case he should apply one of the operators (switch).  If Bob receives bit $0$, then his options are $00, 01, 10$.  Should he stick (to $00$) or switch (to $01$ or $10$)?  To answer this, let us use the notation developed in the Monty Hall protocol.

We have $P(B_{00})= 1$ and $P(A_{ab}) = 1/4$ . Let $d$ in $C_{d}$ be the single bit received by Bob; based on the scenario described above, we have $P(C_{0} \, | \, B_{00} \, A_{00}) = 1,$ $P(C_{0} \, | \, B_{00} , A_{01}) = 1/2$, and $P(C_{0} \, | \, B_{00} , A_{10}) = 1/2$.  We can compute the probability that Bob receives bit $0$:
\begin{equation*}
P(\text{received bit 0}) = \sum_{ab \neq 11}  P(C_{0} , B_{00} , A_{ab}) = \frac{1}{2}.
\end{equation*} 
If Bob decides to always do nothing then this would be like a sticking strategy.  The probability that bit $0$ is received and Bob wins by sticking is given by $P(A_{00} , B_{00} , C_{0}) = 1/4$.  Hence we can compute the conditional probability:
\begin{equation}
P(\text{win if stick} \, | \, \text{received bit 0}) = \frac{1/4}{1/2} = \frac{1}{2}.
\end{equation}  
If an always switching strategy is adopted, then there are two possibilities ($01$ or $10$) each occuring with probability $1/2$.  In this case probability of winning if switched and bit $0$ is received is given by $P(A_{01} , B_{00} , C_{0} , D_{01}) + P(A_{10} , B_{00} , C_{0} , D_{10}) = 1/8$. With that we compute,
\begin{equation}
P(\text{win if switch} \, | \, \text{received bit 0}) = \frac{1/8}{1/2} = \frac{1}{4}. 
\end{equation}
It is an advantage to stick ie Bob should do nothing.  This strategy may serve to be useful in an error-correcting design for reliability issues in practical quantum networks~\cite{simon2017towards, ren2017ground}

\subsection{Quantum cryptography}

\textbf{a) Preliminaries: }  The secure exchange of messages in classical communications is mainly carried out using public key cryptography, which was described in Chapter \ref{chap: classical}.  However, as we shall show in the next section on quantum computing, a dramatic result is that a scalable quantum computer would be able to break public key cryptography by solving prime factorization.  This discovery has radically changed the landscape of cryptographic research.  There are various investigations that aim to build information security systems based on mathematical problems that many believe a quantum computer would not be able to solve.  These are collectively referred to as post-quantum cryptography \cite{chen2016report, alagic2019status}.  Perhaps the greatest drawback of this set of solutions is that their durability can be questioned; it may be the case that one finds a way for a quantum computer to solve such problems in the future; more precisely stated, there are no formal proofs that such solutions are secure against a quantum computing attack.

Besides public key cryptography, another classical method to encrypt and decrypt messages is through private key cryptography like the one-time pad.  In this scenario, Alice and Bob each possess an identical copy of a random string of bits known as the private key.  More crucially, only they are aware of the key values and keep them in secret.  As long as the key is kept in secrecy, this method is is shown to be provably secure.  When Alice wants to transmit a secure message to Bob, she encrypts the message using this private key by adding the random key bits to the message.  When Bob receives the encrypted message, he decrypts it by subtracting the key bits, using his own copy of the private key.  Hence the problem of transmitting secure messages can be reduced to the problem of how can Alice and Bob acquire pre-established perfectly correlated random keys that an adversary would not be able to acquire.  This subroutine can be accomplished using quantum information, and is known a quantum key distribution (QKD) or quantum cryptography \cite{gisin2002quantum}.  The keys generated by this task are guaranteed to be secure through the properties of quantum information, and hence through the laws of physics; this is in vast contrast to the security of public key cryptography which is based on the difficulty of solving certain mathematical problems.  We discuss two protocols that implement QKD, where the second involves an entanglement in space. 

\textbf{b) BB84 protocol: } In this scenario \cite{bennett1984quantum, mcmahon2007quantum}, our task is for Alice and Bob to acquire identical private keys. There are two communication channels between Alice and Bob. The first is a quantum communication channel that transmits qubits, while the second is a public classical channel for transmitting bits.  In this protocol, we consider the two bases, $\{\ket{0}, \ket{1}\}$ and $\{\ket{+}, \ket{-}\}$.  They are related to one another in the following way: 
\begin{equation}\label{bb84basis1}
\ket{0} = \frac{\ket{+}+\ket{-}}{\sqrt{2}}, \quad \ket{1} = \frac{\ket{+}-\ket{-}}{\sqrt{2}}, 
\end{equation}  
\begin{equation}\label{bb84basis2}
 \ket{+} = \frac{\ket{0}+\ket{1}}{\sqrt{2}}, \quad \ket{-} = \frac{\ket{0}-\ket{1}}{\sqrt{2}}.
\end{equation} 
We also use the encoding that logical $0$ is represented by states $\ket{0}$ and $\ket{+}$, whereas logical $1$ is represented by states $\ket{1}$ and $\ket{-}$.  Alice creates a random string of classical bits.  She encodes this into a corresponding string of qubits using the code.  She sends these qubits through the quantum channel to Bob.  From there, Bob chooses to measure each qubit in either the $\{\ket{0}, \ket{1}\}$ basis or the $\{\ket{+}, \ket{-}\}$ basis; he makes this choice, for each qubit, randomly.  

All the four states, in (\ref{bb84basis1}) and (\ref{bb84basis2}), are not mutually orthogonal, and therefore there is no quantum measurement to distinguish each of them with certainty.  This creates two cases.  If Alice and Bob used the same bases for their respective tasks, then their results are perfectly correlated.  As an example if Alice prepared state $\ket{0}$ and Bob measures in basis $\{\ket{0}, \ket{1}\}$, then he will find state $\ket{0}$ with certainty.  If on the other hand they used different bases, there is a chance of an error.  As an example if Alice prepared state $\ket{0}$ and Bob measures in the $\{\ket{+}, \ket{-}\}$ basis, then there is a probability of $1/2$ of Bob obtaining the incorrect state $\ket{1}$.  From these measurements and the code, Bob obtains a corresponding string of classical bits.   

From there, Alice and Bob proceed to employ the classical channel to tell each other what basis was used at each position.  They discard the bits in their strings where they used a different basis for their respective quantum tasks.  As a consequence, both Alice and Bob end up with perfectly correlated classical private keys whose values are only known to them.

\textbf{c) Security analysis: } Suppose an adversary, named Eve, is attempting to obtain information about the private key.  There are a number of features of quantum information that make this impossible.  We have seen that Alice and Bob use the classical communication channel to share what basis was used at each position.  This information can be public since it cannot be used to infer what the prepared and measured value of the qubit was at the respective position. In regards to the quantum communication channel, Eve cannot copy the qubits transmitted due to the no-cloning theorem.  Even more striking is that it is impossible for Eve gain any information on non-orthogonal qubits without introducing a disturbance on the signal.  This is why we have used non-orthogonal states in the protocol.  More broadly speaking,
\begin{proposition}\label{BB84nodisturbance}
	\textbf{(Information gain implies disturbance)} 
	In any attempt to distinguish between two non-orthogonal quantum states, information gain is only possible at the expense of introducing disturbance to the signal.
	(See e.g. \cite{nielsen2002quantum}.) 
\end{proposition}
Hence at the end of the protocol, Alice and Bob select a subset of bits from their final strings to compare the values.  If more than an acceptable number disagree, they abort the protocol and try again.  

\textbf{d) E91 protocol: }  Having introduced the need for correlations in the BB84 protocol, it seems appropriate to ask whether the \textit{interdependence in an entanglement in space} can be used for a QKD protocol?  It turns out that such an answer was first developed in \cite{ekert1991quantum}.  This is known as the E91 protocol. It utilizes the Bell state
\begin{equation}\label{E91Bellstate}
\ket{\beta_{00}} = {1\over\sqrt2}(\ket{00}+\ket{11}),
\end{equation} 
which can be re-written in terms of the $\{\ket{+}, \ket{-}\}$ basis as
\begin{equation}\label{E91Bellstate2}
\ket{\beta_{00}} = {1\over\sqrt2}(\ket{++}+\ket{--}).
\end{equation} 
Once again the task is for Alice and Bob to acquire identical private keys made up of random values.  In this protocol, one qubit from state $\ket{\beta_{00}}$ is held by Alice and the other qubit by Bob.  By considering both (\ref{E91Bellstate}) and (\ref{E91Bellstate2}), it appears that if Alice and Bob measure in the same basis, then their results are perfectly correlated and yet random.  Furthermore, to obtain an appropriate key length, we suppose that Alice and Bob share many copies of $\ket{\beta_{00}}$ and repeat the measurement procedure over many rounds; they use a public classical channel to randomly agree to measure in either basis $\{\ket{0}, \ket{1}\}$ or $\{\ket{+}, \ket{-}\}$.

In the BB84 protocol, the key is fundamentally produced by Alice and sent to Bob (before measurement and the removal of some results).  In the E91 protocol, Alice and Bob can measure their qubits simultaneously and obtain their respective keys.  \textit{The lack of a time interval involved} suggests that the key is not fundamentally distributed, in any way, from one location to another like in BB84.  Rather identical keys are generated at same time at two different locations, and whose values cannot be pre-determined by Alice nor Bob!  

\textbf{e) Security analysis: } We provide a brief outline of a security proof \cite{coles2017entropic, berta2010uncertainty} for E91.  In order to do accomplish this result, we have to show that the following two statements are mutually exclusive:
\begin{enumerate}[noitemsep, topsep=0pt, label=\roman*)]
	\item The measurement results between Alice and Bob agree on most rounds. 
	\item An adversary, whom we can name Eve, possesses a large amount of information on the results of either Alice or Bob.      
\end{enumerate}
In the protocol, we assumed that Alice and Bob share state $\ket{\beta_{00}}$.  However, it may be the case that Eve interfered.  Hence, let $\rho_{ABE}$ represent a density operator where $A$ represents Alice's qubit, $B$ represents Bob's qubit and $E$ signifies any quantum subsystem acquired by Eve.  Let $\Theta$ be a mixed state whose role is to be a binary register which signifies whether the qubits are to be measured in basis $\{\ket{0}, \ket{1}\}$ or $\{\ket{+}, \ket{-}\}$.  Furthermore, let $Y$ denote the measurement results of Alice, and let $\bar{Y}$ denote the measurement results Bob obtains.  Alice and Bob measure their system in the basis as indicated by $\Theta$; we assume that Eve also holds state $\Theta$.  We first consider an analysis on Alice's results.  Using entropic concepts, we obtain
\begin{align}
S(Y|B\Theta) =& \frac{1}{2} \, S(X|B) + \frac{1}{2} \,S(Z|B), \\
S(Y|E\Theta) =& \frac{1}{2}\,  S(X|E) + \frac{1}{2} \, S(Z|E).
\end{align} 
Applying the tripartite entropic uncertainty principle (\ref{tripartiteentropicrelation}) with quantum memory results in,
\begin{equation}\label{entropicuncertaintycrypto}
S(Y|B\Theta) + S(Y|E\Theta) \geq 1,
\end{equation}
given $q_{MU}=1$ for bases $\{\ket{0}, \ket{1}\}$ and $\{\ket{+}, \ket{-}\}$.  In Bob's case, the following result can be derived 
\begin{equation}
S(Y|B\Theta) \leq H(Y|\bar{Y}). 
\end{equation}
This brings us to the final result 
\begin{equation}\label{E91security}
S(Y|E\Theta) \geq 1 - H(Y|\bar{Y}).
\end{equation}
The interpretation of (\ref{E91security}) is that the von Neumann entropy relating to Eve's uncertainty is small given that the conditional entropy between Alice and Bob $H(Y|\bar{Y})$ is large; this provides the necessary expression to show that the two statements of the proof are mutually exclusive as required.  The argument can be extended to multiple rounds.

\textbf{f) Comments: }
\begin{enumerate}[noitemsep, topsep=0pt, label=\roman*)]
	\item QKD protocols are rigorously proven to be secure using the laws of physics.  A formal definition for security along with a proof can be found in \cite{nielsen2002quantum}. Hence, unlike classical cryptography, quantum cryptography provides a guaranteed level of protection against a quantum computing attack.
	\item It is of noteworthy interest that the initial idea of using quantum physics for cryptography was first formulated in a design for money bank notes that would be impossible to forge.  However this work was rejected for publication.  For a historic and broad review of the subfield of quantum cryptography, we refer the reader to \cite{gisin2002quantum}.
	\item Other than Bell states, the more general GHZ states have also been employed in cryptographic settings.  One notable example is in a quantum version of the classical secret sharing protocol \cite{hillery1999quantum}.
	\item Quantum key distribution systems have moved from theory to commercial reality.  There are a number of quantum cryptographic companies deploying these systems in the public and private sector.
\end{enumerate}

\section{Application: Quantum Computing}

We have witnessed the use of entanglement in space in quantum information to perform communication tasks that would be classically unimaginable.  In this section, we introduce the notion of a quantum computer \cite{benioff1980computer, yu1980computable, feynman1982simulating} that harnesses quantum information, which includes an entanglement in space resource \cite{jozsa2003role, vidal2003efficient, vedral2006introduction}, to solve computational problems. There are a number of different models for quantum computation such as the gate model \cite{nielsen2002quantum}, the adiabatic model \cite{farhi2000quantum, johnson2011quantum}, the topological model \cite{pachos2012introduction} and the one-way measurement model \cite{raussendorf2003measurement}.  Our sole focus is on the gate model which is based on the quantum circuit model presented in Chapter \ref{chap: QInfo}.  Moreover, we present three quantum algorithms that remarkably outperform the best known classical algorithms for the same task.  However, it is not formally proven that quantum computers are more powerful than classical computers; it may very well be the case that we find classical algorithms that are equivalent in computational performance.  Nevertheless, small-scale quantum computing systems have been experimentally realized and shown \cite{arute2019quantum} to drastically outpeform the world's best classical supercomputers.  For a broader survey of quantum algorithms, we refer the reader to \cite{montanaro2016quantum, childs2010quantum}.      

\subsection{Quantum search}

\textbf{a) Grover's algorithm: } The classical computational problem of search involves finding $M$ solutions in a search space of $N$ elements, where $1 \leq M \leq N$. To solve this problem using a quantum computer \cite{grover1997quantum}, we encode each of these $N$ elements into a quantum state $\ket{x}$, and create the following state that is in an equal superposition
\begin{equation}\label{groverinitial}
\ket{\psi} = \frac{1}{N^{\frac{1}{2}}}\sum_{x=0}^{N-1}\ket{x}.
\end{equation}   
Suppose a single solution is marked as $x'$.  Then the goal of the quantum computer is to transform state $\ket{\psi}$ into the state $\ket{x'}$ using the fewest number of operations and measurements.  To perform this task, we simply need to construct an operator known as the Grover operator,
\begin{equation}
G = (2\ket{\psi}\bra{\psi} - I)O,
\end{equation}
where $(2\ket{\psi}\bra{\psi} - I)$ can be constructed using (\ref{groverinitial}).  The operator $O$ is known as the oracle; the action of the oracle is given by 
\begin{equation}
O\ket{x} = (-1)^{f(x)}\ket{x},
\end{equation}
where $f(x')=1$, and otherwise the function evaluates to zero for all other $x$.  It is important to note that the oracle can only recognize the solution to the search problem.  There is a clear distinction between recognizing the solution and knowing the solution.  The former does not mean the latter.  

The quantum algorithm is straightforward in that it consists of repeatedly applying the Grover operator $\pi \sqrt{N}/4$ times on state $\ket{\psi}$.  This transforms the quantum information in $\ket{\psi}$ such that when measured gives with high probability the result $\ket{x'}$.  To see this to be the case, let $\sum_{x}^{'}$ represent the sum over all $x$ which are solutions, and $\sum_{x}^{''}$ represent the sum over all $x$ which are not solutions.  We can construct the following normalized quantum states
\begin{align}\label{groveralphabeta}
\ket{\alpha} &= \frac{1}{\sqrt{N-M}}\sum_{x}^{''}\ket{x}, \\
\ket{\beta} &= \frac{1}{\sqrt{M}}\sum_{x}^{'}\ket{x}.
\end{align} 
We can re-express initial state $\ket{\psi}$ of the quantum computer, represented in (\ref{groverinitial}), in terms of $\ket{\alpha}$ and $\ket{\beta}$,
\begin{equation}\label{initialinalphabeta}
\ket{\psi} = \sqrt{\frac{N-M}{N}}\ket{\alpha} + \sqrt{\frac{M}{N}}\ket{\beta}.
\end{equation}
Furthermore, let
\begin{equation}
\cos \frac{\theta}{2} = \sqrt{\frac{N-M}{M}},
\end{equation}
so that we can write (\ref{initialinalphabeta}) as
\begin{equation}
\ket{\psi} = \cos \frac{\theta}{2}\ket{\alpha} + \sin \frac{\theta}{2}\ket{\beta}.
\end{equation}
The effect of applying the Grover operator on the initial state results in
\begin{equation}
G\ket{\psi} = \cos \frac{3 \theta}{2}\ket{\alpha} + \sin \frac{3 \theta}{2}\ket{\beta}.
\end{equation}
The repeated iteration of the operator on the state computes to
\begin{equation}
G^{k}\ket{\psi} = \cos\Bigg({\frac{(2k+1)\theta}{2}}\Bigg)\ket{\alpha} + \sin\Bigg({\frac{(2k+1)\theta}{2}}\Bigg)\ket{\beta}.
\end{equation}
This has the requirement of transforming the state $\ket{\psi}$ to $\ket{\beta}$. More precisely, the number of iterations required is upper bounded by 
\begin{equation}
\frac{\pi}{4}\sqrt{\frac{N}{M}}.
\end{equation}
After this repeated application of $G$ on initial state $\ket{\psi}$, a measurement in the computational basis provides the answer to the problem with a high probability. 

\textbf{b) Analysis of algorithm: }  In Chapter \ref{chap: classical}, we provided a brief overview of the asymptotic notation.  Using those tools, it can be said that Grover's algorithm requires $O(\sqrt{N/M})$ operations for an $N$ item search problem with $M$ solutions.  For the case of a single solution, this equates to $O(\sqrt{N})$.  A classical computer for the same single solution case requires $O(N)$ operations.  To highlight the shocking aspect of this situation, consider a search space of a million items; a classical computer would need to, at worst, go through all million of them whereas a quantum computer simply needs to search through, at worst, a thousand of them; this is remarkable as there does not seem to be any geometric structure in the problem to offer such a quadratic speed up.

\subsection{Quantum factoring}

\textbf{a) Preliminaries: } Perhaps the most influential result in quantum information science is Shor's algorithm \cite{shor1994algorithms, shor1999polynomial}:  

\begin{enumerate}[noitemsep, topsep=0pt, label=\roman*)]
	\item This is a quantum algorithm that can efficiently derive the prime factorization of an integer.  In Chapter \ref{chap: classical}, we saw that the reliability of public key cryptography is based on the hypothesis that prime factorization cannot be computed in any reasonable time.  Hence, Shor's algorithm has dramatic consequences on the information infrastructure of the modern world.
	\item Shor's algorithm provided an concrete instantiation of the notion that quantum computer could be far more powerful than classical computers on real-world problems.  The faith in this idea led to a drastic growth in the theoretical and experimental progress of the quantum computation.
	\item Shor's algorithm has to some degree directed the attention of the subject of cryptography, rooted in number theory and abstract algebra \cite{galbraith2012mathematics, katz2014introduction} towards quantum physics.  Hence, broader development of quantum information science has, in some part, been predicated on cryptographic aims.
\end{enumerate}

Before stating the computational steps of Shor's algorithm, we aim to discuss two of its subroutines.  The first is the implementation of a quantum version of a discrete Fourier transform.  Suppose we have a quantum computer represented by a Hilbert space with an orthonormal basis $\ket{0}, \dots, \ket{N-1}$.  Then the action of the quantum Fourier transform on arbitrary state of the computer is given by
\begin{equation}
\sum_{j=0}^{N-1} x_{j}\ket{j} \rightarrow \sum_{k=0}^{N-1} y_{k}\ket{k},
\end{equation}
where
\begin{equation}
y_{k} \equiv \frac{1}{\sqrt{N}}\sum_{j=0}^{N-1}x_{j}e^{2\pi ijk/N}.
\end{equation}
The quantum information $y_{k}$ is the discrete Fourier transform of the quantum information $x_{j}$.  The quantum Fourier transform can be expressed in terms of a sequence of qubit gates and can be shown to be unitary.  For the case of a single basis state and where $N=2^{n}$, the action of the quantum Fourier transform can be written as
\begin{equation}\label{QFTranform}
\ket{j} \rightarrow \frac{1}{2^{n/2}}\sum_{k=0}^{2^{n}-1}e^{2\pi ijk/2^{n}}\ket{k}.
\end{equation}
We adopt the following two notations: For a state $\ket{j}$, we express it in terms of binary representation $j=j_{1}j_{2}\dots j_{n}$ meaning $j=j_{1}2^{n-1} + j_{2}2^{n-2} + \dots j_{n}2^{0}$; we also use the notation $0.j_{l}j_{l+1}\dots j_{m}$ to represent $j_{l}/2 + j_{l+1}/4 + \dots j_{m}/2^{m-l+1}$.  

Hence, we can expand the output in (\ref{QFTranform}) as
\begin{align}
\ket{j} &\rightarrow \frac{1}{2^{n/2}}\sum_{k=0}^{2^{n}-1}e^{2\pi ijk/2^{n}}\ket{k} \\
&= \frac{1}{2^{n/2}} \sum_{{k}_{1}=0}^{1} \cdots \sum_{{k}_{n}=0}^{1}e^{2\pi ij(\sum_{l=1}^{n}k_{l}2^{-l})}\ket{k_{1}\dots k_{n}} \\
& = \frac{1}{2^{n/2}} \sum_{{k}_{1}=0}^{1} \cdots \sum_{{k}_{n}=0}^{1} \bigotimes_{l=1}^{n}e^{2\pi ijk_{l}2^{-l}}\ket{k_{l}} \\
&= \frac{1}{2^{n/2}} \bigotimes_{l=1}^{n} \Bigg( \sum_{{k}_{l}=0}^{1} e^{2\pi ijk_{l}2^{-l}}\ket{k_{l}}\Bigg) \\
&= \frac{1}{2^{n/2}}  \bigotimes_{l=1}^{n} \Bigg(\ket{0} + e^{2\pi ij 2^{-l}}\ket{1} \Bigg) \\
\end{align}
This can be expanded into what is known as the product representation of the quantum Fourier transform
\begin{equation}\label{QFTproductform}
\frac{(\ket{0} + e^{2\pi i0.j_{n}}\ket{1})(\ket{0} + e^{2\pi i0.j_{n-1}j_{n}}\ket{1})\cdots (\ket{0} + e^{2\pi i0.j_{1}j_{2}\dots j_{n}}\ket{1})}{2^{n/2}}.
 \end{equation}  
The second subroutine in Shor's algorithm is known as phase estimation; the computational task is that given a unitary operator $U$ with an eigenvector $\ket{u}$, find the unknown value $\varphi$ in the corresponding eigenvalue $e^{2\pi i \varphi}$.  For simplicity, assume that $\varphi$ can be written in $t$ bits as $\varphi = 0.\varphi_{1}\dots \varphi_{t}$.  The quantum computer starts with in the state
\begin{equation}
\ket{0}^{\otimes t}\ket{u},
\end{equation}
where the first register contains $t$ qubits in state $\ket{0}$ and the second register contains the eigenvector.  We proceed to apply the Hadamard transform to the first register to obtain
\begin{equation}
\Bigg(\frac{\ket{0} + \ket{1}}{\sqrt{2}} \Bigg)^{\otimes t} \ket{u} = \frac{1}{2^{t/2}}\sum_{j=0}^{2^{t}-1}\ket{j}\ket{u}.
\end{equation}
Recall that $U \ket{u} = e^{2\pi i \varphi} \ket{u}$. This implies that when we apply a controlled-$U$ operation on the second register with $U$ raised to successive powers of two, the state results in 
\begin{equation}
\frac{(\ket{0} + e^{2\pi i0.\varphi_{t}}\ket{1})(\ket{0} + e^{2\pi i0.\varphi_{t-1}\varphi_{t}}\ket{1})\cdots (\ket{0} + e^{2\pi i0.\varphi_{1}\varphi_{2}\dots \varphi_{t}}\ket{1})}{2^{t/2}} \ket{u}.
\end{equation}
After this step, we apply the inverse  of the quantum Fourier transform (\ref{QFTproductform}) to obtain the desired output $\ket{\varphi_{1} \dots \varphi_{t}}$.

\textbf{b) Shor's algorithm: } The classical computational problem is to find the prime factorization of an integer $N$.  This problem is equivalent to the order-finding problem which can be described as follows.  Suppose $x$ and $N$ are positive integers with no common factors and where $x < N$.  The order of $x$ modulo $N$ is the smallest positive integer, $r$, such that 
\begin{equation}
x^{r} = 1 \, (\text{mod }N).
\end{equation}
The order-finding problem is that given $x$ and $N$, determine $r$.  Showing the mathematical equivalence of these two problems is beyond the scope of this thesis.  Assuming this equivalence, Shor's algorithm can be seen as a classical algorithm with a quantum subroutine for order-finding.  Every step in the following algorithm can be performed efficiently on a classical computer except the quantum subroutine.  Over the course of repeating the algorithm, the complete prime factorization of $N$ can be computed.

The first step of the algorithm is check if $N$ is even, and if so return the factor $2$.  The second step is to use a known classical algorithm determine whether $N=a^{b}$ for integers $a\geq 1$ and $b \geq 2$, and if so return the factor $a$.  The third step is to randomly choose an $x$ in the range $1$ to $N-1$.  If $\text{gcd}(x, N)>1$, then return the factor $\text{gcd}(x, N)$.  The fourth step is the quantum order-finding subroutine to derive the order $r$ of $x$ modulo $N$. The last step is if $r$ is even and $x^{r/2}\neq -1(\text{mod} N)$, then compute both $\text{gcd}(x^{r/2}-1, N)$ and  $\text{gcd}(x^{r/2}+1, N)$; check if any one of these is a factor and output that factor; otherwise the algorithm fails.  

Shor's algorithm crucially depends on the quantum subroutine for order-finding, which we now describe.  We encode the order-finding problem into the quantum computer as unitary operator,
\begin{equation}
U\ket{y} \equiv \ket{xy(\text{mod N})},
\end{equation}
where $y \in \{0,1\}^{L}$.  The eigenvectors of $U$ can expressed as
\begin{equation}
\ket{u_{s}} \equiv \frac{1}{\sqrt{r}}\sum_{k=0}^{r-1}\text{exp}\Bigg(\frac{-2\pi i s k}{r}\Bigg) \ket{x^{k}\text{mod} \, N},
\end{equation}
for integer $0 \leq s \leq r-1$. Then the eigenvalues can be written in the following equation as
\begin{equation}
U\ket{u_{s}} = \text{exp}\Bigg(\frac{2\pi i s}{r}\Bigg) \ket{u_{s}}.
\end{equation}
After the encoding, we can use phase estimation to obtain $s/r$ in the eigenvalues $\text{exp}(2\pi i s/r)$.  After that, we can use a procedure known as the continued fractions algorithm to efficiently obtain the order $r$.

\textbf{c) Analysis of algorithm: } The best known classical algorithm for the task of prime factorization of an $n$-bit integer is the number field sieve which requires $\text{exp}(\Theta (n^{1/3} \text{log}^{2/3} n))$ operations.  Shor's algorithm is exponentially faster than this as it can be shown that it performs the same task in $O(n^{2} \text{log} \, n \, \text{log} \, \text{log} \, n)$ operations.  The concrete output at the end of this algorithm makes the fundamental question of what is quantum information unavoidable to easily dismiss; it is natural to ask, in this case, where do the quantum computations in Shor's algorithm physically take place \cite{deutsch1998fabric}?  In this way, quantum information science can be seen to provide a resurgence on the historical inquiry \cite{infeld1971evolution} regarding the fundamental nature of quantum physics as a whole.

\subsection{Quantum machine learning}

\textbf{a) Preliminaries: } Machine learning \cite{goodfellow2016deep} is a relatively new field of computer science with the goal to significantly advance artificial intelligence.  The central tenet of the field is that computational machines can `learn' from large data sets to perform tasks (traditionally assigned to only humans) as opposed to being explicitly programmed to do so.  The area has found many real-world applications, as well as exhibiting its progress in the domain of games; recently \cite{silver2016mastering, silver2017mastering} a machine learning system defeated the world champion in the game of Go.  Machine learning can be crudely separated into supervised and unsupervised learning. In the former, each piece of data with a corresponding labelled category, collectively known as the training set, is provided to the machine; using this, the machine is supposed to carry out the task of correctly labelling data that exists outside of the training set.  This is in contrast to unsupervised learning where the training set does not contain any categories; rather the machine is supposed to find natural categories in which the data could be indexed under; furthermore the machine is then tasked with classifying data outside of that training set. 

Within quantum information science, there has been an effort to investigate whether quantum computers could outperform classical computers to implement machine learning \cite{biamonte2017quantum, schuld2019quantum, havlivcek2019supervised}.  One prominent example of this is in relation to a supervised learning algorithm known as a support vector machine.  A quantum support vector machine was designed in \cite{rebentrost2014quantum} with a drastic improvement over the classical case. Central to that work as well as many other quantum machine learning algorithms is the HHL algorithm.  This is a quantum algorithm that efficiently performs matrix inversion on data such as a training set.

\textbf{b) HHL algorithm: } The classical computational problem is that given $N \times N$ complex matrix $A$ and a vector  $\vec{b} \in \mathbb{C}^{n}$, solve for $\vec{x}\in \mathbb{C}^{n}$ in the equation $A\vec{x} = \vec{b}$.  In other words, derive $A^{-1}$ to compute $\vec{x} = A^{-1}\vec{b}$.  For our discussion, assume that $A$ is also Hermitian.  However, this can easily be generalized if we let
\begin{equation}
C= \begin{pmatrix}
0 & A \\
A^{\dagger} & 0
\end{pmatrix},
\end{equation}
where one can solve the equation
\begin{equation}
C\vec{y} = 
\begin{pmatrix}
\vec{b}  \\
0 
\end{pmatrix}.
\end{equation}
This results in the solution 
\begin{equation}
\vec{y} = 
\begin{pmatrix}
0  \\
\vec{x} 
\end{pmatrix}.
\end{equation}
Of noteworthy importance is result that matrix $A$ has an eigenvalue $\lambda$ if and only if $A^{-1}$ has eigenvalue $\lambda^{-1}$.  Hence, if $A$ is diagonalizable then computing the inverse of the eigenvalues allows us to construct $A^{-1}$ in an analogous way.  

We want to encode this classical problem onto a quantum computer.  Let $A$ be the corresponding Hermitian operator with eigenbasis $\ket{u_{j}}$ with corresponding eigenvalues $\lambda_{j}$. Moreover, we encode the $N$ variable vector $\vec{b}$ into a quantum state,
\begin{equation}
\ket{b} = \sum_{i=1}^{N} b_{i} \ket{i},
\end{equation}
using $\log_{2} N$ qubits.  Our goal is to construct 
\begin{equation}
\ket{x} = A^{-1}\ket{b},
\end{equation}   
where $\ket{x}$ encodes the solution $\vec{x}$ over $\log_{2} N$ qubits.

The first step of the algorithm is to use a version of phase estimation to decompose the state $\ket{b}$ into the eigenbasis of $A$ and obtain the eigenvalues of $A$.  Roughly this amounts to applying the unitary operator $e^{iAt}$ on state $\ket{b}$ for a superposition of different times $t$.  After this phase estimation stage, we can represent $\ket{b}$ as
\begin{equation}
\ket{b} = \sum_{j=1}^{N} \beta_{j} \ket{u_{j}},
\end{equation}     
and the total state can be, informally, written as
\begin{equation}\label{HHLinformalstate}
 \sum_{j=1}^{N} \beta_{j} \ket{u_{j}}\ket{\lambda_{j}}.
\end{equation}
For a more precise description of this step, we introduce the state 
\begin{equation}
\ket{\Psi_0} := \sqrt{\frac{2}{T}}\sum_{\tau=0}^{T-1}\sin{\frac{\pi(\tau + \frac{1}{2})}{T}}\ket{\tau},
\end{equation}
for some large $T$; this is to minimize a quadratic loss function which we will not concern us here.  Hence we can express the process discussed more accurately as
\begin{equation}
\sum_{\tau=0}^{T-1} \ket{\tau}\bra{\tau} \otimes e^{\frac{iA\tau t_0}{T}}  (\ket{\Psi_0} \otimes \ket{b}),
\end{equation}
where $t_{0}$ is dependent on the condition number (ratio between $A$'s largest and smallest eigenvalues) and the additive error achieved in output state $\ket{x}$.  This is followed by a Fourier transform on the first register which gives the state   
 \begin{equation}
\sum_{j=1}^{N}\sum_{k=0}^{T-1}\alpha_{k|j}\beta_j\ket{k}\ket{u_j},
\end{equation}
where states $\ket{k}$ represent the Fourier basis states; the value $\lvert \alpha_{k|j} \rvert$ is large if and only if $\lambda_j \approx (2 \pi k)/(t_0)$.  We proceed to define $\bar{\lambda_k} \equiv (2 \pi k)/(t_0)$ and re-express our $\ket{k}$ register as
		\begin{equation}\label{HHLfirststepaccurate}
			\sum_{j=1}^{N}\sum_{k=0}^{T-1}\alpha_{k|j}\beta_j\ket{\bar{\lambda_k}}\ket{u_{j}}
			\end{equation}
The second step of the algorithm is acquire the inverse of the eigenvalues into the quantum information; this is the critical step as it allows us construct $A^{-1}$.   This is roughly accomplished by performing a linear map taking state $\ket{\lambda_{j}}$ in (\ref{HHLinformalstate}) to 
\begin{equation}\label{HHLrough2nd}
C\lambda_{j}^{-1}\ket{\lambda_{j}},
\end{equation}
where $C$ is some normalization constant.   A more precise description of this step can be described by adding an extra qubit in state $\ket{0}$ to (\ref{HHLfirststepaccurate}).  This extra qubit will be rotated conditioned on state $\ket{\bar{\lambda_k}}$ to produce 	
		\begin{equation}\label{HHLexact2nd}
			\sum_{j=1}^{N}\sum_{k=0}^{T-1}\alpha_{k|j}\beta_j\ket{\bar{\lambda_k}}\ket{u_j}\Biggl(\sqrt{1-{\frac{C^2}{\bar{\lambda_{k}^2}}}}\ket{0} + \frac{C}{\bar{\lambda_k}}\ket{1}\Biggr),
		\end{equation}
where $C$ is chosen based on the condition number of the matrix. This procedure is not unitary so it does have a probability of failure.

The third step of the algorithm is to uncompute the $\ket{\lambda_{j}}$ register in (\ref{HHLrough2nd}) and the quantum computer outputs a state proportional to
\begin{equation}
\sum_{j=1}^{N} \beta_{j} \lambda_{j}^{-1}\ket{u_{j}} = A^{-1}\ket{b} = \ket{x}.
\end{equation}
The more precise description, of this third step, following (\ref{HHLexact2nd}) is to undo phase estimation to uncompute $\ket{\bar{\lambda_k}}$.  For the case where phase estimation is perfect, we have the value $\alpha_{k|j} = 1$ if $\bar{\lambda_k} = \lambda_{j}$, and $0$ otherwise.  Supposing this case, we can write the resulting state 
\begin{equation}
\sum_{j=1}^{N}\beta_j\ket{u_j}\Biggl(\sqrt{1-{\frac{C^2}{\lambda_{j}^2}}}\ket{0} + \frac{C}{\lambda_j}\ket{1}\Biggr),
\end{equation}
and from there measure the last qubit.  Conditioned on seeing the result $1$, we have the final state 
\begin{equation}
\sqrt{\frac{1}{\sum_{j=1}^{N}C^{2}\lvert \beta_{j} \rvert^{2}/ \lvert \lambda_{j} \rvert^{2}}} \sum_{j=1}^{N}\beta_{j}\frac{C}{\lambda_{j}}\ket{u_{j}},
\end{equation}
which corresponds to the state
\begin{equation}
\ket{x} = \sum_{j=1}^{n}\beta_j\lambda_{j}^{-1}\ket{u_j},
\end{equation}
up to a normalization.

The final step is that we can make measurement $M$ which provides us with expectation value $\bra{x}M\ket{x}$; this could be used to estimate features of $\vec{x}$ that we may be interested in.
	
\textbf{c) Analysis of algorithm: } The best known classical algorithm for the task of finding $\vec{x}$ is $O(N \log N)$, where $N$ is the number of variables.  By contrast the HHL quantum algorithm is exponentially better in that it requires takes $O((\log N)^{2})$ steps to find $\ket{x}$.  However, due to the hidden nature of quantum information, we can only output some expectation value $\braket{x|M|x}$ for a measurement $M$, rather than $\ket{x}$ itself.  Nevertheless, the broad applicability of quantum information to develop artificial intelligence is proving to be a promising area of research.  Perhaps more exciting are the investigations on whether quantum information could form the direct basis for biological intelligence \cite{hameroff1996orchestrated, stuart1998quantum, penrose2000large}.

\section{Entanglement in Time}

To introduce an entanglement in time, it is perhaps useful to consider the properties of an entanglement in space.  The latter entanglement involves an interdependence of quantum information systems across a spatial distance.  This trivially implies that the systems cannot be in the same location.  In an entanglement in time, the interdependence of quantum information systems is across a temporal interval.  In its strictest form, this implies that the entanglement is between systems that do not coexist!  Rather remarkably, such an entanglement has recently been experimentally realized \cite{megidish2013entanglement, megidish2012resource, megidish2017quantum}.  

Our aim is to describe this entanglement in time.  Moreover, we will highlight examples of it through a temporal Bell state, a temporal GHZ state and a temporal graph state.  Despite the experimental realization of such states, the role of this temporal effect in quantum information science is largely unknown.  

We aim to provide some insight by arguing for the following overarching theme that was derived solely on the basis of examining the collected literature. \textit{The interdependence in any entanglement in space is shocking due to the absense of a time interval involved.  The interdependence in any entanglement in time is shocking due to the existence of a time interval involved.}  The latter statement is rather non-trivial as an interdependence across time exists for even classical information systems through a causal dependence.  Nevertheless, we will articulate that the interdependence across time for this entanglement is stronger than could ever exist between classical information systems.

We proceed to describe the experimental realization of this temporal effect through the theoretical tools of qubits and density operators.  For a general survey on the experimental procedures involving temporal quantum information systems, we refer reader to \cite{victora2019time}.

\subsection{Preliminaries}

\textbf{a) Entanglement swapping: } To generate an entanglement in time, between subsystems that do not coexist, one uses a modification of an entanglement in space procedure known as entanglement swapping \cite{zukowski1993event}.  The spatial entanglement is essentially transferred from a composite system to a different composite system.  This procedure can also be viewed as a teleportation protocol for a spatially entangled state \cite{zeilinger2017light}.  It has been generalized to multiple swappings \cite{goebel2008multistage}, and entropic analysis of the procedure can be found in \cite{witten2018mini}. It also allows one to extend the range of quantum communication networks through repeaters \cite{briegel1998quantum}.  Of noteworthy importance is a delayed choice version of this procedure \cite{peres2000delayed}, which has recently been experimentally realized \cite{ma2012experimental}.  In this delayed choicce experiment, the choice to transfer the entanglement to a desired composite system is made after the desired system has already been measured.  This results in a portrayal of the total system exhibiting retrocausality (the influence of future actions on past events).  However all the subsystems in this procedure coexisted; the entanglement in time that we will describe between subsystems that do not coexist will prove to be far more bizarre.  For a broad overview on delayed choice experiments, we refer the reader to \cite{ma2016delayed}.

\textbf{b) Quantum optics: } To mathematically describe the optical experimental generation of this entanglement in time, we provide a brief `dictionary' that explains the experimental physics in terms of the theoretical physics: 
\begin{enumerate}[noitemsep, topsep=0pt, label=\roman*)]
	\item A photon with horizontal (h) and vertical (v) polarization states can be seen as a physical instantiation of a qubit,
	\begin{equation}\label{photonqubit}
	\ket{\psi} = \alpha \ket{h} + \beta \ket{v}.
	\end{equation}
	where $\lvert \alpha \rvert^{2} + \lvert \beta \rvert^{2}=1$.  The left and right circularly polarized states are respectively,
	\begin{align}\label{lrstates}
	\ket{l} &\equiv \frac{1}{\sqrt{2}} (\ket{h} + i\ket{v}), \\
	\ket{r} &\equiv \frac{1}{\sqrt{2}} (\ket{h} - i\ket{v}).
	\end{align}
	\item Wave plates are devices that mathematically transform (\ref{photonqubit}) to
	\begin{equation}
	\ket{\psi} = \alpha \ket{h} + e^{i\phi}\beta \ket{v},
	\end{equation}
	where $\phi = \pi$ represents a half-wave plate (HWP) and $\phi = \pi/2$ represents a quarter-wave plate (QWP).
	\item A beam splitter (BS) transforms the two spatially separated inputs $\ket{T}$ and $\ket{B}$ (which refer to top and bottom respectively) to two spatially separated outputs 
	\begin{align}
	\ket{T} &\rightarrow \frac{1}{\sqrt{2}} (\ket{T} + \ket{B}), \\
	\ket{B} &\rightarrow \frac{1}{\sqrt{2}} (\ket{T} - \ket{B}).
	\end{align}  
	\item A polarizing beam splitter (PBS) directs the $\ket{h}$ photons in one direction and directs the $\ket{v}$ photons in another direction.  After the PBS, it is typical that there are two detectors after that measure the photons of the two different polarizations.
	\item Spontaneous parametric down-conversion (SPDC) is a method where a photon with frequence $v$ is converted to two photons with respective frequencies $v_{1}$ and $v_{2}$ such that $v = v_{1} + v_{2}$.  
	\item Post-selection is also known as conditional detection.  The process of SPDC is not certain to happen given that the photon detectors do not have perfect efficiency.  Hence the event is recorded only if two photons are detected.  
\end{enumerate}

\subsection{Through qubits}

\textbf{a) Temporal bipartite systems: }  The aim is to manifest the intended entanglement in time through a bipartite system.  More precisely, the experiment \cite{megidish2013entanglement} generates a temporal version of a Bell state between a pair of photons that do not coexist.  Our review of the experiment begins by considering a PDC to create polarized photons in any of the four spatial Bell states 
\begin{align}
\ket{\phi^{\pm}} &= \frac{1}{\sqrt{2}}(\ket{h_{a}h_{b}} \pm \ket{v_{a}v_{b}} ), \\
\ket{\psi^{\pm}} &= \frac{1}{\sqrt{2}}(\ket{h_{a}v_{b}} \pm \ket{v_{a}h_{b}}).
\end{align}  
The labels $h$ and $v$ represent the respective polarization states, and the spatial modes are denoted $a$ or $b$.  These can simply be viewed as a physical instantiation of the Bell states described in (\ref{Bellstatebasis}) and (\ref{Bellstatebasis2}).  

In order to observe the intended effect, the experiment generates two pairs of photons (1-2 and 3-4) separated by a well-defined time interval $\tau$. Hence, there are a total of four photons across a span of time. The quantum state of such a system can be described as
\begin{equation}
\ket{\psi^{-}}_{a,b}^{0,0} \otimes \ket{\psi^{-}}_{a,b}^{\tau,\tau} = \frac{1}{{2}}(\ket{h_{a}^{0} v_{b}^{0}} - {\ket{v_{a}^{0}h_{b}^{0}}}) \otimes (\ket{h_{a}^{\tau} v_{b}^{\tau}} - {\ket{v_{a}^{\tau}h_{b}^{\tau}}}).
\end{equation}
In this case, the spatial modes are located in the subscripts and the time labels of the photons are in the superscripts.

The aim is to perform a Bell state projection on the second photon of the first pair and the first photon of the second pair.  To achieve this the former particle is delayed in a delay line.  The same delay is also used on the second photon of the second pair.  The resulting state can be reordered and written as
\begin{align}\label{temporalBellstateproject}
	\ket{\psi^{-}}_{a,b}^{0,\tau} \ket{\psi^{-}}_{a,b}^{\tau,2\tau} &= \frac{1}{2}(\ket{\psi^{+}}_{a,b}^{0,2\tau} \ket{\psi^{+}}_{a,b}^{\tau,\tau} - \ket{\psi^{ -}}_{a,b}^{0,2\tau} \ket{\psi^{-}}_{a,b}^{\tau,\tau} \\ \nonumber
	&- \ket{\phi^{+}}_{a,b}^{0,2\tau} \ket{\phi^{+}}_{a,b}^{\tau,\tau} + \ket{\phi^{ -}}_{a,b}^{0,2\tau} \ket{\phi^{-}}_{a,b}^{\tau,\tau}).
\end{align}

We now describe the explicit sequence of measurements undertaken by the experimentalists.  They measure the first photon of the first pair (1) immediately after it is created, while the second photon of that pair (2) is delayed by temporal interval $\tau$ in a free-space delay line.  The length of the delay line is chosen so that there is adequate time for the measurement of the first photon (1) before the second pair of photons is created (3-4).

After the second pair (3-4) of photons is generated, the first photon of that pair (3) is projected onto a Bell state with the delayed photon of the first pair (2).  The last photon (4), which is the second photon of the second pair, is delayed by an interval $\tau$ through the same delay line.  Moreoever, the last photon (4) is measured only after that delay period. 

\begin{figure}[!htbp]
	\begin{center}
		\includegraphics[scale=0.8]{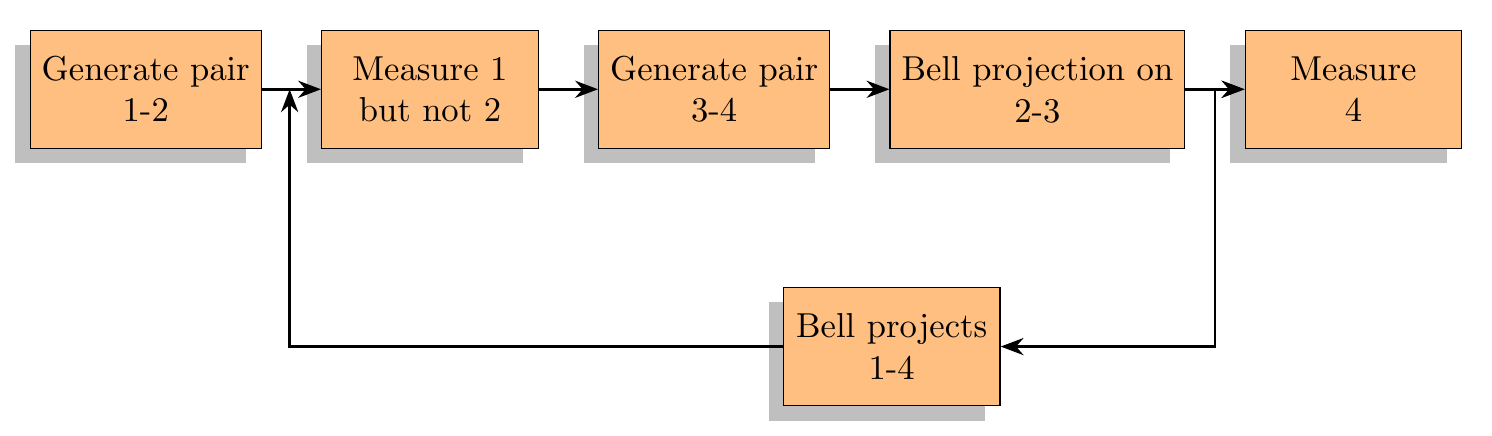}
		\caption{Time is increasing to the right}
		\label{TemporalBellDiagram}
	\end{center}
\end{figure}
When the photons at time $\tau$ (2-3) are projected onto any Bell state, the first and last photons (1-4), which share no prior interdependence, also collapse into the same state and the entanglement is `swapped.'  The first photon (1) and the last photon (4) become entangled.  It is important to emphasize that the first photon (1) was measured before the last photon (4) was even created.  We can write these \textit{temporal Bell states} more explicity as
\begin{empheq}[box=\widefbox]{align}\label{temporalBellstate}
\ket{\phi^{\pm}}_{a,b}^{0, 2\tau} &= \frac{1}{\sqrt{2}}(\ket{h_{a}^{0}h_{b}^{2\tau}} \pm \ket{v_{a}^{0}v_{b}^{2\tau}} ), \\
\ket{\psi^{\pm}}_{a,b}^{0, 2\tau} &= \frac{1}{\sqrt{2}}(\ket{h_{a}^{0}v_{b}^{2\tau}} \pm \ket{v_{a}^{0}h_{b}^{2\tau}}).
\end{empheq}
This is an \textit{entanglement in time} between subsystems that do not coexist.  The mathematical description as in (\ref{temporalBellstate}) describe the state in terms of the particle that existed just prior to $t=0$ and the particle that existed just prior to $t=2\tau$; we emphasize that the particle at $t=0$ did not coexist with the particle at $t=2\tau$.  

We see that the time at which each photon is measured has no effect on the final outcome.  Hence the timing of each photon simply serves an additional label to differentiate between the various photons.  From the view of this thesis, it can be alternatively described as that the quantum information (held in these photons) is labelled in both space ($a$ and $b$) and time ($0$ and $2\tau$). 

However, a Bell state measurement using linear optical elements can only simultaneously discriminate between two of the four Bell states.  Relaying this restriction back to the experiment, we can describe the procedure in more detail as follows.  The delayed photon of the first pair (2) and the first photon of the second pair (3) are projected onto a Bell state by combining them at a PBS.  Moreover, postselection is carried out in that the photons must exit the PBS at different ports and that the photons must be indistinguishable.  Conditioned on that requirement, the photons are rotated by HWPs to the polarization basis $\ket{p/m} = 1/\sqrt{2}(\ket{h} \pm \ket{v})$.  If the the polarization of the middle photons (2-3) were measured to be correlated (hh or vv), then they were projected onto a $\ket{\phi^{+}}_{a,b}^{\tau,\tau}$ state.  This results in the first and last photon being entangled through the temporal Bell state $\ket{\phi^{+}}_{a,b}^{0,2\tau}$.  However, if the polarization of the middle photons (2-3) were anti-correlated (hv or vh), then they were were projected onto the $\ket{\phi^{-}}_{a,b}^{\tau,\tau}$ state.  In an analogous manner, this resulted in the first and last photon being entangled through the temporal Bell state $\ket{\phi^{-}}_{a,b}^{0,2\tau}$.  In closing, this experimental procedure provided a means of generating temporal Bell states between subsystems that do not coexist.   

\textbf{b) Temporal multipartite systems:} We will now describe the generation of temporal multipartite entangled systems whose subsystems existed at different times.  This was experimentally realized in \cite{megidish2012resource}.  To describe the generation of such an effect, first consider two spatially entangled photons in the following Bell state
\begin{equation}
\ket{\phi_{12}^{+}} = \frac{1}{\sqrt{2}}(\ket{h_{1}h_{2}} + \ket{v_{1}v_{2}}).
\end{equation} 
Using two pairs of such states   
\begin{equation}
\ket{\phi_{12}^{+}} \otimes \ket{\phi_{34}^{+}} = \frac{1}{2}(\ket{h_{1}h_{2}} + \ket{v_{1}v_{2}}) \otimes (\ket{h_{3}h_{4}} + \ket{v_{3}v_{4}}),
\end{equation}
one can fuse these states, using a PBS, into a four-photon spatial GHZ state
\begin{equation}\label{spatialGHZ}
\ket{\Psi_{GHZ}^{(4)}} = \frac{1}{\sqrt{2}}(\ket{h_{1}h_{2}h_{3}h_{4}} + \ket{v_{1}v_{2}v_{3}v_{4}}).
\end{equation}
This state can be viewed as a specific physical instantiation of (\ref{GHZ}).  One can fuse further entangled photon pairs to create a growing GHZ state.

In the experimental work in \cite{megidish2012resource}, they were able to realize a temporal version of the above GHZ state using the same experimental setup that used in the temporal bipartite case \cite{megidish2013entanglement}.  Once again, they considered pairs of photons in spatial Bell states generated at consecutive intervals of time.  The first photon of a pair was directed to a PBS whereas the second photon enters a delay line of time $\tau$.  This second photon met the first photon of the next pair which was then fused at the PBS.  Using post-selection, this projected the two entangled pairs onto a temporal four-photon GHZ state.  This resulted in an entanglement between four photons that were across different spatial modes and existed at different times!  This \textit{entanglement in time }can be mathematically written as
\begin{empheq}[box=\widefbox]{align}\label{temporalGHZ4}
\ket{\Psi_{GHZ}^{(4)}} = \frac{1}{\sqrt{2}}(\ket{h_{1'}^{0}h_{2}^{\tau}h_{1}^{\tau}h_{2'}^{2\tau}} + \ket{v_{1'}^{0}v_{2}^{\tau}v_{1}^{\tau}v_{2'}^{2\tau}}).
\end{empheq}
From this, we see that there were two spatial modes ($1$ and $2$) after the projecting PBS and $1'$ and $2'$ before the projecting PBS) and three temporal modes ($0$, $\tau$, $2\tau$).  Contrasting this equation to the spatial case (\ref{spatialGHZ}), the different spatial modes that would exist for the different photons are replaced in the temporal case by different time slots for only two spatial modes.  In principle, one can create larger GHZ states using this technique, solving many of the scalability problems encountered in alternative experimental setups.  The most general \textit{temporal GHZ state}, from $n$ photon pairs, would take the form
\begin{empheq}[box=\widefbox]{align}\label{temporalGHZN}
\ket{\Psi_{GHZ}^{(2n)}} &= \frac{1}{\sqrt{2}}(\ket{h_{1'}^{0}h_{2}^{\tau}h_{1}^{\tau}\cdots h_{2}^{(n-1)\tau}h_{1}^{(n-1)\tau}h_{2'}^{n\tau}} \\ \nonumber
&+\ket{v_{1'}^{0}v_{2}^{\tau}v_{1}^{\tau}\cdots v_{2}^{(n-1)\tau}v_{1}^{(n-1)\tau}v_{2'}^{n\tau}}).
\end{empheq}
The experimental procedure was also able to generate temporal graph states, using polarization rotations on the respective photons.  One such example was a six-photon \textit{temporal graph state} with an `H-shape'
\begin{empheq}[box=\widefbox]{align}
 \ket{\Psi_{H}^{(6)}} &= \frac{1}{\sqrt{2}}(\ket{h_{1'}^{0}h_{2}^{\tau}h_{1}^{\tau}h_{2}^{2\tau}h_{1}^{2\tau}h_{2'}^{3\tau}} + \ket{h_{1'}^{0}h_{2}^{\tau}h_{1}^{\tau}v_{2}^{2\tau}v_{1}^{2\tau}v_{2'}^{3\tau}} \\ \nonumber
 &+ \ket{v_{1'}^{0}v_{2}^{\tau}v_{1}^{\tau}h_{2}^{2\tau}h_{1}^{2\tau}h_{2'}^{3\tau}} - \ket{v_{1'}^{0}v_{2}^{\tau}v_{1}^{\tau}v_{2}^{2\tau}v_{1}^{2\tau}v_{2'}^{3\tau}}).
\end{empheq} 

There have been other experiments that have generated an entanglement in time (given this definition); one remarkable achievement was generating a time-multiplexed cluster state containing more than $10,000$ entangled modes \cite{yokoyama2013ultra}.

\subsection{Through density operators}

\textbf{a) Temporal bipartite systems: } We saw the generation of temporal Bell states (\ref{temporalBellstate}) in the experimental setup described in \cite{megidish2013entanglement}.  An example of such a state was $\ket{\phi^{+}}_{a,b}^{0,2\tau}$ which is just one instantiation of the entanglement between the first and last photon of the experiment.  The density operator framework allows us to articulate this entanglement in an alternative way.

The experimentalists constructed a density matrix to characterize this temporal bipartite system (\ref{temporalBellstate}).  More precisely, the density matrix of the first and last photons was constructed, conditioned on the measurement outcome of the projection of the two photons at time $\tau$.  This was accomplished using a modification of quantum state tomography (\ref{quantumtomography}).  Moreoever, the modified tomography required projection measurements that used states such as $\ket{hv}$ as well as states such as $\ket{lr}$ described in (\ref{lrstates}).  For the experimental details, refer to  \cite{megidish2013entanglement}.  

Using such a density matrix, they were able to detect and measure entanglement.  This was accomplished a posteriori, meaning only after the measurement of all the photons in the experiment.  We briefly list only some of the experimental results of the measured matrices:
\begin{enumerate}[noitemsep, topsep=0pt, label=\roman*)]
	\item The fidelity (\ref{fidelity}) between the measured and theoretical density operators were $(77\pm 1)\%$.  Entanglement is said to be demonstrated when the fidelity exceeds $50\%$.
	\item The CHSH (\ref{CHSH}) value was $2.04 \pm 0.04$ which was a marginal violation to demonstrate entanglement.
	\item The PPT criterion (\ref{PPT}) was $-0.28 \pm 0.01$.  This aligns with the result as a negative value is needed to demonstrate entanglement.
	\item The concurrence (\ref{concurrence}) was $0.57 \pm 0.03$.  This once again demonstrates entanglement as the number needs to be positive for such an effect.
\end{enumerate}
Though the two photons in a temporal Bell state do not coexist, their quantum state is entangled.  This is experimentally expressed through the measured density matrix of the two photons conditioned on the result of the Bell state projection measurement concerning (\ref{temporalBellstateproject}).

However, if the Bell state projection is not carried out perfectly  on the photons at time $\tau$ (e.g. indistinguishability is introduced into the projected photons), then it can be measured that the first and last photon do not become entangled.  Rather they share classical correlations. This is an important observation as it emphasizes that prior to the Bell projection on the middle photons, the first and last photon do not somehow share any entanglement.

\textbf{b) Temporal multipartite systems: } Our aim is to describe the work in \cite{megidish2017quantum} which uses the density operator formalism to efficiently characterize temporal GHZ states (\ref{temporalGHZN}).  Moreover, we focus on the theoretical aspects.  As described earlier, these temporal multipartite states were experimentally generated in \cite{megidish2012resource}.  We recall the four photon case (\ref{temporalGHZ4}) whose generation involved two spatial Bell states with a delay line followed by a fusion at the PBS.  This process can be expressed with alternative initial Bell states as  
\begin{align}
\ket{\psi +}_{a,b}^{0,0} \otimes \ket{\psi +}_{a,b}^{\tau,\tau} &\xrightarrow{delay} \ket{\psi +}_{a,b}^{0,\tau} \otimes \ket{\psi +}_{a,b}^{\tau,2\tau} \\  \nonumber
&= \frac{1}{{2}}(\ket{h_{a}^{0} v_{b}^{\tau}} + {\ket{v_{a}^{0}h_{b}^{\tau}}}) \otimes (\ket{h_{a}^{\tau} v_{b}^{2\tau}} + {\ket{v_{a}^{\tau}h_{b}^{2\tau}}}) \\ &\xrightarrow{PBS} \frac{1}{2}(\ket{h_{a}^{0} v_{b}^{\tau}v_{a}^{\tau}h_{b}^{2\tau}} + \ket{v_{a}^{0} h_{b}^{\tau}h_{a}^{\tau}v_{b}^{2\tau}}) = \ket{GHZ}_{1,2,3,4}.
\nonumber
\end{align}
We adopt the labels $1,2,3,4$ to allow for a more compact notation.  The final four photon state can be re-expressed in terms of a density operator 
\begin{equation}
\rho_{1,2,3,4} = E(\rho_{1,2}\otimes \rho_{3,4})E^{\dagger},
\end{equation}
where $\rho_{i,j}$ is the density matrix of the $i$th and $j$th photons, and $E$ is the operator that represents the four-photon entangling process.  Recall that at time $\tau$, only the photons $2$ and $3$ are interacting at the PBS.  This implies that we can decompose $E$ as
\begin{equation}
E = \sigma_{0}^{1}F_{2,3}\sigma_{0}^{4},
\end{equation}
where
\begin{align}
F_{2,3} &= (\ket{h_{2}h_{3}}\bra{h_{2}h_{3}} + \ket{v_{2}v_{3}}\bra{v_{2}v_{3}}) \\
&= \frac{1}{2}(\sigma_{0}^{2}\sigma_{0}^{3} + \sigma_{3}^{2}\sigma_{3}^{3}).
\end{align}
Here $\sigma_{0}^{i}$ and $\sigma_{3}^{i}$ are the identity and Pauli $z$ operator which are applied to the $i$th photon.  The recursive nature of the experimental set up implies that all the entangled pairs originate from the same source and the fusion process operation is also identical.  This means that by measuring $\rho_{1,2}$ and $F_{2,3}$, the density matrix of any temporal GHZ state can be computed by combining identical two-photon states with identical projections
\begin{align}
\rho_{1,2, \dots , n} &= \sigma_{0}^{1}F_{2,3} \cdots F_{n-2,n-1}\sigma_{0}^{n} \\ \nonumber
&(\rho_{1,2} \otimes \cdots \otimes \rho_{n-1, n}) \\ \nonumber
&(\sigma_{0}^{1}F_{2,3} \cdots F_{n-2,n-1}\sigma_{0}^{n})^{\dagger}. 
\end{align}
One can obtain the entire information about a GHZ state containing any number of photons without getting their full statistics or even observing them.  This provides a far more efficient method to characterize the state than standard quantum state tomography (\ref{quantumtomographyN}).  For the experimental details, refer to \cite{megidish2017quantum}.

\subsection{Implications} 
 
\textbf{a) For the theme:} We shall now consider how properties of this entanglement, between subsystems that do not coexist, relate to the formulation of the overarching theme of this thesis.  Namely that \textit{the interdependence in any entanglement in time is shocking due to the existence of a time interval involved}.  For simplicity, we devote our analysis on the bipartite case \cite{megidish2013entanglement}.  Recall for spatial Bell states, a measurement on one of its subsystems instantaneously affects the other spatially distant subsystem. Similarly, an analogous effect exists for temporal Bell states, such as for $\ket{\phi^{+}}_{a,b}^{0,2\tau}$, where the state represents an entanglement between the first photon and last photon which do not coexist.  In \cite{megidish2013entanglement}, this was stated more directly and we quote, ``\textit{In the standard entanglement [in space] case, the measurement of any one of the particles instantaneously changes the physical description of the other. This result was described by Einstein as ‘‘spooky action at a distance.’’ In the scenario we present here, measuring the last photon affects the physical description of the first photon in the past, before it has even been measured. Thus, the ‘‘spooky action’’ is steering the system’s past. Another point of view that one can take is that the measurement of the first photon is immediately steering the future physical description of the last photon. In this case, the action is on the future of a part of the system that has not yet been created}.''  These interpretations are shocking because of a time interval between the measurement of the first photon and generation of the last photon. If one were to diminish the time interval to zero, then this entanglement in time loses its perplexing character. 

Our thesis aims to provide a finer level of analysis on these matters.  It needs to be articulated that the collapses of states are a mathematical description rather than a physical observable.  Denying the quantum state a physical reality allows for a pragmatic approach in terms of measurement correlations.  The measurement of the first photon yielded a definite outcome, and the measurement of the last photon is affected by it in a way that is stronger than could ever exist between classical systems.  The inability to replicate these correlations using classical systems is quantitatively captured using the CHSH (Bell) inequalities.

However if one assumes that the quantum state has a physical reality, then the interpretation quoted above does indeed highlight that quantum physics allows for influences to propagate into the past.  Perhaps this is seen far more concretely in that prior to the Bell state measurements of the photons at time $\tau$ (photons 2 and 3), the outcome of the Bell state measurement is unknown.  Using linear optical elements, the two possible outcomes are $\ket{\phi^{+}}_{a,b}^{\tau,\tau}$ or $\ket{\phi^{-}}_{a,b}^{\tau,\tau}$.  This implies that two possible states that the first and last photon can collectively collapse to are either $\ket{\phi^{+}}_{a,b}^{0,2\tau}$ or $\ket{\phi^{-}}_{a,b}^{0,2\tau}$.  But this is determined only much after the measurement of the first photon (see Figure \ref{TemporalBellDiagram})!

\textbf{b) For relativity: } Similar to the spatial case, the individual measurement result of a subsystem in the entanglement in time is probabilistically random.  Hence this does not violate causality (in a strict sense).  But there may be consequences for approaching the quantum physics of gravitation (the problem of quantum gravity), or at the very least of viewing relativity with an alternative perspective.  It is often the case that spacetimes which contain closed timelike curves (CTC) are regarded as pathological.  This experimentally verified effect of entanglement in time can be interpreted to exhibit some similar properties to CTCs.  Hence, for the pursuit towards quantum gravity, the premature dismissal of such pathological spacetimes may prove to be erroneous.

\textbf{c) For the nature of quantum physics: } As mentioned in Chapter \ref{chap: QInfo}, a most fundamental mystery is what is quantum information?  With respect to its historical origins, this problem is referred to as the interpretation issue of quantum mechanics.  We shall briefly describe how the above effect of entanglement in time gives support to a subset of the proposed interpretations. One of the proposals is known as the transactional interpretation \cite{cramer1986transactional} which involves sending signals back in time.  In this experiment of entanglement between subsystems that do not coexist, we have so far seen that it possible to some extent to non-classically influence the past.  Hence, this effect certainly adds considerable weight to furthering the transactional interpretation or some modification of it.  An alternative interpretation of quantum mechanics is known as the sum over paths approach \cite{feynman2005feynman}.  One of the insights (gained from my supervisor) is that if the the framework of quantum mechanics is taken seriously along with the sum over paths approach, then it can be seen that an event has multiple histories and that there is no single definite past.  The latter statement aligns with the description presented in this entanglement in time where the past is seen to be unsettled with respect to quantum properties.  The final interpretation that aligns well with the experiment is the two-state vector formalism \cite{twostatevector, twostatevector2} where the present is affected by both the past and the future.  In fact this was shown \cite{nowakowski} to be related to the entangled histories formalism, which was then used to describe an alternative temporal version of a GHZ state.

\textbf{d) For the nature of time: }  One of the most fundamental areas in the philosophy of time \cite{callender2011oxford} is in regards to answering the question:  Is there more to the world than the present moment?  The classification of answers within this particular sub-branch of metaphysics can be crudely categorized in three groups.  The first are the ``eternalists'' who believe that the past, present, and future are real.  The second call themselves the ``possibilists'' who claim that the past and present are real, but the future is not.  The remaining category are known as ``presentists'' who hold the position that only the present is real.  From the perspective of modern physics, it can easily be seen that the theories of relativity (which we shall review in Chapter \ref{chap: RQI}) are in conflict, to a large degree, with presentism.  However, a defense \cite{zimmerman2011presentism} was put forth that included the dismissal of the relativistic attack since those theories are challenged by quantum physics.  Hence it is surprising that in our investigation of entanglement in time, it is precisely a phenomena in quantum physics that shows from the present one can non-classically affect the past or immediately affect the future that has not yet been created.  This alludes to denying presentism, and rather taking the eternalists' view in that the past and future are as real as the present.

\section{Application: Quantum Blockchain}

A central aim in the field of quantum information science is the creation of new applications.  Both quantum communications and quantum computers are established applications of entanglement in space.  An open question in the field was whether entanglement in time is also a resource in quantum information?  More concretely, would it enhance the advantage of established applications?  Or more astonishingly,  would it lead to the development of novel applications and thereby open the door to new areas of research?

There have been various proposals to modify parts of quantum communications and quantum computers to adopt an entanglement in time.  For example, it was noted in \cite{megidish2013entanglement} that a memory system in a communication network would benefit from using this entanglement in time.  The entanglement in time which we will review in Chapter \ref{chap: QFound} has found applications for its use in a communication protocol \cite{taylor2004entanglement} as well as in the analysis of computing \cite{morikoshi2006information}.  In Chapter \ref{chap: RQI}, we examine yet another entanglement in time with proposals for memory systems \cite{sabin2012extracting}, teleportation  \cite{olson2012extraction}, and quantum key distribution\cite{ralph2015quantum}.

\textit{In this thesis, we make an original contribution} (which was done in collaboration with my supervisor) by designing one of the first quantum information applications of entanglement in time, namely a quantum blockchain \cite{rajan2019qblock}.  Our primary innovation is in encoding the blockchain into a temporal GHZ state.  It will be shown that the entanglement in time, as opposed to an entanglement in space, provides the crucial quantum advantage over a classical blockchain. More shockingly, the information encoding procedure in this quantum blockchain can be interpreted as non-classically influencing the past, and hence the system can be viewed as a `quantum time machine.' Furthermore, all the subsystems of this design have already been shown to be experimentally realizable \cite{megidish2013entanglement, megidish2012resource, megidish2017quantum, mccutcheon2016experimental}.  

However, the scope of our original research into the quantum blockchain is limited to only specifying a conceptual design.  This can be seen as the major step before providing a fully detailed design.  On a coarse level, there are three phases to the design of a quantum information system.  The first phase is to identify and extract the most essential task that characterizes the information system under consideration.  As an example in certain classical cryptographic systems, the most fundamental aspect is to have identical private keys at two different locations.  The second phase is to articulate the particular quantum system to be used to represent that information task along with an appropriate encoding method.  As an example in the E91 quantum protocol, it was realized that a spatial Bell state was the desirable structure for generating private keys at two different locations.  The third phase is to specify the dynamics of the quantum information system.  In the E91 protocol this is a clear sequence of steps carried out by Alice and Bob in terms of quantum operations, quantum measurements, classical communication and classical processing to generate that key.  For our conceptual design we will concern ourselves with the first two phases and merely provide an outline for how it could be used for developing the third phase.

\subsection{Preliminaries}

\textbf{a) Quantum computing attacks: } Before describing the quantum blockchain, we want to convey that one instantiation of the interplay between blockchains and quantum information arises as a security threat.  Quantum computers pose a significant threat to the security features of a classical blockchain, thereby potentially invalidating it as an information security system.  We refer the reader to \cite{aggarwal2017quantum} for an in-depth analysis on this topic whose crucial points we briefly outline.  In Chapter \ref{chap: classical} on classical blockchains, it was emphasized that the certain quantities associated with cryptographic hash functions would be infeasible to classically compute.  The security of the system crucially depends on such properties.  However, a quantum computer running Grover's search algorithm can perform quadratically fewer computations for this problem than is needed by classical computing.  Therefore over time,  this vulnerability will pose an imminent threat.  The second risk to classical blockchains posed by quantum computers derives from Shor's factoring algorithm.  The classical blockchains require public key cryptography for various operations involving digital signatures.  Hence Shor's algorithm and its attack on public key cryptography poses a dramatic risk to the classical system.  

\textbf{b) Post-quantum cryptography: } One can address the second risk of Shor's algorithm by developing classical blockchains which replace the public key cryptography with a post-quantum cryptographic component \cite{chen2016report, alagic2019status}.  Recall that these protocols are classical cryptographic systems which utilize mathematical problems that many believe a quantum computer would not be able to solve.  However, the durability of these solutions can be questioned in that there are no formal proofs supporting this hypothesis.  Nevertheless, these post-quantum blockchains have been proposed \cite{gao2018secure, torres2018post, li2018new}. 

\textbf{c) Quantum cryptography: } In a previous section, we described how quantum key distribution provides an alternative solution towards the risk posed by Shor's algorithm.  In \cite{kiktenko2018quantum}, a classical blockchain with a quantum key distribution subroutine was proposed.  It was also experimentally realized among a small number of network nodes.  In addition to this work, classical blockchains with various added quantum features have also been put forward in \cite{jogenfors2016quantum, sapaev2018quantum, behera2018quantum, tessler2017bitcoin, ikeda2018qbitcoin, kalinin2018blockchain}.

\textbf{d) Design methodology: } Recall that a classical blockchain system stores data securely over time and in a decentralized manner.  It is composed of two parts, namely the temporal blockchain data structure and a decentralized network consensus algorithm.  Our aim is to redesign the classical blockchain system into a full quantum information application to not only protect it from a quantum computing attack, but highlight further superior advantages as an information security system.  To do so our conceptual design focuses on creating quantum analogues of the blockchain data structure as well as the network consensus algorithm.  However, a number of low level design gaps do exist, but the intention was to open up a novel area where at least the core functionalities are covered.

\subsection{Quantum data structure}

\textbf{a) Description: } In this section, our aim is to replace the data structure component of the classical blockchain with a quantum information system which harnesses an entanglement in time.  In the classical case, records are chained in a chronological order through cryptographic hash functions.  In the quantum information case, we will capture the notion of the chain through the non-separability (entanglement) of quantum systems. For a spatially bipartite system $\ket{\psi}_{AB}$, this means that
\begin{equation}
\ket{\psi}_{AB} \ne \ket{a}_{A}\ket{b}_{B}, 
\end{equation}
for all single qubit states $\ket{a}$ and $\ket{b}$; the subscripts refer to the respective Hilbert spaces.  In particular multipartite GHZ states  are ones in which all subsystems contribute to the shared entangled property.  This enables us to create the concept of a chain.  However we need a method to encode the records into the chain, and develop a temporal structure to identify the chronological order.

To create the appropriate code to utilize this chain, it is helpful to use a concept from superdense coding~\cite{bennett1992communication}.  In this protocol, recall that a code (\ref{superdensecode}) converts classical information into spatially entangled Bell states; two classical bits, $xy$, where $xy = 00, 01, 10$ or $11$, are encoded to the state

\begin{equation}
\ket{\beta_{xy}} = \frac{1}{\sqrt{2}}(\ket{0}\ket{y} + (-1)^{x}\ket{1}\ket{\bar{y}}), 
\end{equation}
where $\bar{y}$ is the negation of $y$.  Given that Bell states are orthonormal, they can be distinguished by quantum measurements.  This decoding process allows one to extract the classical bit string, $xy$, from $\ket{\beta_{xy}}$.  

We still need a temporal structure to encode the chronological order.  This can be accomplished using an entanglement in time rather than a spatial entanglement.  For our conceptual design, we temporarily simplify the data characterizing the records in the classical block to a string of two bits.  Our encoding procedure converts each block with its classical record, say $r_{1}r_{2}$, into a temporal Bell state , generated at a particular time, say $t=0$:

\begin{equation}
\ket{\beta_{r_{1}r_{2}}}^{0, \tau} = \frac{1}{\sqrt{2}}(\ket{0^{0}}\ket{r_{2}^{\tau}} + (-1)^{r_{1}}\ket{1^{0}}\ket{\bar{r_{2}}^{\tau}}).
\end{equation}
From the entanglement in time section, it was seen that the superscripts in the kets signify the time at which the photon is absorbed; notice that the first photon of a block is absorbed immediately. For our purposes, this provides a way to do time stamps for each block. 

Recall such temporal Bell states were experimentally generated in the work by \cite{megidish2013entanglement} which we described in the last section.  In their procedure, spatially entangled qubits were represented through polarized photons,
	
\begin{equation}
	\ket{\phi \pm} = \frac{1}{\sqrt{2}}(\ket{h_{a}h_{b}} \pm {\ket{v_{a}v_{b}}}), \quad \ket{\psi \pm} = \frac{1}{\sqrt{2}}(\ket{h_{a}v_{b}} \pm {\ket{v_{a}h_{b}}}),
\end{equation}       
	
	where $h_{a}$ ($v_{a}$) represent the horizontal (vertical) polarization in spatial mode $a$ ($b$).  To create the temporally entangled states, consecutive pairs of spatially entangled pairs were generated at well-defined times separated by time interval $\tau$:
	
	\begin{equation}
	\ket{\psi -}_{a,b}^{0,0} \otimes \ket{\psi -}_{a,b}^{\tau,\tau} = \frac{1}{{2}}(\ket{h_{a}^{0} v_{b}^{0}} - {\ket{v_{a}^{0}h_{b}^{0}}}) \otimes (\ket{h_{a}^{\tau} v_{b}^{\tau}} - {\ket{v_{a}^{\tau}h_{b}^{\tau}}}),
	\end{equation}
	
	where the added superscripts provide the time labels for the photons.  In the experiment, a delay line of time $\tau$ is introduced to one of the photons of each entangled pair.  This resulting state equated to
	
	\begin{align}
	\ket{\psi -}_{a,b}^{0,\tau} \ket{\psi -}_{a,b}^{\tau,2\tau} &= \frac{1}{2}(\ket{\psi +}_{a,b}^{0,2\tau} \ket{\psi +}_{a,b}^{\tau,\tau} - \ket{\psi -}_{a,b}^{0,2\tau} \ket{\psi -}_{a,b}^{\tau,\tau} \\ \nonumber
	&- \ket{\phi +}_{a,b}^{0,2\tau} \ket{\phi +}_{a,b}^{\tau,\tau} + \ket{\phi -}_{a,b}^{0,2\tau} \ket{\phi -}_{a,b}^{\tau,\tau}).
	\end{align} 
	
	When Bell projection was carried out on two photons at time $t= \tau$, entanglement is created between the photon absorbed at $t=0$ and the photon absorbed at $t=2\tau$; this is despite the fact that the latter two photons have never coexisted.

Going back to our design, as records are generated, the system encodes them as blocks into temporal Bell states; these photons are then created and absorbed at their respective times.  A specific example of such blocks would be:
\begin{align}
\ket{\beta_{00}}^{0, \tau} &= \frac{1}{\sqrt{2}}(\ket{0^{0}}\ket{0^{\tau}} + \ket{1^{0}}\ket{{1}^{\tau}}), \\
 \ket{\beta_{10}}^{\tau, 2\tau} &= \frac{1}{\sqrt{2}}(\ket{0^{\tau}}\ket{0^{2\tau}} - \ket{1^{\tau}}\ket{{1}^{2\tau}}), \\
\ket{\beta_{11}}^{2\tau,3\tau} &= \frac{1}{\sqrt{2}}(\ket{0^{2\tau}}\ket{1^{3\tau}} - \ket{1^{2\tau}}\ket{{0}^{3\tau}}),
\end{align}
and so forth.  To create the desired quantum design, the system should chain the bit strings of the Bell states together in chronological order, through an entanglement in time. Such a task can be accomplished by using a fusion process~\cite{megidish2012resource}, described in the last section, in which temporal Bell states are recursively projected into a growing temporal GHZ state.  Physically, the fusion process is carried out through the entangled photon-pair source, a delay line and a polarizing beam splitter (PBS).  As an example, two Bell states can be fused into the following four-photon GHZ state:
\begin{align}
	\ket{\psi +}_{a,b}^{0,0} \otimes \ket{\psi +}_{a,b}^{\tau,\tau} &\xrightarrow{delay} \ket{\psi +}_{a,b}^{0,\tau} \otimes \ket{\psi +}_{a,b}^{\tau,2\tau} \\  \nonumber
	&= \frac{1}{{2}}(\ket{h_{a}^{0} v_{b}^{\tau}} + {\ket{v_{a}^{0}h_{b}^{\tau}}}) \otimes (\ket{h_{a}^{\tau} v_{b}^{2\tau}} + {\ket{v_{a}^{\tau}h_{b}^{2\tau}}}) \\ &\xrightarrow{PBS} \frac{1}{2}(\ket{h_{a}^{0} v_{b}^{\tau}v_{a}^{\tau}h_{b}^{2\tau}} + \ket{v_{a}^{0} h_{b}^{\tau}h_{a}^{\tau}v_{b}^{2\tau}}) = \ket{GHZ}^{0,\tau, \tau, 2\tau}.
	\nonumber
	\end{align}
Recall that in this GHZ state, entanglement exists between the four photons that propagate in different spatial modes and exist at different times.  Implementing  this procedure in our design, the state of the \textit{quantum blockchain}, at $t=n\tau$ (from $t=0$) is given by
\begin{empheq}[box=\widefbox]{align}\label{QBLOCK}
&\ket{GHZ_{r_{1}r_{2}\ldots r_{2n}}}^{0,\tau, \tau, 2\tau, 2\tau \ldots , (n-1)\tau, (n-1)\tau, n\tau}   \\ \nonumber
&= \frac{1}{\sqrt{2}}(\ket{0^{0}r_{2}^{\tau}r_{3}^{\tau}\ldots r_{2n}^{n\tau}} + (-1)^{r_{1}}\ket{1^{0}\bar{r}_{2}^{\tau}\bar{r}_{3}^{\tau}\ldots \bar{r}_{2n}^{n\tau}}).
\end{empheq}
The subscripts on the LHS of (\ref{QBLOCK}) denote  the concatenated string of all the blocks, while superscripts refer to the time stamps.  The time stamps allow each blocks' bit string to be differentiated from the binary representation of the temporal GHZ basis state.  Note that at $t=n\tau$, there is only one photon remaining.

The dynamics of this procedure can be illustrated with our example above.  Out of the first two blocks, $\ket{\beta_{00}}^{0, \tau}$ and $\ket{\beta_{10}}^{\tau, 2\tau}$, the system creates the (small) blockchain,
\begin{equation}
\ket{\beta_{00}}^{0, \tau} \otimes \ket{\beta_{10}}^{\tau, 2\tau} \rightarrow \ket{GHZ_{0010}}^{0, \tau, \tau, 2\tau}
\end{equation}
Concatenating the third block $\ket{\beta_{11}}^{2\tau,3\tau}$  produces 
\begin{equation}
\ket{GHZ_{001011}}^{0, \tau, \tau, 2\tau, 2\tau, 3\tau} = \frac{1}{\sqrt{2}}(\ket{0^{0}0^{\tau}1^{2\tau}0^{2\tau}1^{2\tau}1^{3\tau}} + \ket{1^{0}1^{\tau}0^{2\tau}1^{2\tau}0^{2\tau}0^{3\tau}}). 
\end{equation}
The decoding process extracts the classical information, $r_{1}r_{2}\ldots r_{2n}$, from the state (\ref{QBLOCK}).  As mentioned in the previous section, it was shown \cite{megidish2017quantum} how to characterize any such temporally generated GHZ state efficiently compared to standard tomography techniques.  This can be accomplished without measuring the full photon statistics, or even detecting them.

\textbf{b) Security analysis: } Recall that in the classical blockchain system the relevant performance metric is nontampering for the data structure.  This is accomplished by the data structure being extremely sensitive to tampering through the \textit{interdependence of classical blocks achieved by cryptographic hash functions}.  If one attempts to modify even a single block, the extreme sensitivity is such that is invalidates all future blocks following the tampered block.  This provides a tamper proof system for storing records because tampering with it can easily be detected.  In the quantum blockchain, the sensitivity to tampering is achieved through \textit{the interdependence of the quantum blocks in an entanglement in time}.  

To elaborate, for the quantum blockchain we have replaced the important functionality of time stamped blocks and hash functions linking them, by a temporal GHZ state with an entanglement in time.  The quantum advantage is that the sensitivity towards tampering is significantly amplified, meaning that the blockchain is destroyed if one tampers with a single block (due to entanglement); on a classical blockchain only the blocks after the tampered block are destroyed (due to cryptographic hash functions) which leaves it open to vulnerabilities.  For the classical case, it is often stated that the farther back the block was time stamped in, the more "secure" it is; this is precisely because of the above invalidation.  Even if we had used an entanglement in space (with all the photons coexisting) that would still have provided an advantage since if an attacker tries to tamper with any photon, the full blockchain would be invalidated immediately; this already provides a benefit over the classical case where only the future blocks of the tampered block are invalidated.  The temporal GHZ blockchain (\ref{QBLOCK}) adds a far greater benefit in that the attacker cannot even attempt to access the previous photons since they no longer exist.  They can at best try to tamper with the last remaining photon, which would invalidate the full state.  Hence in this application of quantum information, we see that the entanglement in time provides a far greater security benefit than an entanglement in space. There still needs to be a careful case by case analysis of potential tampering with the ultimately classical measurement results, but that would entail full security proofs which is left for future work.  

\textbf{c) Comments: }
\begin{enumerate}[noitemsep, topsep=0pt, label=\roman*)]
	\item The temporal GHZ state, that we use in our design, involve an entanglement between photons that do not share simultaneous coexistence, yet they share non-classical interdependence.  This temporal interdependence, between two entangled photons that existed at different times, was interpreted in~\cite{megidish2013entanglement} as follows: ``\textit{...measuring the last photon affects the physical description of the first photon in the past, before it has even been measured. Thus, the ``spooky action" is steering the system's past}".  Stated more shockingly, in our quantum blockchain, we can interpret our encoding procedure as linking the current records in a block, not to a record \textit{of} the past, but linking it to the actual record \textit{in} the past, a record which does not exist anymore. Hence the system can be viewed as a `quantum time machine.'
	\item Much of the performance of the quantum blockchain data structure is simply due to the properties of a temporal GHZ state.  The non-trivial aspect was in obtaining the appropriate quantum structure and finding an efficient encoding method.  This phase of design is comparable to realizing that a spatial Bell state was a useful structure for key generation in E91
	\item We imagine that future designs of quantum blockchains may harness the other entanglement in time effects discussed in Chapter \ref{chap: QFound} and Chapter \ref{chap: RQI}.
	\item This conceptual design presented the case that a security advantage exists given that the previously existing photons are not able to be accessed (since they no longer exist).  However a full security proof of this remains to be worked out.  Of great interest would be whether from an operational point of view there is a security advantage between using photons that do not coexist as opposed to using photons from a single Bell-pair in which a measurement of one particle is simply delayed with respect to the other.  Such an understanding may provide the necessary basis for the development of further temporal based quantum information protocols.
	
\end{enumerate}

\subsection{Quantum consensus protocol}

\textbf{a) Description: } Our aim in this section is to develop a quantum analogue of the network consensus protocol.  The $\theta$-protocol~\cite{mccutcheon2016experimental} was originally designed to verify GHZ entanglement in a quantum network.  Given that we have encoded a quantum blockchain into a temporal GHZ state we can harness the $\theta$-protocol as a consensus algorithm for blocks.     
  
To provide some elaboration, recall that a classical blockchain system has a number of different components.  A blockchain data structure, a copy of this data structure at each node of a classical network, and a consensus network algorithm to verify the correctness of new blocks (before adding that new block to a blockchain).  In our design, we replace the classical network with a quantum network and with that, digital signatures would be covered by a quantum key distribution (QKD) protocol.  In fact, others have used this way of reasoning when introducing new quantum protocols.  For example in the $\theta$-protocol~\cite{mccutcheon2016experimental} the authors also simply assume a QKD layer before moving onto their original work.  We quote their paper, ``\textit{it is assumed that the verifier and each of the parties share a secure private channel for the communication.  This can be achieved by using either a one-time pad or a quantum key distribution.}" Furthermore, in this design, each node on the quantum network would host a copy of the quantum blockchain (\ref{QBLOCK}); hence if a node tampers with its own local copy, it does not affect the copies at the other nodes analogous to the classical case.   New blocks (that come from a sender) need to be verified for their correctness, before being copied and added to each node's blockchain.  Since correct blocks are GHZ entangled states, one needs a verification test to do it.

At this stage of the design, we assume that newly generated blocks are spatial GHZ states (converting this to the related temporal case is at this stage of the design process unnecessary, and is left for future work).  As in the classical case, the objective is to add valid blocks in a decentralized manner. The challenge is that the network can consist of dishonest nodes, and the generated blocks can come from a dishonest source.  To solve this problem, the quantum network uses the $\theta$-protocol~\cite{mccutcheon2016experimental}, which is a consensus algorithm where a random node in the quantum network can verify that the untrusted source created a valid block (ie spatial GHZ state).  More crucially, this is accomplished in a decentralized way by using other network nodes, who may also be dishonest (ie Byzantine nodes). 

To start off this verification protocol, we need to pick a randomly chosen verifier node (analogous to proof-of-stake or proof-of-work); this can be accomplished through a low level sub-algorithm involving a quantum random number generator.  The untrusted source shares a possible valid block, an $n$-qubit state, $\rho$. Since it knows the state, it can share as many copies of the block as is needed without running afoul of the no-cloning theorem.  For verification, it distributes each of the qubits to each node, $j$. The verifying node generates random angles $\theta_{j} \in \left[0, \pi \right) $ such that $\sum_{j}\theta_{j}$ is a multiple of $\pi$.  The (classical) angles are distributed to each node, including the verifier. They respectively measure their qubit in the basis,
\begin{eqnarray}
\ket{+_{\theta_{j}}} = \frac{1}{\sqrt{2}}\left(\ket{0}+e^{i\theta_{j}}\ket{1}\right), 
\\
\ket{-_{\theta_{j}}} = \frac{1}{\sqrt{2}}\left(\ket{0}-e^{i\theta_{j}}\ket{1}\right).
\end{eqnarray}
The results, $Y_{j}=\{0,1\}$, are sent to the verifier. If the $n$-qubit state was a valid block, ie a spatial $n$-qubit GHZ state, the necessary condition 
\begin{equation}
\oplus_{j}Y_{j} = \frac{1}{\pi} \sum_{j}\theta_{j} \quad   (\text{mod 2}), 
\end{equation}
is satisfied with probability $1$. The protocol links the verification test to the state that is used; the paper~\cite{mccutcheon2016experimental} explicitly mentions this and we quote, ``\textit{It is important to remark that our verification protocols go beyond merely detecting entanglement; they also link the outcome of the verification tests to the state that is actually used by the honest parties of the network with respect to their ideal target state. This is non-trivial and of great importance in a realistic setting where such resources are subsequently used by the parties in distributed computation and communication applications executed over the network.}''  Hence the block can be copied and distributed to each node on the network to be added onto their blockchain.

\textbf{b) Security analysis: } We refer the reader to \cite{mccutcheon2016experimental} for an in-depth security analysis whose results we briefly outline.  Let $P(\rho)$ denote probability of passing the verification test.  Furthermore, let the fidelity of the shared state $\rho$ with respect to an ideal GHZ state be computed as
\begin{equation}
F(\rho) = \braket{GHZ_{n}|\rho|GHZ_{n}}.
\end{equation}
One can obtain a lower bound on the passing the verification test.  It can be proven that if the $n$ parties are honest, then we have the relationship,
\begin{equation}
F(\rho) \geq 2 P(\rho)-1.
\end{equation}
When the protocol is performed in the presence of dishonest nodes (byzantine nodes), then the results are modified.  Suppose we have $n-k$ nodes that apply local or joint unitary operation $U$ to their state.  This encodes the various ways these nodes may attempt to cheat the system.  The modified fidelity for this scenario is given by
\begin{equation}
F'(\rho) = \text{max}_{U} F((I_{k} \otimes U_{n-k})\rho(I_{k} \otimes U_{n-k}^{\dagger})),
\end{equation} 
and the associated lower bound for pass probability can be computed to be
\begin{equation}
F'(\rho) \geq 4 P(\rho)-3.
\end{equation}
Compared to other quantum verification protocols, the $\theta$-protocol can be shown to be more sensitive to detecting dishonest nodes.

\textbf{c) Comments: }
\begin{enumerate}[noitemsep, topsep=0pt, label=\roman*)]
	\item Combining the data structure component with the network consensus protocol provides us with the conceptual design of a quantum blockchain. 
	\item Our work provided a conceptual design.  This is the major step before providing a fully detailed protocol design.  The latter is left for future work.  However there are some challenges that we foresee: Standard blockchain protocols do not easily fit into the traditional framework of distributed computing \cite{narayanan2017bitcoin} and proof of their security functionalities in a rigorous manner is not well articulated. Hence developing a detailed quantum blockchain protocol with security proofs would be predicated on also undertaking many research problems from the classical case.
	\item Given the rise of classical blockchains and the development of a quantum network, we hope this conceptual design may potentially open the door to a new research frontier in quantum information science.
\end{enumerate}


\chapter{Quantum Foundations}\label{chap: QFound}

\begin{chapquote}{Charles Bennett, co-inventor of quantum teleportation}
	``Quantum information is more like the information in a dream.''
\end{chapquote}

\textsf{QUANTUM INFORMATION SCIENCE} is based on the framework of quantum theory.  In an almost paradoxical manner, quantum theory provides an extraordinary degree of applicability, and yet its fundamental structures remain deeply mysterious.  Quantum foundations is a field that is devoted to examining the nature of these structures.  Perhaps the two great mysteries are:  What is the nature of the quantum state? And how does the quantum state `collapse' upon measurement?  The first question stems as a generalization concerning the unknown physical representation of quantum information.  Whereas the second question arises from the unarticulated notion that an undefined observer causes an instant transformation from quantum information to classical information.

In this thesis, we highlight how concepts from the entanglements in both space and time can progress us towards these two questions. To elaborate, we will focus our study on a particular \textit{interdependence} witnessed in the entanglements, known as non-locality. We will see that it a stricter form of non-classical interdependence than entanglement.  Our focus in this chapter is to describe non-locality across space as well as across time.  We aim to show how these properties can shed at least a partial understanding on the two questions.

\section{Quantum Measurements}

Our focus will initially be on the question of quantum measurement.  The approach will involve deriving the condition for non-locality in space, and using it to unravel issues regarding measurements.  This procedure requires two points: 

\begin{enumerate}[noitemsep, topsep=0pt, label=\roman*)]
	\item The probabilistic aspect of quantum theory.
	\item A concept known as realism.
\end{enumerate}
Before moving to non-locality in space, we provide an aid by illuminating results concerning these two points.  Gleason's theorem is a result concerning the first point, and the Kochen–Specker theorem elaborates on the second point.  For a comprehensive overview on these results, refer to \cite{hemmick2011bell}.  

\subsection{Gleason's theorem}

The probabilistic aspect of quantum theory is conveyed through the measurement postulate such as the Born rule (\ref{povmproprob}).  Roughly speaking, Gleason's theorem \cite{gleason1957measures} states that if one is given the non-probabilistic structure of quantum theory (e.g. Hilbert spaces, projection operators) and one also assumes that the theory requires a probabilistic character, then that character must be expressed in no other way than the Born rule.  An alternative view is that if one requires non-Born rule quantum probabilities, then one must give up using projection operators to describe measurements.  In this sense, Gleason's theorem can be interpreted as a `derivation' of the Born rule.  However, it is important to emphasize that the assumption of a probabilistic aspect is still needed and the underlying nature of it is currently unknown.  Thus at present, the Born rule cannot be derived solely from the non-probabilistic postulates of quantum theory.

\textbf{a) Preliminaries: } Within quantum foundations as well certain areas of pure mathematics, an extensive investigation of Gleason's theorem has been carried out \cite{cooke1985elementary,hellman1993gleason,billinge1997constructive, pitowsky1998infinite, richman1999constructive, richman2000gleason, busch2003quantum, caves2004gleason, buhagiar2009gleason, wright2018gleason, hamhalter2013quantum, cohen2012introduction}.  The theorem addresses the minimal, (in fact, quite surprisingly minimal), assumptions required to deduce the existence of a quantum density matrix, (a unit trace Hermitian matrix encoding the notion of quantum probability), and as mentioned underlies the theoretical justification for adopting the Born rule.  Early proofs of Gleason’s theorem were implicit and non-constructive, and for some time there was controversy as to whether a constructive proof was even possible \cite{hellman1993gleason, billinge1997constructive, richman1999constructive, richman2000gleason}. With hindsight, disagreement on what methods are legitimately to be deemed “constructive” is the key point of the constructivist debate. Even with modern constructive (in principle) proofs, the construction is not particularly explicit, and often very little is said as to what the quantum density matrix actually looks like. Traditionally the analysis stops, and the theorem is complete, once the existence of the quantum density matrix is established.

In this subsection, we will now present work \cite{rajan2019explicit} this is \textit{part of the original component of this thesis} (which was done in collaboration with my supervisor).  It has very little to say about the theorem and proof themselves, focusing more on the implications: We shall say a little more about the density matrix itself — and shall provide two constructions (one implicit, one explicit) for the density matrix.

\textbf{b) Gleason's theorem: } An explicit statement of the theorem runs thus \cite{gleason1957measures}:
\begin{theorem}\label{gleasonstheorem}
	\textbf{(Gleason's theorem)} \\
	Suppose $H$ is a separable\footnote{A Hilbert space is separable if and only if it has a countable orthonormal basis.} Hilbert space, (either real or complex). \\[5pt]
	A measure on $H$ is defined to be a function $v(\cdot)$ that assigns a nonnegative real number to each closed subspace of $H$ in such a way that:
	If   $\{A_{i}\}$ is any countable collection of mutually orthogonal subspaces of $H$, and the closed linear span of this collection is $B$, then $v(B)=\sum _{i}v(A_{i})$. Furthermore we normalize to  $v(H)=1$. \\[5pt]
	Then if the Hilbert space $H$ has dimension at least three, (either real or complex), every measure $v(\cdot)$ can be written in the form $v(A)=\mathrm{tr}(\rho \, P_{A})$, where $\rho$ is a positive semidefinite trace class operator with $\mathrm{tr}(\rho)=1$, and 
	$P_{A}$ is the orthogonal projection onto~$A$.\hfill$\Box$
\end{theorem}
(Physicists would almost immediately focus on complex Hilbert spaces; but some of the mathematical literature also works with real Hilbert spaces.)  The original theorem gives one very little idea of what the density matrix might look like, and it is this topic we shall address. Indeed, the original theorem spends many pages proving that the valuation $v(P)$ uniformly continuous;  while this is certainly an extremely useful result, most physicists, (and applied mathematicians for that matter), would simply assume continuity on physical grounds.

\textbf{c) Elementary observations: }  Our first observation is that since $\rho$ is Hermitian we can diagonalize it and define
\begin{equation}
\rho = \sum_i \lambda_i \; Q_i.
\end{equation}
Here the $Q_i$ are taken to be 1-dimensional subspaces, and the $\lambda_i$ are to be repeated with the appropriate multiplicity.
Per Gleason's theorem, 
\begin{equation}
v (Q_j) = \text{tr}(\rho Q_j) = \text{tr}\left( \left[\sum_i \lambda_i  Q_i\right]\; Q_j \right) =  \sum_i \lambda_i \;\text{tr}\left( Q_i \; Q_j\right) =  \lambda_j.
\end{equation}
So actually
\begin{equation}
\rho = \sum_i v(Q_i) \; Q_i,
\end{equation}
which does not (yet) help unless you can somehow extract the $Q_i$ in terms of the underlying valuation function $v(\cdot)$. Furthermore note that for each 1-dimensional subspace $Q_i$ we can identify
\begin{equation}
Q_i \sim |\psi_i\rangle \; \langle \psi_i|
\end{equation}
where $|\psi_i\rangle$ is any arbitrary vector in the 1-dimensional subspace $Q_i$. Then
\begin{equation}
v(Q_i) =  \langle \psi_i| \rho  |\psi_i\rangle.
\end{equation}
Now let $P_i$ be any arbitrary collection of orthogonal 1-dimensional projection operators
\begin{equation}
v\left( \sum_i P_i\right) =  \sum_i  v(P_i) = 1.
\end{equation}
Using Gleason's theorem, we can calculate
\begin{equation}
v(P_j) = \text{tr}(\rho P_j) = \text{tr}\left( \left[\sum_i v(Q_i)  Q_i\right]\; P_j \right) =  \sum_i v(Q_i) \; \text{tr}\left( Q_i \; P_j\right) =  \sum_i v(Q_i) \; S_{ij},
\end{equation}
with $S_{ij}= \text{tr}\left( Q_i \; P_j\right) $ a bi-stochastic matrix (which is a square matrix of non-negative real numbers with each row and column summing to unity). That is,  Gleason's theorem implies
\begin{equation}
v(P_j) = \sum_i v(Q_i) \;  S_{ij}; 
\qquad\qquad\hbox{with} \qquad\qquad 
S_{ij} = | \langle q_i | p_j \rangle |^2 = |U_{ij}|^2.
\end{equation}
So we see that the matrix $S_{ij}$ is actually unitary-stochastic (which is a bi-stochastic matrix whose entries are the squares of the absolute values of the entries of some unitary matrix); both unitary and unitary-stochastic matrices drop out automatically.

\def\ss{\hbox{\large $\displaystyle\$ $}}
Now pick some random basis $P_i$ and construct
\begin{equation}
\rho_P = \sum_i v(P_i) \; P_i.
\end{equation}
This is not $\rho$ itself, but it is what you get from $\rho$ by hitting it with $\ss_P$, the decoherence super-scattering operator with respect to the basis $P_i$ \cite{alonso2017coarse}.  (At a basic level, a super-scattering operator can be viewed as a trace-preserving linear mapping from density matrices to density matrices.) To see this note
\begin{equation}
\ss_P \;\rho =  \sum_i P_i \;\text{tr}( P_i \,\rho) =    \sum_i P_i \; v( P_i) = \rho_P.
\end{equation}
Finally consider what happens if you average over the $P_i$:
\begin{equation}
\left\langle \ss_P\right\rangle \;\rho =  \left\langle \sum_i P_i \;\text{tr}( P_i \,\rho) \right\rangle 
=    \left\langle \sum_i P_i \,\; v( P_i) \right\rangle = \left\langle \rho_P \right\rangle.
\end{equation}
In $d$ dimensions for a uniform average over the $(P_i)_{ab}$ we have
\begin{equation}
\left\langle \sum_i (P_i)_{ab} \; (P_i)_{cd} \right\rangle = {\delta_{ac}\delta_{bd} + \delta_{ab} \delta_{cd}\over d+1}.
\end{equation}
This arises from symmetry plus the normalization condition $\langle I_{d\times d}\rangle = I_{d\times d}$. But then we can reconstruct
\begin{equation}
\rho = (d+1) \left\langle \rho_P \right\rangle - I_{d\times d}. 
\end{equation}
(Note this does have the correct trace, $\text{tr}(\rho)=1$.)
So if you know all possible ways in which the density matrix decoheres $\rho \to \rho_P$, and uniformly average over all choices of decoherence basis, then one can reconstruct the full density matrix.  While certainly an elegant result, this is by no means explicit. 

\textbf{d) Implicit construction: }  Let us now set up a reasonably explicit construction of the density matrix $\rho$ directly from the valuation function $v(P)$.  To construct $\rho$ proceed as follows: First for any 1-dimensional subspace note $Q\sim |n\rangle\; \langle n|$ where $n$ can be taken to be a unit vector in $S^{d-1}$. This defines a valuation $v(n)$ on $S^{d-1}$. Then find a ${n_1}$ such that $v(Q_{n_1})= \max_{n\in S^{d-1}} \{v(P_n)\} =  \max_{n\in S^{d-1}} \langle n|\rho|n\rangle$. 

Now consider the $S^{d-2}$ perpendicular to $n_1$: Proceed as follows --- find a $n_2$ such that $v(Q_{n_2})= \max_{n\in S^{d-2}} \{v(P_n)\}$. By construction $n_1 \perp n_2$ and $P_{n_1} P_{n_2} = 0$. Iterate this construction: Consider the $S^{d-i}$ perpendicular to $n_1$, $n_2$, \dots, $n_{i-1}$: Find a $n_i$ such that $v(Q_{n_{i}})= \max_{n\in S^{d-i}} \{v(P_n)\}$.  By construction the $n_j$ for $j\in\{1,2,\cdots, i\}$ are mutually perpendicular, and $P_{n_j} P_{n_k} = 0$ for $j\neq k$ and $j,k\in\{1,2,\cdots, i\}$.  Ultimately we have $n_d = \max_{n\in S^{0}} \{v(P_n)\} = \min_{n\in S^{d-1}} \{v(P_n)\}$. The construction terminates after $d$ steps with an orthonormal basis $n_1$, $n_2$, \dots, $n_d$, and the corresponding valuations $v(Q_{n_i})$. Now construct
\begin{equation}
\rho = \sum_{i=1}^d v(Q_{n_i}) \; Q_{n_i}.
\end{equation}
This is the density matrix you want. \hfill $\Box$
\begin{proof}
It is clearly a density matrix; it only remains to check that it is the density matrix.  But this is obvious from the construction --- the $n_i$ are the simply eigenvectors of $\rho$, with the corresponding projection operators $Q_{n_i}$, and the $v(Q_{n_i})$ are the eigenvalues. (Basically the construction above is just an application of the  Rayleigh--Ritz min-max variational theorem for  finding eigenvectors/eigenvalues of Hermitian matrices.) The density matrix is constructed in terms of the values, $v(Q_{n_i})$, and locations, $n_i$, of the maximum, minimum, and extremal points of the valuation function $v(\cdot)$.  
\end{proof}
Note the construction is still rather implicit. Once Gleason's theorem guarantees the existence of the density matrix, this construction implicitly allows one to determine the density matrix. The more purist of constructivist mathematicians might not call this constructive, but most others would.  On the other hand, as we shall now show, much better can be done in terms of a fully explicit construction.

\textbf{e) Explicit construction: } This second construction is completely explicit but considerably more subtle.  We assert that within the framework of Gleason's theorem, for any arbitrary basis on complex Hilbert space we can write:
\begin{eqnarray}
&&\rho =  \sum_j |n_j\rangle \; v(n_j) \; \langle n_j|
\\
&&+{1\over2} \sum_{j\neq k}  |n_j\rangle
\left\{ {v\left(n_j+n_k\over\sqrt{2}\right) - v\left(n_j-n_k\over\sqrt{2}\right)} -i \; {v\left(n_j+in_k\over\sqrt{2}\right)+iv\left(n_j-in_k\over\sqrt{2}\right)}\right\}     
\; \langle n_k|.
\nonumber
\end{eqnarray}
That is, to reconstruct the full density matrix we need only determine the valuations $v(\cdot)$, which is a collection of real numbers, on the specific set of unit vectors
\begin{equation}
n_j; \qquad \left(n_j\pm n_k\over\sqrt{2}\right); \qquad \left(n_j\pm in_k\over\sqrt{2}\right).
\end{equation}
There are a total of $d + d(d-1) + d(d-1) = 2d^2-d$ such unit vectors to deal with.  This formula for the density matrix can also be rearranged as follows
\begin{eqnarray}
\rho &=&  \sum_j  v(n_j) |n_j\rangle \; \langle n_j| 
\nonumber\\
&&+{1\over2} \sum_{j< k}  
\left\{ {v\left(n_j+n_k\over\sqrt{2}\right) - v\left(n_j-n_k\over\sqrt{2}\right)} \right\}   
\left( \vphantom{\Big{|}} |n_j\rangle \; \langle n_k| +  |n_k\rangle \; \langle n_j| \right)  
\nonumber\\
&&-{i\over2} \sum_{j< k}  
\left\{  \; {v\left(n_j+in_k\over\sqrt{2}\right)+v\left(n_j-in_k\over\sqrt{2}\right)}\right\}    
\left( \vphantom{\Big{|}} |n_j\rangle \; \langle n_k| -  |n_k\rangle \; \langle n_j| \right).  
\end{eqnarray}
In this form, Hermiticity of the density matrix is manifest.  The situation for a real Hilbert space is considerably simpler:
\begin{equation}
\rho =  \sum_j |n_j\rangle \; v(n_j) \; \langle n_j|
+{1\over2} \sum_{j\neq k}  |n_j\rangle
\left\{ {v\left(n_j+n_k\over\sqrt{2}\right) - v\left(n_j-n_k\over\sqrt{2}\right)} \right\}     
\; \langle n_k|.\qquad
\end{equation}
There are now only a total of $d + d(d-1) = d^2$ unit vectors to deal with.  This formula for the (real) density matrix can also be rearranged as follows
\begin{eqnarray}
\rho &=&  \sum_j  v(n_j) |n_j\rangle \; \langle n_j| 
\nonumber\\
&&+{1\over2} \sum_{j< k}  
\left\{ {v\left(n_j+n_k\over\sqrt{2}\right) - v\left(n_j-n_k\over\sqrt{2}\right)} \right\}   
\left( \vphantom{\Big{|}} |n_j\rangle \; \langle n_k| +  |n_k\rangle \; \langle n_j| \right).
\end{eqnarray}
In this form, symmetry of the (real) density matrix is manifest.  To start the construction, following \cite{richman1999constructive},  we extend the valuation $v(P)\longleftrightarrow v(n)$ from $S^{d-1}$ to all of $H$ as follows:
\begin{equation}
f(n) =  ||n||^2 \; v\left( n\over ||n||\right),
\end{equation}
Now, again following \cite{richman1999constructive},
\begin{equation}
\langle x| \rho| y \rangle = {f(x+y) - f(x-y)\over 4} -i \; {f(x+iy)-f(x-iy)\over 4},
\end{equation}
which in the real case reduces to
\begin{equation}
\langle x| \rho| y \rangle = {f(x+y) - f(x-y)\over 4}.
\end{equation}
In \cite{richman1999constructive}, it asserts the equivalence of:
\begin{itemize}[noitemsep, topsep=0pt]
	\item $ \langle a x| \rho| b y \rangle  = \overline{a}\,b\, \langle x| \rho| y \rangle $.
	\item $\langle x| \rho| y \rangle  = \overline{ \langle y| \rho| x \rangle }$.
	\item $\langle x| \rho| y_1+y_2 \rangle = \langle x| \rho| y_1 \rangle + \langle x| \rho| y_2 \rangle$.
\end{itemize}
where the overline signifies the complex conjugation. This is needed to verify that $\langle x| \rho| y \rangle $ actually represents a bilinear form.  Then the density matrix $\rho$ can itself be defined by
\begin{equation}
\rho =  \sum_j \sum_k |n_j\rangle \; \langle n_j| \rho| n_k \rangle \; \langle n_k|.
\end{equation}
So
\begin{equation}
\rho =  \sum_j \sum_k |n_j\rangle \; 
\left\{ {f(n_j+n_k) - f(n_j-n_k)\over 4} -i \; {f(n_j+in_k)-f(n_j-in_k)\over 4}\right\}     
\; \langle n_k|.
\end{equation}
Whence, splitting the sum into diagonal and off-diagonal pieces, and noting that both $||n_j\pm n_k||^2 = 2 = ||n_j\pm in_k||^2$, while $ \widehat{n_j\pm n_k} = (n_j\pm n_k)/\sqrt 2$, and finally $ \widehat{n_j\pm i n_k} =  (n_j\pm i n_k)/\sqrt 2$, we have:
\begin{eqnarray}
&&\rho =  \sum_j |n_j\rangle \; v(n_j) \; \langle n_j|
\\
&&+{1\over2} \sum_{j\neq k}  |n_j\rangle
\left\{ {v\left(n_j+n_k\over\sqrt{2}\right) - v\left(n_j-n_k\over\sqrt{2}\right)} -i \; {v\left(n_j+in_k\over\sqrt{2}\right)+iv\left(n_j-in_k\over\sqrt{2}\right)}\right\}     
\; \langle n_k|.
\nonumber
\end{eqnarray}
That is, in terms of the decohered density matrix $\rho_P$ we have:
\begin{eqnarray}
&&
\rho =  \rho_P
\\
&&
+{1\over2} \sum_{j\neq k}  |n_j\rangle
\left\{ {v\left(n_j+n_k\over\sqrt{2}\right) - v\left(n_j-n_k\over\sqrt{2}\right)} -i \; {v\left(n_j+in_k\over\sqrt{2}\right)+iv\left(n_j-in_k\over\sqrt{2}\right)}\right\}     
\; \langle n_k|.
\nonumber
\end{eqnarray}
For a real Hilbert space this reduces to
\begin{equation}
\rho =  \rho_P
+{1\over2} \sum_{j\neq k}  |n_j\rangle
\left\{ {v\left(n_j+n_k\over\sqrt{2}\right) - v\left(n_j-n_k\over\sqrt{2}\right)} \right\}     
\; \langle n_k|.
\end{equation}
One aspect of the ``miracle'' of Gleason's theorem is  that this construction is actually independent of the specific basis chosen.  To see why this construction works, note that from Gleason's theorem, for unit vectors 
\begin{equation}
\hat x \sim |\hat x\rangle ={|x\rangle\over ||x||}\sim {x\over||x||},
\end{equation}
we have
\begin{equation} 
v(\hat x) = \langle \hat x | \rho |\hat x\rangle = { \langle  x | \rho |x\rangle\over ||x||^2},
\end{equation} 
or more prosaically
\begin{equation} 
\langle  x | \rho |x\rangle = ||x||^2 v(\hat x).
\end{equation} 
But then
\begin{equation} 
\langle  x+y | \rho |x+y\rangle = ||x+y||^2 v(\widehat{x+y})  = 
\langle x | \rho |x\rangle +  \langle  y | \rho |y\rangle 
+ ( \langle  x | \rho |y\rangle  + \langle  y | \rho |x\rangle ),
\end{equation} 
and 
\begin{equation} 
\langle  x-y | \rho |x-y\rangle = ||x-y||^2 v(\widehat{x-y})  = 
\langle x | \rho |x\rangle +  \langle  y | \rho |y\rangle 
- ( \langle  x | \rho |y\rangle  + \langle  y | \rho |x\rangle ),
\end{equation} 
whence
\begin{equation} 
\langle  x | \rho |y\rangle  + \langle  y | \rho |x\rangle
= {1\over2} \left\{||x+y||^2 v(\widehat{x+y}) - ||x-y||^2 v(\widehat{x-y})  \right\}.
\end{equation} 
(In a real Hilbert space we could stop here since then $\langle  x | \rho |y\rangle  = \langle  y | \rho |x\rangle$.)  Similarly, in a complex Hilbert space, 
\begin{equation} 
\langle  x+iy | \rho |x+iy\rangle = ||x+iy||^2 v(\widehat{x+iy})  = 
\langle x | \rho |x\rangle +  \langle  y | \rho |y\rangle 
+ i ( \langle  x | \rho |y\rangle  - \langle  y | \rho |x\rangle ),
\end{equation} 
and
\begin{equation} 
\langle  x-iy | \rho |x-iy\rangle = ||x-iy||^2 v(\widehat{x-iy})  = 
\langle x | \rho |x\rangle +  \langle  y | \rho |y\rangle 
- i( \langle  x | \rho |y\rangle  - \langle  y | \rho |x\rangle ),
\end{equation} 
whence
\begin{equation} 
\langle  x | \rho |y\rangle  - \langle  y | \rho |x\rangle
= -{i\over2} \left\{||x+iy||^2 v(\widehat{x+iy}) - ||x-iy||^2 v(\widehat{x-iy})  \right\}.
\end{equation} 
Combining these results
\begin{eqnarray} 
\langle  x | \rho |y\rangle &=&
+  {1\over4} \left\{||x+y||^2 v(\widehat{x+y}) - ||x-y||^2 v(\widehat{x-y})  \right\}
\nonumber
\\
&& -{i\over4} \left\{||x+iy||^2 v(\widehat{x+iy}) - ||x-iy||^2 v(\widehat{x-iy})  \right\}.
\end{eqnarray} 
This finally justifies our construction of the density matrix $\rho$ as presented above.

\textbf{f) Two dimensions: } Although Gleason's theorem does not apply in two dimensions, there are improved versions of Gleason's theorem based on POVMs, see \cite{busch2003quantum, caves2004gleason}, that do apply to 2-dimensional Hilbert space. In this case the formalism simplifies even further: Let $\hat x$ and $\hat y$ be any orthonormal basis for the 2-dimensional Hilbert space. Then in terms of the valuation $v(\cdot)$ the density matrix is
\begin{eqnarray}
\rho &=&  v(\hat x) \; |\hat x\rangle \; \langle \hat x| +  v(\hat y) \; |\hat y\rangle \; \langle \hat y| 
\nonumber\\
&&+{1\over2} 
\left\{ {v\left(\hat x+\hat y\over\sqrt{2}\right) - v\left(\hat x-\hat y\over\sqrt{2}\right)} \right\}   
\left( \vphantom{\Big{|}} |\hat x\rangle \; \langle \hat y| +  |\hat y\rangle \; \langle \hat x| \right)  
\nonumber\\
&&-{i\over2}  
\left\{  \; {v\left(\hat x+i\hat y\over\sqrt{2}\right)-v\left(\hat x-i\hat y\over\sqrt{2}\right)}\right\}    
\left( \vphantom{\Big{|}} |\hat x\rangle \; \langle \hat y| -  |\hat y\rangle \; \langle \hat x| \right).  
\end{eqnarray}
If desired one can further rewrite this in terms of the Pauli $\sigma$ matrices
\begin{eqnarray}
\rho &=&  {v(\hat x) + v(\hat y)\over2} \; I_{2\times2} + {v(\hat x) - v(\hat y)\over2}\; \sigma_z 
\nonumber\\
&+&{1\over2} 
\left\{ {v\left(\hat x+\hat y\over\sqrt{2}\right) - v\left(\hat x-\hat y\over\sqrt{2}\right)} \right\}   \sigma_x
-{i\over2}  
\left\{  \; {v\left(\hat x+i\hat y\over\sqrt{2}\right)-v\left(\hat x-i\hat y\over\sqrt{2}\right)}\right\}    \sigma_y.
\qquad
\end{eqnarray}
For real 2-dimensional Hilbert space this further simplifies to
\begin{eqnarray}
\rho &=&  v(\hat x) \; |\hat x\rangle \; \langle \hat x| +  v(\hat y) \; |\hat y\rangle \; \langle \hat y| 
\nonumber\\
&&+{1\over2} 
\left\{ {v\left(\hat x+\hat y\over\sqrt{2}\right) - v\left(\hat x-\hat y\over\sqrt{2}\right)} \right\}   
\left( \vphantom{\Big{|}} |\hat x\rangle \; \langle \hat y| +  |\hat y\rangle \; \langle \hat x| \right).  
\end{eqnarray}
(For completeness, note that for one dimension the valuation trivializes to $v(\cdot)\equiv 1$, and so the density matrix trivializes to $\rho \equiv I_{1\times 1}$.)

\textbf{g) Comments: } 
\begin{enumerate}[noitemsep, topsep=0pt, label=\roman*)]
	\item We have not attempted to provided a new proof of Gleason's theorem. We have in mind a much more modest attempt at trying to understand what the density matrix actually looks like directly in terms of the probability valuations $v(\cdot)$ on a limited number of subspaces of the Hilbert space.  
	\item Gleason's theorem is profound that it shapes the probabilistic nature of quantum theory resulting in the Born rule.  It places strong constraints on any attempts to modify this probabilistic formalism.  However, it still requires the assumption of a probabilistic aspect for its derivation.
	\item Future work regarding this explicit construction of the density operator may involve applications to quantum information science.  This may reveal interesting links between quantum foundations, and to the fundamental quantum information results such as no-cloning or no-broadcasting.
\end{enumerate}

\subsection{Kochen–Specker theorem}

In this subsection, we want to articulate a concept known as realism.  Realism is the view that physical properties have definite values which exist independent of observation.  (This of course seems obvious to classical intuition.)  It is also known as value definiteness \cite{zeilinger2017quantum} where it is said that the properties of physical objects always have definite values even if they are not measured or accessible for any observer.  In quantum theory, values of physical objects are revealed at the moment of measurement; prior to that we only have access to the quantum state and are not given a physical picture of the world.  The crucial question is whether there could be a value definite structure underlying quantum theory?

The Kochen-Specker theorem (sometimes called the Bell–Kochen–Specker theorem) can be crudely stated that if a theory reproduces the results of quantum theory and also has value definiteness, then that theory must be contextual.  Contextuality is the property that the result of a measurement can depend on what combination of measurements we chose to do!  In other words, the outcome of a question depends on what other questions we are simultaneously trying to answer alongside it.  For a classical analogy, suppose one is trying to measure a person's height.  Then contextuality in this scenario implies one gets a different value for height if one measured the person's weight along with it than one would get if one measured the person's shoe size along with it!  To avoid a contextual characteristic to a theory, the alternative method is to give up value definiteness.  In this case the values do not exist before one does a measurement!  

\textbf{a) Preliminaries: }  Our aim is to state the Kochen-Specker theorem and provide a proof.  We shall phrase our discussion in terms of real Hilbert spaces, noting that a complex Hilbert space can always be viewed as a real Hilbert space of double the dimensionality $\mathbb{C}^n \sim \mathbb{R}^{2n}$.  One view of the Kochen–Specker theorem is that it demonstrates the impossibility of consistently assigning $\{0, 1\}$ truth values to quantum propositions.  It was originally proved some fifty years ago by explicitly finding a set of $117$ distinct projection operators on $3$-dimensional Hilbert space \cite{kochen1975problem, bell1966problem}, and then showing that there was no way to consistently assign values in $\{0, 1\}$ to these projection operators. (That is, these 117 “quantum questions” that one might ask could not be consistently assigned yes-no answers.) A later version of the Kochen–Specker theorem reduced the number of projection operators to $33$ \cite{peres1991two}. This was further reduced to 24 \cite{peres1991two}, to 20 \cite{kernaghan1994bell}, and then to 18 \cite{cabello1997proof, cabello1996bell}, at the cost of slightly increasing the dimension of the Hilbert space to 4. Ultimately the number of projection operators was further reduced to 13 in an 8-dimensional Hilbert space in reference \cite{kernaghan1995kochen}. Interest in these foundational issues has continued unabated \cite{mermin1990s, mermin1990simple}, with at least two ``geometrical" proofs that avoid explicit construction of sets of projection operators \cite{gill1996geometric, calude2014kochen}.

In this subsection, we will now present work \cite{rajan2017kochen} this is \textit{part of the original component of this thesis} (which was done in collaboration with my supervisor).  We shall provide yet another even more simplified ``geometrical'' proof of the Kochen–Specker theorem, which, while it is still non-constructive, (proceeding by establishing an inconsistency), is utterly minimal in its technical requirements, and so hopefully instructive.

\textbf{b) Kochen–Specker theorem: }  An explicit statement of the Kochen--Specker theorem, (based on the discussion in the Stanford encyclopaedia of philosophy), runs thus:

\begin{theorem}\label{kockenspecker}
	\textbf{(Kochen-Specker -- mathematical version)} 
	
	Let $H$ be a Hilbert space of quantum state vectors of real dimension $d \geq  3$. Then there is a set $M$ of observables on $H$, containing $n$ elements, such that the following two assumptions are contradictory: 
		\begin{description}
		\item[KS1:]  All $n$ members of $M$ simultaneously have values, that is, they are unambiguously mapped onto real numbers (designated, for specific observables $A, B, C$, ..., by values $v(A)$, $v(B)$, $v(C)$, ...).
		\item[KS2:]  Values of observables conform to the following constraints:
		\vspace{-5pt}
		\begin{description}
			\itemsep-3pt
			\item[(a)] If $A, B, C$ are all compatible and $C = A+B$, then $v(C) = v(A)+v(B)$.
			\item[(b)] If  $A, B, C$ are all compatible and $C = AB$, then $v(C) = v(A)v(B)$.
			\item[(c)] $\exists$ at least one observable $X$ with $v(X)\neq 0$.\\
			(Here ``compatible'' means that the observables commute.)
		\end{description}
	\end{description}\hfill$\Box$ 
\end{theorem}

The statement {\bf KS1} essentially captures the notion of value definiteness (or realism).  The assumptions {\bf KS2a} and {\bf KS2b} are respectively referred to as the sum rule and product rule.  Both of these assumptions are based on what is known as the functional composition principle which is in turn a consequence of non-contextuality.  (The explicit connection among these various statements can be found in the Stanford encyclopaedia of philosophy).

There are several technical issues with the above presentation. Without condition {\bf KS2c} the theorem is actually false --- the trivial valuation where for all observables $X$ one sets $v(X)= 0$ provides an explicit counter-example. 
Without condition {\bf KS2c}, $v(I) = v(I^2) = v(I)^2$ only implies $v(I)\in\{0,1\}$.
With condition {\bf KS2c} we have the stronger statement that  $v(X) = v(IX) = v(I)v(X)$, which since $v(X)\neq 0$  implies $v(I)=1$.

A more subtle issue is this: 
Physically, we would like to have $v(zA)=z\, v(A)$, for any $z\in \mathbb{C}$.
But using the conditions  {\bf KS2a} and  {\bf KS2b}  we could only deduce this for rational numbers. 
Extending this to the complex numbers requires us to first construct the real numbers ``on the fly'' using Dedekind cuts, and then to formally construct the complex numbers as an algebraic extension of the field of real numbers --- while this is certainly possible, in a physics context it is rather pointless --- it would seem more reasonable to start with the complex numbers as being given, even if you then need slightly stronger axioms.

{\bf Improved KS2 axioms:}
\vspace{-7pt}
\begin{description}
	\itemsep-3pt
	\item[(a)] If $[A,B]=0$ and $a,b\in \mathbb{C}$, then $v(aA+bB) = a\,v(A)+b\,v(B)$.
	\item[(b)] If $[A,B]=0$ then $v(AB) = v(A)\,v(B)$.
	\item[(c)] $\exists$ at least one observable $X$ with $v(X)\neq 0$.
\end{description}
\vspace{-7pt}

If one accepts these improved {\bf KS2} axioms then immediately
\begin{equation}
v(I)=1; \qquad v(aI)=a; 
\end{equation}
and for any analytic function with a non-zero radius of convergence
\begin{equation}
v(f(A))= f(v(A)).
\end{equation}
Note that this last condition, $v(f(A))= f(v(A))$, is where physics discussions of the Kochen--Specker theorem often start. 
Indeed let us write $A=\sum_i  a_i P_i$ where the $a_i$ are real and the $P_i$ are projection operators onto 1-dimensional subspaces; so the projection operators $P_i = |\psi_i\rangle\,\langle\psi_i|$ can be identified with the vectors $|\psi_i\rangle$ which form a basis for the Hilbert space. Then
\begin{equation}
v(A) = v\left(\sum_i  a_i P_i\right) = \sum_i  a_i \; v(P_i).
\end{equation}
This now focusses attention on the valuations $v(P_i)$. Since $P_i^2=P_i$, condition {\bf KS2b} implies that $v(P_i)\in\{0,1\}$; the valuation must be a yes-no valuation. Now consider the identity operator $I = \sum_i P_i$ and note
\begin{equation}
\sum_i v(P_i) = v(I) = 1.
\end{equation}
It is customary to identify the projectors $P_i = |n_i\rangle\, \langle n_i | $ with the corresponding unit vectors $n_i$, (defined up to a sign), with the $n_i$ forming a basis for Hilbert space, and in $d$ dimensions write
\begin{equation}
\sum_{i=1}^d v(n_i)  = 1;    \qquad \qquad v(n_i)  \in\{0,1\}; \qquad\quad v(-n)=v(n). 
\end{equation}
It is the claimed existence of this function $v(n)$, having the properties stated above for any arbitrary basis of Hilbert space, which is the central point of the {\bf KS1} and {\bf KS2} conditions.  This discussion allows us to rephrase the Kochen--Specker theorem in terms of the non-existence of such a valuation.

\begin{theorem}\label{kockenspecker2}
	\textbf{(Kochen-Specker --- physics-based  version)} 
	
	For $d\geq3$ there is no valuation $v(n)  : S^{d-1} \to \{0, 1\}$, where $S^{d-1}$
	is the unit hypersphere, such that $v(-n) = v(n)$ for all $n$ and
	\begin{equation}
	\sum_{i=1}^d v(n_i)  = 1,
	\end{equation}
	for every basis (frame, $d$-bein) of orthogonal unit vectors  $n_i$. \hfill$\Box$ 
\end{theorem}

It is this statement about bases in Hilbert space that is often more practical to work with, rather than the formulation at the start of this section --- of course without that initial formulation it would be less than clear why the basis formulation is physically interesting. 

We will start by looking in a non-traditional place, by considering one-dimensional and two-dimensional Hilbert spaces, before dealing with three-dimensional Hilbert space, (which then settles things for any higher dimensionality).  Since one is trying to prove an inconsistency result, there will be an infinite number of ways of doing so; the question is whether one learns anything new by coming up with a different proof. We shall do so with a modified and simplified ``descent'' argument, one that requires only two steps in the descent process.

\textbf{c) One dimensions: } There is no Kochen--Specker no-go result in one dimension, since in one dimension all operators are multiples of the identity, $A= a I$, and then
\begin{equation}
v(f(A)) = v(f(aI)) = v(f(a)I) = f(a)\, v(I) = f(a). 
\end{equation}
In particular, as long as $f(a)\not\equiv 0$,  (which is implied by the {\bf KS2c} axiom),  then for the (unique) normalized basis vector $n$ we have $v(n) = 1$.  Conversely if we are considering a one-dimensional subspace of a higher-dimensional Hilbert space then the {\bf KS2c} axiom tells us nothing;  for the (unique) normalized basis vector we merely have $v(n)\in\{0,1\}$, and we have no further constraint on the valuation. 

\textbf{d) Two dimensions: } There is no Kochen--Specker no-go result in two dimensions, but there are still quite interesting things to say. Consider the valuation $v : S^1 \to \{0, 1\}$ (where $S^1$ is the unit circle) such that $v(-n) = v(n)$ for all $n$ and
\begin{equation}
v(n_1) + v(n_2) = 1 
\end{equation}
for every dyad (every pair of orthogonal unit vectors) $n_1$, $n_2$.  Indeed in two dimensions we can construct such a valuation. Re-characterize $n_1$ and $n_2$ in terms of the angle they make with (say) the $x$ axis; then the constraints we want to impose are
\begin{equation}
v(\theta)=v(\theta+\pi); \qquad\qquad v(\theta) + v(\theta\pm\textstyle{\pi\over2}) = 1.
\end{equation}
But these conditions are easily solved: Let $g(\theta)$ be some arbitrary (not necessarily continuous) function mapping the interval $\left[0,\textstyle{\pi\over2}\right) \to \{0,1\}$,  and define
\begin{equation}
v(\theta) =  \left\{  \begin{array}{cl} 
\vphantom{\Big |} g(\theta) & \hbox{ for } \theta \in  [0,\textstyle{\pi\over2});\\
\vphantom{\Big |} 1-g(\theta- \textstyle{\pi\over2}) & \hbox{ for } \theta \in  [\textstyle{\pi\over2},\pi);\\
\vphantom{\Big |} g(\theta-\pi) & \hbox{ for }\theta \in  [\pi,\textstyle{3\pi\over2});\\
\vphantom{\Big |} 1- g(\theta - \textstyle{3\pi\over2}) & \hbox{ for }\theta \in  [\textstyle{3\pi\over2},2\pi).\\
\end{array}\right. 
\end{equation}
So the existence of a Kochen--Specker valuation is easily verified in two dimensions, and because points separated by $\pi/2$ radians must be given opposite valuations, the image $v(S^1)$ is automatically 50\%--50\% zero-one.  Note in particular that the function $v(\theta)$ cannot be everywhere continuous. (We will recycle these results repeatedly when we turn to three and higher dimensions.)

\textbf{e) Three dimensions: } It is in 3 dimensions that things first get interesting. We are interested in  valuations $v : S^2 \to \{0, 1\}$, (where $S^2$
is the unit 2-sphere), such that $v(-n) = v(n)$ for all $n$ and
\begin{equation}
v(n_1) + v(n_2) + v(n_3) = 1 
\end{equation}
for every triad (every triplet of orthogonal unit vectors) $n_1$, $n_2$, $n_3$.  In the argument below we shall make extensive use of the great circles $S^1$ in the unit 2-sphere $S^2$.

{\bf Lemma:}
On any great circle in $S^2$, under the conditions given above, the valuation is either  50\%--50\% zero-one (as in two dimensions), or is 100\% zero (identically zero). \hfill $\Box$
\begin{proof}
	Pick any great circle and for convenience align it with the equator. \\
	Now look at the poles:
	\begin{itemize}
		\itemsep-3pt
		\item 
		If $v(\hbox{poles})=1$, then $v(\hbox{equator})\equiv 0$ is identically zero.\\
		(Since points on the equator will be part of some triad that includes the unit vector pointing to the poles.)
		\item 
		If $v(\hbox{poles})=0$, then any dyad lying in the equator will satisfy the conditions of the two dimensional argument given above, and so will be 50\%--50\% zero-one.
	\end{itemize}
\end{proof}
Now bootstrap this to a modified ``great circle descent'' argument,  one that needs only two steps in the descent process.  We start with a purely geometrical result.  From the argument above if we arrange $v(\hbox{poles})=1$, then $v(\hbox{equator})\equiv 0$, and for each line of longitude $v(\hbox{meridian})$ will be 50\%--50\% zero-one. (See figure \ref{F:longitudes}.)
\begin{figure}[!htbp]
	\begin{center}
		\includegraphics[scale=1.00]{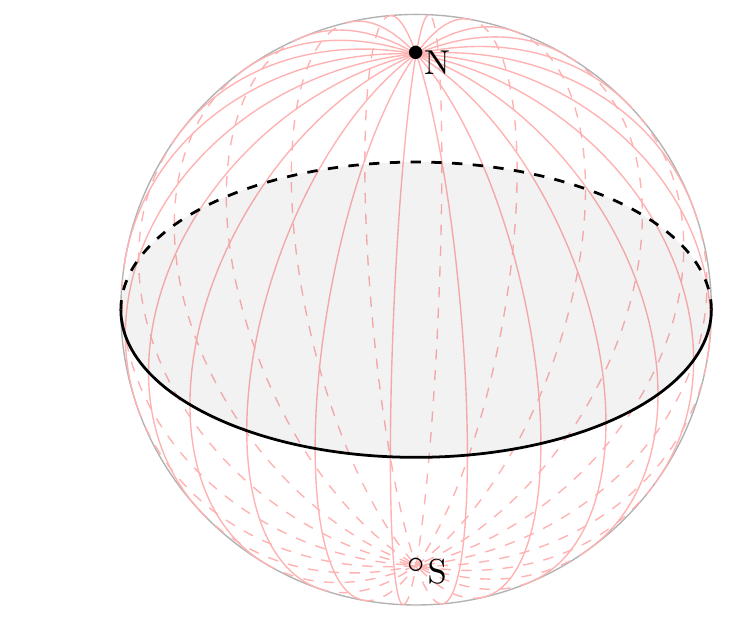}
		\caption{Setup with $v(\hbox{poles})=1$, $v(\hbox{equator})\equiv 0$, and $v(\hbox{meridians})$ 50\%--50\% zero-one.}
		\label{F:longitudes}
	\end{center}
\end{figure}

We define a ``great circle descent''  $C(p)$ through a point $p$ on the sphere as a great circle that starts off at constant latitude.
(So the point $p$ is either the northernmost or southernmost point on the great circle.  See figures~\ref{F:descent-1} and~\ref{F:descent-2}.)

\begin{figure}[!htbp]
	\begin{center}
		\includegraphics[scale=1.00]{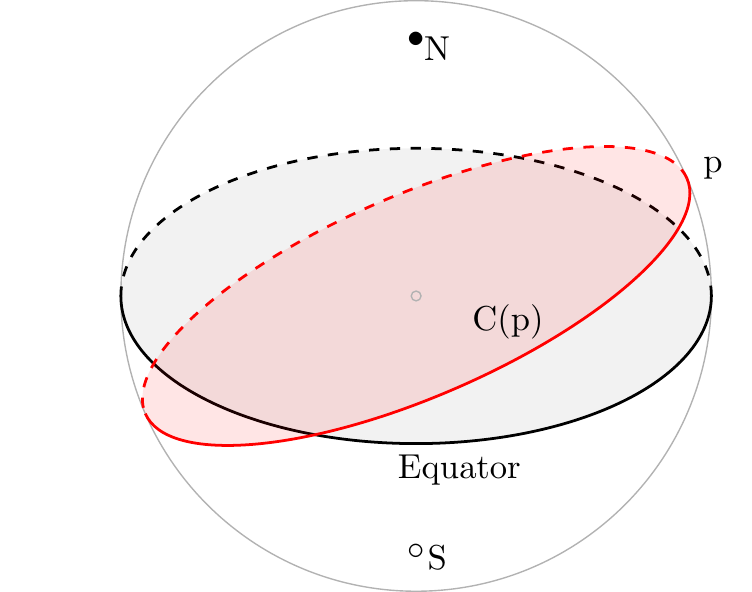}
		\caption{Descent great circle $C(p)$, with northernmost point at $p$.}
		\label{F:descent-1}
	\end{center}
\end{figure}

\begin{figure}[!htbp]
	\begin{center}
		\includegraphics[scale=0.50]{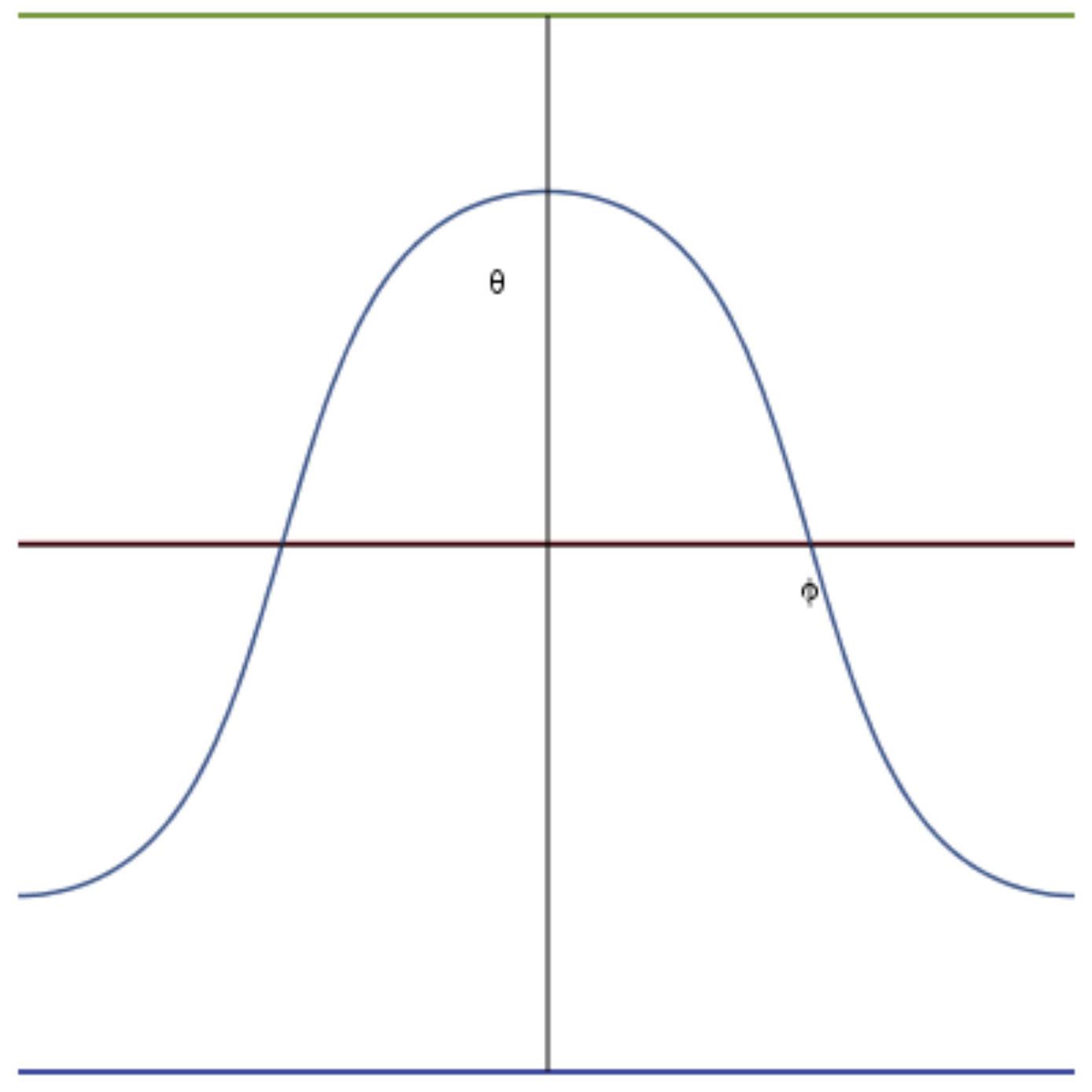}
		\caption{Descent great circle represented in terms of $\theta(\phi)$.}
		\label{F:descent-2}
	\end{center}
\end{figure}

{\bf Lemma:} Let $q$ be any other point at the same longitude as $p$ (the same meridian) that is closer to the equator than $p$. Then there exists a point $r$ such that $r$ lies on the great circle descent through $p$, and $q$ lies on the great circle descent through $r$. \hfill $\Box$

That is $r\in C(p)$ and $q\in C(r)$, so one can always zig-zag directly to towards the equator via exactly two great circle descents. Note that this is a much easier geometric result than that used in the Gill--Keane \cite{gill1996geometric} or Calude--Hertling---Svozil \cite{calude2014kochen} approaches where a finite but possibly large number of great circle descents is used to get to any point closer to the equator, not necessarily at the same longitude. (See figure~\ref{F:descent-2-step}.)

\begin{figure}[!htbp]
	\begin{center}
		\includegraphics[scale=0.750]{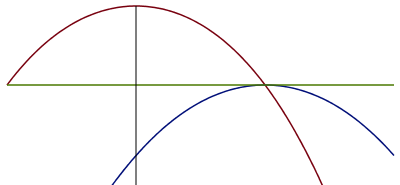}
		\caption{Example of a two-step descent with net motion along a meridian.}
		\label{F:descent-2-step}
	\end{center}
\end{figure}

\begin{proof}
Using spherical coordinates $(\theta,\phi)$ let the generic point $x$ be represented by the 3-vector 
\begin{equation} 
\vec x = (\cos\theta \cos\phi, \cos\theta \sin\phi,\sin\theta).
\end{equation} 
(Somewhat unusually, we adopt conventions close to the usual latitude nomenclature:  $\theta=+\pi/2$ represents the north pole, $\theta=0$ represents the equator, while while $\theta=-\pi/2$ represents the south pole. Doing this simplifies some of the formulae below.)\\
Now let the specific point $p$ of interest be represented by the 3-vector 
\begin{equation} 
\vec p = (\cos\theta_p \cos\phi_p, \cos\theta_p \sin\phi_p,\sin\theta_p).
\end{equation} 
Consider the great circle descent $C(p)$. This great circle will be orthogonal to the vector
\begin{equation} 
\vec p_\perp = (\sin\theta_p \cos\phi_p, \sin\theta_p \sin\phi_p,-\cos\theta_p).
\end{equation} 
The entire great circle $C(p)$ will be characterized by $\vec p_\perp \cdot \hat x(\theta,\phi) =0$, that is
\begin{equation} 
\sin\theta_p\cos\theta (\cos\phi_p \cos\phi + \sin\phi_p \sin \phi_p) - \cos\theta_p \sin\theta  = 0,
\end{equation} 
implying
\begin{equation} 
\sin\theta_p\,\cos\theta\, \cos(\phi - \phi_p) = \cos\theta_p \sin\theta.
\end{equation} 
That is
\begin{equation} 
\tan \theta = {\tan \theta_p \cos(\phi-\phi_p)},
\end{equation} 
or more explicitly 
\begin{equation} 
\theta(\phi) = \tan^{-1} \left({\tan \theta_p\; \cos(\phi-\phi_p)}\right).
\end{equation} 
This explicitly yields $\theta(\phi)$ along the entire descent circle $C(p)$.

Note that this descent circle crosses the equator at $\theta=0$, implying $(\phi-\phi_p)=\pm\pi/2$. This occurs at the points $s$ such that $\vec s=\pm(-\sin\phi_p,\cos\phi_p,0)$.  In particular, for the three points $p$, $r$, $q$, (and using $\phi_p=\phi_q$ because we want $p$ and $q$ to have the same longitude), we have
\begin{equation} 
\tan \theta_r = {\tan \theta_p\; \cos(\phi_r-\phi_p)}; \qquad \tan \theta_q = {\tan \theta_r\; \cos(\phi_r-\phi_p)};
\end{equation} 
implying
\begin{equation} 
\tan \theta_q = {\tan \theta_p\; \cos^2(\phi_r-\phi_p)}.
\end{equation} 
That is
\begin{equation} 
\cos^2(\phi_r-\phi_p) = {\tan \theta_q\over \tan\theta_p}.
\end{equation} 
Alternatively
\begin{equation}
|\phi_r-\phi_p| = \cos^{-1} \sqrt {\tan \theta_q\over \tan\theta_p}.
\end{equation}
The azimuthal difference $|\phi_r-\phi_p|$ tells you exactly how much you have to zig-zag along the descent circles for the net motion to be directly along the line of longitude towards the equator. Note $|\phi_r-\phi_p|$ is real only if you move towards (rather than away from) the equator.		
\end{proof}

\begin{figure}[!htbp]
	\begin{center}
		\includegraphics[scale=0.5]{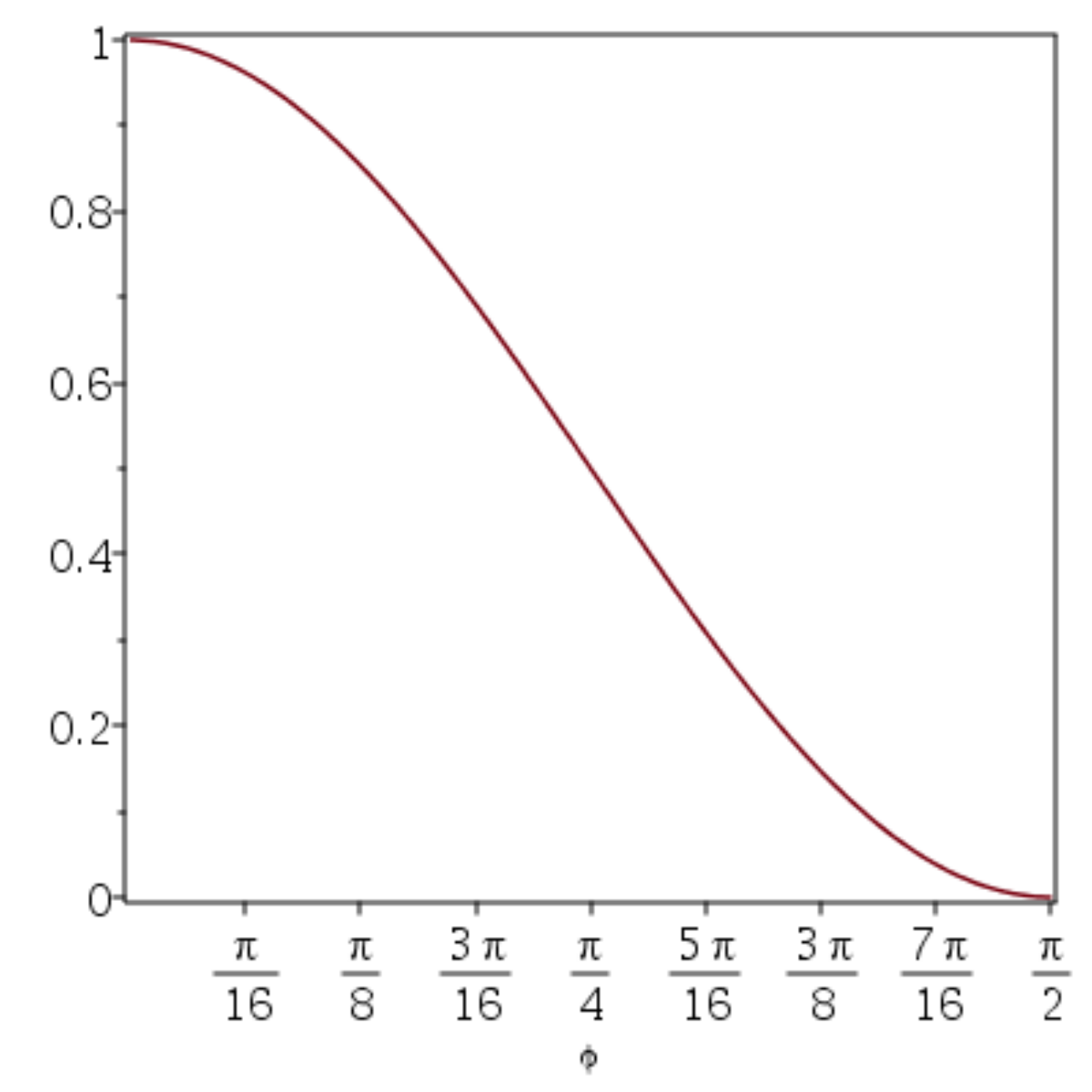}
		\caption{Quantifying $\delta\phi$ in terms of $\cos^{-1} \sqrt {\tan \theta_q\over \tan\theta_p}$ for two-step descent towards the equator.}
		\label{F:delta-phi}
	\end{center}
\end{figure}

{\bf Application to the Kochen-Specker theorem:}\\
Consider any point $p_0$ such that $v(p)=1$ and rotate to put it at the north or south pole. Then by hypothesis $v(equator)=0$ for any point on the equator.  Now consider any other point $p$ such that $v(p)=0$ and $p$ is not on the equator.  Consider the descent circle $C(p)$; we have $v(p)=0$ by hypothesis, and $v(s)=0$ at the perpendicular point $s$ with $\vec s=(0,-\sin\phi_0,\cos\phi_0)$ where $C(p)$ crosses the equator. Therefore $v(C(p))\equiv 0$ everywhere on this descent circle.  But in particular this implies that $v(r)=0$.
Thence 
$v(C(r))\equiv 0$ everywhere on this descent circle. Thence $v(q)=0$. This means we have proved:

{\bf Lemma:} If $v(pole)=1$ and $v(p)=0$ then also $v(q)=0$ for $q$ any point on the same line of longitude (same meridian) as $p$ that is closer to the equator than $p$. 
\hfill $\Box$

Consequently, for any line of longitude for which $v(poles)=1$,  we see that $v^{-1}(0)$ is path connected. Specifically this implies that  $\exists \; \pi/2\leq \theta_*\leq 0$ such that either
\begin{equation}
v(\theta) =  \left\{  \begin{array}{cl} 
\vphantom{\Big |} 1 & \hbox{ for } {\pi\over2}\geq \theta \geq \theta_*;\\
\vphantom{\Big |} 0 & \hbox{ for }  \theta_* > \theta \geq \theta_* - {\pi\over2};\\
\vphantom{\Big |}  1 & \hbox{ for } \theta_* - {\pi\over2} > \theta \geq -{\pi\over2};\\
\end{array}\right.
\end{equation}
or
\begin{equation}
v(\theta) =  \left\{  \begin{array}{cl} 
\vphantom{\Big |} 1 & \hbox{ for } {\pi\over2}\geq \theta > \theta_*;\\
\vphantom{\Big |} 0 & \hbox{ for }  \theta_* \geq \theta > \theta_* - {\pi\over2};\\
\vphantom{\Big |}  1 & \hbox{ for } \theta_* - {\pi\over2} \geq \theta \leq -{\pi\over2}.\\
\end{array}\right.
\end{equation}
(See figure \ref{F:piece-wise-connected}.)
\begin{figure}[!htbp]
	\begin{center}
		\includegraphics[scale=2.5]{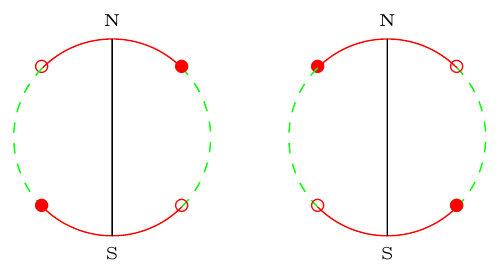}
		\caption{Assuming $v(poles)=1$, as a consequence of the two-step descent argument \emph{any} meridian can be put into one of these two forms for \emph{some} value of $\theta_*$.}
		\label{F:piece-wise-connected}
	\end{center}
\end{figure}

Now pick any specific line of longitude, by interchanging the north and south poles we can without loss of generality assert
\begin{equation}
v(\theta) =  \left\{  \begin{array}{cl} 
\vphantom{\Big |} 1 & \hbox{ for } {\pi\over2}\geq \theta \geq \theta_*;\\
\vphantom{\Big |} 0 & \hbox{ for }  \theta_* > \theta \geq \theta_* - {\pi\over2};\\
\vphantom{\Big |}  1 & \hbox{ for } \theta_* - {\pi\over2} > \theta \geq -{\pi\over2}.\\
\end{array}\right.
\end{equation}
Now rotate the sphere $S^2$ around the polar axis so that the line of longitude we have chosen lies on the zero meridian $\phi_*=0$ (the prime meridian). Then furthermore rotate the sphere $S^2$ around the axis perpendicular to the zero meridian so that point $p_*= (\sin\theta_*, 0, \cos\theta_*)$ is moved to the north pole. That is:

{\bf Lemma:}
Without any loss of generality we can choose the zero meridian to satisfy
\begin{equation}
v(\theta) =  \left\{  \begin{array}{cl} 
\vphantom{\Big |} 1 & \hbox{ for } \theta = {\pi\over2};\\
\vphantom{\Big |} 0 & \hbox{ for }  {\pi\over2} > \theta \geq 0;\\
\vphantom{\Big |}  1 & \hbox{ for }  0> \theta \geq -{\pi\over2}. \\
\end{array}\right.
\label{E:key}
\end{equation}
(See figure~\ref{F:standardized}.)

\begin{figure}[!htbp]
	\begin{center}
		\includegraphics[scale=3.0]{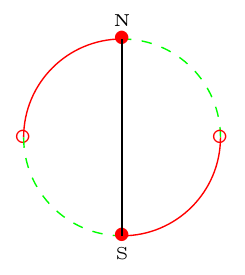}
		\caption{Assuming $v(poles)=1$,  after suitable rotations the prime meridian can always be put into this standardized form.}
		\label{F:standardized}
	\end{center}
\end{figure}
This will now quickly lead to a contradiction. \hfill $\Box$

First consider all the descent great circles $C(p)$ based on this particular choice of zero meridian. These descent great circles will (in the northern hemisphere) sweep out the entire half-hemisphere $\phi\in(-\pi/2,+\pi/2)$ and $\theta\in(\pi/2,0)$. Similarly, in the southern hemisphere these decent circles will in turn sweep out  the complementary half-hemisphere $\phi\in(+\pi/2,\pi] \cup [-\pi,-\pi/2)$ and $\theta\in(0,-\pi/2)$. But, following previous arguments,  since $v(\theta)=0$ at the apex of all these descent great circles, $v(C(p))=0$ for all these descent great circles. That is:

{\bf Lemma:}
Without loss of generality we have chosen the zero meridian such that (except possibly at the poles themselves) 
\begin{equation}
v(\theta>0,|\phi|<\pi/2) =0; \qquad\hbox{and} \qquad v(\theta<0,|\phi|<\pi/2) =1; 
\label{E:zero1}
\end{equation}
\begin{equation}
v(\theta<0,|\phi|>\pi/2) =0; \qquad\hbox{and} \qquad v(\theta>0,|\phi|>\pi/2) =1.
\label{E:zero2}
\end{equation}

\begin{figure}[!htbp]
	\begin{center}
		\includegraphics[scale=1.0]{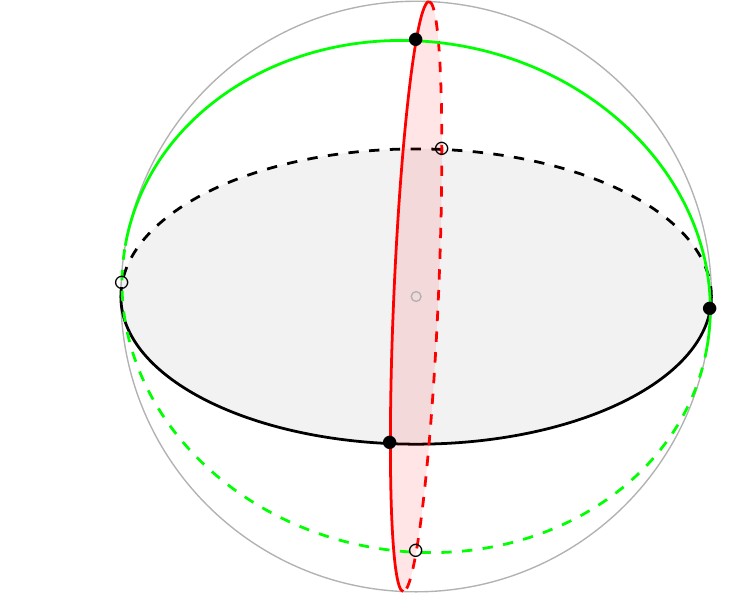}
		\caption{Let green denote the prime meridian at $\phi=0$, black the equator at $\theta=0$, and red the meridians at $\phi=\pm\pi/2$. The equator and red meridians split the sphere into four segments, with two of these segments having valuation zero, and two segments having valuation unity.}
		\label{F:segments}
	\end{center}
\end{figure}

Thus the valuation $v(\cdot)$ is 50\%--50\% zero-one over the entire 2-sphere $S^2$. \hfill $\Box$  

Completing the inconsistency argument can now be done in many ways (in fact, an infinite number of ways). Consider any meridian with $\phi_*\neq 0$ and $|\phi_*| <\pi/2$. On the one hand this meridian will also have the same valuation, equation~(\ref{E:key}), as the zero meridian. On the other hand  by considering the descent great circles based on this new meridian we have
\begin{equation}
v(\theta>0,|\phi-\phi_*|<\pi/2) =0; \qquad\hbox{and} \qquad v(\theta>0,|\phi-\phi_*|<\pi/2) =1; 
\end{equation}
\begin{equation}
v(\theta<0,|\phi-\phi_*|>\pi/2) =0; \qquad\hbox{and} \qquad v(\theta<0,|\phi-\phi_*|>\pi/2) =1.
\end{equation}
But this is incompatible with the behaviour based on the zero meridian, equations~(\ref{E:zero1}) and (\ref{E:zero2}), so we have a contradiction.  This completes the proof of Kochen--Specker in three dimensions. We feel that this is a nice simple proof of Kochen--Specker that does not rely on finding explicit bases for the Hilbert space --- it also seems to us  to be considerably simpler than the other geometric or colouring arguments.

\textbf{f) Four dimensions and higher: } What happens in a $d>3$-dimensional Hilbert space? The 3-dimensional logic carries over with utterly minimal modifications. 
\begin{itemize}
	\item 
	In $d=4$ one needs to study the unit 3-sphere $S^3$. Pick any point $n$ on $S^3$ such that $v(n)=0$. This can always be done. Then consider the 2-sphere perpendicular to chosen point $n$.
	On that 2-sphere the 4-dimensional Kochen--Specker theorem will reduce to the 3-dimensional Kochen--Specker theorem, which we have already established. So nothing more need be done.
	\item
	In $d\geq 4$ dimensions one needs to study the unit $(d-1)$-sphere $S^{d-1}$. Pick any $d-3$ mutually-orthogonal points $n_i$ on $S^3$ such that $v(n_i)=0$.  If this cannot be done then the existence of the claimed valuation $v(\cdot)$  already fails at this elementary level so that the $d$-dimensional Kochen--Specker theorem is established; so without loss of generality we can assume this can be done.
	Then consider the 2-sphere perpendicular to all the $n_i$.
	On that 2-sphere the $d$-dimensional Kochen--Specker theorem will reduce to the 3-dimensional Kochen--Specker theorem, which we have already established. So nothing more need be done.
\end{itemize}
It is interesting to note that 3-dimensions is the key part of the theorem; in 1 and 2 dimensions related results are trivial. In 4 or more dimensions the Kochen--Specker theorem follows immediately from the 3-dimensional result. 

\textbf{g) Comments: } 
\begin{enumerate}[noitemsep, topsep=0pt, label=\roman*)]
	\item We have presented a geometric approach where one constructs and exploits the properties of great circles on a $n$-sphere.  This has the power to significantly simplify the argument, while maintaining the validity of the theorem for a minimum dimension of three.
	\item The  Kochen--Specker theorem is  more basic and fundamental than Gleason's theorem. Indeed, if one assumes Gleason's theorem then the Kochen--Specker theorem is trivial. The point is that once one asserts that the valuation $v(\cdot)$ is inherited from a density matrix $v(n) = \langle n| \rho| n \rangle$, then one knows that the valuation is continuous. But no function from the connected space $S^n$ to the discrete set $\{0,1\}$ (with its implied discrete topology) can possibly be continuous.   
	\item The main implication of the result is that quantum theory fails to allow a underlying non-contextual model.  More precisely, it states that it is impossible for the predictions of quantum mechanics to be in line with measurement outcomes which are pre-determined in a non-contextual manner.  Hence this would rule out a large class of models that might otherwise seem at first sight to be intuitive representations of  the physical world.
	\item With respect to quantum information science, there has been recent evidence that contextuality may be the primary reason for the speedup for quantum computation.  This has been shown through `magic' state injection \cite{howard2014contextuality}.
\end{enumerate}

\section{Non-locality across Space}

Non-locality across space is the characteristic that an action on a subsystem can instantaneously influence another subsystem at an arbitrarily far spatial location.  We have seen this strange property exemplified in the previous chapter regarding the entanglement in space.  In this section we will describe two different non-localities across space.  One requires an entanglement in space, and is known as Bell non-locality. The other does not require the entanglement and is known as the violation of preparation independence.

By considering both the probabilistic aspects of quantum theory and the concept of realism, one is led to a mathematical formulation of the Bell non-locality across space.  This will expressed through what is known as Bell's theorem \cite{bell1964einstein}.  We will show that Bell non-locality across space represents a stricter form of non-classical interdependence than an entanglement in space.  The particular version of Bell's theorem we will focus our attention on is the Bell-CHSH or known simply as the CHSH (Clauser-Horne-Shimony-Holt) inequality \cite{clauser1969proposed}.  For the questions in quantum foundations, the CHSH inequality sheds a partial understanding on the nature of quantum measurement.  Due to its implications, Bell's theorem has also been viewed by some as the most profound discovery of science \cite{whitaker1998john}.  For a thorough review of Bell non-locality across space, we refer the reader to \cite{brunner2014bell}.

The second non-locality across space is known as the violation of preparation independence.  It does not require entanglement and applies to product states.  It will be expressed through the PBR (Pusey-Barrett-Rudolph) theorem \cite{pusey2012reality}.  The original aim of the theorem was to shed a partial understanding on the nature of the quantum state.  Due to its implications, the PBR theorem has been referred to as the most important theorem in quantum foundations since Bell's theorem \cite{reich2011quantum}. For a comprehensive reviews of the PBR theorem and the violation of preparation independence, we refer the reader to \cite{leifer2014quantum, jennings2016no}

Our aim in this section is therefore to present the CHSH inequality and the PBR theorem.  We will also articulate both of these results through the lens of a game.

\subsection{Bell-CHSH inequality}

\textbf{a) CHSH inequality: } The CHSH inequality will be used to demonstrate a non-locality across space.  It will be derived without any reference to quantum theory.  Suppose there are three parties who are each spatially apart named Alice, Bob and Charlie.  Charlie prepares two particles and sends one particle to Alice and the other one to Bob.  Each particle can be measured in two quantities.  For Alice's particle we denote these quantities as $A_{1}$ and $A_{2}$, and similarly for Bob's particle we have quantities $B_{1}$ and $B_{2}$.  Each of these can take either value $+1$ or $-1$.  We assume realism, and hence the values are objective properties which exist independent of observation; these values are merely revealed by measurement. 

Both Alice and Bob each choose to measure their respective particles at the same time. With this constraint, we can assume that the measurement of one particle cannot effect the result of the other particle.  This is known as the assumption of locality.   Furthermore, we also require that each choose to measure their particle randomly using their two options.  This is also known as the free will assumption.

We proceed to consider the quantity
\begin{equation}
A_{1}B_{1} + A_{2}B_{1} + A_{2}B_{2} - A_{1}B_{2}.
\end{equation}
This can be re-expressed as
\begin{equation}
A_{1}B_{1} + A_{2}B_{1} + A_{2}B_{2} - A_{1}B_{2} = (A_{1} + A_{2})B_{1} + (A_{2} - A_{1})B_{2}.
\end{equation}
Given that $A_{1}, A_{2} = \pm1$, we have that
\begin{align}
(A_{1} + A_{2})B_{1} = 0,
\end{align}
or  
\begin{align}
(A_{2} - A_{1})B_{2} = 0.
\end{align} 
For both cases, we obtain
\begin{equation}
A_{1}B_{1} + A_{2}B_{1} + A_{2}B_{2} - A_{1}B_{2} = \pm 2.
\end{equation}
Let $p(a_{1},a_{2}, b_{1},b_{2})$ denote the joint probability that before the measurements are performed the total system is in state $A_{1}=a_{1}, A_{2}=a_{2}, B_{1}=b_{1},$ and $B_{2}=b_{2}$.  Using the expectation value (\ref{expectation}), we have
\begin{align}
&\mathbb{E}(A_{1}B_{1} + A_{2}B_{1} + A_{2}B_{2} - A_{1}B_{2}) \\ &= \sum_{a_{1}a_{2}b_{1}b_{2}}p(a_{1},a_{2}, b_{1},b_{2})(a_{1}b_{1} + a_{2}b_{1} + a_{2}b_{2} - a_{1}b_{2}) \\ 
&\leq \sum_{a_{1}a_{2}b_{1}b_{2}} p(a_{1},a_{2}, b_{1},b_{2}) (2) \\ 
&= 2.\label{CHSHderivation1}
\end{align}
We can also deduce that
\begin{align}
&\mathbb{E}(A_{1}B_{1} + A_{2}B_{1} + A_{2}B_{2} - A_{1}B_{2}) \\ &=\sum_{a_{1}a_{2}b_{1}b_{2}}p(a_{1},a_{2}, b_{1},b_{2})a_{1}b_{1} + \sum_{a_{1}a_{2}b_{1}b_{2}}p(a_{1},a_{2}, b_{1},b_{2})a_{2}b_{1} \\
&+ \sum_{a_{1}a_{2}b_{1}b_{2}}p(a_{1},a_{2}, b_{1},b_{2})a_{2}b_{2} 
- \sum_{a_{1}a_{2}b_{1}b_{2}}p(a_{1},a_{2}, b_{1},b_{2})a_{1}b_{2} \\
&= \mathbb{E}(A_{1}B_{1}) + \mathbb{E}(A_{2}B_{1}) + \mathbb{E}(A_{2}B_{2}) - \mathbb{E}(A_{1}B_{2}). \label{CHSHderivaton2}
\end{align}
Using both (\ref{CHSHderivation1}) and (\ref{CHSHderivaton2}), we obtain the CHSH inequality  
\begin{equation}
\mathbb{E}(A_{1}B_{1}) + \mathbb{E}(A_{2}B_{1}) + \mathbb{E}(A_{2}B_{2}) - \mathbb{E}(A_{1}B_{2}) \leq 2.
\end{equation}
This can be re-written using the quantum theoretic notation for expectation value
\begin{empheq}[box=\widefbox]{align}\label{CHSHinfoundations}
\langle A_{1}B_{1} \rangle  + \langle A_{2}B_{1} \rangle + \langle A_{2}B_{2} \rangle - \langle A_{1}B_{2} \rangle \leq 2.
\end{empheq}
This is the equation we saw earlier (\ref{CHSH}) as a means to detect entanglement.  More precisely, using the probabilistic aspects of quantum theory, we saw the CHSH inequality violated using Bell state (\ref{CHSHviolationBell}), resulting in the equation
\begin{equation}
\langle A_{1}B_{1} \rangle  + \langle A_{2}B_{1} \rangle + \langle A_{2}B_{2} \rangle - \langle A_{1}B_{2} \rangle = 2\sqrt{2}.
\end{equation}
The value of $2\sqrt{2}$ is the maximum quantum value and is known as Tsirelson’s bound.  This violation has been experimentally verified \cite{freedman1972experimental, aspect1982experimental} in numerous quantum scenarios.  Hence, these measurement correlations are stronger than could ever exist in classical systems.  It implies a profound consequence in that these quantum correlations overthrow the classical picture of the world; at least one of the three assumptions made to derive the CHSH inequality is wrong.

\textbf{b) Implications: } It is dominantly viewed that the assumption of locality is the one that is incorrect.  Hence, the mathematical characterization for Bell \textit{non-locality across space} is expressed as the violation of (\ref{CHSHinfoundations}). Note that this Bell notion of locality (\ref{CHSHinfoundations}) is distinct from the term locality used in other areas of quantum physics which describes the case that operators defined in spacelike separated regions commute. In this section, when we refer to non-locality we shall mean a Bell non-locality across space.  

It also is common in the literature to interchange between the terms entanglement and non-locality.  Such use may in principle be sufficient for a large number of cases, but falls short of the precision required for an adequate scientific taxonomy.  We proceed to emphasize the differences between an entanglement and non-locality.  The most obvious difference is that former is an algebraic property residing in the mathematics of quantum theory (\ref{bipartiteseparable}), whereas the latter is rooted from the measurement outcomes/correlations of experiments (\ref{CHSHinfoundations}). 

Nevertheless, to obtain non-local correlations from measurements on a quantum state, it is necessary that the state is entangled.  This implies that the observations of non-local correlations means the state is entangled.  Hence our use in (\ref{CHSH}). In a converse direction, it only true that all pure entangled states are non-local.  This means for any entangled pure state one can obtain local measurements such that the measurement correlations violate the CHSH inequality.  (The only pure states that do not violate it are product states.)  However, there are entangled mixed states, such as (\ref{Wernerstates}), that do not violate the CHSH inequality.  Therefore, not all entangled states are non-local.

In the language of quantum information, we can say that the interdependence of certain quantum information systems, that violate the CHSH inequality, would be impossible to replicate by classical information systems, which cannot violate the inequality. One of the utilities of this is that it allows one to detect entangled quantum information systems directly from measurement data without any reference to the physical experiment. This is known as device independence.

From the perspective of quantum foundations, non-locality across space suggests that for a subset of entangled cases, a quantum measurement on one system has the ability to \textit{instantaneously} influence another system that can be arbitrarily spatially far.  Hence an alternative perspective to gain is that it sheds a partial understanding on the non-trivial properties of certain quantum measurements.  Furthermore, this instantaneous characteristic implies the \textit{lack of a time interval} in this scenario.  With a time interval involved, the non-local influence could be explained away by some hidden causal signal.  Hence the \textit{interdependence in this non-locality across space is shocking due to the absense of a time interval involved}.

The result also has a influence on philosophy, which can be highlighted by the subject being termed by some as `experimental metaphysics' \cite{shimony1984contextual}.  We provide a brief discussion.  The decision to forgo the assumption of locality so to explain the experimental violation of the CHSH inequality is not based on any rigorous evidence.  There is no mathematical or experimental proof to warrant such a decision.  It may very well be the case that our concept of physical realism needs to be radically altered.  From the perspective of this thesis, we find that there is more weight to the argument that one should drop the free will assumption.  Entanglement in time already suggests the eternalists' view that the past and future are as real as the present.  This provides an ideal scaffold to build a argument for the loss of free will, also known in this context as superdeterminism \cite{davies1993ghost, vervoort2013bell}. 

\textbf{c) Multipartite systems: } The definition of Bell non-locality across space has been extended to more than two systems.  Furthermore, it can be shown that all pure entangled $n$-partite states are are non-local \cite{popescu1992generic}.  

Another important point to discuss within multipartite scenarios is what is referred to as the monogomy of entanglement \cite{toner2008monogamy, ringbauer2018multi}.  Let the left hand side of (\ref{CHSHinfoundations}) be denoted
\begin{equation}
S_{CHSH}^{AB} \equiv \langle A_{1}B_{1} \rangle  + \langle A_{2}B_{1} \rangle + \langle A_{2}B_{2} \rangle - \langle A_{1}B_{2} \rangle. 
\end{equation} 
One property of this spatial non-locality is that a violation of the CHSH inequality precludes a simultaneous violation with another spatially separated system.  This is mathematically characterized as
\begin{equation}\label{monogomyCHSH}
S_{CHSH}^{AB} + S_{CHSH}^{BC} \leq 4,
\end{equation}
for systems $A$, $B$, and $C$.  A similar set of inequalities (\ref{monogomyCHSH}) hold for combinations $(AB, AC)$ and $(AC, BC)$.  

\textbf{d) Entropic version: } In \cite{braunstein1988information}, an information-theoretic CHSH inequality was put forth.  This provides a perspective in terms of systems storing information, as opposed to measurement correlations.  They assumed the same scenario as in the original case involving the two spatially separated parties.  Once again Alice has observables $A_{1}$ and $A_{2}$, whereas Bob has observables $B_{1}$ and $B_{2}$.  These have respective values $a_{1}, a_{2}, b_{1}$ and $b_{2}$.  The assumption of local realism (along with free will) is used to establish the existence of the joint probability $p(a_{1}, a_{2}, b_{1}, b_{2})$.  Using the Shannon conditional entropy (\ref{conditionalentropy}), the information-theoretic CHSH inequality can be expressed as
\begin{empheq}[box=\widefbox]{align}\label{entropicCHSH}
H(A_{1}|B_{1}) \leq H(A_{1}|B_{2}) + H(B_{2}|A_{2}) + H(A_{2}|B_{1}).
\end{empheq}
To derive such a quantity, one makes use of the assumption that four objective quantities cannot carry less information than two of them,
\begin{equation}
H(A_{1}, B_{1}) \leq H(A_{1}, B_{2}, A_{2}, B_{1}).
\end{equation}
Nevertheless, certain quantum entangled systems violate (\ref{entropicCHSH}).  An alternative entropic version can be found in \cite{cerf1997entropic}.

\subsection{CHSH game}

\textbf{a) Preliminaries: } We have seen the use of guessing games in articulating the entropic uncertainty relations (\ref{entropicuncertainty}).  More broadly the relationship between quantum theory and game theory is explored in \cite{eisert1999quantum, benjamin2001multiplayer, khan2018quantum}.  Pertinent to this section is that Bell's theorem (CHSH inequality) have also been viewed through the lens of game.  These are commonly referred to as nonlocal games, and the best known example is the CHSH game which we will briefly describe below; in this scenario the participants can win the game at a higher probability with quantum resources, as opposed to having access to only classical resources.  There has also been work on the relationship between Bell's theorem and Bayesian game theory \cite{brunner2013connection, roy2016nonlocal, banik2019two}; in a subset of cases it was shown that quantum resources provide an advantage, and lead to quantum Nash equilibria. In \cite{pappa2015nonlocality}, it was shown that quantum nonlocality can outperform classical strategies in games where participants have conflicting interests. However, in \cite{almeida2010guess}, a nonlocal game was constructed where quantum resources did not offer an advantage. 

\textbf{b) CHSH game: } In this game, we consider spatially separated players Alice and Bob, as well as an outside party known as the referee that plays against Alice and Bob. Based on some probability distribution,  
\begin{equation}
\pi: X \times Y \rightarrow [0,1],
\end{equation}
the referee chooses a question $x \in X$ for Alice and $y \in Y$ for Bob from some set of possible questions $X$ and $Y$.  With respect to the CHSH inequality, these questions can be thought of as labels for measurement settings.  After receiving the questions, Alice and Bob respectively return answers $a\in R_{A}$ and $b \in R_{B}$ from some set of possible answers $R_{A}$ and $R_{B}$.  Relaying this to the CHSH inequality, one can view the answers as measurement outcomes.  The referee is also tasked with deciding whether these answers are the winning answers for the questions that was posed according to the rules of the game.  These rules are expressed through
\begin{equation}
V: R_{A} \times R_{B} \times X \times Y \rightarrow \{0,1\},
\end{equation}    
where $V(a,b,x,y)=1$ if and only if Alice and Bob win against the referee by giving answers $a$ and $b$ for questions $x$ and $y$.  In this game, Alice and Bob have access to both the rules $V$ and the probability distribution $\pi$.  However, the constraint they face is that they cannot communicate once the game starts.  This implies that each player is unaware of what question is given to the other player.

To see the direct relationship to the CHSH inequality (\ref{CHSHinfoundations}), let $X=Y=\{0,1\}$ and $R_{A} = R_{B} = \{0,1\}$.  The rules of the game are such that Alice and Bob win if and only if
\begin{equation}
x \cdot y = a \oplus b,
\end{equation}
where $\oplus$ represents modulo $2$ addition. From this one can compute that the winning probability for a CHSH game is
\begin{equation}
p_{Win}^{CHSH} = \frac{1}{2}\Bigg(1 + \frac{S}{4}\Bigg),
\end{equation}
where $S$ is the CHSH expression (\ref{CHSHinfoundations}).  This provides us with an alternative view of the non-classical features of quantum resources.  The probability that Alice and Bob win using only classical resources is at most probability $0.75$, given $S\leq 2$.  This is in contrast to utilizing quantum resources where Alice and Bob have the ability to win the game at a probability of almost $0.85$ since $S=2\sqrt{2}$.

\subsection{PBR theorem}

\textbf{a) PBR theorem: } From the perspective of this thesis, the PBR theorem demonstrates the discovery of a new quantum non-locality across space.  However, the original intention of the theorem was to answer the foundational question:  What is the nature of the quantum state (or quantum information)?  The answer to this question can be aided by philosophical terminology.  An ontic state refers to a state of reality meaning something that exists objectively in the world independent of an observer; it can be thought of as realism for the system in consideration.  An epistemic state is a state of knowledge and refers to only what an observer currently knows about a physical system.  The PBR theorem answers the question: Is the quantum state ontic or epistemic? 

The mathematical characterization of these concepts is carried out through the framework of ontological models \cite{harrigan2010einstein}.  It can be thought of as a refinement of the hidden variable models found in the literature regarding Bell's theorem \cite{bell2004speakable}.  In the ontological model, when a system is prepared in some quantum state $\ket{\Psi}$, it is really in an ontic state $\lambda$, which describes a state of reality.  The set of ontic states is denoted $\Lambda$.  Due to our ignorance on what ontic state the system is in, the model assigns each quantum state $\ket{\Psi}$ an epistemic state $\mu_{\Psi}$, which is a probability distribution over $\Lambda$.  These satisfy
\begin{equation}\label{PBRmodelstates}
\mu_{\Psi}(\lambda) \geq 0, \quad \text{and} \quad \int \mu_{\Psi}(\lambda) d\lambda = 1.
\end{equation}
It also models a measurement and the outcome of that measurement in terms of the ontic state.  For a measurement $M$ we can denote the probability of obtaining outcome $f$ in the state $\lambda$ as $\xi_{M}(f|\lambda)$.  These satisfy the conditions
\begin{equation}\label{PBRmodelmeasure}
\xi_{M}(f|\lambda) \geq 0, \quad \text{and} \quad \sum_{f}\xi_{M}(f|\lambda) = 1.
\end{equation}   
In order to reproduce the measurement predictions of quantum theory (\ref{Bornruleeigenbasis}), we demand that
\begin{equation}\label{PBRmodelsimulate}
\int_{\Lambda} \xi_{M}(f|\lambda)\,\mu_{\Psi}(\lambda)\, d\lambda = \lvert \braket{f|\Psi}\rvert^{2},
\end{equation}
for all $\ket{\Psi}$ and $f$.  It is important to emphasize that this ontological model includes standard quantum theory as a special case.  Furthermore, note that the assumption of realism is implicit through the existence of an ontic state.   

We now have the required tools to mathematically define what it means for a quantum state to be a state of reality or a state of knowledge.  We say that an ontological model is $\Psi$-epistemic if there exists at least one pair of distinct quantum states $\ket{\Psi_{1}}$ and $\ket{\Psi_{2}}$, such that the corresponding epistemic states $\mu_{\Psi_{1}}$ and $\mu_{\Psi_{2}}$ have a non-zero overlap.  If a model is not $\Psi$-epistemic, then it is $\Psi$-ontic.  

When we say non-zero overlap we mean,
\begin{equation}
1-\delta_{C}(\mu_{\Psi_{1}}, \mu_{\Psi_{2}}) > 0,
\end{equation}  
where the classical trace distance is defined as
\begin{equation}
\delta_{C}(p,q) \equiv \frac{1}{2} \int \lvert p(x) - q(x) \rvert dx,
\end{equation}
for probability distributions $p(x)$ and $q(x)$.  The underlying idea is that if there is no overlap in the epistemic states then distinct quantum states refer to distinct ontic states, thereby warranting the quantum state itself as a state of reality.  However if there is an overlap in the epistemic states, then a single ontic state can relate to two different quantum states through the two respective epistemic states.  Hence, a unique quantum state cannot be associated with the ontic state.  In this case, a quantum state signifies itself merely as a state of knowledge.

The aim of the PBR theorem is to show that models must be $\Psi$-ontic.  The proof for the PBR theorem starts by assuming a $\Psi$-epistemic model and then arriving at a contradiction.  More precisely suppose that for two quantum states $\ket{\Psi_{1}}$ and $\ket{\Psi_{2}}$, the corresponding epistemic states $\mu_{\Psi_{1}}$ and $\mu_{\Psi_{2}}$ overlap.  This implies that there exists an ontic state $\lambda_{*} \in \Lambda$ where
\begin{equation}
\mu_{\Psi_{1}}(\lambda_{*}) > 0, \quad \text{and} \quad \mu_{\Psi_{2}}(\lambda_{*}) > 0. 
\end{equation} 
In this case, even if one had access to the underlying ontic state $\lambda_{*}$, it would be impossible to tell which of the two quantum states was prepared.  Alternatively, regardless of which of these quantum states were prepared, the ontic state $\lambda_{*}$ will be occupied a non-zero fraction $P_{*}>0$ of the time (where the value of $P_{*}$ does not need to be specified).  Next, let two copies of the system be prepared in one of the four quantum (separable or product) states,
\begin{eqnarray}\label{states}
\ket{\Psi_{11}} &= \ket{\Psi_{1}} \otimes \ket{\Psi_{1}},  \quad \ket{\Psi_{12}} = \ket{\Psi_{1}} \otimes \ket{\Psi_{2}},
\nonumber
\\
\ket{\Psi_{21}} &= \ket{\Psi_{2}} \otimes \ket{\Psi_{1}},  \quad \ket{\Psi_{22}} = \ket{\Psi_{2}} \otimes \ket{\Psi_{2}}.
\end{eqnarray}
These two systems are prepared spatially far apart, and the choice to prepare either $\ket{\Psi_{1}}$ or $\ket{\Psi_{2}}$ is made independently at each spatial location. For this task, we make use of the assumption of \textit{preparation independence}.  This comprises of two components.  The first is that each system obtains its own copy of $\Lambda$.  The total state space of the two systems is the product of two copies of $\Lambda$, and therefore the ontic states are written as
\begin{empheq}[box=\widefbox]{align}\label{PI1}
(\lambda_{1} \times \lambda_{2}) \in \Lambda \times \Lambda.
\end{empheq}
This implies that the quantum state $\ket{\Psi_{jk}}$ corresponds to epistemic state $\mu_{\Psi_{jk}}(\lambda_{1},\lambda_{2})$ and that joint measurements take the form $\xi_{M}(\Phi|\lambda_{1}, \lambda_{2})$ for some vector $\ket{\Phi}$ in a measurement basis.  The second component is that the epistemic state $\mu_{\Psi_{jk}}(\lambda_{1}, \lambda_{2})$ associated with quantum state $\ket{\Psi_{jk}}$ factorizes as
\begin{empheq}[box=\widefbox]{align}\label{PI2}
\mu_{\Psi_{jk}}(\lambda_{1}, \lambda_{2})  = \mu_{\Psi_{j}}\,(\lambda_{1})\mu_{\Psi_{k}}(\lambda_{2})
\end{empheq}
where $\mu_{\Psi_{j}}(\lambda_{1})$ is the epistemic state for $\ket{\Psi_{j}}$ and $\mu_{\Psi_{k}}(\lambda_{2})$ for $\ket{\Psi_{k}}$.  Notice the resemblance to (\ref{classicalindependence}) and (\ref{bipartiteseparable}) through its factorization.   

To illustrate the core points, let us first consider the simple case of qubits.  Suppose that $\ket{\Psi_{1}} = \ket{0}$ and $\ket{\Psi_{2}} = \ket{+} \equiv (\ket{0} + \ket{1})/\sqrt{2}$.  If we prepare the two systems in one of the four states $\ket{\Psi_{jk}}$ and use preparation independence, then at least $P_{*}^{2}$ of the time we arrive at the situation that total system will be in ontic state $(\lambda_{*,1}, \lambda_{*,2})$.  In this scenario, both systems are in ontic state $\lambda_{*}$; it will be impossible to decide whether $\ket{\Psi_{1}}$ or $\ket{\Psi_{2}}$ was prepared.  Next we introduce the following two-qubit measurement using basis  
\begin{eqnarray}\label{PBRmeasure}
\ket{\Phi_{11}} &= \frac{1}{\sqrt{2}}(\ket{0} \otimes \ket{1} + \ket{1} \otimes \ket{0}),
\nonumber
\\ 
\ket{\Phi_{12}} &= \frac{1}{\sqrt{2}}(\ket{0} \otimes \ket{-} + \ket{1} \otimes \ket{+}),
\nonumber
\\
\ket{\Phi_{21}} &= \frac{1}{\sqrt{2}}(\ket{+} \otimes \ket{1} + \ket{-} \otimes \ket{0}),
\nonumber
\\
\ket{\Phi_{22}} &= \frac{1}{\sqrt{2}}(\ket{+} \otimes \ket{-} + \ket{-} \otimes \ket{+}),
\end{eqnarray}
where $\ket{-}=(\ket{0}-\ket{1})/\sqrt{2}$.  The four states are antidistinguishable (\ref{antidisting}) using this measurement basis.  In other words, we have that
\begin{equation}\label{antipbr}
\lvert \braket{\Phi_{jk}|\Psi_{jk}}\rvert^{2} = 0
\end{equation} 
for every choice of $j,k = 1,2$.  We see therefore that quantum theory predicts that the measurement outcome $\ket{\Phi_{jk}}$ never occurs when the quantum state $\ket{\Psi_{jk}}$ is prepared.  Referring back to the ontological model we have the certainty that whichever quantum state is prepared, a fraction $P_{*}^{2}$ of the time the system is in ontic state $(\lambda_{*,1}, \lambda_{*,2})$.  We also see that if the system is in this ontic state, we may get outcome $\ket{\Phi_{jk}}$ when we measure in basis (\ref{PBRmeasure}).  Moreoever, this ontic state occurs when the quantum state $\ket{\Psi_{jk}}$ is prepared a non-zero fraction of the time.  However to reproduce the predictions of quantum theory (\ref{antipbr}), the measurement outcome $\ket{\Phi_{jk}}$ should never occur for this ontic state.  Therefore, a non-zero fraction of the time the measurement device contradicts the predictions of quantum theory.  This provides us with the desired contradiction.  This argument can be generalized using similar concepts (as described below).  Therefore, the PBR theorem can be stated as: Any ontological model that reproduces quantum predictions and satisfies preparation independence is $\Psi$-ontic. 

\textbf{b) Implications: } An equally weighted perspective is that at least one of the assumptions used to arrive at the contradiction must be false.  This allows for the position (held by most physicists) that the quantum state is simply a state of knowledge ($\Psi$-epistemic); such a view is desirable in dissolving away many conundrums including the measurement collapse which is only a problem if the quantum state has a physical existence.  To decipher which assumption must be incorrect, we relay our thoughts back to the CHSH inequality.  It was of consensus in that scenario to maintain realism and adopt non-locality.  In the PBR case, an analogous choice is to therefore adopt the \textit{violation of preparation independence}.  This resulting effect can be described as a new type of non-locality across space, and mathematically characterized as a violation of (\ref{PI1}) and/or (\ref{PI2}).

The violation of preparation independence is a far more perplexing spatial interdependence than the Bell non-locality across space.  First, it applies to product states, and hence does not require entanglement (unlike Bell non-locality).  The second point to note is that preparation independence is perhaps the most natural assumption to make in that two spatially separated systems should possess their own separate states of reality; such a notion of separability should be natural for product states.  As an example, if one person prepares a system in one part of the universe and another person prepares a system in the other part of the universe, then there should be no correlations between these preparations; if this was not the case as insisted above, then an extrapolation on this effect is that one requires every system in the universe in order to determine all the parameters that are of relevance for a system prepared on Earth.  A non-locality of such magnitude would lead to a radical destruction of basic assumptions.

Einstein wrote \cite{einstein1948quanten, howard1985einstein} about the dangers of abandoning such assumptions (but within another context), and we quote ``\textit{Further, it appears to be essential for this arrangement of the things introduced in physics that, at a specific time, these things claim an existence independent of one another, insofar as these things `lie in different parts of space.'  Without such an assumption of the mutually independent existence of spatially distant things, an assumption which originates in everyday thought, physical thought in the sense familiar to us would not be possible." }  From the perspective of the theme of this thesis, notice that Einstein mentions the words \textit{specific time} which signifies the \textit{lack of a time interval}. Whether we accept the quantum state as a state of reality, or instead give up realism, or introduce the violation of preparation independence, the PBR theorem has most certainly emphasized the large gap in our fundamental understanding of quantum physics.

\textbf{c) Multipartite systems: }  Using certain assumptions, we have shown that the epistemic states for $\ket{0}$ and $\ket{+}$ cannot overlap.  Generalizing this to any pair of quantum states implies that a quantum state can uniquely correspond to an ontic state, thereby signifying itself as a physical property of the system.  To prove this, we can let 
\begin{align}
\ket{\Psi_{0}} &= \cos\Big(\frac{\theta}{2}\Big)\ket{0} + \sin\Big(\frac{\theta}{2}\Big)\ket{1}, \\
\ket{\Psi_{1}} &= \cos\Big(\frac{\theta}{2}\Big)\ket{0} - \sin\Big(\frac{\theta}{2}\Big)\ket{1},
\end{align}
represent arbitrary non-orthogonal qubits, where $0 < \theta < \pi/2$.  As in the previous case, suppose there is a non-zero probability of at least $P_{*}$ that the ontic state of the system is compatible with either preparation.  This means the corresponding epistemic states overlap.  Suppose we prepare $n$ of these systems independently. The total system can be described by one of the quantum states,
\begin{equation}
\ket{\Psi(x_{1}, \dots x_{n})} = \ket{\Psi_{x_{1}}} \otimes \cdots \otimes \ket{\Psi_{x_{n-1}}} \otimes \ket{\Psi_{x_{n}}},
\end{equation}
where $x_{i} \in \{0,1\}$ for each $i$.  Assuming preparation independence, we have the probability that at least $P_{*}^{n}$ that the ontic state is compatible with any one of the quantum states.  The contradiction to quantum theory is obtained if we can derive a measurement basis that makes these quantum states antidistinguishable.  This can be achieved if the number of systems $n$ satisfies
\begin{equation}
2 \, \text{arctan} \, (2^{1/n} - 1) \leq \theta.
\end{equation}
Furthermore, the exact measurement circuit consists of a unitary evolution,  
\begin{equation}
U_{\alpha,\beta} = H^{\otimes n} R_\alpha {Z_\beta}^{\otimes n},
\end{equation}
where 
\begin{equation}
Z_\beta = \begin{pmatrix}1 & 0 \\ 0 & e^{i\beta}\end{pmatrix}.
\end{equation}
and where $H$ is Hadamard gate (\ref{hadamard}); the gate $R_{\alpha}$ operates as $R_{\alpha}\ket{0\cdots 0} = e^{i\alpha}\ket{0 \cdots 0}$ and acts as an identity operator on the other computational basis states.  To achieve the desired result the unitary evolution is chosen based on certain values of $\alpha$ and $\beta$.  This is followed by a measurement of each qubit in the computational basis states (\ref{compbasis}).  The result is that each outcome has zero probability given one of the $2^{n}$ possible preparations.  

More precisely, the probability of obtaining the basis state $\ket{x_{1}, \dots x_{n}}$ given that the state $\ket{\Psi(x_{1}, \dots x_{n})}$ is prepared is the squared absolute value of 
\begin{equation}
\bra{x_1 \ldots x_n} H^{\otimes n}R_\alpha Z_\beta^{\otimes n}\ket{\Psi_{x_1}}\otimes \cdots\otimes \ket{\Psi_{x_n}}
\end{equation}
which can be shown to equate to
\begin{equation}
\frac{1}{\sqrt{2^n}}\left(\cos\frac\theta2\right)^n \left(e^{i\alpha} + \left(1+e^{i\beta}\tan\frac\theta2\right)^n - 1\right).
\end{equation}
Moreover, $\alpha$ and $\beta$ can be derived so that
\begin{equation}
e^{i\alpha} + \left(1+e^{i\beta}\tan\frac\theta2\right)^n - 1 = 0.
\end{equation}
Hence, the required quantum probabilities are zero, and as a result we found a measurement that provides antidistinguishability for these quantum states.

\textbf{d) Entropic version: } From the perspective of quantum information science, it also interesting to note that the PBR theorem been interpreted through the language of classical and quantum communication protocols \cite{montina2012epistemic, montina2015communication}.  This program crucially involved the use of the Shannon entropy (\ref{Shannon}).  In a related work antidistinguishability, which is a core concept in the PBR theorem, was used to provide an advantage in a two-player communication task \cite{havlivcek2019simple}.  

Other notable developments on the PBR theorem and $\Psi$-epistemic models have been carried out in \cite{lewis2012distinct, schlosshauer2012implications, aaronson2013psi, patra2013no, schlosshauer2014no, mansfield2016reality} including on the issue of quantum indistinguishability \cite{leifer2014psi, barrett2014no, branciard2014psi}.

\subsection{PBR theorem as a Monty Hall game}

\textbf{a) Preliminaries: } Analogous to the game formulation of CHSH inequality, a desirable construction is to view the PBR theorem through the lens of a game.  One instantiation of this is in an exclusion game where the participant's goal is to produce a particular bit string \cite{perry2015communication, bandyopadhyay2014conclusive}; this has been shown to be related to the task of quantum bet hedging \cite{arunachalam2013quantum}.  Furthermore, concepts involved in the PBR proof have been used for a particular guessing game \cite{myrvold2018psi}.  

In this subsection, we reformulate the PBR theorem into a Monty Hall game \cite{rajan2019quantum}, which is \textit{part of the original component of this thesis} (which as done in collaboration with my supervisor). This particular gamification of the theorem highlights that winning probabilities, for switching doors in the game, depend on whether it is a $\Psi$-ontic or $\Psi$-epistemic game; we also show that in certain $\Psi$-epistemic games switching doors provides no advantage.  This may have consequences for an alternative experimental test of the PBR theorem

\textbf{b) PBR elements: } We extract certain parts of the PBR proof.  Recall that  two quantum systems are prepared independently, and each system is prepared in either state $\ket{0}$ or state $\ket{+}$.  This means that the total system is in one of the four possible non-orthogonal quantum states (\ref{states}) which we rewrite as:
\begin{eqnarray}\label{PBRgamestates}
\ket{\Psi_{1}} &= \ket{0} \otimes \ket{0},  \quad \ket{\Psi_{2}} = \ket{0} \otimes \ket{+},
\nonumber
\\
\ket{\Psi_{3}} &= \ket{+} \otimes \ket{0},  \quad \ket{\Psi_{4}} = \ket{+} \otimes \ket{+}.
\end{eqnarray}
The total system is brought together and measured in the basis (\ref{PBRmeasure}) which we re-label as:
\begin{eqnarray}\label{PBRgamemeasures}
\ket{\Phi_{1}} &= \frac{1}{\sqrt{2}}(\ket{0} \otimes \ket{1} + \ket{1} \otimes \ket{0}),
\nonumber
\\ 
\ket{\Phi_{2}} &= \frac{1}{\sqrt{2}}(\ket{0} \otimes \ket{-} + \ket{1} \otimes \ket{+}),
\nonumber
\\
\ket{\Phi_{3}} &= \frac{1}{\sqrt{2}}(\ket{+} \otimes \ket{1} + \ket{-} \otimes \ket{0}),
\nonumber
\\
\ket{\Phi_{4}} &= \frac{1}{\sqrt{2}}(\ket{+} \otimes \ket{-} + \ket{-} \otimes \ket{+}),
\end{eqnarray}
Invoking the Born probabilities, $\lvert \braket{\Phi_{i}|\Psi_{h}}\rvert^{2}$, where $i,h=1,2,3,4$, we found that for $i=h$, $\lvert \braket{\Phi_{i}|\Psi_{i}}\rvert^{2} = 0$. This means that for any value $i$, the outcome $\ket{\Phi_{i}}$ never occurs when the system is prepared in quantum state $\ket{\Psi_{i}}$.  The PBR proof showed that in $\Psi$-epistemic models there is a non-zero probability $q$ (whose value does not need to be specified) that outcome $\ket{\Phi_{i}}$ occurs when state $\ket{\Psi_{i}}$ is prepared, thereby contradicting the predictions of quantum theory; hence one can infer that the quantum state corresponds to a $\Psi$-ontic model.

\textbf{c) $\Psi$-ontic Monty Hall game: } Antidistinguishability, where there is a measurement for which each outcome identifies that a specific member of a set of quantum states was definitely not prepared, is highlighted in the PBR proof by $\lvert\braket{\Phi_{i}|\Psi_{i}}\rvert^{2} = 0$ for all $i$. We will exploit this to construct our game, which can be thought of as a quantum Ignorant Monty Hall game (\ref{IgnorantMontyHallgame}).

For state $\ket{\Psi_{1}}$ in (\ref{PBRgamestates}), we have
\begin{eqnarray}\label{Born}
\lvert\braket{\Phi_{1}|\Psi_{1}}\rvert^{2} &= 0,  \quad \lvert\braket{\Phi_{2}|\Psi_{1}}\rvert^{2} &= 1/4,
\nonumber
\\
\lvert\braket{\Phi_{3}|\Psi_{1}}\rvert^{2} &= 1/4,  \quad \lvert\braket{\Phi_{4}|\Psi_{1}}\rvert^{2} &= 1/2.
\end{eqnarray}
For the other states in (\ref{PBRgamestates}), the same probability distribution ($0$, $1/4$, $1/4$, $1/2$) occur but across the different outcomes (\ref{PBRgamemeasures}); hence we will focus our game on $\ket{\Psi_{1}}$, but similar constructions hold for the other states.  

The Monty Hall gamification is as follows:  There are four doors labelled $\{1,2,3,4\}$, and these correspond to the different measurement outcomes listed in (\ref{PBRgamemeasures}).  The prize door $A_{i}$, where $i$ takes one of the door labels, is the outcome $\ket{\Phi_{i}}$ that the state $\ket{\Psi_{1}}$ collapses to upon measurement.  For a $\Psi$-ontic game, through the Born probabilities (\ref{Born}), we have $P(A_{i}) = \lvert\braket{\Phi_{i}|\Psi_{1}}\rvert^{2}$.  

The contestant on the show does not know what state from (\ref{PBRgamestates}) is used, and is only aware of the possible measurement outcomes (\ref{PBRgamemeasures}).  Based on this limited information, the contestant randomly picks one of the doors which we denote $B_{j}$ where $j$ is the corresponding door label; hence we have $P(B_{j}|A_{i})= 1/4$, for all values $i,j$.

Monty's decision corresponds to the predictions of quantum theory.  He is aware that state  $\ket{\Psi_{1}}$ was used, and has access to the Born probabilities (\ref{Born}).  The door opened by Monty is denoted $C_{k}$ where $k$ is one of the door labels.  The main insight to construct this game is that when Monty opens a goat door, he is opening a door that has probability zero of having a prize in it.  And for our game, a door that definitely does not have a prize in it corresponds to outcome $\ket{\Phi_{1}}$ as $P(A_{1}) = \lvert\braket{\Phi_{1}|\Psi_{1}}\rvert^{2} = 0$.  Hence in this game, Monty will open door $C_{1}$ unless the contestant has already chosen this door as their pick (as Monty cannot open the door chosen by the contestant); in that case Monty will open one of the other remaining doors with equal probability, and there is a chance he may open up the prize door as in the Ignorant Monty Hall game.  From these factors, one can compute,
\begin{equation}\label{PBRMonty}
P(C_{k} \, | \, B_{j} \cap A_{i}) = 
\begin{cases}
\frac{1}{3},& \text{if } j =1 \text{ and } k=2,3, 4,\\
1,              & \text{if } j \neq 1 \text{ and } k=1, \\
0, & \text{otherwise,}
\end{cases}
\end{equation}
where we have adopted the notation for joint probabilities as $P(A, B) \equiv P(A \cap B)$.  The probability that Monty opens the prize door is
\begin{equation}
P(\text{opens prize door}) = \sum_{i = k \neq j} P(A_{i} \cap B_{j} \cap C_{k}) = \frac{1}{12}.
\end{equation}  
This implies that the probability that he opens a goat door is $11/12$.  Monty then offers the option to stick or switch.  Suppose the contestant always sticks with the initial choice.  Then the probability of winning if sticking and Monty opening a goat door is 
\begin{equation}
\sum_{i = j \neq k} P(A_{i} \cap B_{j} \cap C_{k}) = \frac{1}{4}.
\end{equation} 
With that, we can compute the conditional probability
\begin{equation}
P(\text{win if stick} \, | \, \text{opens goat door}) = \frac{1/4}{11/12} = \frac{3}{11}.
\end{equation}
Suppose the contestant decides to always switch to one of the other two unopened doors with equal probability $1/2$.  Let $\ket{\Phi_{l}}$ be the outcome switched to and let $D_{l}$ be the corresponding door.  With that, we can compute $P(A_{i} \cap B_{j} \cap C_{k} \cap D_{j}) =  P(D_{l} \, | C_{k} \cap B_{j} \cap A_{i})P(C_{k} \, | \, B_{j} \cap A_{i})P(B_{j}|A_{i})P(A_{i})$.  Hence, the probability of winning if switching and Monty opening a goat door is       
\begin{equation}
\sum_{i = l \neq j \neq k} P(A_{i} \cap B_{j} \cap C_{k} \cap D_{j}) = \frac{1}{3}.
\end{equation}
From that, one can calculate
\begin{equation}
P(\text{win if switch} \, | \, \text{opens goat door}) = \frac{1/3}{11/12} = \frac{4}{11}.
\end{equation}
Hence, we see in a $\Psi$-ontic game, switching provides an advantage.

\textbf{d) $\Psi$-epistemic Monty Hall game: } In the PBR proof, for the $\Psi$-epistemic model, there is a non-zero probability $q$ that outcome $\ket{\Phi_{1}}$ occurs when state $\ket{\Psi_{1}}$ is prepared.  This implies that in a $\psi$-epistemic game, $P(A_{1}) = q \neq 0$.  To allow for a comparison with the $\Psi$-ontic game, let $q = q_{1} + q_{2} + q_{3}$, and with that let the other prize door probabilities take values $P(A_{2}) = (1/4) - q_{1}$, $P(A_{3}) = (1/4) - q_{2}$ and $P(A_{4}) = (1/2) - q_{3}$.  

As in the $\psi$-ontic game, $P(B_{j}|A_{i})= 1/4$, for all values $i,j$.  Monty as a character corresponds to the predictions of quantum theory (\ref{Born}); he will assume $C_{1}$ is definitely a goat door since $\lvert\braket{\Phi_{1}|\Psi_{1}}\rvert^{2} = 0$.  This means the probabilities in  (\ref{PBRMonty}) apply in this game as well.  Hence, the probability that Monty opens the prize door
\begin{equation}
P(\text{opens prize door}) = \sum_{i = k \neq j} P(A_{i} \cap B_{j} \cap C_{k}) = \frac{1}{12} + \frac{2q}{3}.
\end{equation}  
This implies that the probability that Monty opens a goat door is $(11/12) - (2q/3)$.  The probability of winning if always sticking and that Monty opens a goat door is 
\begin{equation}
\sum_{i = j \neq k} P(A_{i} \cap B_{j} \cap C_{k}) = \frac{1}{4}.
\end{equation} 
From this we compute
\begin{equation}
P(\text{win if stick} \, | \, \text{opens goat door}) = \frac{3}{11-8q}.
\end{equation}
If a switching strategy is adopted then:      
\begin{eqnarray}
\sum_{i = l \neq j \neq k} P(A_{i} \cap B_{j} \cap C_{k} \cap D_{j}) = \frac{1}{3} - \frac{q}{3}, \\
P(\text{win if switch} \, | \, \text{opens goat door})  = \frac{4-4q}{11-8q}.
\end{eqnarray}
Thus the  probabilities depend on whether the game is a $\Psi$-ontic or $\Psi$-epistemic game.  For value $q=1/4$, we can calculate that $P(\text{win if switch} \, | \, \text{opens goat door}) = P(\text{win if stick} \, | \, \text{opens goat door})$; hence for certain $\Psi$-epistemic games, switching offers no advantage.

\textbf{e) Experimental implications: } Comparing a $\Psi$-ontic game to a $\Psi$-epistemic game, Monty opens the prize door less often.  This corresponds to certain probabilities in the PBR proof being zero; some work on the experimental tests \cite{pusey2012reality, nigg2015can, miller2013alternative, ringbauer2015measurements, liao2016experimental} of PBR discuss this exact zero probability as an experimental difficulty. Through our game, we provide another viewpoint;  the difference in the probabilities of winning conditioned that a goat door is opened are simply different for the two physical scenarios. This may provide insights to alternative experimental designs to test PBR.

\section{Non-locality across Time}

Our aim is to explore how the effects of non-locality across space extend into the temporal regime.  In particular, one can qualitatively describe a non-locality across time as a characteristic where an action on a subsystem can instantaneously influence the same subsystem at a later or earlier time!  We have seen this shocking property portrayed in the section regarding entanglement in time.  However in this chapter, our focus be on the case of a single system across multiple times.  This non-locality across time will be expressed mathematically through the violation of a temporal version of the Bell-CHSH inequality.  The most prominent of these are known as Leggett-Garg (LG) inequalities \cite{leggett1985quantum}.  

Some refer to the effects, that we shall describe, as an `entanglement in time.' However in this thesis we thread carefully and refer to these effects exclusively as a non-locality across time.  There are three reasons for this taxonomy and they stem from the extensive review of the spatial case.  The first reason is that the LG and related temporal inequalities are about measurement outcomes, and not about an algebraic property within the mathematics of quantum theory; in the spatial case, measurement outcomes related directly to non-locality, whereas the algebraic property defined the concept of entanglement; setting the spatial case as precedence allows us to forgo using the words entanglement in time to describe these scenarios of (temporal) measurement correlations.  The second and perhaps the more cautious reason is that not all spatially entangled states are spatially Bell non-local; hence assuming the nature of the relationship between an entanglement in time and non-locality across time prior to rigorous results is not very prudent.  The third reason is that the necessary algebraic construction to help define an entanglement in time between a single system over various times is met with considerable technical problems \cite{horsman2017can}.  It must therefore be emphasized that there is a large degree of unknown aspects to this area.  However, for an extensive review on LG inequalities, refer to \cite{emary2013leggett}.

In the spatial case, we reviewed Bell-CHSH, PBR and its games.  In this section we will articulate non-locality across time using the concepts in Bell-CHSH, PBR, and games.  This serves to provide a systematic view into the subject.

\subsection{Through Bell-CHSH concepts}

\textbf{a) LG inequalities: } The LG inequalities \cite{leggett1985quantum} can be thought of as a temporal version of the Bell-CHSH inequalities (\ref{CHSHinfoundations}).  It was derived within the context of macroscopic coherence which can be thought of as property of an object, consisting of many quantum particles, existing in superpositions of macroscopically distinct states. (A fictional example is the Schr\"{o}dinger's cat.)  The result largely follows the same style of derivation that was used in the spatial Bell-CHSH case.  We start with the following classically intuitive assumptions:
\begin{enumerate}[noitemsep, topsep=0pt, label=\roman*)]
	\item (A1) Macroscopic Realism (MR):  a macroscopic system with two or more macroscopically distinct states available to it will at all times be in one or the other of these states.
	\item (A2) Non-Invasive Measurability (NIM) at the macroscopic level:  it is possible, in principle, to determine the state of the system with arbitrarily small perturbation on its subsequent dynamics.
	\item (A3) Induction: the outcome of a measurement on the system cannot be affected by what will or will not be measured on it later. 
\end{enumerate}
NIM has also been described as that a measurement of an observable at any instant of time does not influence its subsequent evolution \cite{devi2013macrorealism}.  NIM has also been described as nondisturbance in that a measurement can be performed such that it does not influence the outcome of a measurement on the same system at a later time \cite{fedrizzi2011hardy}.  Another temporal Bell-CHSH inequality \cite{taylor2004entanglement} was more direct to replace NIM with the assumption that the results of measurements performed at some time is independent of any other measurement at another time; they referred to this as a locality in time.  The assumptions used in the LG inequalities are still of great debate.  In this thesis, we view both (A2) and (A3) as the assumption of locality in time, with (A1) taking the role of realism.  Notice the resemblance with the Bell-CHSH case where it was a locality in space paired with realism.

Using these assumptions, we can define a macroscopic dichotomic variable $Q=\pm 1$ for a system.  We aim to measure its two-time correlation function
\begin{equation}
C_{ij} = \langle Q(t_{i})\, Q(t_{j}) \rangle.
\end{equation}
This quantity is computed from the joint probability $P_{ij}(Q_{i}, Q_{j})$ of obtaining $Q_{i} = Q(t_{i})$ and $Q_{j} = Q(t_{j})$ from measurement times $t_{i}$, $t_{j}$ as
\begin{equation}\label{LGcorrelation}
C_{ij} = \sum_{Q_{i}, Q_{j} = \pm 1} Q_{i}\, Q_{j}\, P_{ij}(Q_{i}, Q_{j}).
\end{equation}
The assumption (A1) implies the observable has a defined value at all times regardless of whether it is measured.  Hence one can obtain a two-time probability as the marginal of a three-time distribution as follows
\begin{equation}
P_{ij}(Q_{i}, Q_{j}) = \sum_{Q_{k};k\neq i, j}P_{ij}(Q_{3}, Q_{2}, Q_{1}).
\end{equation}
Using (A2) and (A3), we find that the three probabilities $P_{21}(Q_{3}, Q_{2}, Q_{1})$, $P_{32}(Q_{3}, Q_{2}, Q_{1})$ and $P_{31}(Q_{3}, Q_{2}, Q_{1})$ become the same.  Hence we can write this simply as
\begin{equation}
P(Q_{3}, Q_{2}, Q_{1}) = P_{ij}(Q_{3}, Q_{2}, Q_{1}).
\end{equation}
One can proceed to use this single probability to compute the following correlation functions:
\begin{align}
C_{21} &= P(+1, +1, +1) - P(+1, +1, -1) - P(-1, -1, +1) + P(-1, -1, -1) \\
&- P(+1, -1, +1) + P(+1, -1, -1) + P(-1, +1, +1) - P(-1, +1, -1); 
\nonumber
\end{align} 
\begin{align}
C_{32} &= P(+1, +1, +1) + P(+1, +1, -1) + P(-1, -1, +1) + P(-1, -1, -1) \\
&- P(+1, -1, +1) - P(+1, -1, -1) - P(-1, +1, +1) - P(-1, +1, -1); 
\nonumber
\end{align} 
\begin{align}
C_{31} &= P(+1, +1, +1) - P(+1, +1, -1) - P(-1, -1, +1) + P(-1, -1, -1) \\
&+ P(+1, -1, +1) - P(+1, -1, -1) - P(-1, +1, +1) + P(-1, +1, -1); 
\nonumber
\end{align}
Given that
\begin{equation}
\sum_{Q_{3}, Q_{2}, Q_{1}} P(Q_{3}, Q_{2}, Q_{1}) \equiv 1,
\end{equation} 
this implies
\begin{align}\label{LGK3}
K_{3} &\equiv C_{21} + C_{32} - C_{31} \\ 
&= 1 - 4 (P(+1, -1, +1) + P(-, +, -)). \nonumber
\end{align} 
If $P(+1, -1, +1) = P(-, +, -) = 0$, then $K_{3} = 1$ which is the upper bound.  On the other hand, the choice $P(+1, -1, +1) + P(-, +, -) = 1$ gives lower bound $K_{3} \geq -3$.  From this, we obtain the simplest LG inequality,
\begin{empheq}[box=\widefbox]{align}\label{LGsimplest}
-3 \leq K_{3} \leq 1.
\end{empheq}  
This LG inequality has been violated through various quantum systems.  It can be shown that the maximum violation by a two-level quantum system (qubit) is $K_{3}^{max} = 3/2$; at least one of the three assumptions made to derive the LG inequality is wrong.

\textbf{b) Implications: }  If one takes the spatial Bell-CHSH case as an analogy but also as precedence, then we leave (A1) intact.  This means that assumptions (A2) and (A3), which express locality in time, are incorrect.  Hence the mathematical characterization of \textit{non-locality across time} is expressed as the violation of (\ref{LGsimplest}).

\textbf{c) Multi-measurements: } The LG inequality has been extended to $n$-meaurements.  Let us denote the variable,
\begin{equation}
K_{n} = C_{21} + C_{32} + C_{43} + \cdots +C_{n(n-1)} - C_{n1}.
\end{equation}
Using the assumptions (A1-3), one can obtain the following LG inequalities
\begin{empheq}[box=\widefbox]{align}\label{LGmultiple}
-n \leq \, &K_{n} \leq n-2 \quad \quad  n \geq 3, \, \text{odd}; \\
-(n-2) \leq \, &K_{n} \leq n-2 \quad \quad  n \geq 4, \, \text{even},
\end{empheq} 
where the only requirement on the variable is to be bounded $\lvert Q \rvert \leq 1$.  Using various symmetry properties one derive further inequalities.  One in particular is written as
\begin{empheq}[box=\widefbox]{align}\label{LGfour}
-2 \leq C_{21} + C_{32} + C_{43} - C_{41} \leq 2.
\end{empheq} 
Note that (\ref{LGfour}) and other LG inequalities describe a situation where there are a set of measurements on the same operator at $n \geq 3$ different times.  There is in fact another temporal Bell-CHSH inequality \cite{taylor2004entanglement} that considers a different physical scenario.  In this case, there are two different times but different operator choices at each time.  More precisely, in this scenario Alice measures at time $t_{1}$ while Bob measures at time $t_{2} > t_{1}$.  These measurements involve dichotomic variables.  Each of them have two measurement choice $(i=1,2)$, which can be denoted $A_{i}$ and $B_{i}$ for Alice and Bob respectively. Using assumptions (A1-3) where the locality of time assumption was explicity stated, one can derive the following temporal CHSH inequality,  
\begin{empheq}[box=\widefbox]{align}\label{temporalCHSH}
\lvert \langle B_{1}A_{1} \rangle + \langle B_{1}A_{2} \rangle + \langle B_{2}A_{1} \rangle - \langle B_{2}A_{2} \rangle \rvert \leq 2.   
\end{empheq} 
Despite the physical differences, the equation (\ref{temporalCHSH}) can be obtained directly from the LG inequality (\ref{LGfour}) by setting 
\begin{equation}
Q(t_{1}) = B_{2}, \quad Q(t_{2}) = A_{1}, \quad Q(t_{3}) = B_{1}, \quad Q(t_{4}) = A_{2}.
\end{equation}
Once again the violation of (\ref{temporalCHSH}) provides a mathematical characterization of a \textit{non-locality across time.}  A qubit can be shown to violate (\ref{temporalCHSH}) with a maximum value of $2\sqrt{2}$.  Notice the resemblance of temporal CHSH case (\ref{temporalCHSH}) to the spatial CHSH case (\ref{CHSHinfoundations}).

\textbf{d) Entropic version: } We have seen an entropic version (\ref{entropicCHSH}) of the spatial CHSH inequality.  A natural question to consider is whether such a possibility exists for the temporal LG case.  Such a curiosity has been answered in the affirmative in the works by \cite{devi2013macrorealism, morikoshi2006information}. We provide a derivation of this entropic LG inequalities which utilizes the Shannon entropy (\ref{Shannon}).

In the LG scenario, we have a macroscopic system where $Q(t_{i})$ represents an observable at time $t_{i}$.  Let the outcome be denoted $q_{i}$ with corresponding probability $P(q_{i})$.  Using assumption (A1), we have the existence of a joint probability distribution $P(q_1,q_2,\ldots )$ due to the notion that the outcomes of observable at all instants of time exist whether the system has been measured or not.  Using (A2) and (A3), we have the result that measurement at an earlier time $t_{i}$ has no influence on the value at a subsequent time $t_{j} > t_{i}$; this implies that joint probabilities  are written as convex combinations involving a hidden variable probability distribution $\rho(\lambda)$,        
\begin{equation} 
P(q_1,q_2,\dots , q_n) = \sum_{\lambda}\, \rho(\lambda)\, P(q_1\vert \lambda)P(q_2\vert \lambda)\dots P(q_n\vert \lambda)
\end{equation}
where
\begin{equation}
 0\leq \rho(\lambda)\leq 1, \quad  \sum_\lambda \rho_\lambda =1,
\end{equation}
and 
\begin{equation}
  0\leq P(q_i\vert \lambda)\leq 1 \quad \sum_{q_i}P(q_i\vert \lambda)=1.
\end{equation}
One can harness the joint Shannon entropy (\ref{(joint)}) to an observable at two different times $t_{k}$ and $t_{k+l}$, resulting in
\begin{equation}
H(Q_k,Q_{k+l})=-\sum_{q_k,q_{k+l}}\, P(q_k,q_{k+l})\, \log_2\, P(q_k,q_{k+l}).
\end{equation}
Using the conditional Shannon entropy (\ref{conditionalentropy}), we can examine the information held by observable $Q_{k+l}$ at time $t_{k+l}$ given it had the values $Q_{k}=q_{k}$ at a previous time $t_{k}$.  This can be shown to equate to
\begin{equation}
H(Q_{k+l}\vert Q_k=q_k)=-\sum_{q_{k+l}}\, P(q_{k+l}\vert q_{k})\, \log_2\, P(q_{k+l}\vert q_{k}),
\end{equation}
where the conditional probability is expressed as 
\begin{equation}
P(q_{k+l}\vert q_{k})=\frac{P(q_{k}, q_{k+l})}{P(q_k)}. 
\end{equation}
From this, one can easily derive the full conditional Shannon entropy,
\begin{align}
H(Q_{k+l}\vert Q_{k})&= \sum_{q_k}\, P(q_k)\,  H(Q_{k+l}\vert Q_k=q_k) \\
&= H(Q_k,Q_{k+l})-H(Q_k).
\end{align}
Re-arranging this, we obtain
\begin{equation}\label{useinentropicderivation}
H(Q_k,Q_{k+l}) = H(Q_{k+l}\vert Q_{k}) + H(Q_k)
\end{equation}
One further set of inequalities that will be of use is given by the properties intrinsic to the Shannon entropy
\begin{equation}\label{useinentropicderivation2}
H(Q_{k+l}\vert Q_{k})\leq H(Q_{k+l})\leq H(Q_k,Q_{k+l}),
\end{equation}
where the right-hand inequality signifies that two variables can never hold less information than held by one of them.  By combining (\ref{useinentropicderivation}) and (\ref{useinentropicderivation2}) and extending it to three variables, we obtain
\begin{align}
H(Q_k, Q_{k+m})&\leq& H(Q_k, Q_{k+l}, Q_{k+m})=H(Q_{k+m}\vert Q_{k+l}, Q_{k})+ H(Q_{k+l}\vert Q_{k})+H(Q_{k}).
\end{align}
This results in the entropic LG relation
\begin{empheq}[box=\widefbox]{align}\label{entropicLG}
H(Q_{k+m}\vert Q_k)\leq  H(Q_{k+m}\vert Q_{k+l})+ H(Q_{k+l}\vert Q_{k}), 
\end{empheq} 
for times $t_{k} < t_{k+l} < t_{k+m}$.  A similar line of argument allows one to obtain an $n$-measurement entropic LG inequality,
\begin{empheq}[box=\widefbox]{align}\label{entropicLGmulti}
H(Q_n\vert Q_1)\leq H(Q_n\vert Q_{n-1})+H(Q_{n-1}\vert Q_{n-2})+\ldots +H(Q_{2}\vert Q_{1}),
\end{empheq} 
for consecutive measurements  $Q_1, Q_2, \ldots, Q_n$ for the various times  $t_1<t_2<\ldots <t_n$.  

Once again a violation of (\ref{entropicLG}) or (\ref{entropicLGmulti}) is a mathematical characterization of \textit{non-locality across time}.  Such a violation has been exhibited by quantum systems.  Of great interest to this thesis is that \textit{the interdependence of this non-locality across time is shocking due to the existence of a time interval}.  To elaborate, let us consider the equation (\ref{entropicLG}).  It suggests the information content of the observable at three different times $t_{k} < t_{k+l} < t_{k+m}$ can never be smaller than the information content at two time instants.  A quantum violation suggests the perplexing narrative that in fact having the knowledge of an observable at three different times corresponds to less information than knowing the observable at only two different times!  From the view of this thesis, the added time interval to introduce the third time point makes this interdependence across time truly shocking.

\subsection{Through PBR concepts}

\textbf{a) LG inequalities: } One can provide an alternative derivation of the LG inequalities (\ref{LGsimplest}) using ontological models \cite{harrigan2010einstein, leifer2014quantum, emary2013leggett}.  Recall the use of this framework in proving the PBR theorem.  In this section, our aim is use these models to compute the correlation functions
\begin{equation}\label{PBRLGcorrelation}
C_{ij} = \sum_{Q_{i}, Q_{j} = \pm 1} Q_{i}\, Q_{j}\, P_{ij}(Q_{i}, Q_{j}),
\end{equation}
and thereby re-derive the LG inequality.  We start by describing the state of the system as outlined in (\ref{PBRmodelstates}).  This is denoted by an epistemic state $\mu(\lambda)$ over the set of ontic states $\lambda$.  Note that the ontic states capture assumption (A1). Next, a measurement (\ref{PBRmodelmeasure}) at time $t_{i}$ is represented as 
\begin{equation}
\xi_{i}(Q_{i}|\lambda),
\end{equation}
which signifies the probability that of outcome $Q_{i}$ given ontic state $\lambda$.  We denote the probability of disturbance by the measurement on the ontic state $\lambda \rightarrow \lambda'$ as
\begin{equation}
\gamma_{i}(\lambda'|Q_{i}, \lambda).
\end{equation}
In the ontological framework, the joint probability function of two measurements is then written as
\begin{equation}\label{LGontologicaltwo}
P(Q_{i}, Q_{j}) = \int d\lambda' \, d\lambda \, \xi_{j}(Q_{j}|\lambda') \,\gamma_{i}(\lambda'|Q_{i}, \lambda) \, \xi_{i}(Q_{i}|\lambda) \, \mu(\lambda).
\end{equation}
Using (A2) and (A3), we have the condition that the disturbance does not affect the ontic state.  This can be expressed generally as
\begin{equation}
\gamma_{M}(\lambda'|Q, \lambda) = \delta(\lambda' - \lambda).
\end{equation}
Hence (\ref{LGontologicaltwo}) equates to
\begin{equation}\label{LGontologicalthree}
P(Q_{i}, Q_{j}) = \int  d\lambda \, \xi_{j}(Q_{j}|\lambda) \, \xi_{i}(Q_{i}|\lambda) \, \mu(\lambda).
\end{equation} 
By substituting (\ref{LGontologicalthree}) into (\ref{PBRLGcorrelation}), we get
\begin{align}
\langle Q_{i}\, Q_{j} \rangle &= \int d\lambda \sum_{Q_{i}, Q_{j} = \pm 1} Q_{i} \, Q_{j} \, \xi_{j}(Q_{j}|\lambda) \, \xi_{i}(Q_{i}|\lambda) \, \mu(\lambda) \\
&= \int d\lambda \langle Q_{i} \rangle_{\lambda} \, \langle Q_{j} \rangle_{\lambda},
\end{align}
where $\langle \dots \rangle_{\lambda}$ denotes the expectation value for a given ontic state $\lambda$.  We can then express (\ref{LGK3}),
\begin{equation} 
K_{3} \equiv C_{21} + C_{32} - C_{31},
\end{equation}
in the following way
\begin{equation}
K_{3} = \int d\lambda \, \mu(\lambda) \Bigg( \langle Q_{2} \rangle_{\lambda} \,  \langle Q_{1} \rangle_{\lambda} +  \langle Q_{3} \rangle_{\lambda} \,  \langle Q_{2} \rangle_{\lambda} \, -  \langle Q_{3} \rangle_{\lambda} \, \langle Q_{1} \rangle_{\lambda} \Bigg).
\end{equation}
Given that the expectation value of $Q_{i}$ is bounded in magnitude by unity, the value $K_{3}$ is once again seen to satisfy the inequality
\begin{empheq}[box=\widefbox]{align}\label{LGsimplest2}
-3 \leq K_{3} \leq 1.
\end{empheq}  

\textbf{b) Comments: }
\begin{enumerate}[noitemsep, topsep=0pt, label=\roman*)]
	\item This expresses the notion that LG inequalities are valid for ontological models (A1) with locality in time, (A2) and (A3).  
	\item Macroscopic realism can be considered as a specific form of the ontological framework through the formula
	\begin{equation}
	\mu(\lambda) = \sum_{k} p_{k} \nu_{k} (\lambda),
	\end{equation}  
	where $\nu_{k} (\lambda)$ is a  distribution of states that all share macroscopic property $k$ with respect to the relevant measurement $M$.
	\item The framework of ontological models has also found use in other temporal settings.  In \cite{leifer2017time}, these models were used in arguing that a time symmetric interpretation of quantum theory is not possible without retrocausality.
\end{enumerate}

\subsection{Through Games}

\textbf{a) Preliminaries: } There are two considerable problems with the LG inequalities.  The first is that the LG inequalities were designed for macroscopic systems, as opposed to a single evolving system.  The second problem is that correlation functions that lead to violations of (\ref{temporalCHSH}) can be classically simulated using a temporal version of the Toner-Bacon protocol \cite{brierley2015nonclassicality,toner2003communication}.  To counter these points, we will describe the use of games to develop a new formulation \cite{brierley2015nonclassicality} of Bell's theorem for temporal correlations.  We consider the case of a single quantum system measured at $n$ points in time.  The focus will be on a novel definition of nonclassicality for these temporal correlations, and provide the needed advantages over the LG inequalities. 

\textbf{b) modulo-(m,d) games: } The particular $n$-player game which we will utilize are known as modulo-$m,d$ games \cite{boyer2004extended}.  Each $n\geq 1$ players is given an integer $X_{j} \in [0, \dots, d-1 ]$ for some fixed integer $d\geq 2$.  The players are promised that $d$ divides their sum 
\begin{equation}
\sum_{j=1}^{n} X_{j} \equiv 0 \quad \text{mod} \, d.
\end{equation} 
The players are allowed to give answers in the form of integers $Y_{j} \in [0, \dots , m-1]$ for some fixed integer $m \geq 2$.  The condition for winning the game is if the answers satisfy
\begin{equation}
\sum_{j=1}^{n} Y_{j} \equiv \frac{\sum_{j}X_{j}}{d} \quad \text{mod} \, d.
\end{equation}
One can think of these games as a distributed computing task.  It can be shown that these games cannot be solved with certainty using classical randomized algorithms.  However, it possible to solve these games with certainty using a quantum GHZ state. 

For a temporal scenario, a sequential version of the modulo-$m,d$ game is desired.  This can be achieved as follows.  A sequential mudolo-$m,d$ game is a communication task in which $n$ separate players are given $(\log \, d)$-bit inputs $X_{k}$ with the condition that
\begin{equation}
\sum_{k=1}^{n} X_{k} \text{mod} \, d = 0. 
\end{equation}
The requirement of the players is to provide values $Y_{k} \in [0,\dots m-1]$ that satisfies 
\begin{equation}
d \sum_{k=1}^{n}Y_{k}\equiv \sum_{k=1}^{n}X_{k} \, \text{mod} \, (md)
\end{equation}
in a sequential protocol.  In this sequential case, the $k$th stage allows the $k$th player to produce their local output $Y_{k}$ and communicate a $c_{k}$-bit message $M_{k}$ to the $(k+1)$st player.  

\textbf{c) Temporal correlations: } Temporal correlations that have the same form as spatial correlations of an $n$-qu$m$it GHZ state 
\begin{equation}
\ket{GHZ} = \frac{1}{\sqrt{m}}\sum_{i=1}^{m} \ket{i}^{\otimes n},
\end{equation}
are referred to as \textit{temporal GHZ correlations}.  It can be shown that the sequential mudolo-$(m,d)$ game can be solved exactly using a sequence of POVM measurements on a single qu$m$it state which produces temporal GHZ correlations \cite{brierley2015nonclassicality}.

\textbf{d) Nonclassicality of temporal correlations: } We describe a new definition to capture the nonclassical properties of these temporal quantum correlations, and relate this later to the game.  The impetus for this definition comes from the notion that an $m$-level physical system has a classical information capacity of $\text{log}_{2}\, m$.  The other motivation is that one wants to decide if a set of correlations is nonclassical purely based on the correlation function.  We write the temporal correlation function as 
\begin{equation}\label{temporalcorrelationfunction}
E(Y_{1} \dots, Y_{N} | X_{1} \dots, X_{N}),
\end{equation}
where a sequence of $N$ consecutive measurements on a single quantum system with measurement settings provided by inputs $X_{k}$ and measurement results given by numbers $Y_{k}$.  We define a temporal correlation function (\ref{temporalcorrelationfunction}) of the $m$-level physical system as \textit{nonclassical} if all classical algorithms that simulate the function require more than $\text{log}_{2}\, m$ bits of classical communication at some step of the simulation.  In other words, the correlation function is nonclassical if every classical simulation of it requires more communication that the classical communication capacity of the physical system in at least one stage of the simulation.

\textbf{e) A temporal ``Bell inequality": } A key result is that every classical protocol that solves the sequential modulo-$(m,d)$ games with certainty uses at least 
\begin{empheq}[box=\widefbox]{align}\label{nonclassicalitynumber}
c_{k} = \text{log} \, \frac{d}{m}
\end{empheq}  
bits of communication in all stages of the protocol except at most $md-1$ (not necessarily consecutive) stages when $d$ is an integer power of $2$ and $m$ is even. 

This result can be thought of as a Bell inequality in that it limits what one can do with classical resources and also us to exhibit the nonclassicality of temporal quantum correlations.  This latter piece can be described using the following result:  The temporal GHZ correlations arising from the sequential measurements on a single qu$m$it, where $m$ is even, are nonclassical for $n \geq 2m^{3}$.  To prove this one simply uses result (\ref{nonclassicalitynumber}) and also shows the result that there exists a sequential modulo-$(m,d)$ game for some $d$ and $n$ for which classical simulation uses in at least one stage of the protocol more than $\text{log} \, m$ bits of communication.

\subsection{Other works}

In this subsection, we provide a brief overview on some other interesting works regarding non-locality across time.

\textbf{a) Temporal Hardy's paradox: } A temporal version of Hardy's non-locality paradox was proposed \cite{fritz2010quantum} and experimentally verified \cite{fedrizzi2011hardy}.  In this scenario, let Alice and Bob measure one after the other to signify the temporal property.  Let $P(r,s|k,l)$ denote the probability that Alice obtains result $r$ and Bob obtains result $s$ given they chose detector settings $a_{k}$ and $b_{l}$ respectively.  The temporal Hardy paradox is that under the LG assumptions (A1-3), the probabilities
\begin{align}
P(+1, +1 |1,1) &= 0, \\
P(-1, +1 |1,2) &= 0, \\
P(+1, -1 |2,1) &= 0, \\
P(+1, +1 |2,2) &> 0,
\end{align}
are mutually inconsistent.  Quantum theory on the other hand provides a way where these probabilities can be simultaneously be fulfilled.  

\textbf{b) Indefinite causal structures: } There are many frameworks that employ the use of the Choi-Jamiolkowski isomorphism.  One example of this is in the framework of quantum indefinite causal structures \cite{brukner2014quantum, oreshkov2012quantum}.  It provides a framework that does not assume a pre-defined global causal structure but only that quantum theory holds locally. Central to the framework is the ``process matrix" which can be thought of as a generalization of a density matrix.  Of interest to the subject of this thesis is that this framework has been used to analyze temporal quantum correlations \cite{costa2018unifying, ringbauer2018multi}; in this work they experimentally observed multi-time quantum correlations that cannot be replicated by any spatial quantum state of equal dimension.

\textbf{c) Pseudo-density matrix: } Another generalization of a density matrix is known as the pseudo-density matrix \cite{fitzsimons2015quantum}.  This framework has been used to analyze various temporal quantum correlations including a weaker version than non-locality across time known as temporal steering \cite{ku2018hierarchy}. In addition to that it has been used in identifying the relationship between temporal correlations and aspects of quantum communications \cite{pisarczyk2019causal}.

\textbf{d) Quantum causal models: } Classical causal models have found a wide range of use in areas of machine learning \cite{koller2009probabilistic}.  There have been various advances \cite{allen2017quantum, barrett2019quantum} on quantum generalizations of classical causal models.  This may lead to a deeper understanding of how quantum causality differs from classical causality.  Moreover, from the perspective of this thesis, it may provide a platform for the development of temporal quantum machine learning algorithms.  

\textbf{e) Entangled histories framework: } The entangled histories framework \cite{nowakowski2017quantum, nowakowski2016monogamy} (and its consistent histories framework) are based on a analogous version of a unitary evolution operator known as the bridging operator.  The concept of entanglement in time is introduced within this framework with a focus on studying the property of monogomy of entanglement (\ref{monogomyCHSH}).


\chapter{Relativistic Quantum Information}\label{chap: RQI}

\begin{chapquote}{Philip Pullman, \textit{His Dark Materials}}
	``You know why we have come together: we must decide what to do about
	these new events. The universe is broken wide, and Lord Asriel has opened the way
	from this world to another.''
\end{chapquote}

\textsf{THE UNIVERSE} contains both quantum physics as well as relativistic effects.  However, a single theoretical description of these diverse phenomena remains elusive.  More recently, there have been investigations on whether the conceptualization of quantum information could play a crucial role for this unification.  On a coarse level, such research activities can be categorized in two directions.  The first is known as relativistic quantum information (RQI), and it examines the effects of relativity on the concepts of quantum information science  \cite{fuentes2013lecture, alsing2012observer, martin2014entanglement}.  Besides fundamental reasons, this has important applications most notably to satellite based quantum communications \cite{bruschi2014spacetime}.  The second direction explores how the concepts of quantum information science could be used to study relativistic structures \cite{harlow2016jerusalem, nishioka2018entanglement, witten2018aps}. A large motivation for this path stems from the holographic principle.  Both directions use quantum field theory \cite{srednicki2007quantum, mukhanov2007introduction, birrell1984quantum} which represents a partial unification; this is in contrast to the standard use of non-relativistic quantum mechanics to articulate quantum information.  In this thesis we focus on the first direction of RQI, and explore how an entanglement in space and an entanglement in time manifest themselves in such a setting.

\section{Review of Relativity}

To understand the relativistic effects on quantum information science, we first provide a brief review on the subject of relativity.  For a thorough introduction, we refer the reader to \cite{carroll2004introduction}. 

\subsection{Special relativity}

Perhaps the most shocking temporal effect in special relativity is time dilation.  This can be mathematically described as
\begin{equation}\label{timedilation}
\Delta t = \gamma \Delta \tau,
\end{equation}
where $\Delta t$ and $\Delta \tau$ represents the time intervals under consideration, and
\begin{equation}\label{timedilationgamma}
\gamma \equiv \frac{1}{\sqrt{1-\frac{v^{2}}{c^{2}}}}.
\end{equation} 
Ultimately, time dilation and other special relativistic effects are a consequence of the invariance of the spacetime interval between two events
\begin{equation}\label{spacetimeinterval}
(\Delta S)^{2} = c^{2}(\Delta t)^{2} - (\Delta x)^{2} - (\Delta y)^{2} - (\Delta z)^{2},
\end{equation}
where $c$ represents the speed of light and where we have used coordinates $(t, x, y, z)$.  This implies that observers who are in motion to each other should always agree on the value of the spacetime interval (\ref{spacetimeinterval}) despite them disagreeing on the individual spatial intervals and time interval.  The invariance of this particular combination of spatial and temporal intervals is what leads to the statement that space and time form one `object' called spacetime.

A spacetime interval is called timelike if $(\Delta S)^{2}>0$.  This means there is some frame of reference (coordinate system) where the events occur at the same spatial location, and there is no frame of reference where the events occur at the same instant of time.  Moreover, an event that occurs first in one frame of reference, occurs first in all frames of reference.  A spacetime interval is called spacelike if $(\Delta S)^{2}<0$.  This implies that there is some frame of reference where the events occur at the same instant of time, and there is no frame of reference where the events occur at the same spatial location.  Lastly, a spacetime interval is called lightlike or null if $(\Delta S)^{2}=0$.  This has the consequence that there is no frame of reference where the events occur at either the same instant of time or at the same spatial location.  Furthermore, the event that occurs first in one frame of reference, occurs first in all frames of reference.

A spacetime diagram is one where the vertical axis corresponds to time (and where we set $c=1$) and horizontal axis corresponds to one of the spatial coordinates, say $z$.  The origin is an event, denoted by say $E$.  Light rays move on lines $z=t$ and $z=-t$, which defines a light cone.  Special relativity says that nothing can travel faster than the speed of light.  This can be depicted in an alternative way in that spacelike intervals are the regions outside the light cone.  These are the sets of events that are causally unrelated to event $E$.  

The infinitesimal version of the spacetime interval (\ref{spacetimeinterval}) takes the form
\begin{equation}\label{flatspacetime}
ds^{2} = c^{2}dt^{2} - dx^{2} - dy^{2} - dz^{2}.
\end{equation}

\subsection{General relativity}

The theory of general relativity says that spacetime line element (\ref{flatspacetime}) is one of many possible spacetime line elements.  Each describes a different spacetime.  The particular case of (\ref{flatspacetime}) is known as flat spacetime, and the theory of general relativity includes curved spacetimes.  To adequately describe such curvature, the subject utilizes the mathematics of differential geometry.  The central tenet of differential geometry is that an intrinsic description of space could be accomplished by distance measurements made within that space.

For our brief review, we employ the standard use of the Einstein summation convention where one omits the summation symbol whenever a pair of contravariant and covariant indices appears in one term.  We usually let the indices range over the four spacetime dimensions unless otherwise stated.  Hence, (\ref{flatspacetime}) can be written as
\begin{equation}\label{flatspacetime2} 
ds^{2} = \eta_{ab} \, dx^{a}dx^{b},
\end{equation}
for coordinates $dx^{a}$ ($dx^{0} = c dt$) and the quantity  $\eta_{ab}$ is known as the Minkowski metric,
\begin{equation}
 \eta_{ab} = 
\begin{pmatrix}
1 & 0 & 0 & 0 \\
0 & -1 & 0 & 0 \\
0 & 0 & -1 & 0 \\
0 & 0 & 0 & -1
\end{pmatrix}.
\end{equation}
More generally, we can represent an arbitrary line element as
\begin{equation}\label{metrictensor}
ds^{2} = g_{ab} \, dx^{a}dx^{b},
\end{equation}
where $g_{ab}$ is known as the metric tensor or simply as metric.  A Lorentzian metric is a metric with signature  $(+---)$.  This means that any given point in spacetime we can find coordinates such that 
\begin{equation}
g_{ab} = \eta_{ab} = 
\begin{pmatrix}
1 & 0 & 0 & 0 \\
0 & -1 & 0 & 0 \\
0 & 0 & -1 & 0 \\
0 & 0 & 0 & -1
\end{pmatrix}.
\end{equation}
In the case of flat spacetime, the metric $g_{ab} = \eta_{ab}$ everywhere.

General relativity postulates that spacetime is a four-dimensional manifold equipped with a Lorentzian metric, $g_{ab}$.  A manifold can be thought of as a collection of points which locally looks like $\mathbb{R}^{4}$.  The quantity $ds^{2} = g_{ab}\:dx^{a}\:dx^{b}$ is invariant.  The metric tensor plays the crucial role of determining the geometry of the manifold and the important geometric quantities are built from this tensor and its derivatives.  The connection (or Christoffel symbol) is given by 
\begin{equation}
\Gamma_{\;\;bc}^{\:a} \equiv \frac{1}{2}\:g^{ad}\:(g_{db,c} + g_{dc,b} - g_{bc,d}).
\end{equation}
The matrix inverse of metric $g_{ab}$ is denoted $g^{ab}$.  Furthermore, the commas denote partial derivatives: $X_{a,b} \equiv \partial_{b}X_{a} \equiv \partial X_{a}/\partial x^{b}$.  We say a vector is timelike if $g_{ab}X^{a}X^{b} > 0$, spacelike if $g_{ab}X^{a}X^{b} < 0$ and lightlike or null if $g_{ab}X^{a}X^{b} = 0$  

This connection can be defined in terms of the covariant derivative of a tensor
\begin{equation}
\nabla_{b} \: V^{a} = \partial_{b}\:V^{a} + \Gamma_{\;\;cb}^{\:a}\:V^{c}.
\end{equation}
This is a generalization of taking a derivative in curved spaces. Notice the deviation from flat space is represented by the connection. The Riemann curvature tensor is a quantity which measures the extent to which the covariant derivative fails to commute, and in that sense, the information about the curvature is located in the components of this tensor.  The explicit formula for this tensor is given by
\begin{equation}
R_{\;\;\;bcd}^{\;a} \equiv \partial_{c}\:\Gamma_{\;\;bd}^{\:a}\: - \partial_{d}\:\Gamma_{\;\;bc}^{\:a}\: + \Gamma_{\;\;ec}^{\:a}\:\Gamma_{\;\;bd}^{\:e}\: - \Gamma_{\;\;ed}^{\:a}\:\Gamma_{\;\;bc}^{\:e}.
\end{equation}  
The Ricci tensor, Ricci scalar and Einstein tensor are respectively built out of the Riemann tensor as
\begin{align}
&R_{ab} \equiv R_{\;\;\;acb}^{c}, \\
&R \equiv g^{ab}\:R_{ab}, \\
&G_{ab} \equiv R_{ab} - \frac{1}{2} \:R\:g_{ab}.
\end{align}
From this, the theory of general relativity postulates the Einstein field equations
\begin{equation}\label{Einsteinfieldequations}
G_{ab} = \frac{8\pi G}{c^{2}} T_{ab},
\end{equation}
where $G$ is Newton's constant of universal gravitation.  The equations relate the curvature of spacetime (quantified by $G_{ab}$) to the distribution of matter and energy (as quantified by the stress-energy tensor $T_{ab}$).  

The Einstein field equations allows one to obtain a spacetime from a given matter-energy distribution.  Vacuum spacetimes are solutions where $T_{ab} = 0$ in (\ref{Einsteinfieldequations}).  This can be shown to be equivalent to the statement that $R_{ab} = 0$ and is known as a Ricci-flat solution.  The flat spacetime (\ref{flatspacetime}) one such solution.  Another Ricci-flat solution is known as the Schwarzschild metric which in coordinates $(t, r, \theta, \phi)$ is written as
\begin{equation}\label{Schmetric}
ds^{2}= \Bigg(1-\frac{2\:m}{r}\Bigg)\:{dt}^{2}-\frac{1}{1-\frac{2\:m}{r}}dr^{2}-{r}^2\, \Bigg({d\theta}^{2}+{{\sin}^{2}\theta}\:{d\phi}^2\Bigg).
\end{equation}

The parameter $m$ measures the amount of mass inside the radius $r$, and in the region $r \leq 2m$ the metric describes a black hole region.  The solution blows up at $r=0$ and $r=2m$; the former is known as a physical singularity whereas the latter is known as a coordinate singularity as it is simply an artefact of the use of this particular coordinate system.

In this thesis, we have seen the manifestation of the Schr\"{o}dinger equation (\ref{schrodinger}).  Along with the Einstein field equations (\ref{Einsteinfieldequations}), these two pieces form the fundamental equations of modern theoretical physics.  The aim to unify these descriptions is known as the problem of quantum gravity, and has so far remained unsolved despite considerable efforts.  We make a few remarks on the similarities of these equations.   Both require an energy input; the first through the energy-momentum tensor and the the second from the Hamiltonian.  The first equation describes the dynamics of spacetime while the second equation describes the dynamics of a quantity whose direct relationship to the physical world is unknown. Nevertheless, both output solutions that describe point particles that behave in the most bizarre manner; general relativity says that point particles cannot exist but form singularities; quantum theory provides point particles with the most bizarre properties such as entanglement.  Of particular relevance is that the authors Einstein and Rosen of the EPR paper \cite{einstein1935can}, wrote another paper that same year titled ``The particle problem in general relativity" \cite{einstein1935particle}.  In it they attempted (but failed) to build a model of a point particle without a singularity.  The work was later termed the Einstein-Rosen bridge, and provided the pathway for the most shocking temporal structures in relativity, namely wormholes \cite{visser1996lorentzian}.

\section{Quantum Fields}

\subsection{Quantum field theory}

We showed in Chapter \ref{chap: QInfo} how the quantum circuit model is based on the postulates of quantum theory.  As a framework, quantum theory (ie those set of postulates) does not specify the state space, the state vector, or the Hamiltonian of a specific physical system under consideration.  It merely provides the mathematical framework for the construction of various physical theories.  The specification of such quantities allows the physics to arise, resulting in different physical theories.  Non-relativistic quantum mechanics is only one such theory; the quantum circuit model corresponds to a non-relativistic two-level quantum system.  Quantum field theory is another subset of quantum theory which describes (special) relativistic quantum particles.  In this latter sense, quantum field theory can be viewed as a unification of quantum theory and special relativity. 

In non-relativistic quantum mechanics, we have a position and a momentum operator; however the existence of these operators are not part of the postulates.  In quantum field theory, these operators are not well-defined, and position is described as a label. We have an operator at each point in space and the collection of these position-dependent operators is known as a quantum field.  Each quantum field has what is known as a conjugate momentum density which is also a function of the spatial label.

The framework of quantum theory can be expressed in the Schr\"{o}dinger picture (where operators are time independent and states are time dependent), the Heisenberg picture (where operators are time dependent and states are time independent), or the Dirac picture (which is an intermediate of the two).  In the Heisenberg picture, the quantum field then is also a function of time.  So far in this thesis we have been using the Schr\"{o}dinger picture which involves the equation (\ref{schrodinger}).  Portraying quantum field theory using the Schr\"{o}dinger picture results the Hamiltonian expressed in terms of infinitely many degrees of freedom, and (\ref{schrodinger}) taking the form of a functional differential equation.  The quantum field state (or wave functional as its known in this case) is a function of time but also a functional of the classical field configuration.  And the square of the wave functional gives the probability density for measuring a certain field configuration.  On a related note, all the foundational mysteries regarding quantum theory, such as the measurement problem, still remain.

Despite that the Heisenberg picture is rarely used to introduce the subject of non-relativistic quantum mechanics, it is precisely the Heisenberg picture that is often used to introduce quantum field theory.  We aim to provide the most basic tools of this subject in order to progress towards to the entanglements in RQI.

\subsection{Quantization}

The procedure of quantization allows one to obtain a physical theory of a quantum system from an analogous classical system.  As an example, it allows one to obtain a quantum  Hamiltonian operator from a classical Hamiltonian function.

\textbf{a) Harmonic oscillator: } In non-relativistic quantum mechanics, one often quantizes a harmonic oscillator.  A classical harmonic oscillator with external force $J(t)$ satisfies the equation of motion
\begin{equation}
\ddot{q} = -\omega^{2} q + J(t),
\end{equation} 
where Hamiltonian is written as
\begin{equation}\label{classicalHamiltonian}
H(p,q) = \frac{p^{2}}{2} + \frac{\omega^{2}q^{2}}{2} - J(t)q,
\end{equation}
where $q$ is the spatial coordinate and $p$ is the momentum.  Quantization involves turning $q$ and $p$ into respective operators $\hat{q}(t)$ and $\hat{p}(t)$ that satisfy the following commutation relation
\begin{equation}
[\hat{q}, \hat{p}] = i,
\end{equation}
where we have set $\hbar = 1$.  From these quantities, one can define the annihilation and creation operators which can respectively be expressed as
\begin{equation}\label{creationannihilationHO}
\hat{a}(t) \equiv \sqrt{\frac{\omega}{2}} \, \Big[\hat{q}(t) + \frac{i}{\omega}\hat{p}(t)\Big], \quad \hat{a}^{\dagger}(t) \equiv \sqrt{\frac{\omega}{2}} \, \Big[\hat{q}(t) - \frac{i}{\omega}\hat{p}(t)\Big].
\end{equation}
These satisfy 
\begin{equation}
[\hat{a}(t), \hat{a}^{\dagger}(t)] = 1,
\end{equation}
at every moment of time.  Through various computations and a final substitution into (\ref{classicalHamiltonian}), one can obtain the quantum Hamiltonian operator
\begin{align}
\hat{H} = \frac{\omega}{2}(\hat{a}^{\dagger} \hat{a} + \hat{a}\hat{a}^{\dagger}) - \frac{\hat{a}^{\dagger} + \hat{a}}{\sqrt{2 \omega}} J(t).
\end{align}
One can proceed to construct a basis for the corresponding Hilbert space.  This assumes the existence of normalized state $\ket{0}$ (note this is not the element from the computational basis states (\ref{compbasis})) where 
\begin{equation}\label{HOvacuum}
\hat{a} \ket{0} = 0.
\end{equation}
This state is known as the vacuum state.  One can create excited states 
\begin{equation}
\ket{n} = \frac{1}{\sqrt{n!}}(\hat{a}^{\dagger})^{n}\ket{0},
\end{equation}   
for $n \geq 1$.  All possible quantum states of the oscillator can be written as
\begin{equation}
\ket{\psi} = \sum_{n=0}^{\infty}\psi_{n} \, \ket{n}, \quad \sum_{n=0}^{\infty} \lvert \psi_{n} \rvert^{2} = 1.
\end{equation}

\textbf{b) Field quantization: } Classical relativistic fields can be described by equations such as the Klein-Gordon equation, the Dirac equation, and the Maxwell equations.  Their role is analogous to harmonic oscillator in that they provide a classical Hamiltonian for quantization.  The quantization of these classical field equations provides the quantum field theory (which in turn results in a description of relativistic quantum particles).  More precisely, the Schr\"{o}dinger equations corresponding to each of the classical field equations articulates the different quantum field theories.  In this section, we will describe this quantization using the Heisenberg picture for the specific case of the Klein-Gordon equation.  The scalar field satisfying this equation can be thought of as a set of infinitely many harmonic oscillators.  Hence our quantization method will relate to the method used to quantize a harmonic oscillator.  To start this procedure, we have classical scalar field $\phi$ which satisfies the Klein-Gordon equation
\begin{equation}
\square \phi = 0,
\end{equation}
where the d'Alambertian operator $\square$ is defined as
\begin{equation}
\square \phi \equiv \frac{1}{\sqrt{-g}}\partial_{\mu}(\sqrt{-g}g^{\mu \nu}\partial_{\nu}\phi).
\end{equation}
We have used notation $g = \text{det}(g_{ab})$.  In flat two-dimensional spacetime (\ref{flatspacetime2}), the metric is given by $g_{\mu \nu} = \eta_{\mu \nu} = \{+-\}$.  Hence we have
\begin{equation}
\square \phi = (\partial_{t}^{2} - \partial_{z}^{2})\phi,
\end{equation}
for coordinates $(t,z)$.  The solution to the Klein-Gordon equation are plane waves in Minkowski spacetime $M$, 
\begin{equation}
u_{\omega,M}(t,z)=\frac{1}{\sqrt{4\pi\omega}}e^{-i\omega(t-\epsilon z)},
\end{equation}
where $\epsilon$ equates to $+1$ for positive momentum modes and to $-1$ for negative momentum modes. These plane wave solutions are known as global field modes.  The modes where $\omega>0$ are orthonormal with respect to a Lorentz invariant inner product
\begin{equation}
(\phi,\psi)=-i\int_{\Sigma}(\psi^*\partial_{\mu}\phi-(\partial_{\mu}\psi^*)\phi)d\Sigma^{\mu}.
\end{equation}
Our next step is to quantize this scalar field.  To do so we require a time-like Killing vector field.  We say that $K$ is a Killing vector field if
\begin{equation}
\mathcal{L}_{K} g_{\mu \nu} = 0
\end{equation}
where the Lie derivative of the metric tensor is defined as
\begin{equation}
\mathcal{L}_{K} g_{\mu \nu} \equiv K^{\lambda}\partial_{\lambda}g_{\mu \nu} + g_{\mu \lambda} \partial_{\nu} K^{\lambda} + g_{\nu \lambda} \partial_{\mu} K^{\lambda}.
\end{equation}
If a spacetime has a Killing vector field, then one can find a basis for the plane wave solutions of the Klein-Gordon equation such that
\begin{equation}\label{specialbasissolution}
\mathcal{L}_{K}u_{k,M} = K^{\mu}\partial_{\mu}u_{k, M} = -i \omega u_{k,M}.
\end{equation}
It can be shown that if $K$ is a time-like Minkowski vector field, then the Lie derivative corresponds to $\partial_{t}$.  Then (\ref{specialbasissolution}) takes the form
\begin{align}
\partial_{t} u_{k, M} &= -i \omega u_{k, M} \\
\partial_{t} u_{k, M}^{*} &= -i \omega u_{k, M}^{*}
\end{align}
where we identify $\omega > 0$ with a frequency.  One can classify the plane wave solutions to the Klein-Gordon equation as
\begin{align}
u_{k} \quad &\rightarrow \quad \text{positive frequency solutions} \\
u_{k}^{*} \quad &\rightarrow \quad \text{negative frequency solutions}
\end{align}
The quantized field can be obtained using these positive and negative frequency solutions.  More precisely, the quantized field satisfies equation
\begin{equation}
\square \hat{\phi} = 0,
\end{equation}
and is given by the operator value function
\begin{equation}\label{quantumfieldfunction}
\hat{\phi}=\int(u_{k,M}a_{k,M}+u_{k,M}^*a_{k,M}^{\dag})dk.
\end{equation}
The operators $a_{k,M}^{\dag}$ and $a_{k,M}$ are creation and annihilation operators which satisfy
\begin{equation}
[a_{k,M},a^{\dagger}_{k^{\prime},M}]=\delta_{k, k^{\prime}}.
\end{equation} 
Observe that positive frequency solutions are associated with annihilation operators, whereas negative frequency solution correspond to creation operators.  Furthermore, these creation and annihilation operators are analogous to case of the harmonic oscillator (\ref{creationannihilationHO}).  We also have a vacuum state for the quantum field which is defined by
\begin{equation}
a_{k,M}{\ket{0}}^{\mathcal{M}}=0.
\end{equation}  
This is analogous (\ref{HOvacuum}).  In fact, the vacuum state can be expressed as
\begin{equation}\label{Minkowskivacuum}
\ket{0}^{\mathcal{M}} = \prod_{k} \ket{0_{k}}^{\mathcal{M}},
\end{equation}
where $\ket{0_{k}}^{\mathcal{M}}$ is the ground state of mode $k$.  The vacuum can be physically thought of as `empty space.'  The action of the creation operators on the vacuum allows one to define particle states
\begin{equation}
\ket{n_1,...,n_{k}}^{\mathcal{M}}= (n_1!,...,n_{k}!)^{-1/2} \, a_{1,M}^{\dag n_1}...a_{k,M}^{\dag n_k}\ket{0}^{\mathcal{M}}.
\end{equation} 
This procedure implies that only when there exists a time-like Killing vector field is the notion of a particle well-defined.

\textbf{c) Bogoliubov transformation: } For the scenario that a spacetime admits a time-like Killing vector field, the vector field is generally not unique.  In our procedure we used $\partial_{t}$, and another such vector field could be denoted by $\partial_{\hat{t}}$.  For each case, one can obtain a basis for the solutions which we respectively denote $\{u_{k}, u_{k}^{*}\}$ and $\{\bar{u}_k,\bar{u}_k^*\}$.  Positive and negative frequency solutions can be identified for the other basis as well.  Hence, the field can be equivalently quantized in both bases,
\begin{equation}
\hat{\phi}=\int(u_ka_k+u_k^*a_k^{\dag})dk=\int(\bar{u}_{k^{\prime}}\bar{a}_{k^{\prime}}+\bar{u}_{k^{\prime}}^*\bar{a}_{k^{\prime}}^{\dag})dk^{\prime}.
\end{equation}
By utilizing the inner product, it is possible to obtain a transformation, known as the Bogoliubov transformation, between the representations for the creation and annihilation operators
\begin{equation}
a_k=\sum_{k^{\prime}}(\alpha^{\ast}_{kk^{\prime}}\bar{a}_{k^{\prime}}-\beta^{\ast}_{kk^{\prime}}\bar{a}_{k^{\prime}}^{\dag}),
\end{equation}
where $\alpha_{kk^{\prime}}=(u_{k},\bar{u}_{k^{\prime}})$ and $\beta_{kk^{\prime}}=-(u_{k},\bar{u}^{\ast}_{k^{\prime}})$ which are known as the Bogoliubov coefficients.  Both vacuum states are defined as
\begin{equation}
{a}_k{\ket{0}}={\bar{a}}_k{\bar{\ket{0}}}=0
\end{equation}
and hence it is possible to derive a transformation between the two vacuum states.

\textbf{d) For RQI: } The field from the above procedure is just one type of quantum field theory.  Collectively, quantum field theories provide a description of relativistic quantum systems. The field of RQI uses quantum field theory, as opposed to non-relativistic quantum mechanics, to express quantum information and its information tasks.  In this sense, it allows one to investigate the effects of relativity on the concepts of quantum information science.

\subsection{Locality in RQI}

In Chapter \ref{chap: QFound}, we defined locality as systems that obey the Bell-CHSH inequality (\ref{CHSHinfoundations}).  It is important to use the more precise terminology of Bell locality in relation to (\ref{CHSHinfoundations}) given that there are other mathematical characterizations of the notion of locality \cite{brunner2014bell}.  In quantum field theory, locality is quantitatively expressed in a different manner to (\ref{CHSHinfoundations}), and is also more commonly referred to as causality \cite{tongquantum}.  It captures the notion that a measurement at one spatial location say $x$ cannot affect a measurement at another spatial location $y$, when $x$ and $y$ are not causally connected.  

More rigorously, \textit{causality} is the requirement that all operators commute for spacelike separation
\begin{empheq}[box=\widefbox]{align}\label{localqft}
[O(x), O(y)] = 0 \quad \text{for} \quad (x-y)^{2} < 0.
\end{empheq} 
One often says that the theory is causal if the commutators vanish outside the light cone.  Nevertheless, (\ref{localqft}) does not make a distinction between the forward light cone and backward light cone.  In  \cite{donoghue2019arrow} it was shown that there is an implied arrow of causality (meaning what is the past and what is the future) which is connected to the sign of the imaginary number in the quantization procedure.  Reversing the sign of the factors of $i$ leads to a causal theory with the consequence of an arrow of causality running from large times to small times. 

\subsection{Entanglement in RQI}

In Chapter \ref{chap: QEnt}, we defined entanglement as the nonseparability of a state (\ref{bipartiteseparable}).  This nonseparability was also expressed through the density operator as (\ref{bipartitedensityseparable}).  Despite RQI harnessing the framework of quantum field theory, it utilizes the same algebraic nonseparability definition of entanglement as non-relativistic quantum mechanics.  An open question in RQI \cite{alsing2012observer} is whether there exists a more general notion of quantum interdependence for relativistic quantum systems which maps to the standard notion of entanglement in the non-relativistic regime.    

To articulate the current definition \cite{olson2012extraction}, consider a spacetime manifold $M$, with two subsets of it denoted $R_{1}$ and $R_{2}$.  The states of the field restricted to each subset are described by respective Hilbert spaces, $H_{R_{1}}$ and $H_{R_{2}}$.  If the field operators commute between the two regions then one can say that $H_{R_{1}}$ and $H_{R_{2}}$ represent independent systems.  Then the state of the quantum field $\rho$ restricted to $R_{1} \cup R_{2}$ is called \textit{entangled} if it is not separable, meaning if it cannot be represented as     
\begin{empheq}[box=\widefbox]{align}\label{entanglementqft}
\rho = \sum_{i}p_{i} \rho_{1}^{i} \otimes \rho_{2}^{i},
\end{empheq}
where $\rho_{1}^{i}$ are density operators on $H_{R_{1}}$ and $\rho_{2}^{i}$ are density operators on $H_{R_{2}}$, with $p_{i} \geq 0$.  Notice that this definition is analogous to (\ref{bipartitedensityseparable}).

Our aim in the next section is to show that the vacuum state (\ref{Minkowskivacuum}) of a quantum field is an entangled state.  We will do this through the state as opposed to a density operator framework.  It is important to emphasize that whether a state is entangled or not depends on the tensor-product decomposition that is chosen for the total Hilbert space \cite{martin2014entanglement}.  Hence the concept of entanglement in the field is to be understood within this context.  To elaborate on matter, let $H_{\text{field}}$ denote the Hilbert space associated to the free scalar field in Minkowski spacetime (\ref{quantumfieldfunction}).  If we decompose that Hilbert space into plane wave modes, then we have decomposition 
\begin{equation}
H_{\text{field}} = \bigotimes_{k} L^{2}(R)_{k}
\end{equation}  
where $L^{2}(R)_{k}$ is the countably infinite harmonic oscillator state space with mode $k$.  Using this, the Minkowski vacuum is not entangled but can be decomposed into product state 
\begin{equation}
\ket{0}^{\mathcal{M}} = \bigotimes_{k}  \ket{0_{k}}^{\mathcal{M}},
\end{equation} 
which is equivalent to (\ref{Minkowskivacuum}).  However, as we shall describe one can also decompose the field into a left and right half (known as Rindler wedges) 
\begin{equation}
H_{\text{field}} = H_{\text{left}} \bigotimes H_{\text{right}}.
\end{equation}  
Then the Minkowski vacuum state is a tensor product of two-mode squeezed (TMS) states in pairs of (Rindler) modes indexed by $\omega$ 
\begin{equation}\label{Minkentangled}
\ket{0}^{\mathcal{M}} = \bigotimes_{\omega}  \ket{TMS}_{(\omega, I); (\omega, II)}.
\end{equation} 
TMS states are a central topic in the area of quantum optics \cite{lvovsky2015squeezed}; physically in squeezed states the noise of the electric field at certain phases falls below that of the vacuum state; however our focus is solely on the mathematical description.  

The above decomposition represents a bipartite entanglement across the left-right cut, and the explicit form of equation (\ref{Minkentangled}) will be articulated in the next section.  We also would like to point out that the analysis in the next sections are all in $1+1$ dimensions.

\section{Spacelike Entanglement}

\subsection{Definition}

Recall the RQI definition of entanglement as expressed through (\ref{entanglementqft}) using subsets $R_{1}$ and $R_{2}$ of spacetime $M$.  A further distinction can be made \cite{olson2012extraction}.  If all the points in $R_{1}$ are spacelike separated with respect to all points in $R_{2}$, then we say that the quantum field in $R_{1}$ is \textit{spacelike entangled} with respect to the quantum field in $R_{2}$.  In other words, this nonseparability of the state of the quantum field can be regarded as an entanglement in space as described in Chapter \ref{chap: QEnt}.  However in this relativistic setting, the spatial aspect is articulated far more precisely by using the light cone structure i.e. spacelike intervals.

\subsection{Left-Right entanglement}

The two-dimensional Minkowski spacetime $(t,z)$ can be broken up into regions using the light cone structure.  The spacelike regions outside the light cone $\lvert t \rvert < -z$ and $\lvert t \rvert < z$ are respectively known as the left Rindler wedge and right Rindler wedge.  We want to show that the Minkowski vacuum can be written as a spacelike entangled state between the left and right Rindler modes \cite{crispino2008unruh}.

\textbf{a) Independent systems: } For the possibility of entanglement, we want that the fields within the left and right Rindler wedges are considered as independent systems.  This is a requirement as we want to quantize them separately.  Such a condition gets fulfilled if the commutators vanish
\begin{equation}
[\hat{\phi}(x_{L}),\hat{\phi}(x_{R}) ] = 0.
\end{equation} 
For spacelike intervals, this vanishing holds for both massive and massless fields.

\textbf{b) Minkowski plane waves: } Our derivation rests on the following statement:  Ordinary plane waves in Minkowski spacetime cover the spacetime.  Our coordinates for Minkowski spacetime is $(t,z)$ and hence the massless scalar field in two dimensions satisfies
\begin{equation}\label{waveequation} 
\Big(\frac{\partial^{2}}{\partial t^{2}} - \frac{\partial^{2}}{\partial z^{2}} \Big) \hat{\phi} = 0.
\end{equation}
If we were to use light-cone coordinates
\begin{equation}
U = t-z, \quad \quad \quad \quad V = t+z,  
\end{equation}
then we can write the field in terms left and right moving sectors
\begin{equation}
\hat{\phi}(t,z) = \hat{\phi}_{-}(U) + \hat{\phi}_{+}(V).
\end{equation}
Given that the left and right moving sectors do not interact, we can discuss the effect for the left-moving sector to simply the exposition.  Through expansion, one finds that
\begin{equation}
\hat{\phi}(V) = \int_{0}^{\infty} dk [\hat{b}_{+k}\, u_{k}(V) +\hat{b}_{+k}^{\dagger}u_{k}^{*}(V)],
\end{equation}
where
\begin{equation}\label{minkplanewaves}
u_{k}(V) = (4\pi k)^{-1/2}e^{-ikV}.
\end{equation}
The Minkowski spacetime plane waves are given by (\ref{minkplanewaves}) and its complex conjugate.  The Minkowski vacuum which we denote $\ket{0_{M}}$ is defined as
\begin{equation}
\hat{b}_{+k}\ket{0_{M}} = 0,
\end{equation}
for all $k$.

\textbf{c) Rindler plane waves:} For the right Rindler wedge $(0<V)$, we have the coordinate transformation
\begin{equation}
t = a^{-1} e^{a\epsilon}\sinh(a\tau),  \quad \quad \quad \quad z = a^{-1} e^{a\epsilon}\cosh(a\tau).
\end{equation}
Due to the conformal invariance of the massless wave equation in two dimension, the wave equation takes the same form as (\ref{waveequation}),
\begin{equation}
\Big(\frac{\partial^{2}}{\partial \tau^{2}} - \frac{\partial^{2}}{\partial \epsilon^{2}} \Big)_{R} \hat{\phi} = 0.
\end{equation}
Obtaining analogous light-cone coordinates 
\begin{equation}
\chi = \tau + \epsilon, \quad \quad \quad \quad \kappa = \tau - \epsilon,
\end{equation}
we can express the left-moving sector as 
\begin{equation}
\hat{\phi}_{+}(V) =  \int_{0}^{\infty}d\omega [\hat{a}_{+\omega}^{R}g_{\omega}^{R}(\chi)  +\hat{a}_{+\omega}^{R\dagger}g_{\omega}^{R*}(\chi)].
\end{equation}
The mode solutions or plane waves in Rindler coordinates are
\begin{equation}\label{Rindler1}
g_{\omega}^{R}(\chi) = (4\pi \omega)^{-1/2}e^{-i\omega\chi}.
\end{equation}
Without repeating all the details, similar calculations and results can be made for analogous light cone coordinate, $\overline{\chi}$, in the left Rindler wedge $(V<0<U)$.  In particular the mode function takes the form
\begin{equation}\label{Rindler2}
g_{\omega}^{L}(\overline{\chi}) = (4\pi \omega)^{-1/2}e^{-i\omega\overline{\chi}}.
\end{equation}
The vacuum state for the left Rindler wedge and right Rindler wedge are the identical.  It is known as the Rindler vacuum $\ket{0_{R}}$ and it is defined through
\begin{equation}
\hat{a}_{+\omega}^{R}\ket{0_{R}} = \hat{a}_{+\omega}^{L}\ket{0_{R}} = 0,
\end{equation}
for all $\omega$.

\textbf{d) Bogoliubov transformation: } The Minkowksi light cone coordinates are related to the analogous Rindler light cone coordinates through
\begin{equation}
V = a^{-1}e^{a\chi},  \quad \quad \quad \quad V = -a^{-1}e^{-a\overline{\chi}}.
\end{equation}
Since the modes are complete in their region, we can expand the Rindler modes $g_{\omega}^{R}(\chi)$ and $g_{\omega}^{L}(\overline{\chi})$ in terms of Minkowski plane waves, $u_{k}(V)$ and $u_{k}^{*}(V)$, because the plane waves are defined over all spacetime.  Hence, we can express one set of modes in terms of the other (i.e. using Bogoliubov transformations).  The Heaviside function is used to make the expression valid in their respective quadrant,  
\begin{align}\label{bogo1}
\theta(V)\, g_{\omega}^{R}(\chi) &=  \int_{0}^{\infty}dk \, ( \alpha_{\omega k}^{R} u_{k} (V) + \beta_{\omega k}^{R} u_{k}^{*} (V)), \\ \label{bogo2}
\theta(-V)\, g_{\omega}^{L}(\overline{\chi}) &=  \int_{0}^{\infty}dk \, ( \alpha_{\omega k}^{L} u_{k} (V) + \beta_{\omega k}^{L} u_{k}^{*} (V)).
\end{align}
These Rindler modes form a superposition of Minkowski plane waves and the coefficients, $\alpha$ and $\beta$, are the Bogoliubov coefficients.  Solving equations (\ref{bogo1}) and (\ref{bogo2}) gives the relations
\begin{equation}\label{bogorelations}
\beta_{\omega k}^{L} = -e^{-\pi \omega/a}\alpha_{\omega k}^{R*} \quad \quad \quad \quad \beta_{\omega k}^{R} = -e^{-\pi \omega/a}\alpha_{\omega k}^{L*}
\end{equation}

\textbf{e) L-R entanglement: } By substituting (\ref{bogorelations}) back into (\ref{bogo1}) and (\ref{bogo2}), we can define a new set of modes known as the Unruh modes,
\begin{align}
G_{\omega}(V) &= \theta(V)g_{\omega}^{R}(\chi) + \theta(-V)e^{-\pi\omega/a}g_{\omega}^{L*}(\overline{\chi}), \\
\overline{G}_{\omega}(V) &= \theta(-V)g_{\omega}^{L}(\overline{\chi}) + \theta(V)e^{-\pi\omega/a}g_{\omega}^{R*}({\chi}).
\end{align}
These share the Minkowski vacuum
\begin{equation}
\hat{a}_{G \omega} \ket{0_{M}} = \hat{a}_{\overline{G} \omega} \ket{0_{M}} = 0.
\end{equation}
The Unruh mode annihilation operators can be written in terms of the Rindler annihilation and creation operators as follows
\begin{align}
\hat{a}_{G \omega} &= (\hat{a}_{\omega}^{R} - e^{-\pi\omega/a}\hat{a}_{\omega}^{L\dagger}), \\
\hat{a}_{\overline{G} \omega} &= (\hat{a}_{\omega}^{L} - e^{-\pi\omega/a}\hat{a}_{\omega}^{R \dagger}). 
\end{align}
These expressions can be combined to produce the equation
\begin{equation}
(\hat{a}_{\omega}^{R \dagger}\hat{a}_{\omega}^{R} - \hat{a}_{\omega}^{L \dagger}\hat{a}_{\omega}^{L}) \ket{0_{M}} = 0.
\end{equation}
We proceed to use the approximation that $\omega$ is discrete.  From the previous equation we obtain,
\begin{empheq}[box=\widefbox]{align}\label{spacelikeentanglement}
\ket{0_{M}} = \prod_{i}C_{i} \sum_{n_{i}=0}^{\infty}\frac{e^{-\pi n_{i} \omega_{i}/a}}{n_{i}!}(\hat{a}_{\omega_{i}}^{R\dagger}\hat{a}_{\omega_{i}}^{L\dagger})^{n_{i}}\ket{0_{R}}
\end{empheq}
where $C_{i} = \sqrt{1 - e^{-2\pi \omega_{i}/a}}$.  This is a \textit{spacelike entanglement} of the Minkowski vacuum in terms of the left and right Rindler modes.  One can re-express equation (\ref{spacelikeentanglement}) as
\begin{equation}
\ket{0_{M}} = \prod_{i}C_{i} \sum_{n_{i}=0}^{\infty}{e^{-\pi n_{i} \omega_{i}/a}}\ket{n_{i}^{R}} \otimes \ket{n_{i}^{L}}
\end{equation} 
where $\ket{n_{i}^{R}}$ is the state of a Rindler mode restricted to the right wedge, containing $n$ excitations of frequency $\omega_{i}$.  In an analogous manner, $\ket{n_{i}^{L}}$ is the Rindler mode restricted to the left wedge.  This state is entangled as it is nonseparable between the left and right wedges.

\subsection{Implications}

We discuss some observations of this spacelike entanglement.  One can form a density operator using (\ref{spacelikeentanglement}) and trace over either one of the regions.  This results in a thermal state with temperature,
\begin{equation}
T = \frac{a \hbar}{2\pi k c},
\end{equation}
where $\hbar$ is the reduced Planck's constant and $k$ is the Boltzmann's constant.  This quantity is commonly referred to as the Unruh temperature.  By utilizing the Rindler coordinates, one can intepret the temperature in the following way:  While inertial observers describe the field to be the vacuum, observers in uniform acceleration observe a state thermalized with particles at the Unruh temperature.  Hence the particle content of a field is observer dependent!  This Unruh temperature is in fact analogous to the famous Hawking temperature of a black hole
\begin{equation}\label{hawkingtemperature}
T_{H} = \frac{\hbar c^{3}}{8 \pi G M k},
\end{equation}
where $M$ is the mass of the black hole (\ref{Schmetric}).  From an RQI point of view, one can say that the section of the quantum vacuum trapped behind the event horizon $r=2m$ is spacelike entangled with that outside.  From the temperature (\ref{hawkingtemperature}), one can derive the Bekenstein-Hawking entropy of a black hole
\begin{equation}\label{hawkingentropy}
S_{BH} = \frac{c^{3}Ak}{4G\hbar},
\end{equation}
where $A$ is the surface area of the black hole (ie area of the event horizon).  However it is not known what the microscopic nature the black hole entropy is, and many consider this formula as the crucial clue to quantum gravity.

In terms of quantum information protocols, in \cite{alsing2003teleportation} it was shown that this Unruh effect reduces the fidelity of quantum teleportation.  This alluded to the notion that entanglement is degraded in non-inertial frames.  This was clearly shown in \cite{fuentes2005alice} where degradation of spacelike entanglement was quantified when one of the observers moved in uniform acceleration.  Such a case was also mapped into the scenario of an observer falling into a black hole resulting in a similar degradation.  These results imply that entanglement is an observer dependent phenomenon!  In Chapter \ref{chap: QEnt}, we mentioned the procedure of entanglement swapping.  This notion carries over to RQI and is referred to as entanglement extraction or entanglement harvesting \cite{martin2014entanglement}.  It is the process of extracting field entanglement by local quantum systems interacting with the field in a spacelike separated way.  This harvesting procedure can be generalized into entanglement farming.

\section{Timelike Entanglement}

\subsection{Definition}

Recall the RQI definition of entanglement as expressed through (\ref{entanglementqft}) using subsets $R_{1}$ and $R_{2}$ of spacetime $M$.  Analogous to the spacelike case, one can provide a timelike version \cite{olson2012extraction}.  If all the points in $R_{1}$ are timelike separated with respect to all points in $R_{2}$, then we say that the quantum field in $R_{1}$ is \textit{timelike entangled} with respect to the quantum field in $R_{2}$.  In other words, this nonseparability of the state of the quantum field can be regarded as an entanglement in time as described in Chapter \ref{chap: QEnt}.  However in this relativistic setting, the temporal aspect is articulated far more precisely by using the light cone structure i.e. timelike intervals.

\subsection{Future-Past entanglement}

We want to focus on the regions inside the light cone in the two-dimensional Minkowski spacetime $(t,z)$.  The regions $t>\lvert z \rvert$ and $t < -\lvert z \rvert$ are known respectively as the Future and Past. We proceed to describe the work in \cite{olson2011entanglement} which showed the Minkowski vacuum can be written as a timelike entangled state between the future and past modes.  It follows an analogous procedure to the spacelike case.

\textbf{a) Independent systems: } For massless fields, the commutator vanishes for timelike intervals 
\begin{equation}
[\hat{\phi}(x_{F}),\hat{\phi}(x_{P}) ] = 0.
\end{equation}
Hence we can quantize the massless fields in $F$ and $P$ as independent systems.  The concept of independent system also remains valid as an approximation when the commutator is small.

\textbf{b) Minkowski plane waves: } Exactly the same content as the spacelike case.

\textbf{c) Future-Past plane waves: } For the future quadrant $F$ we have the coordinate transformation
\begin{equation}
t = a^{-1} e^{a\eta}\cosh(a\zeta),  \quad \quad \quad \quad z = a^{-1} e^{a\eta}\sinh(a\zeta).
\end{equation}
For the past quadrant $P$ we have the coordinate transformation
\begin{equation}
t = -a^{-1} e^{a\overline{\eta}}\cosh(a\overline{\zeta}),  \quad \quad \quad \quad z = -a^{-1} e^{a\overline{\eta}}\sinh(a\overline{\zeta}).  
\end{equation}
Due to conformal invariance of the massless wave equation, the wave equations take the same form as (\ref{waveequation}).  With analogous light cone coordinates,
\begin{equation}
\nu = \eta+\zeta, \quad \quad \quad \quad \overline{\nu} = - \overline{\eta} - \overline{\zeta},  
\end{equation}   
we obtain mode functions in these coordinates (like in the spacelike case).  These modes are called conformal modes and are written as
\begin{align}
g_{\omega}^{F}(\nu) &= (4\pi \omega)^{-1/2}e^{-i\omega\nu}, \\
g_{\omega}^{P}(\overline{\nu}) &= (4\pi \omega)^{-1/2}e^{-i\omega\overline{\nu}}.
\end{align}
These conformal modes resemble the Rindler modes (\ref{Rindler1}) and (\ref{Rindler2}).  In fact we will show that these are the Rindler modes.  Hence their annihilation operators define the Rindler vacuum
\begin{equation}
\hat{a}_{+\omega}^{F}\ket{0_{R}} = \hat{a}_{+\omega}^{P}\ket{0_{R}} = 0,
\end{equation}
for all $\omega$.

\textbf{d) Bogoliubov transformation: } The main construct we require, regardless of spacelike or timelike case, are the Minkowski plane waves.  For the spacelike case, we utilized the coordinate transformation
\begin{equation}
V = a^{-1}e^{a\chi},  \quad \quad \quad \quad V = -a^{-1}e^{-a\overline{\chi}},
\end{equation}
and for the timelike case we find that
\begin{equation}
V = a^{-1}e^{a\nu},  \quad \quad \quad \quad V = -a^{-1}e^{-a\overline{\nu}}.
\end{equation}
This shows that the light cone coordinate, $V$, has the same functional relationship to $\chi$ as to $\nu$.  More precisely, $g_{\omega}^{R}(\chi)$ and $g_{\omega}^{F}(\nu)$ are identical functions of $V$ since $\chi(V)=\nu(V)$.  This implies that they are made of the same combination of Minkowski plane waves.  Mathematically, these are the same quantities.  A similar relationship holds between $g_{\omega}^{L}(\chi)$ and $g_{\omega}^{P}(\nu)$.  In other words, these conformal modes are alternative expression for Rindler modes.  Therefore when we expand these modes into Minkowski plane waves
\begin{align}
\theta(V)\, g_{\omega}^{F}(\nu) &=  \int_{0}^{\infty}dk \, ( \alpha_{\omega k}^{F} u_{k} (V) + \beta_{\omega k}^{F} u_{k}^{*} (V)), \\ 
\theta(-V)\, g_{\omega}^{P}(\overline{\nu}) &=  \int_{0}^{\infty}dk \, ( \alpha_{\omega k}^{P} u_{k} (V) + \beta_{\omega k}^{P} u_{k}^{*} (V)),
\end{align}
and compare to (\ref{bogo1}) and (\ref{bogo2}) we obtain relationships
\begin{equation}
\alpha_{\omega k}^{F} = \alpha_{\omega k}^{R}, \quad \beta_{\omega k}^{F} = \beta_{\omega k}^{R}, \quad \alpha_{\omega k}^{P} = \alpha_{\omega k}^{L}, \quad \beta_{\omega k}^{P} = \beta_{\omega k}^{L}.
\end{equation}
\textbf{e) F-P entanglement: } From here, we carry the same procedure as the spacelike case except to replace labels $R$ to $F$, and $L$ to $P$.  This produces the final result,
\begin{empheq}[box=\widefbox]{align}\label{timelikeentanglement}
\ket{0_{M}} = \prod_{i}C_{i} \sum_{n_{i}=0}^{\infty}\frac{e^{-\pi n_{i} \omega_{i}/a}}{n_{i}!}(\hat{a}_{\omega_{i}}^{F\dagger}\hat{a}_{\omega_{i}}^{P\dagger})^{n_{i}}\ket{0_{R}}
\end{empheq}
where $C_{i} = \sqrt{1 - e^{-2\pi \omega_{i}/a}}$.  This shows a \textit{timelike entanglement} between the past and the future in the Minkowski vacuum.

\subsection{Implications}

We make a few remarks on the $F$-$P$ entanglement expressed in (\ref{timelikeentanglement}), in particular on its similarity with the temporal entanglement in Chapter \ref{chap: QEnt}.  Although the quantum field is causally disconnected between $F$ and $P$, measurements in $F$ (for example, projections onto $g_{\omega}$-particle number) can collapse the state of the field in $P$.  Similarly measurements in $P$ should collapse the state in $F$.  This is analogous to the properties of entanglement in time discussed in Chapter \ref{chap: QEnt}.  The similarities to temporal Bell states (\ref{temporalBellstate}) become more striking when in a recent work \cite{sabin2012extracting} it was theoretically shown that this future-past entanglement (\ref{timelikeentanglement}) could be extracted to a pair of qubits that do not coexist at the same time. 

Another extraction for this timelike entanglement was proposed in  \cite{olson2012extraction}.  This involved the use of two detectors.  One of the detectors interacted with the vacuum in the past while the other detector waits and interacts with the future vacuum.  The two detectors end up becoming entangled.  More precisely the timelike entanglement in the Minkowski vacuum is converted into bipartite entanglement between detectors at a constant time.  However, the procedure requires a particular time correlation for the extraction to optimally occur.  This was stated more shockingly through an example as, \textit{``a detector that is switched on and off in the vicinity of a quarter to 12:00 can become entangled with a detector interacting with the field at the same spatial location in the future, but only if the later detector waits to be switched on and off at a quarter past 12:00." }  It is therefore not surprising to see that the existence of a time interval (in this case thirty minutes) is what makes the interdependence of this timelike entanglement shocking.


\chapter{Conclusion}\label{chap: Conclusion}
\begin{chapquote}{Carl Jung, \textit{The Black Books}}
	``My soul, my soul, where are you? Do you hear me? I speak, I call you–are you there? I have returned, I am here again. I have shaken the dust of all the lands from my feet, and I have come to you, I am with you. After long years of long wandering, I have come to you again...''
\end{chapquote}

\textsf{IN THE REALM} of quantum physics, we have witnessed a shocking interdependence across time in a system of multiple particles (in Chapter \ref{chap: QEnt}), a single particle (in Chapter \ref{chap: QFound}), and zero particles (in Chapter \ref{chap: RQI}).  We found that the entanglement in time (in Chapter \ref{chap: QEnt}) and the timelike entanglement (in Chapter \ref{chap: RQI}) are similar in nature.  Both utilize the same algebraic definition of nonseparability.  Both can be expressed in terms of an entanglement between qubits that do not coexist.  However, the timelike entanglement articulates itself more clearly through the light cone terminology.  Despite the common use of referring to non-locality in time (in Chapter \ref{chap: QFound}) as an entanglement in time, we argue for the distinction to be clearly made.  Non-locality in time is not an algebraic definition of nonseperability but rather the property of experimental measurement correlations.  Moreover, further care needs to be taken in the RQI case where the term locality is characterized in a vastly different manner.  In this chapter, we summarize the main achievements of this thesis and discuss future projects that relate to its topic.  They range from conservative next steps to imaginatively speculative paths.

\section{Summary}

This thesis provides one of the first systematic expositions on the concept of entanglement in time of quantum systems.  Furthermore, the thesis contrasts it with the more familiar concept of entanglement in space.  The similarities and differences between these concepts are examined in the context of quantum information, quantum foundations and relativistic quantum information.

The thesis also achieved various original contributions:

\textbf{a) Quantum blockchain: } We designed a quantum information application of entanglement in time, namely a quantum blockchain \cite{rajan2019qblock}.  Most other applications of quantum information harness an entanglement in space.  Though the literature does refer to a few other applications of entanglement in time, they are fundamentally modifications of the spatially entangled case.  Therefore, our work can be regarded as the first novel application of entanglement in time.

\textbf{b) Monty hall teleportation: } We designed a Monty Hall version of quantum teleportation \cite{rajan2019quantum}.  The teleportation protocol is one of the most explored topics in quantum information and our work has added novel techniques into this area.  Future work may involve porting these techniques to other quantum protocols.

\textbf{c) Teleportation involving noise: } We developed a variation of the teleportation protocol for the effect of noise on teleportation \cite{rajan2019quantum}.  This work could be of great applicability to practical quantum communication networks.

\textbf{d) PBR game: } We provided one of the first (if not the first full) reformulation of the Pusey-Barrett-Rudolph theorem into a quantum game \cite{rajan2019quantum}. Given that game-theoretic versions of the CHSH inequalities have played a non-trivial role in quantum information and its foundation, time will only tell how impactful our gamification of this recent foundation result will be.

\textbf{e) Density matrix in Gleason's theorem: }  We provided an explicit construction of the density matrix in Gleason's theorem \cite{rajan2019explicit}.  Such a construction was missing in the vast literature concerning the foundations of quantum physics.

\textbf{f) Geometric proof of KS theorem: } We constructed a simplified geometrical proof \cite{rajan2017kochen} of the Kochen-Specker theorem in quantum foundations.

\section{Quantum Time Machines}

Our exploration moved from classical information (in Chapter \ref{chap: classical}),  towards quantum information (in Chapter \ref{chap: QInfo}), to finally an examination of the entanglements within quantum information science (in Chapter \ref{chap: QEnt}).  Much like classical information was the resource for the development of the `Information Age', one can imagine that quantum information may transform the world to a `Quantum Information Age.'  More pragmatically, the applications of quantum information can be regarded as a technological frontier.  Others have harnessed quantum information to build teleportation systems or create the most powerful computers.  In this thesis,  we have used quantum information to design a quantum blockchain, which can be viewed as a `quantum time machine.'  

Classically, a time machine is any system that permits one to travel into the past, and a rigorous definition can be found in \cite{visser1996lorentzian}.  Using quantum physics, we believe that a broader class of time machines may be possible; this includes functionalities that we have not yet imagined.  To be more precise, we define a \textit{quantum time machine as any quantum system that can perform information tasks across time in classically impossible ways}.  A more rigorous definition could perhaps be formulated with the use of constructs known as steering inequalities \cite{wiseman2007steering, chen2014temporal}.

Our view is that quantum time machines (instead of quantum teleportation or quantum computers) will be the most shocking applications of quantum information, and the most exciting technologies for the world's transformation to the Quantum Information Age.  We proceed to outline three possible projects regarding these temporal-based quantum information technologies.

\subsection{Temporal cryptography}

In Chapter \ref{chap: QEnt}, entanglement in space was described through the tools of qubits, density operators and entropy.  The entropic analysis was particularly useful in the construction of entropic uncertainty relations with a spatially entangled memory (\ref{entropicuncertaintyrelationsentangled}).  These relations were ultimately related to the security of quantum cryptographic protocols (\ref{entropicuncertaintycrypto}).  

For the entanglement in time in Chapter \ref{chap: QEnt}, a description through entropy was not found in the literature.  Rather than providing an analysis using the Shannon or von Neumann entropy, a novel next step is to make use of certain temporal entropic quantities.  In classical information theory, two relatively recent quantities in the analysis of temporal data are the transfer entropy \cite{schreiber2000measuring, staniek2008symbolic, barnett2012transfer} and the past entropy \cite{di2002entropy, nanda2006some}.  Modifications of these quantities for the quantum case may allow one to better analyze the entanglement in time.  Furthermore, it may allow for the derivation of a unique set of entropic uncertainty relations with a temporally entangled memory.  Following the spatial case, this may lead to the development of temporal quantum cryptographic protocols, which secure information across time in classically impossible ways.

\subsection{Network consensus} 

In our quantum blockchain, the entanglement in time was harnessed for the data structure component.  However, we believe that one of the best applications of entanglement in time could be in the network consensus component. 

Prior to the inception of blockchain systems, network consensus was considered to be central topic within the subject of distributed algorithms \cite{lynch1996distributed}.  An important aspect to this subject is the timing model which captures the timing of events in a distributed computer network.  This can be synchronous (processors performing communication and computation in perfect lock-step synchrony), completely asynchronous (taking steps at arbitrary speeds and arbitrary order), or partially synchronous (where processors have partial information about the timing of events).  Given this temporal environment,  one research direction would be to harness an entanglement in time to develop quantum distributed consensus algorithms that can outperform the classical algorithms within each of scenarios of the timing model.

A far more interesting path would be when one considers the network consensus protocols that blockchain systems have recently introduced.  Advanced blockchain systems such as proof-of-elapsed-time systems \cite{chen2017security} and hedara hashgraph \cite{baird2016swirlds} have a significant temporal property to their design.  But far more important is that blockchain consensus overall have a probabilistic aspect that makes their system operational.  Given that an entanglement in time is a temporal phenomenon with (quantum) probabilistic properties, it may be the case that developing probabilistic consensus protocols is their `killer app.'  One can imagine that these protocols would allow a network to achieve consensus across time in classically impossible ways.

\subsection{Temporal logical machines}

One can think of a digital computer as a machine that carries out boolean logic.  Quantum computers or more precisely the quantum circuit model can be thought of as a quantum analogy of boolean operators.  As an example, one refers to the Pauli operator $\sigma_{x}$ as the quantum NOT operator.  There has been recent work to reformulate the quantum circuit model into other frameworks; one such example utilizes category theory \cite{abramsky2004categorical, heunen2013quantum} which provides various advantages.   

Using a similar line of reasoning, we speculate that an entanglement in time may not be best captured through the quantum circuit model which is founded on boolean logic.  There exists a well developed field known as temporal logic \cite{krger2008temporal, galton1990temporal, gabbay1994temporal} which involves various temporal logical operators.  An ambitious path would be to develop a model that can be considered a quantum analogue of temporal logic.  This may better capture the the effect of entanglement in time than the boolean logic inspired quantum circuit model.  We speculate that such a framework may lead to the derivation of radically new types of time machines i.e. temporal logical machines that may be as revolutionary as digital computing.

\section{What is Quantum Information?}

Entanglement in time is a most shocking temporal effect which is fundamentally mysterious.
Quantum information (ie quantum state), such as the complex numbers in (\ref{qubit}), allowed us to mathematically express this entanglement in time.  We viewed similar temporal effects from a foundational perspective in Chapter \ref{chap: QFound} as well as in the relativistic regime in Chapter \ref{chap: RQI}.  If quantum information represents a physical quantity, then an entanglement in time magnifies the disruption of quantum physics onto the classical temporal world far greater than any other effect. Furthermore in Quantum Foundations, holding an epistemic view of quantum information has drastic consequences, and in RQI the unification of quantum information with the Einstein equations is still the most important open problem in theoretical physics.  Therefore, from such a wide exploration we come to appreciate the most fundamental mystery:  What is quantum information?

In this thesis, we take the view that there are two separate problems, a theoretically inclined problem and a physically inclined problem.  The theoretical problem is what does quantum information or the quantum state physically represent?  The physical problem is what is the physical state of the quantum system when it is not observed?  (A refinement of the latter question is:  Where is the mass of a quantum particle located or distributed when we do not observe it?)  The connection between the quantum state and its corresponding unobserved system is not direct within the postulates of quantum theory (note that a quantum state is not a probability distribution; probability requires squaring it first).  This lack of direct connection provides the ambiguity that results in the inception of the two problems described.  We present three directions that relate to entanglement in time that could lead to an advancement towards answering these two problems.

\subsection{Null tetrads}

RQI is currently presented through the metric formulation of general relativity.  There is an equivalent picture of general relativity known as the null tetrad formulation \cite{penrose1984spinors, o2003introduction}.  It views the light cone structure as the fundamental entity, and elevates complex numbers as central quantities within relativity.   

More precisely, one can define a set of light-like or null vectors as a basis.  By tetrad, this implies a basis of four vectors.  Hence, at each point on the spacetime manifold, there are four null vectors ${l^{a}, n^{a}, m^{a}, \overline{m}^{a}}$ with specific properties that we shall describe.  The vectors $l^{a}$ and  $n^{a}$ are real and satisfy $l^{a}\:n_{a} = 1$.  The other two vectors, $m^{a}, \overline{m}^{a}$ are complex null vectors and have the property that they are complex conjugates of each other and satisfy the condition, $m^{a}\:\overline{m}_{a} = -1$.  The relationship of the null tetrad to the metric tensor (\ref{metrictensor}) can be expressed as
\begin{align}
g_{ab} &= l_{a}\:n_{b} + n_{a}\:l_{b} - m_{a}\:\overline{m}_{b} - \overline{m}_{a}\:m_{b}, \\
g^{ab} &= l^{a}\:n^{b} + n^{a}\:l^{b} - m^{a}\:\overline{m}^{b} - \overline{m}^{a}\:m^{b}.
\end{align}
The spacelike and timelike entanglements in RQI were based on regions separated by the light cone structure.  Perhaps by using the null tetrad formulation, one may be able to express these entanglements more efficiently.  This may provide some ideas on generalizing these entanglements in non-trivial curved spacetimes settings.  This may lead to the discovery of a novel set of entanglements.

However, a far more ambitious reason is that since quantum information is complex valued, its unification to relativity may require relativistic structures to be complex valued \cite{rajan2016complex}.  This complex valued unification towards quantum gravity was already outlined in \cite{penrose2007road}.  The novelty would be whether a null tetrad formulation of RQI could lead to a better insight into this program.  This may help us answer the question of what is quantum information from the perspective of complex valued relativistic structures.

\subsection{`Spacetime information theory'}

Entanglement in time involves an interplay of both quantum physics and time at a fundamental level.  In this thesis we have investigated quantum physics through an information-theoretic perspective.  Could time itself also be studied using an information-theoretic method? If so, would that help understand entanglement in time (and ultimately quantum information) in a deeper way?  Such a curiosity aligns to a program set by John Wheeler known as `It from Bit' \cite{wheeler1990information}.  

\textbf{a) It from Bit: } This research program puts forth the notion that the physical world emerges from information.  Wheeler hypothesized that every physical quantity derives its ultimate significance from bits, and we quote, \textit{``...every it -- every particle, every field of force, even the spacetime continuum itself -- derives its function, its meaning, its very existence entirely -- even if in some contexts indirectly -- from the apparatus-elicited answers to yes or no questions, binary choices, bits...all physical things are information-theoretic in origin"}.  He emphasized the goal to carry out such a program and we quote \textit{``Tomorrow we will have learned to understand and express all of physics in the language of information"}.  He proceeds to set an agenda by saying \textit{``...capitalize on the findings and outlooks of information theory...search out every link each has with physics..."}.  

We aim to contribute to Wheeler's program by outlining an \textit{original} \cite{GodTime} speculative path on how time (and space) could be viewed from an information-theoretic perspective.  Before embarking on these ideas, it wise to first identify the non-trivial concept that underpins all information theories.

\textbf{b) Compression: }  Our view is that all information theories are fundamentally about compression and not about information.  To support the previous statement, we examine several information theories and find that the concept of compression does indeed exist as the foundational result for each of them.

For the case of classical information theory, the fundamental result is the noiseless coding theorem (Theorem \ref{Shannoncoding}).  It highlights that the Shannon entropy $H(X)$ can be operationally defined in terms of optimal compression.  More precisely, for a sequence of $n$ two-outcome random variables, uncompressed $n$ bits can be optimally compressed to $H(X)n$ bits:
\begin{equation}\label{classicalcompressionconclusion}
H(X)n_{uncompressed} = n_{compressed}.
\end{equation}
Quantum information theory provides the quantum generalization of classical information theory.  The fundamental result in this theory is the Schumacher’s noiseless coding theorem (Theorem \ref{Schumachercoding}).  This articulates that the von Neumann entropy $S(\rho)$ can be operationally defined in terms of optimal compression.  More specifically, uncompressed $n$ qubits can be optimally compressed to $S(\rho)n$ qubits:
\begin{equation}\label{quantumcompressionconclusion}
S(\rho)n_{uncompressed} = n_{compressed}.
\end{equation}

Therefore we see that compression is the key idea in the fundamental results of these theories, and it serves to provide the operational definition of the entropies.

The field of algorithmic information theory \cite{grunwald2008algorithmic, downey2010algorithmic}, which we have not examined in this thesis, can also be seen to be fundamentally about compression.  This particular information theory is based on a notion of the information content of an individual object using concepts from computability theory.  This is in contrast to classical information theory which is built on probability distributions.  To elaborate, the information content of an object such as a finite binary string $x$ can be captured by the algorithmic entropy $K(x)$ (also known as the Kolmogorov complexity) which is defined as the length of the shortest binary computer program that can produce $x$ as the output.  In other words, the algorithmic entropy of a string provides the length of its shortest possible compression.  As an example, the string \begin{equation}
01010101010101010101010101010101010101010101010101010101010101 
\end{equation}
has the compressed program description ``31 repetitions of 01".  This is opposed to the string
\begin{equation}\label{string}
11001000011000011101111011101100111110100100001001010111100101
\end{equation}
which has no shorter description than writing down the string itself.

Using the algorithmic entropy, one can arrive at a definition of randomness of a single individual sequence which is not possible using probabilistic theories.  The fundamental idea is that randomness equates to incompressibility.  The inability to compress highlights that there is no structure of pattern to provide a concise description (as witnessed in string (\ref{string})).  Stated more formally, a finite string $x$ is random if and only if the algorithmic entropy of the string $K(x)$ is no shorter than the length of the string.

Therefore we see compression existing at a foundational level in all information theories.  However the ambiguity arises when one probes what the information in each of the information theories is supposed to represent; classical information is fit for engineering purposes with no regard to meaning \cite{shannon1948mathematical}; nobody knows what quantum information is \cite{leifer2014quantum}; a similar ambiguity exists on what algorithmic information is about \cite{grunwald2008algorithmic}.  Given each of the information theories has an associated entropy, perhaps that may add some clarification.  Such a hope is quickly diminished when one reads \cite{scales2012information} that Shannon called his term the entropy because von Neumann suggested one reason as ``\textit{...nobody knows what entropy really is, so in a debate you will always have the advantage}." 

Our view to resolve the confusion is remove the focus on information in information theories. Rather it is compression in these theories that is fundamental. Building on these observations, we state that a structure that is compressible is what should constitute information, and the quantity involved in optimal compression is what should be termed the entropy.

\textbf{c) Spacetime compression: } If one takes `It from Bit' as the underlying principle of the Universe, then there ideally must exist an information theory associated to time and space, just as there exists information theories for both classical and quantum systems.  More precisely, this research direction demands the development of a `spacetime information theory.'  From an alternative direction, the notion that spacetime itself contains information is already alluded to by the Bekenstein-Hawking entropy (\ref{hawkingentropy}).  In fact the original inception of that result was based on observations regarding similarities between spacetime physics and aspects of information \cite{bekenstein1972black}. 

Furthermore if one takes compression as the primary information-theoretic technique, then such a spacetime information theory should be built on some mathematical notion of compression.  To explore this idea, we provide a speculative path through a heuristic argument.  We emphasize that these ideas are underdeveloped but proceed with the intention of conveying possibilities.

Our immediate aim is to simply identify whether a notion of compression can be found in the theory of relativity, which is our current framework for understanding spacetime.  And predicated on that, develop various inferences that may lead to insights for developing a formal spacetime information theory on a firm basis.

We suggest the following idea: \textit{A time interval itself should be treated as a form of information, and that time dilation as expressed in (\ref{timedilation}) can then be seen as a form of information-theoretic decompression}.  

\textbf{d) Mathematical details: } To express this idea mathematically, we require a `data compression' entropy, analogous to the other information theories.  We identify this by suggesting that the reciprocal of the Lorentz factor (\ref{timedilationgamma}) is an entropy as it `compresses' time intervals.  Mathematically, this can be expressed as
\begin{equation}\label{postulate2}
\alpha = H(X), 
\end{equation}
where $\alpha \equiv 1/\gamma = \sqrt{1-(v^2/c^2)}$.  (This is analogous to $S(\rho) = H(\lambda_{x})$ in quantum information theory).  For an uncompressed time interval $\Delta t_{uncompressed}$, optimal compression is achieved using the time dilation formula 
\begin{equation}
\alpha \Delta t_{uncompressed} = \Delta t_{compressed}.
\end{equation}
Note that this is analogous to the compression formulas in the classical (\ref{classicalcompressionconclusion}) and quantum (\ref{quantumcompressionconclusion}) cases. The bounds for this $\alpha$-entropy are $0 \leq \alpha \leq 1$, with $\alpha =0$ when $v=c$.  The uncompressed interval can be thought of as the case when maximum entropy occurs. 

A particular physical realization of this is when a given coordinate time $\Delta t$ is contracted to various proper times $\Delta \tau$ depending on the velocity of different observers, ie $\Delta t = \gamma_{1}\Delta \tau_{1}$, $\Delta t = \gamma_{2}\Delta \tau_{2}$, etc.  Different velocities compress a fixed time interval differently, and hence we can identify this entropy with velocity.   

Using this assumption, we can state some similarities to the classical and quantum information theories.  In classical information theory, the mutual information, $H(X{:}Y)$, of $X$ and $Y$ measure how much information $X$ and $Y$ have in common; the quantum mutual information for systems $A$ and $B$ is denoted by $S(A{:}B)$.  This notion of commonality can be captured using the physical scenario of the relativistic velocity, $v_{r} = (v_{1}-v_{2})/(1-(v_{1}v_{2}/c^{2}))$, of two observers, $v_{1}$ and $v_{2}$ with respect to a fixed coordinate time interval.  Each observer can be identified with a compression entropy, $\alpha_{1} = 1/\gamma_{1}$ and $\alpha_{2} = 1/\gamma_{2}$.  Their relative Lorentz factor, $\gamma_{r}=\gamma_{1}\gamma_{2}(1-(v_{1}v_{2}/c^{2}))$, provides the inspiration to define the relativistic mutual information between $\alpha_{1}$ and $\alpha_{2}$:
\begin{equation}
\alpha_{1{:}2} \equiv \frac{1}{\gamma_{r}} = \frac{1}{\gamma_{1}\gamma_{2}(1-\frac{v_{1}v_{2}}{c^{2}})} = \frac{\alpha_{1}\alpha_{2}}{(1-\frac{v_{1}v_{2}}{c^{2}})}.
\end{equation}  
The classical joint entropy, $H(A,B)$, and the quantum joint entropy, $S(A,B)$, help us define the relativistic joint entropy:
\begin{eqnarray}
&H(X,Y) = H(X) + H(Y) - H(X{:}Y), 
\\
&S(A,B) = S(A) + S(B) - S(A{:}B),
\\
&\alpha_{1,2} \equiv \alpha_{1} + \alpha_{2} - \alpha_{1{:}2}.
\end{eqnarray}
Similarly, the relativistic conditional entropy can be developed using the classical and quantum analogue: 
\begin{eqnarray}
&H(X|Y) \equiv H(X,Y) - H(Y),
\\
&S(A|B) \equiv S(A,B) - S(B), 
\\
&\alpha_{1|2} \equiv \alpha_{1,2} - \alpha_{2}.
\end{eqnarray}
From here, we proceed to derive entropic properties concerning two systems and compare this with the classical and quantum case.  For example, it can easily be shown that $\alpha_{1{:}2} = \alpha_{2{:}1}$ and $\alpha_{1,2} = \alpha_{2,1}$ which is like the classical case of $H(X{:}Y) = H(Y{:}X)$ and $H(X,Y)=H(Y,X)$.  Furthermore, $H(Y|X)\geq0$ and $H(X)\leq H(X,Y)$ with equality satisfied for each inequality, if and only if $Y$ is a function of $X$; the last two properties fail for the quantum case, in particular for systems involving entanglement; for our relativistic case, $\alpha_{2|1}\geq0$ and $\alpha_{1}\leq \alpha_{1,2}$ is satisfied as long as $\gamma_{r}\geq \gamma_{2}$, ie physically the relative Lorentz factor has to be greater than or equal to the Lorentz factor of the second observer. The property of subadditivity holds for all three cases: $H(X,Y)\leq H(X)+H(Y)$, $S(A,B)\leq S(A)+S(B)$, and $\alpha_{1,2}\leq \alpha_{1} + \alpha_{2}$ (which can be reduced to $\alpha_{1{:}2}\geq 0$).  Classically, $H(Y|X)\leq H(Y)$ and it can be shown that $\alpha_{1|2}\leq \alpha_{1}$ (which can be reduced to $\alpha_{1{:}2}\geq 0$).  In the last classical equation as well as second to last equation of classical subaddivity, equality is expressed if $X$ and $Y$ are independent variables.  In the last two analogous relativistic equations, equality is expressed if $\alpha_{1{:}2}=0$, which is when one of the observers has velocity $c$.  Thus, the notion of `independence' enters when one of the observers is moving at the speed of light.

To draw out physical implications, we assume this information-theoretic compression has the mathematical backbone of classical and quantum information theory.  In classical information theory, a central result was the construction of $\epsilon$-typical sequences (\ref{typical}) with Theorem \ref{typicaltheorem}.  A similar result held in the quantum information theory with $\epsilon$-typical states with Theorem \ref{typicalsubspacetheorem}.  In this `spacetime information' case, one would ideally need to define an time interval $\Delta t$ that is $\epsilon$-typical.  Hence in an analogous manner, time intervals come in typical and atypical forms.  If this relativistic case is similar to the classical and quantum case, then one can easily see that typical time intervals would be those that can be compressed from $\Delta t$ to $\alpha \Delta t$.  Physically, this means typical time intervals obey the relativistic time dilation formula.  Therefore, atypical time intervals are those that exhibit Lorentz violations.  Furthermore, if a similar information theory result occurs in this case, then the probability that a time interval is $\epsilon$-typical is
\begin{equation}
\text{Pr}\{T_{\epsilon}^{(\Delta t)}\} > 1 - \epsilon
\end{equation}
for sufficiently large duration $\Delta t$.  Hence in this setting, Lorentz violations do occur for large time intervals but very rarely.  In this sense, Lorentz violations is fundamentally not a matter of scale, but rather of probability; hence this violation may be experimentally detectable at large scales given a very large sample size.

It's important to note that this idea also applies to space intervals.  For the $x$-direction, $\alpha \Delta x_{uncompressed} = \Delta x_{compressed}$, where the compressed space interval is due to length contraction.  In the data compression subsection, it can be seen that optimal compression is achieved using $n H(X)$ bits (or $n S(\rho)$ qubits). For the case $n R > n H(X)$, compression without loss of information is achieved, but is not optimal.  For $nR < nH(X)$, information is lost and compression is not reliable.  These results seem to correspond to Lorentz transformations, which can be re-written as $\alpha \Delta t_{uncompressed} = \Delta t \pm \frac{v}{c^{2}}\Delta x$ where $\alpha \Delta t_{uncompressed} = \Delta t_{compressed}$ is the optimally compressed interval.  The Lorentz equation for the the minus case, leads to $\Delta t > \Delta t_{compressed}$; this can be interpreted as compression is achieved but not optimally.  For the plus case, $\Delta t < \Delta t_{compressed}$, which means some information is lost and gone to $\Delta x$. 

\textbf{e) Comments: } 
\begin{enumerate}[noitemsep, topsep=0pt, label=\roman*)]
	\item These set of ideas are merely speculations at this point and were inspired by Wheeler's program of `It from Bit.' In our thesis it is precisely a time interval that was central to the shocking nature of the entanglements.  Whether a time interval could be treated as information in the way we stated is an exciting premise but requires much greater exploration.
	\item There are various problems with this heuristic argument.  The first is regarding the problem on how this idea would work given the relative nature of Lorentz transformations between observers.  A related problem is that information in standard compression is carried on less bits while boosting a particle contracts or expands its time completely relative to the reference frame the particle is observed in.  With respect to these and other issues, it is important to emphasize that this is merely a heuristic argument for providing insights to develop a formal spacetime information theory.  The undertaking of this latter subject would of course require far more sophisticated mathematical constructs along with appropriate postulates.
	 \item However within such a formal spacetime information theory, we do believe compression will be found in its fundamental equations and help us view the Universe as Wheeler intended. Would the expansion of the Universe turn out to be an information-theoretic decompression? Would the Bekenstein-Hawking entropy correspond to an optimal compression? Through a unification with quantum information theory, would it help us achieve quantum gravity?  And would that let us finally answer the question what is quantum information (and its related intrinsic randomness)?
\end{enumerate}

\subsection{Unordered time}

Our current view of time is that it obeys a mathematical real number axis with an ordering from smaller values to larger values.  The final speculation we state is that time may be physically unordered in the quantum realm compared to the classical realm.  We believe that a development of this idea may provide a new interpretation of quantum physics, in particular to answering the physical and theoretical inclined problems that we stated earlier.  Perhaps the shocking properties are nothing intrinsic with the quantum particles but rather a consequence of an unordering of time.  When one observes (from the ordered world) a particle at various times, it gives the impression that the particle is behaving in a paradoxical manner. A measurement can be defined as the moment at which the temporal ordered classical world meets the temporal unordered quantum world.  From a theoretical view, the quantum state (ie quantum information) is a complex valued quantity and complex numbers are mathematically unordered; perhaps the quantum state through its unordered complex numbers is about spacetime and the Schr\"{o}dinger equation is an evolution of spacetime in a unordered temporal manner; the dynamics of this unordering of time may stem from the energy associated with the quantum system in question.  Providing a generalization of this bare idea with a mathematical framework is left for future work.  

We do want to emphasize certain works and ideas that served as stimulus for such a speculation.  First were the results on quantum causality \cite{brukner2014quantum, oreshkov2012quantum}, and in particular a Bell's theorem for temporal order \cite{zych2017bell}.  Relating to the unordered nature of complex numbers, the second influence comes from the use of imaginary time through the Wick rotation, as well as through the use of complex valued spacetime transformations in the Newman-Janis trick \cite{rajan2016complex}; there are no deep reasons at present for why these tricks work.  The third influence comes from the `spacetime information theory' ideas we set out earlier; this suggests that time can be compressed in an information-theoretic manner; given that compression techniques can involve data reordering or a reduction of data this also leads to a notion of unordering in time. The final influence is that the metric fluctuations on the sub-Planckian scale is completely unknown \cite{visser1996lorentzian}; a non-trivial metric fluctuation may provide the necessary basis for this unordered time.  

In closing, we want to provide some clarity on the historical aspect of the subject.  It is often emphasized that Einstein was critical of quantum theory.  But of far greater importance is that it should be stated that he was one of its pivotal founders  \cite{isaacson2011einstein, stone2015einstein}.  In fact it was the discovery of entanglement (in space) in the EPR paper \cite{einstein1935can} that is the most cited of all his works!  In align with the theme of this thesis, it should be noted that he also emphasized a time interval in his quest to truly understand quantum physics: \textit{``All the fifty years of conscious brooding have brought me no closer to answer the question, ‘What are light quanta?’ Of course today every rascal thinks he knows the answer, but he is deluding himself.''}

\bibliographystyle{unsrt}

\bibliography{bibfile} 
\end{document}